\documentclass{ZurichLect}\pdfoutput=1

\usepackage[english]{babel}
\usepackage{esint}
\usepackage[applemac]{inputenc}
\usepackage{xcolor,graphics,graphicx}
\usepackage{cite}
\usepackage{amsmath,amssymb}
\usepackage[cyr]{aeguill}

\makeatletter

\usepackage{hyperref}

\newcommand{\benu}{\begin{enumerate}}
\newcommand{\eenu}{\end{enumerate}}
\newcommand{\blist}{\begin{itemize}}
\newcommand{\elist}{\end{itemize}}

\DeclareMathOperator{\supp}{Supp}
\DeclareMathOperator{\dist}{dist}

\def \bm{\begin{displaymath}}
\def \em{\end{displaymath}}
\def \be{\begin{equation}}
\def \ee{\end{equation}}
\def \beq*{\begin{equation*}}
\def \eeq*{\end{equation*}}
\def \ba{\begin{eqnarray}}
\def \ea{\end{eqnarray}}
\def \ba*{\begin{eqnarray*}}
\def \ea*{\end{eqnarray*}}

\newcommand{\Q}{\mathbb{Q}}
\newcommand{\C}{\mathbb{C}}
\newcommand{\R}{\mathbb{R}}
\renewcommand{\P}{\mathbb{P}}
\newcommand{\T}{\mathbb{T}}

\newcommand{\Z}{\mathbb{Z}}

\def \mc{\mathcal}

\def\({\left(}
\def\){\right)}
\def\1{\mathbf{1}}

\def\bo{\partial \Omega}
\def\D{\displaystyle}

\def\curl{{\rm curl\,}}

\def\div{\mathrm{div} \ }
\def\D{\displaystyle}

\def\eps{\varepsilon}
\def\g{g}
\def\hne{h_{n,\eta}}
\def\hal{\frac{1}{2}}
\def\hciii{{H_{c_3}}}
\def\hcii{{H_{c_2}}}
\def\hci{{H_{c_1}}}

\def\he{{h_{\rm ex}}}

\def\indic{\mathbf{1}}

\def\io{\int_{\Omega}}
\def\j{E}

\def\lep{{|\mathrm{log }\ \eps|}}

\def\loc{{\text{\rm loc}}}

\def\l|{\left|}
\def\muv{\mu_0}

\def\mn{\mathbb{N}}
\def\mn{\mathbb{N}}
\def\mr{\mathbb{R}}
\def\mz{\mathbb{Z}}
\def\nab{\nabla}
\def\om{\Omega}

\def\pa{{\partial}}
\def\Q{{\mathbb{P}_n^\beta}}

\def\ro{{\rho}}

\def\supp{\text{Supp}}

\def\vp{\varphi}

\def\W{\mathbb{W}}

\def\Z{{Z_n^\beta}}
\def\fae{f_{\alpha,\eta}}

\def\Xint#1{\mathchoice
   {\XXint\displaystyle\textstyle{#1}}%
   {\XXint\textstyle\scriptstyle{#1}}%
   {\XXint\scriptstyle\scriptscriptstyle{#1}}%
   {\XXint\scriptscriptstyle\scriptscriptstyle{#1}}%
   \!\int}
\def\XXint#1#2#3{{\setbox0=\hbox{$#1{#2#3}{\int}$}
     \vcenter{\hbox{$#2#3$}}\kern-.5\wd0}}
\def\dashint{\Xint-}

\def\fint{\dashint}

\newtheorem{theo}{Theorem}[chapter]
\newtheorem{prop}{Proposition}[chapter]
\newtheorem{coroll}[prop]{Corollary}
\newtheorem{lemme}[prop]{Lemma}
\theoremstyle{definition}
\newtheorem{example}[prop]{Example}
\newtheorem{rem}[prop]{Remark}
\newtheorem{defini}[prop]{Definition}

\numberwithin{figure}{chapter}

\newtheorem{claim}{Claim}
\newtheorem{conjecture}{Conjecture}[chapter]

\newcommand{\f}{\frac}

\def \lg{\langle}
\def \rg{\rangle}

\newcommand{\interieur}{\overset{\circ}}
\renewcommand{\div}{\mathrm{div}\,}
\renewcommand{\eps}{\varepsilon}

\author{Sylvia Serfaty}
\title{Coulomb Gases and Ginzburg-Landau Vortices}
\date{November 2014}
\begin{document}

\begin{titlepage}
\maketitle
\end{titlepage}

{\bf Acknowledgments}\\
These are the lecture notes of the ``Nachdiplomvorlesung" course that I taught in the Spring of 2013 at the ETH Z\"urich, at the invitation of the Forschungsinstitut f\"ur Mathematik. I would  like to express my deep gratitude to Tristan Rivi\`ere, director of the FIM, for this invitation, and to him,  Thomas Kappeler, Michael Struwe and Francesca da Lio, for making my stay particularly enjoyable and stimulating.

 I am grateful to those who followed the course for their constructive questions and comments. 
A particular acknowledgment goes to Thomas Lebl\'e for taking notes and typing the first draft of the manuscript, as well as for making the figures.

The material covered here is based primarily on works with Etienne Sandier, and with Nicolas Rougerie, I am very grateful to them for the fruitful collaborations we have had. 
 Finally, this manuscript also benefited from comments by Ofer Zeitouni to whom I extend thanks here.

\tableofcontents
\chapter{Introduction}
\label{chap-intro}

These lecture notes are devoted to the mathematical study of two physical phenomena that have close mathematical connections: vortices in  the Ginzburg-Landau model of superconductivity on the one hand, and classical Coulomb gases on the other hand. A large part of the results we shall present originates in joint work with Etienne Sandier (for Ginzburg-Landau and two dimensional Coulomb gases) and in joint work with Nicolas Rougerie (for higher dimensional Coulomb gases), recently revisited in work with Mircea Petrache. In order to simplify the presentation, we have chosen to present the material in reverse chronological order, starting with the more recent  results on Coulomb gases which are simpler to present, and finishing with the more complex study of vortices in the Ginzburg-Landau model. But first, in this introductory chapter,  we start by briefly presenting the two topics  and  the connection between them.

\section{From the Ginzburg-Landau model to the 2D Coulomb gas}
\subsection{Superconductivity and the Ginzburg-Landau model}
The Ginzburg-Landau model is a very famous physics  model for superconductivity. Superconductors are certain metallic alloys,  which, when cooled down below a very low {\it critical  temperature},  lose their resistivity and let current flow without loss of energy. This is explained at the microscopic level  by  the creation of superconducting electron pairs called Cooper pairs (Bardeen-Cooper-Schrieffer or BCS theory), and superconductivity is a macroscopic manifestation of this quantum phenomenon. The Ginzburg-Landau theory, introduced on phenomenological grounds by Ginzburg and Landau in the 1950's, some forty years after superconductivity had first been discovered by Kammerling Ohnes in 1911,  has proven amazingly effective in describing the experimental results and predicting the behavior of superconductors. It is only very recently that the Ginzburg-Landau theory \cite{gl} has been rigorously (mathematically) derived from the microscopic theory of Bardeen-Cooper-Schrieffer \cite{bcs}, also dating from the 50's,  by Frank, Hainzl, Seiringer and Solovej \cite{fhss}.

These superconducting alloys exhibit a particular behavior in the presence of a magnetic field : the superconductor levitates above the magnet. This is explained by the {\it Meissner effect} : the superconductor expells the magnetic field. This only happens when the external field $h_{ex}$ is not too large. There are three critical fields $H_{c_1}, H_{c_2}, H_{c_3}$ for which {\it phase transitions} occur. Below the {\it first critical field} $\hci$, the material is everywhere superconducting.  At $H_{c_1}$, one first observes local defects of superconductivity, called {\it vortices}, around which a superconducting loop of current circulates.  As $h_{ex}$ increases,  so does the number of vortices, so that they become densely packed in the sample. The vortices repel each other, while the magnetic field confines them inside the sample, and the result of the competition between  these two effects is that they arrange themselves in a particular {\it triangular lattice} pattern. It was predicted by Abrikosov \cite{a}, and later observed experimentally, that there should be periodic arrays of vortices appearing in superconductors, and this was later observed experimentally (Abrikosov and Ginzburg earned the 2003 Nobel Prize  for their discoveries on superconductivity), cf. Fig. 
\ref{fig32} (for more pictures, see {\tt www.fys.uio.no/super/vortex/}).
\begin{figure}[ht!]
\begin{center}
\includegraphics[width=5cm]{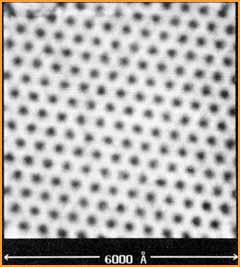}
\caption{Abrikosov lattice, H. F. Hess et al. Bell Labs
{\it Phys. Rev. Lett.} 62, 214 (1989)}
\label{fig32}

\end{center}

\end{figure}

These triangular lattices (originally Abrikosov predicted a square lattice but he had made a small mistake) then became called {\it Abrikosov lattices}. A  part of our study, detailed in this course,  is aimed towards understanding why this particular lattice appears.

The second and third critical fields correspond respectively to the  loss of superconductivity in the sample bulk and to the complete loss of superconductivity. These two transitions are not the focus of our study,  and for more mathematical details on them, we refer to the monograph by Fournais and Helffer \cite{fh}.
For a physics presentation of superconductivity and the Ginzburg-Landau model  we refer to the standard texts \cite{sst,dg,t}, for a mathematical presentation one can see \cite{livre,fh} and references therein. 

In non-dimensionalized form and in a simply connected domain $\om$ of the plane, the model proposed by Ginzburg-Landau can be written as the functional 
\be \label{GL}
G_{\eps}(u, A)  = \f{1}{2} \int_{\Omega} |(\nabla - iA)u|^2 + |\curl  A - \he|^2 + \f{(1 - |u|^2)^2}{2\eps^2}.
\ee
This may correspond to the idealized situation of  an infinite vertical cylindrical sample of cross-section $\Omega$  and a vertical external field of intensity $\he$, or to a thin film.

Here 
\begin{itemize}
\item $u : \Omega \rightarrow \C$, usually denoted by $\psi$ in the physics literature, is called the {\it order parameter}. Its modulus (the density of Cooper pairs of superconducting electrons in the BCS theory) indicates the local state of the material: where $|u|\approx 1$, the material is in the superconducting phase, where $|u|\approx 0$ in the normal phase.
 The vortices correspond to isolated zeroes of $u$, and since $u$ is complex-valued each zero carries an integer topological degree, like a “topological charge.”
\item $A : \Omega \rightarrow \R^2$ is the vector potential of the magnetic field $h = \curl  A$ (defined by $\curl A:= \partial_1 A_2-\partial_2 A_1$), which is thus a real-valued function. 
\item The parameter $\he>0$ is the intensity of the applied (or external) magnetic field.
\item The parameter $\eps>0$ is a material constant, corresponding to the ratio between characteristic lengthscales of the material (the {\it coherence length} over the {\it penetration depth}).
 We will be interested in the asymptotic regime $\eps \to 0$.
 The functional is generally expressed in the physics literature in terms of the inverse of the constant $\eps$, denoted $\kappa$, and called the {\it Ginzburg-Landau parameter.}  Materials with high-$\kappa$ (the case we are interested in) are sometimes called ``extreme type-II superconductors," and the limit $\kappa\to \infty$ is often called the {\it London limit.}
\end{itemize}
When  considering the problem of minimizing the functional $G_{\eps}$, a heuristic examination leads to observing that~: 
\begin{itemize}
\item The term $(1- |u|^2)^2$ favors $u$ close to $1$, hence $u$ should not vanish too often, especially as $\eps \to 0$. A dimensional analysis in fact shows that the regions where $|u|$ is small have lengthscale $\eps$.
\item The quantity $|\curl A - \he|^2$ is smaller when $\curl A = h \approx \he$, that is, when the magnetic field penetrates the material  so that the induced magnetic field equals the external  magnetic field.
\end{itemize} 

Minimizers and critical points of the Ginzburg-Landau functional without boundary constraints solve the associated set of Euler-Lagrange  equations, called the  Ginzburg-Landau equations :
\begin{equation*}
\mathrm{(GL)}\left\lbrace 
\begin{array}{cc}
- \nabla^2_Au = \f{1}{\eps^2} u(1- |u|^2) & \text{in} \ \om
 \\ [2mm]
- \nabla^{\perp}h = \langle i u, \nabla_A u \rangle & \ \text{in} \ \om
\end{array} \right.
\end{equation*} where $\nabla_A := \nabla - iA$, $\langle \cdot, \cdot \rangle$ denotes the scalar product in $\C$ as identified with $\mr^2$, 
$\nabla^{\perp} = (- \partial_2, \partial_1)$ and again $ h= \curl A$;
with natural  boundary conditions
\begin{equation*}
\left\lbrace 
\begin{array}{cc}
\nab_A u \cdot \nu=0 & \text{on} \ \bo\\
h=\he& \text{on} \ \bo.\end{array}\right.
\end{equation*}

\subsection{Reduction to a Coulomb interaction}
More details on the analysis of the  Ginzburg-Landau  model will be given in Chapter \ref{glheuris}, which will be devoted to it, but for now let us try to explain the Coulombic flavor of the phenomenon.

In the regime with vortices (for $\hci \le \he \ll \hcii$), formal computations that will be better detailed  in Chapter \ref{glheuris}  show that in the asymptotic regime $\eps \to 0$, 
the functional $G_{\eps}(u, A)$ behaves as if it were~: 
\be\label{apprx}
G_{\eps} (u, A)\approx \f{1}{2} \int_{\Omega} |\nabla h|^2 + |h - h_{ex}|^2, \text{ where } h = \curl A
\ee
with what is known in the physics literature as the {\it London equation}~: 
\begin{equation}\label{london}
\left\lbrace \begin{array}{ll}
- \Delta h + h \approx 2\pi \sum d_i \delta^{(\eps)}_{a_i} & \text{ in } \Omega \\
h = \he & \text{ on } \partial \Omega,
\end{array}  \right.
\end{equation}
where the $a_i$'s are the centers of the vortices of $u$ and the coefficients $d_i\in \mz$ their (topological) degrees.  One should think of  $\delta_{a_i}^{(\eps)}$ as  being formally a Dirac mass at $a_i$, smoothed out at the scale $\eps$, or some approximation of it.
A large part of our analysis in \cite{livre,ss1} is devoted to giving rigorous statements and proofs of these heuristics.

Inserting the  London equation \eqref{london} into the approximation \eqref{apprx} leads to the following electrostatic analogy:
\be\label{apprx2} G_{\eps} (u, A)   \approx \f{1}{2} \int_{\Omega\times \Omega} G_{\Omega}(x,y) \Big(2\pi \sum_i d_i \delta^{(\eps)}_{a_i} - \he \Big)(x)dx \Big(2\pi \sum_i d_i \delta^{(\eps)}_{a_i} - \he \Big)(y)dy\ee 
where $G_{\Omega}$ is a Green kernel (or more accurately, Yukawa or screened Green kernel), solution to 
\be
\left\lbrace \begin{array}{ll} -\Delta G_{\Omega} + G_{\Omega} = \delta_{y} & \mbox{ in } \Omega \\ G_{\Omega} = 0 & \mbox{ on } \partial \Omega. \end{array} \right. 
\ee
This kernel is  logarithmic to leading order~: we may write\be
G_{\Omega} (x,y)=- \f{1}{2\pi} \log |x-y| + R_{\Omega}(x,y)
\ee
where $R_{\Omega}$ is a  nonsingular function of $(x, y)$. Approximating $G_\Omega$ by $- \frac{1}{2\pi}\log$ gives that the leading terms in \eqref{apprx2} are 
\begin{equation}\label{sumcontr}
G_{\eps}(u,A) \approx - \pi \sum_{i,j} d_i d_j \log |a_i - a_j| 
\end{equation}
which is a sum of pairwise logarithmic or Coulombic interactions, weighted by the degrees  $d_i$. Two such topological charges  repel each other when they have the same sign, and attract each other if they have different signs. Rigorously, this is of course wrong, because we have replaced the smoothed out Diracs by true Dirac masses, leading to infinite contributions  when $i=j$ in \eqref{sumcontr}. One needs to analyze more  carefully the effects of the smearing out, and to remove the infinite self-interaction of each  ``charge" at $a_i$  in \eqref{sumcontr}. One also needs to retain the interaction of these charges with the ``background charge" $-\he\, dx$ appearing in \eqref{apprx2}.  This is what leads to the analogy with the Coulomb gas that we will define and describe just below. 

When looking for a  model that retains these features~: Coulombic interactions of points, combined with the confinement by a background charge, the simplest is to consider a discrete model with all charges equal to $1$, and  consider the Hamiltonian  of a Coulomb gas with confining potential in dimension 2~:
\begin{equation} \label{coulombgashamiltonian}
H_n(x_1, \dots, x_n) = -  \sum_{i \neq j} \log |x_i - x_j| + n \sum_{i=1}^n V(x_i),
\end{equation}
where  $x_i\in \mr^2$, $V$ is the confining potential  (smooth, growing faster than $\log|x|$ at infinity), and the number of points $n$ tends to infinity.

It turns out that this much simpler model (compared to $G_\eps$) does retain many of  the essential features of the vortex interaction, and is also of independent interest  for physics and mathematics, as we will see. The study of \eqref{coulombgashamiltonian} and its higher-dimensional analogues will occupy the largest part of these notes. We will then see how to use the perspective and knowledge gained on this to analyze the Ginzburg-Landau model (again, this is the reverse of the literature chronology, since we first studied the Ginzburg-Landau model and then adapted our analysis to the Coulomb gas situation!).

\section{The classical Coulomb gas}
\subsection{The general setting}
The Hamiltonian given by (\ref{coulombgashamiltonian}) corresponds to the energy of a gas of charged particles in $\R^2$ interacting via the Coulomb kernel in two dimension. To be more precise, $-\log |x-y|$ is a multiple of the Coulomb kernel (or the fundamental solution of the Laplacian in the plane) in dimension $2$. The counterpart in higher dimension corresponds to the $d$-dimensional Coulomb kernel, which is a multiple of $|x|^{2-d}$ for $d\ge 3$.
The Hamiltonian of a classical Coulomb gas in any dimension $d \ge 2$ is thus given by 
\begin{equation}\label{hamiltoniencoulomb}
H_n(x_1, \dots, x_n) = \sum_{i\neq j} g(x_i-x_j) + n \sum_{i=1}^n V(x_i)\end{equation}
where 
\be \label{coulombkernel0}
g(x) = \left\lbrace \begin{array}{cc} - \log |x| & \mbox{for } d =  2 \\ \f{1}{|x|^{d-2}} & \mbox{for } d \geq 3 .\end{array} \right. 
\ee

 The statistical mechanics of a Coulomb gas, also called in physics a {\it two-dimensional one-component plasma},  is described by the corresponding Gibbs measure~: 
\be\label{110b}
d\P_{n, \beta}(x_1, \dots, x_n) := \f{1}{Z_{n, \beta}} e^{-\f{\beta}{2} H_n(x_1, \dots, x_n)} dx_1 \dots dx_n
\ee
where $\beta > 0$ is the {\it inverse temperature} and $Z_{n,\beta}$ is a normalization constant, the {\it partition function}, defined by
\begin{displaymath}
Z_{n, \beta} = \int_{(\R^d)^n} e^{-\f{\beta}{2} H_n(x_1, \dots, x_n)} dx_1 \dots dx_n.
\end{displaymath}
The probability measure $\P_{n,\beta}$ gives the probability of finding the particles at $(x_1, \dots, x_n)$ at (inverse) temperature $\beta$. The object of statistical mechanics is then to analyze possible transitions in the types of states that can be effectively observed (i.e. those that have probability $1$ or almost $1$), according to the value of the inverse temperature $\beta$ (e.g. transitions from ordered to disordered states at critical temperatures, such as liquid to solid phases etc).
    For general reference, we refer to standard statistical mechanics textbooks such as \cite{huang}, and with increasing order of specificity to the books \cite{hansen,forrester}.

This model is one of the most basic statistical mechanics models not confined to a lattice, and it is  considered difficult because of the long-range nature of the electrostatic interaction. Moreover, it can play the role of a toy model of the structure of matter, even if it is a purely classical - and not quantum - model.  Studies in this direction include \cite{sm,LO,aj,jlm,PS}.

The macroscopic distribution of the points as their number $n$ goes to infinity is well understood and  relatively simple to derive,  this will be the object of  Chapter~\ref{leadingorder}. On the other hand, their microscopic distribution, more precisely the one seen  at the scale $n^{-1/d}$, is less understood, and will be the main object of these lectures.

Let us now see some more specific motivations for studying the classical Coulomb gas, many of them being specific to dimension $d=2$.
\subsection{Two-dimensional Coulomb gas}

This is the setting that is the closest to the Ginzburg-Landau setting, as we discussed above. In this setting,  the microscopic distribution of the points in the plane is expected in crystallize (most likely to the Abrikosov lattice triangular pattern) at low temperature. In fact there is some controversy in the physics literature as to whether there is a finite temperature phase transition for this crystallization which is numerically observed, cf. e.g.  \cite{bst,stishov,aj}.

\paragraph{Vortices in superfluids and superconductors}
A first motivation for studying the two-dimensional Coulomb gas is the analysis of vortices in the Ginzburg-Landau model of superconductivity, but also more generally of vortex systems in classical fluids  \cite{CLMP}, in quantum fluids such as in superfluids or Bose-Einstein condensates \cite{nicolas}, and in fractional quantum Hall physics  \cite{Gir,RSY1,RSY2}. All these systems share a lot of mathematics in common, and it is also of interest to understand their statistical physics (critical temperatures and phase transitions).

\paragraph{Fekete sets}
This motivation no longer  comes from physics but rather from a very different area of mathematics: interpolation theory. 
Fekete points are defined to be points that maximize the quantity
\be\label{prodxi}
\prod_{1 \leq i < j \leq n} |x_i - x_j|,
\ee among all families of $n$ points defined on a certain subset of $\R^d$, or a manifold  (or any metric space, replacing modulus by the distances). The Fekete points have the property of  minimizing the error when  interpolating a function by its value at points, see \cite{safftotik} for reference, or \cite{sk} for more details on the motivation, and \cite{brauchartgrabner,grabner} for  surveys of recent results on the sphere. 
A whole literature is also devoted to understanding Fekete points on complex manifolds, possibly in higher dimension, see e.g. \cite{berman,bbn,lev} and references therein.

Of course, in the setting of the Euclidean space, maximizing \eqref{prodxi} is equivalent to minimizing the 
logarithmic interaction
\begin{displaymath}
- \sum_{1 \leq i \neq  j \leq n} \log |x_i - x_j|
\end{displaymath}
which takes us back  to the setting of the two-dimensional Coulomb gas.    Indeed,  Fekete sets confined to a set $K\subset \R^d$ correspond to  minimizers of  $H_n$ with $V$ taken to be $0$ in $K$ and $+\infty$ in $K^c$.
 Minimizers of   \eqref{coulombgashamiltonian} for general $V$'s 
 are in fact  called {\it weighted Fekete sets}, also defined as maximizers  of   
\begin{displaymath}
\prod_{1 \leq i < j \leq n} |x_i - x_j| \prod_{i =1}^n e^{-\f{n}{2} V(x_i)}
\end{displaymath}
where $V$ is the weight. For definitions and the connection to logarithmic potential theory, see again \cite{safftotik} and references therein.
Weighted Fekete sets  are also naturally related to the theory of weighted orthogonal polynomials (cf. the surveys  \cite{simon,koe} or again \cite{safftotik}).

The correspondence can also be made via a mapping, e.g. 
the important  question of finding the Fekete points on the ($2$-)sphere is equivalent, by stereographical projection, to studying the weighted Fekete sets on $\R^2$ with weight $V(x) = \f{1}{2} \log (1 + |x|^2)$, for details see \cite{hardy,betermin}.

\paragraph{Random matrix theory}
Random matrix theory (RMT) is a relatively old theory, pionereed by statisticians and physicists such  as Wishart, Wigner and Dyson, and originally motivated by the understanding of the spectrum of heavy atoms, see \cite{mehta}. For more recent mathematical reference see \cite{agz,deift,forrester}.  An important  model of random matrices is the so-called {\it Ginibre ensemble} \cite{ginibre}~: the law is that of an $n \times n$ complex matrix whose coefficients are i.i.d. complex normal random variables. The main question asked by RMT is~: what is the law of the spectrum of a large random matrix ? In the case of the Ginibre ensemble, the law is known exactly~: upon rescaling the (complex) eigenvalues $x_1, \dots, x_n$ by a factor $\f{1}{\sqrt{n}}$, it  is given by the following density
\begin{equation} \label{loiGinibre}
d\P_n(x_1, \dots, x_n) = \f{1}{Z_n} e^{-H_n(x_1, \dots, x_n)} dx_1 \dots dx_n
\end{equation}
with
\begin{equation} \label{loiGinibre2}
H_n(x_1, \dots, x_n) = - \sum_{i \neq j} \log |x_i - x_j| + n \sum_{i=1}^n |x_i|^2
\end{equation}
and $Z_n$ a normalization constant.
We recognize in \eqref{loiGinibre2}  the 2D Coulomb gas Hamiltonian with potential $V(x) = |x|^2$, and the law $d\P_n$ is the Gibbs measure \eqref{110b} at inverse temperature $\beta = 2$.
This analogy 
 between random matrices and the statistical mechanics of Coulomb gases was first noticed by Wigner \cite{wigner} and Dyson \cite{dyson}, see \cite{forrester} for more on this link.
 Writing the law in the form \eqref{loiGinibre} immediately displays the phenomenon of {\it repulsion of eigenvalues}: eigenvalues in the complex plane interact like Coulomb particles, i.e. they do not ``like" to be too close and repel each other logarithmically.

 At this  specific temperature $\beta=2$, the law of the spectrum acquires a  special algebraic feature : it becomes a determinantal process, part of a wider class of processes (see \cite{bkpv,borodin}) for which the correlation functions are explicitly given by certain  determinants. This allows for many explicit  algebraic computations. However,  many relevant quantities that can be computed explicitly for $\beta = 2$ are not exactly known for the $\beta \neq 2$ case, even in the case of  the potential $V(x) = |x|^2$. In this course, in contrast, we will work for any $\beta$, and with a wide class of potentials.



\subsection{The one-dimensional Coulomb gas and the log gas}
We have not mentioned yet   the one-dimensional Coulomb gas, which corresponds to \eqref{hamiltoniencoulomb} with the Coulomb kernel (up to a constant) $g(x)=|x|$. The reason is that we will not be interested in it,  because it has already been well-understood \cite{len1,len2,kunz,bl,am}. It can be ``solved" almost explicitly  and  crystallization at zero temperature is established.

We are interested however in another one-dimensional model (i.e. with points $x_i\in \mr$), where the two-dimensional logarithmic interaction $g(x)=-\log |x|$ is used in \eqref{hamiltoniencoulomb}. This is usually called a {\it log gas}, and its motivation also comes from Random Matrix Theory (see \cite{forrester}):
one-dimensional counterparts to the Ginibre ensemble are the Gaussian Unitary Ensemble (GUE) and the Gaussian Orthogonal Ensemble (GOE), which are symmetric analogues of it. The law of the GUE (resp. the GOE) is that of a $n\times n$ matrix whose coefficients are complex (resp. real) normal random variables, independent up to a Hermitian (resp. symmetry) condition.  Because of the Hermitian or symmetric nature of the matrix, its eigenvalues lie on the real line (hence the one dimensionality of the  model), but they still repel each other logarithmically: again, the law of the spectrum (the distribution of eigenvalues) can be given explicitly by the following density on $\R$
\begin{equation}\label{loigue}
d\mathbb{P}_n(x_1, \dots, x_n) = \f{1}{Z_n} e^{-\f{\beta}{2} H_n(x_1, \dots, x_n)} dx_1 \dots dx_n,
\end{equation}
where $H_n$ is still defined as \begin{equation}
H_n(x_1, \dots, x_n) = - \sum_{i \neq j} \log |x_i - x_j| + n \sum_{i=1}^n |x_i|^2, \quad x_i \in \R
\end{equation}
with  $\beta = 1$ for the GOE and $\beta =2$ for GUE. This is thus a particular case of a log gas, 
at specific temperature $\beta=1$ or $2$ and with quadratic potential, and the phenomenon of {\it repulsion of eigenlevels} is  visible in the same way as for the Ginibre ensemble. 
Again, in these cases of the GOE and GUE, a lot about \eqref{loigue} can be understood and computed explicitly thanks to the underlying random matrix structure and its determinantal nature. In fact the global and local statistics of eigenvalues are completely understood. 

Considering the coincidence between a statistical mechanics model and the law of the spectrum of a random matrix model for several values of the inverse temperature, it is also natural to ask whether such a correspondence exists for any value of $\beta$. The answer is yes for $\beta = 4$, it corresponds to the Gaussian Symplectic Ensemble (GSE) of Hermitian matrices with quaternionic coefficients, and  for any $\beta$ a somehow complicated model of tridiagonal matrices can be associated to the Gibbs measure of the one-dimensional log gas at inverse temperature $\beta$, see\cite{de}.  This and other methods  allow again  to compute a lot explicitly, and to derive that  the microscopic laws of the eigenvalues are those of a so called {\it sine-$\beta$ process} \cite{vv}.

Generally speaking, much is known for log gases in one dimension, for any value of $\beta$ and a wide class of potentials $V$. In particular, a lot of attention has been devoted to proving that many of the features of the system at the  microscopic  scale are {\it universal}, i.e. independent on the particular choice of $V$. For recent results, see  \cite{bey1,bey2,sh1,sh2,sh3,borotguionnet,bg2,bfg}.
 In contrast, the analogue is true in dimension 2  only for $\beta = 2$ \cite{ginibre,bors,ahm}. 
Thus the topic of log/Coulomb gases does not seem  reducible to just a subset of Random Matrix Theory. 
This is of course even more true in dimension $3$ and higher, which we will also treat, and  where we leave the realm of RMT. 


\chapter{The leading order behavior of the Coulomb gas} \label{leadingorder} 
In this chapter, we study the leading order or ``mean field" behavior of the Coulomb gas Hamiltonian. The results are quite standard and borrowed or adapted from the literature. However, we try to give here a self-contained and general treatment, since results in the literature are a bit scattered between the potential theory literature, the probability and statistical mechanics literature and the PDE literature, and not all situations seem  to be systematically covered.

Let us recall the setting~: 
for $(x_1, \dots, x_n)$ in $(\R^d)^n$, we  define the Hamiltonian
\be \label{hamiltoniencoulomb2}
H_n(x_1, \dots, x_n) = \sum_{i\neq j} g(x_i - x_j) + n \sum_{i=1}^n V(x_i)
\ee
where $g$ is a multiple of the Coulomb kernel in dimension $d\ge 2$ (the fundamental solution of the Laplacian) and  is $-\log $ in dimension $1$ (this choice is made to treat the case of log gases mentioned in the previous chapter):
\be \label{coulombkernel1}
g (x)= \left\lbrace \begin{array}{cc} - \log |x| & \mbox{for } d =  1, 2 \\ \f{1}{|x|^{d-2}} & \mbox{for } d \geq 3. \end{array} \right. 
\ee
The results we present here are in fact valid for a much more general class of radial interaction kernels, as an inspection of the proof shows. 

We need to keep track of the proportionality factor between $g$ and the true Coulomb kernel, and note that, for $d \ge 2$, we have
\begin{equation} \label{coulombkernel2}
- \Delta g = c_d \delta_0
\end{equation}
where $\delta_0$ is the Dirac mass at $0 \in \R^d$, and where the constant $c_d$ is given by~:
\begin{equation}
\label{defcd}
c_2 = 2\pi \  \mbox{ and }  c_d = (d-2) |\mathbb{S}^{d-1}|\  \textrm{ for }d \geq 3,
\end{equation} see e.g. \cite[Chap. 6]{liebloss}.
In dimension $d=1$, instead of \eqref{coulombkernel2}, $g$ solves   the non-local equation
$$-\Delta^{1/2} g =c_1\delta_0\qquad c_1=\pi,$$ 
where $\Delta^s $ is the fractional Laplacian  (in this situation one can check that $\Delta^{1/2}$ coincides with the Dirichlet-to-Neumann operator in the upper half-plane, see e.g. \cite{caffsilvestre}).
We will sometimes   abuse notation by considering also $g$ as a function on $\R$, by writing $g(r)= \begin{cases} -\log r  & \text{for } d=1,2\\ r^{2-d} & \text{for } d\ge 3.\end{cases}$

The function $V$ is called the (confining) potential, precise assumptions on $V$ will be made later.
The Hamiltonian $H_n$ may be physically interpreted as follows~: it  is the sum of  an interaction term $\sum_{i\neq j} g(x_i-  x_j)$ and a confining term $n \sum_{i=1}^n V(x_i)$. The first term describes the pairwise interaction of charged particles of same sign, interacting via (a multiple of) the Coulomb potential in dimension $d$. Since all the charges have the same sign, their natural behavior is to repel  each other, and potentially  escape to infinity. The potential $V$, however, is there to  confine the particles to a compact set of  $\R^d$. Note that the sum of pairwise interactions is expected to scale like the number of pairs of points, i.e. $n^2$, while the sum of the potential terms is expected to scale like $n$ times the number of points, i.e. $n^2$ again. The factor $n$ in front of $V$ in \eqref{hamiltoniencoulomb2} is there precisely so that this happens, in such a way that the opposing effects of the repulsion and of the confinement balance each other. This is called the ``mean-field scaling". It is the scaling in which  the force acting on each particle is   given in terms of the average field generated by the other particles.
For general reference on mean-field theory, see statistical  mechanics textbooks such as \cite{huang}.

The beginning of this chapter is devoted to the analysis of $H_n$ only, via $\Gamma$-convergence, leading to the mean-field description of its minimizers. Later, 
 in Section \ref{LDPsection} we apply these results to the statistical mechanics model associated to the Hamiltonian $H_n$, i.e. to characterizing the states with nonzero temperature.

\section{$\Gamma$-convergence : general definition} \label{sectiongammaconvergence}
The result we want to show about the leading order behavior of $H_n$ can be formalized in terms of  the notion of $\Gamma$-convergence, in the sense of De Giorgi (see\cite{braides},\cite{braides2} for an introduction, or \cite{dalmaso} for an advanced reference). It is a notion of convergence for functions (or functionals) which ensures that minimizers tend to minimizers.  This notion is popular in the community of calculus of variations  and very much used in the analysis of sharp-interface and fracture models, dimension reduction for variational problems, homogeneization...  (see again\cite{braides} for a review) and even more recently in the study of some evolution problems \cite{braides3}.
Using this formalism here is convenient but not essential.

Let us first give the basic definitions.
\begin{defini}[$\Gamma$-convergence] \label{defgcv}We say that a sequence $\{F_n\}_{n}$ of functions on a metric space $X$ $\Gamma$-converges to a function $F : X \rightarrow (-\infty, +\infty]$ if the following two inequalities hold : 
\begin{enumerate}
\item $(\Gamma$-$\liminf)$ If $x_n \rightarrow x$ in $X$, then $\liminf_{n \to + \infty} F_n(x_n) \geq F(x)$.
\item $(\Gamma$-$\limsup)$ For all $x$ in $X$, there is a sequence $\{x_n\}_n$ in $X$ such that $x_n \rightarrow x$ and $\limsup_{n \to + \infty} F_n(x_n) \leq F(x)$. Such a sequence is called a \textit{recovery sequence}.
\end{enumerate}\end{defini}

The second inequality is essentially saying that the first one is sharp, since it implies that  there is a particular sequence $x_n \rightarrow x$ for which the equality $\lim_{n \to + \infty} F_n(x_n) = F(x)$ holds.

\begin{rem} \label{gconvcom}
In practice a compactness assumption is generally needed and sometimes added in the definition, requiring that if $\{F_n(x_n)\}_n$ is bounded, then $\{x_n\}_{n}$ has a convergent subsequence.
A similar compactness requirement also appears in the definition of a good rate function in large deviations theory (see Definition \ref{ratefun} below).
\end{rem}
\begin{rem}
The first inequality is usually proven by functional analysis methods, without making any ``ansatz" on the precise form of $x_n$, whereas the second one is usually obtained by an explicit construction, during which one constructs “by hand” the recovery sequence such that that $F_n(x_n)$ has  asymptotically less energy than $F(x)$. Note also that by a diagonal argument, one may often reduce to  constructing  a recovery sequence for a dense subset of $x$'s.
\end{rem}
\begin{rem}
A $\Gamma$-limit is always lower semi-continuous (l.s.c.) as can be checked. (In particular, a function which is not l.s.c. is a bad candidate for being a $\Gamma$-limit.)
Thus, a functional is not always its own $\Gamma$-limit~: in general $\Gamma$-lim $F = \bar{F}$ where $\bar{F}$ is the l.s.c. envelope of $F$. \end{rem}
\begin{rem}\label{rem4}
The notion of $\Gamma$-convergence can be generalized to the situation where $F_n$ and $F$ are not defined on the same space. One may instead  refer to a sense of convergence of $x_n$ to $x$  which is defined via the  convergence of any specific function of $x_n$ to $x$, which may be a nonlinear function of $x_n$,  cf. \cite{ssgcv,jstern} for instances of this.
\end{rem}

We now state the most important property of $\Gamma$-convergence~: $\Gamma$-convergence sends minimizers to minimizers.
\begin{prop}[Minimizers converge to minimizers under $\Gamma$-convergence]\label{gammaconvmini}  Assume $F_n$ $\Gamma$-converges to $F$ in the sense of Definition \ref{defgcv}.
If for every $n$,  $x_n$ minimizes $F_n$,  and if the sequence $\{x_n\}_n$ converges to some $x$ in $X$, then $x$ minimizes $F$, and   moreover, $\lim_{n\to + \infty} \min_X F_n= \min_X  F$.
\end{prop}
\begin{proof}
Let $y\in X$. By the $\Gamma$-$\limsup$ inequality, there is a recovery sequence $\{y_n\}_n$ converging to $y$ such that $F(y) \geq \limsup_{n\to + \infty} F_n(y_n)$. By minimality of $x_n$,  we have $F_n(y_n) \geq F_n(x_n)$ for all $n$ and by the $\Gamma$-$\liminf$ inequality it follows that  $\liminf_{n\to  + \infty} F_n(x_n) \geq F(x)$, hence $F(y) \geq F(x)$.  Since this is true for every  $y$ in $X$, it proves that  $x$ is a minimizer of $F$. The relation $ \lim_{n\to + \infty} \min F_n= \min F$ follows from the previous chain of inequalities applied with $y=x$.
\end{proof}
\begin{rem}\label{gcvcomplete}
An  additional compactness assumption as in Remark \ref{gconvcom} ensures that if $\{\min F_n\}_{n} $ is bounded then a  sequence $\{x_n\}_n$ of minimizers has a limit, up to extraction. That limit  must then be a minimizer of $F$. If moreover it happens that $F$ has a unique minimizer, then the whole sequence $\{x_n\}_n$ must converge to it.
\end{rem}

\section{$\Gamma$-convergence of the Coulomb gas Hamiltonian}
The example of $\Gamma$-convergence of interest for us here is  that of  the sequence of functions $\{\frac{1}{n^2}H_n\}_n$ defined as in \eqref{hamiltoniencoulomb2}. The space  $\mc{P}(\R^d)$  of Borel probability measures on $\R^d$ endowed with  the topology of weak convergence (i.e. that of the  dual of bounded continuous functions in $\R^d$), which is metrizable,  will play the role of the metric space $X$ above. We may view $H_n$ as being defined on $\mc{P}(\R^d)$ through the map 
\begin{equation}
\left\lbrace \begin{array}{ccc} (\R^d)^n & \longrightarrow & \mc{P}(\R^d) \\
(x_1, \dots, x_n) & \mapsto & \frac{1}{n} \sum_{i=1}^n \delta_{x_i} \end{array}
\right.
\end{equation}
which associates to any configuration of $n$ points the probability measure
 $\f{1}{n}\sum_{i=1}^n \delta_{x_i}$, also called the \textit{empirical measure}, or \textit{spectral measure} in the context of random matrices.  More precisely, for any $\mu$ in $\mc{P}(\R^d)$, we let $H_n(\mu)$ be :
\begin{equation}
H_n(\mu) = \left\lbrace \begin{array}{cl} H_n(x_1, \dots, x_n) & \mbox{ if } \mu \mbox{ is of the form } \frac{1}{n} \sum_{i=1}^n \delta_{x_i} \\
+ \infty & \mbox{otherwise}.
\end{array}
\right.
\end{equation}
The first main result that we prove here is that the sequence $\{\frac{1}{n^2}H_n\}_n$, has an explicit $\Gamma$-limit. It should not be surprising that we first  need to  divide by $n^2$ in order to get a limit, since we have seen that all the terms in $H_n$ are expected to be  proportional to $n^2$.

\begin{prop}[$\Gamma$-convergence of $\frac{1}{n^2}H_n$] \label{gammaconvergenceHn} 
Assume $V$ is continuous and bounded below.  The sequence $\{\f{1}{n^2} H_n\}_n$ of functions (defined on $\mc{P}(\R^d)$ as above)  $\Gamma$-converges as $n \to + \infty$, with respect to the weak convergence of probability measures,  to the function $I : \mc{P}(\R^d) \rightarrow (-\infty, + \infty]$ defined by :
\begin{equation} \label{definitionI}
I(\mu) := \iint_{\R^d \times \R^d} g(x-y) d\mu(x)d\mu(y) + \int_{\R^d} V(x) d\mu(x).
\end{equation}

\end{prop}

Note that $I(\mu)$ is simply a continuous version of the discrete Hamiltonian $H_n$ defined over all $\mc{P}(\R^d)$, which may also take the value $+ \infty$. From the point of view of statistical mechanics, $I$ is the “mean-field” limit energy of  $H_n$, while we will see below in Section \ref{LDPsection} that from the point of view of probability, $I$ also plays the role of a {\it rate function}.

In the next sections, we focus on the analysis of $I$ and its minimization.
This analysis will provide ingredients  for the  proof of Proposition \ref{gammaconvergenceHn}, which  is postponed to Section \ref{seccv}.

\section{Minimizing the mean-field energy   via potential theory} \label{potentialtheoretic}
In this section, we focus on the study of $I$ defined in \eqref{definitionI}, and the associated  minimization problem  of finding
\be \label{minI}
\min_{\mu \in \mc{P}(\R^d)} \iint_{\R^d\times \R^d} g(x-y) d\mu(x) d\mu(y) + \int_{\R^d} V(x) d\mu(x).
\ee

This turns out to be  a classical electrostatic problem, that of finding  the  equilibrium distribution of charges in a  capacitor with an external potential, also called the ``capacitance problem,"  historically studied by Gauss, and rigorously solved by Frostman \cite{frostman}. It thus  is a fundamental question  in {\it potential theory}, which itself grew out of the study of the electrostatic or gravitational potential.  One may see  e.g. \cite{ah,doob,safftotik} and references therein.
The case of $d=2$ and $g(x)=-\log |x|$ is precisely treated in \cite[Chap. 1]{safftotik}. Higher dimensional and more general singular interaction potentials are treated in \cite{cgz}. The general case is not more difficult.

We start with 
\begin{lemme}\label{convexi}
The functional $I$ is strictly convex on $\mc{P}(\R^d)$.
\end{lemme}
\begin{proof}
Since $\mu \mapsto \int V\, d\mu$ is linear, it suffices to notice that the quadratic function $f\mapsto \iint g(x-y)\, df(x)\, df(y)$, defined over all (signed) Radon measures,  is positive. This is true in dimension $\ge 3$ (cf. \cite[Theorem 9.8]{liebloss}), or when $g=-\log $ in dimension $2$ (for a proof cf. \cite[Lemma 3.2]{RSY2}, which itself 
relies on \cite[Chap. 1, Lemma 1.8]{safftotik}.) \end{proof}
Note that 
less restrictive  assumptions, such as  $g>0$ or  $\hat{g}>0$, where $\hat{g}$ stands for the Fourier transform, would also suffice.

As a consequence, there is a unique (if any) minimizer to \eqref{minI}. In potential theory it is called the {\it equilibrium measure} or  the Frostman equilibrium measure, or sometimes the extremal measure.  

The question of the existence of a minimizer is a bit more delicate.
In order to show  it, we start by making the following assumptions on the potential $V$~: 
\begin{description}
\item[(A1)] $V$ is l.s.c. and bounded below.
\item[(A2)] (growth assumption)
$$
   \underset{|x|\to + \infty}{\lim}\( \f{V(x)}{2} + g(x)\) = + \infty
$$
\end{description}
The first condition is there to ensure the lower semi-continuity of $I$ and that $\inf I>-\infty$, the second is made to ensure that $I$ is coercive.
Of course, in the Coulomb case with $d\ge 3$, condition $\textbf{(A2)}$ is equivalent to the condition that $V$ tends to $+\infty$ at infinity.

\begin{lemme}\label{coerI}
Assume $\mathrm{\textbf{(A1)}}$ and $\textbf{(A2)}$ are satisfied, and let $\{\mu_n \}_n$ be a sequence in $\mc{P}(\R^d)$ such that $\{I(\mu_n)\}_n $ is bounded. Then, up to extraction of a subsequence $\mu_n $ converges to some $\mu\in \mc{P}(\R^d)$ in the weak sense of probabilities, and $$\liminf_{n\to \infty} I(\mu_n) \ge I(\mu).$$
\end{lemme}
\begin{proof}
Assume that $I(\mu_n) \le C_1$ for each $n$. 
Given any constant $C_2>0$ there exists a  compact set $K\subset \R^d$ such that 
\begin{equation}\label{vgrand}
\min_{ (K\times K)^c} \left[ g(x-y) + \frac{V}{2}(x)+\frac{V}{2}(y)\right]>C_2.\end{equation}
Indeed, it suffices to  show that $g(x-y) + \frac{V}{2}(x)+\frac{V}{2}(y)\to +\infty$ as $x\to \infty$ or $y \to \infty$.
To check this,
one may separate the cases $d = 1,2$ and $d \geq 3$. In the latter case, the Coulomb kernel $g$ is positive and $V$ is bounded below, so that $g(x-y) + \f{V}{2}(x) + \f{V}{2}(y)$, which is greater than $\f{V}{2}(x) + \f{V}{2}(y)$,  can be made  arbitrarily large if $x$ or $y$ gets large by assumption $\textbf{(A2)}$.
When $d =1,2$, since $g(x-y) = -\log |x -y| \geq - \log 2 - \log \max(|x|, |y|)$, we also have from assumptions $\textbf{(A1)}$ and  $\textbf{(A2)}$ that $\f{1}{2}(V(x) + V(y)) + g(x-y)$ is arbitrarily large if $|x|$ or $|y|$ is large enough.

In addition, by \textbf{(A1)}, \textbf{(A2)}, we have  in all cases that $g(x-y)+\hal V(x)+ \hal V(y)$ is bounded below on $\R^d \times \R^d$ by a constant, say $-C_3$, with $C_3>0$.
 Rewriting then  $I$ as
 \begin{equation}
I(\mu) = \iint_{\R^d} \left[g(x-y) + \f{V}{2}(x) + \f{V}{2}(y)\right] d\mu(x) d\mu(y),
\end{equation}
the relation \eqref{vgrand} and our assumption on $\mu_n$ imply that 
$$C_1\ge I(\mu_n) \ge- C_3+  C_2(\mu_n\otimes \mu_n) ( (K\times K)^c)\ge -C_3 + C_2 \mu_n(K^c).$$
Since $C_2$ can be made arbitrarily large,  $\mu_n(K^c)$ can be made arbitrarily small, which means precisely that $\{\mu_n\}_n$ is a tight sequence of probability measures.
By Prokhorov's theorem, it thus has a convergent subsequence (still denoted $\{\mu_n\}_n$), which  converges to some probability $\mu$.
For any $n$ and any $M > 0$, we may then write
\begin{equation}
\iint g(x-y) d\mu_n(x) d\mu_n(y) \geq \iint (g(x-y) \wedge M) d\mu_n(x) d\mu_n(y) 
\end{equation}
where $\wedge$ denotes the minimum of two numbers. 
For each given $M$,
 the integrand in the right-hand side is  continuous, and thus the weak convergence of $\mu_n$ to $\mu$, which implies the weak convergence of $\mu_n \otimes \mu_n$ to  $  \mu \otimes \mu$, yields    
$$
\liminf_{n \to + \infty} \iint g(x-y) d\mu_n(x) d\mu_n(y) \geq \iint (g(x-y) \wedge M) d\mu(x) d\mu(y).
$$A monotone convergence theorem argument allows one to let $M \rightarrow + \infty$, and using assumption \textbf{(A1)} for the weak l.s.c of the potential part of the functional, we conclude that 
 \begin{equation}
 \underset{n\to  + \infty} {\liminf}  \,  I(\mu_n) \geq I(\mu). 
 \end{equation}

\end{proof}

We have seen above that $\inf I >-\infty$ (indeed the integrand in $I$ is bounded below thanks to assumption $\textbf{(A1)}$).
The next question is to see whether  $\inf I <+\infty$, i.e. that there exist probabilities with finite $I$'s. This is directly related to the notion of (electrostatic, Bessel, or logarithmic) capacity, whose definition we now give.
 One may find it \cite{safftotik,eg,ah} or \cite[Sec. 11.15]{liebloss}, the formulations differ a bit but are essentially equivalent.  It is usually not formulated this way in dimension $1$ but it can be extended to that case with our choice of $g=- \log $ without trouble.

\begin{defini}[Capacity of a set] \label{definitioncapacite1} We define the capacity of a compact set $K\subset \R^d$ by 
\begin{equation} \label{definitioncapacite2}
\mathrm{cap}(K) : = \Phi\(  \inf_{\mu \in \mc{P}(K)} \iint_{\R^d} g(x-y) d\mu(x) d\mu(y)\),
\end{equation}
with $\Phi(t)= e^{-t}$ if $d=1,2$ and $\Phi(t)= t^{-1}$ if $d \ge 3$, and 
where $\mc{P}(K)$ denotes the set of probability measures supported in $K$.  Here  the $\inf$ can be $+\infty$ if there exists no 
 probability measure $\mu\in \mc{P}(K)$ such that
$
\iint_{\R^d} g(x-y) d\mu(x) d\mu(y) < + \infty$. 
For a general set $E$, we define $\mathrm{cap}(E)$ as the supremum of $\mathrm{cap}(K)$ over the compact sets $K$ included in $E$. It is easy to check that capacity is increasing with respect to the inclusion of sets. 
\end{defini}

A basic fact is that a set of zero capacity also has zero Lebesgue measure (see the references above).
\begin{lemme} If $\mathrm{cap}(E) = 0$, then $|E| = 0$. 
\end{lemme}In fact  $\mbox{cap}(E)=0$ is stronger than $|E|=0$, it implies for example that the perimeter of $E$ is also $0$. 
A property is said to hold “quasi-everywhere” (q.e.) if it holds everywhere except on a set of capacity zero. By the preceding lemma a property that holds q.e. also holds Lebesgue-almost everywhere (a.e.), whereas the converse is, in general, not true.

For the sake of generality, it is interesting to consider potential $V$'s which can take the value $+\infty$ (this is the same as imposing the constraint that the probability measures only charge  a specific set, the set where $V$ is finite).
We then need to place a third assumption 
\begin{description}
\item[(A3)]  $\{ x \in \R^d | V(x) <+\infty\}$ has positive capacity.
\end{description}
\begin{lemme}\label{infiborne}
Under assumptions $ \textbf{(A1)}$---$\textbf{(A3)}$, we have $\inf I <+\infty$.
\end{lemme}
\begin{proof}
Let us define for any $\varepsilon > 0$ the set $\Sigma_{\varepsilon} = \{x\ |\ V(x) \leq \f{1}{\varepsilon}\}$. Since $V$ is l.s.c the sets $\Sigma_{\varepsilon}$ are closed, and it is easy to see that assumption \textbf{(A2)} implies that they are also bounded, since $V(x)$ goes to $+ \infty$ when $|x| \rightarrow + \infty$.

The capacity of $\Sigma_0 = \{x\in \R^d | V(x) < +\infty \}$ is positive by assumption. It is easily seen that the sets $\{\Sigma_{\eps}\}_{\eps > 0}$ form a decreasing family of compact sets with $\bigcup_{\eps > 0} \Sigma_{\eps} = \Sigma_0$, and by definition (see Definition \ref{definitioncapacite1} or the references given above) the capacity of $\Sigma_0$ is the supremum of capacities of compact sets included in $\Sigma_0$. Hence we have that $\mbox{cap}(\Sigma_{\varepsilon})$ is positive for $\varepsilon$ small enough. Then by definition there exists  a probability measure $\mu_{\eps}$ supported   in $\Sigma_{\eps}$ such that 
\begin{equation}
\iint g(x-y) d\mu_{\eps}(x) d\mu_{\eps}(y) < + \infty.
\end{equation}
Of course, we also have $\int V d\mu_{\eps} < + \infty$ by definition of $\Sigma_{\eps}$. Hence $I(\mu_{\eps}) < + \infty$, in particular $\inf I < + \infty$.

\end{proof}

We may now give the main existence result, together with the characterization of the minimizer.

\begin{theo}[Frostman \cite{frostman}, existence  and characterization of the equilibrium measure] \label{theoFrostman} Under the assumptions \textbf{(A1)}-\textbf{(A2)}-\textbf{(A3)}, the minimum of $I$ over $\mc{P}(\R^d)$ exists, is finite and is achieved by a unique $\mu_0$, which  has a compact support of  positive capacity. In addition $\mu_0$ is uniquely characterized by the fact that 
\begin{equation}
\label{EulerLagrange}
\left\lbrace
\begin{array}{cc} h^{\mu_0} +\D \f{V}{2} \geq c & \mbox{q.e. in } \R^d  \vspace{3mm} \\ 
 h^{\mu_0} + \D \f{V}{2}= c & \mbox{q.e. in the support of }\mu_0 \end{array} \right.
\end{equation}
where \be \label{defhmu0}
h^{\mu_0}(x) := \int_{\R^d} g(x - y) d\mu_0(y)
\ee is the electrostatic potential generated by $\mu_0$;
 and then the constant $c$ must be 
 \begin{equation}
\label{defc1} c = I(\mu_0) - \hal \int_{\R^d} V(x) d\mu_0(x).
 \end{equation}
 \end{theo}

\begin{rem} \label{rem8} Note that by (\ref{coulombkernel2}), in dimension $d\ge 2$,  the function $h^{\mu_0}$ solves :
\begin{displaymath}
- \Delta h^{\mu_0} = c_d \mu_0.
\end{displaymath} 
\end{rem}

\begin{example}[Capacity of a compact set] Let $K$ be a compact set of positive capacity,  and let $V = 0$ in $K$ and $V = + \infty$ in $K^{c}$. In that case the minimization of $I$ is the same as the computation of the capacity of $K$ as in \eqref{definitioncapacite2}. The support of the equilibrium measure $\mu_0$ is contained in $K$, and the associated Euler-Lagrange equation (\ref{EulerLagrange}) states that  the electrostatic potential (if in dimension $d\ge 2$) $h^{\mu_0} $   is constant  q.e. on the support of $\mu_0$ (a well-known result in physics). If $K$ is sufficiently regular, one can apply the Laplacian on both sides of this equality, and in view of Remark \ref{rem8},  find that $\mu_0 = 0$ q.e. in $K$, which indicates  that $\mu_0$ is supported on $\partial K$. 
\end{example}

\begin{example}[$C^{1,1}$ potentials and RMT examples]  In general,  the relations \eqref{EulerLagrange} say that the total potential $h^{\mu_0}+ \frac{V}{2}$ is constant on the support of the charges. Moreover, in dimension $d \ge 2$,  applying the Laplacian on both sides of \eqref{EulerLagrange}  and using again Remark~\ref{rem8} gives that, on the interior of the support of the equilibrium measure, if $V\in C^{1,1}$, \be \label{densmu0}
c_d \mu_0 = \f{\Delta V}{2} 
\ee
(where $c_d$ is the constant defined in (\ref{defcd})), i.e. the density of the measure on the interior of its support is given by $\frac{\Delta V}{2c_d}$. This will be proven in Proposition  \ref{proequivpb}. For example if $V$ is quadratic, then the measure $\mu_0 = \f{\Delta V}{2c_d} $ is constant  on the interior of its support.  This corresponds to the important examples of the Hamiltonians that arise in random matrix theory, more precisely~:

\begin{itemize}
\item in dimension $d=2$, for $V(x) = |x|^2$, one may check that $\mu_0 = \f{1}{\pi} \mathbf{1}_{B_1}$  where $\mathbf{1}$ denotes a characteristic function and $B_1$ is the unit ball, i.e.  the equilibrium measure is the normalized Lebesgue measure on the unit disk (by uniqueness, $\mu_0$ should be radially symmetric, and the combination of \eqref{densmu0} with the constraint of being a probability measure imposes the support to be $B_1$). This is known as the {\it circle  law} for the Ginibre ensemble in the context of Random Matrix Theory (RMT). Its derivation (which we will see in Section \ref{LDPsection} below) is attributed to   Ginibre, Mehta, an unpublished paper of Silverstein and  Girko \cite{girko1}.
\item in dimension $d\ge 3$, the same  holds, i.e.  for $V(x)=|x|^2$  we have $\mu_0 = \frac{d}{c_d} \indic_{B_{(d-2)^{1/d}   }}$ by the same reasoning.
\item in dimension $d=1$, with $g =-\log |\cdot |$ and $V(x) = x^2$, the equilibrium measure is $\mu_0 (x)= \f{1}{2\pi} \sqrt{4-x^2} \mathbf{1}_{|x|\leq 2}$, which corresponds in the context of RMT (GUE and GOE ensembles) to {\it   Wigner's semi-circle law}, cf. \cite{wigner,mehta}.
\end{itemize}
\end{example}
We now turn to the proof of Theorem \ref{theoFrostman}, adapted from\cite[Chap. 1]{safftotik}.
\begin{proof}
The existence of a minimizer $\mu_0$ follows directly from Lemmas \ref{coerI} and \ref{infiborne}, its uniqueness from Lemma \ref{convexi}. It remains to show that $\mu_0$ has compact support  of finite capacity and that \eqref{EulerLagrange} holds.
\medskip

\noindent{\bf Step 1.} We prove that $\mu_0$ has compact support.
  Using \eqref{vgrand}, we may find a compact set  $K$ such that 
 $g(x-y)+\frac{V}{2}(x)+\frac{V}{2}(y) \ge  I(\mu_0) +1$ outside of $K\times K$.

Assume that $\mu_0$ has mass outside $ K  $,  i.e. assume $\mu_0(K) <1$, and define the new probability measure 
\begin{equation}
\tilde{\mu} := \f{(\mu_0)_{| K}}{\mu_0(K)}.
\end{equation}
We want to show that $\tilde{\mu}$ has less or equal energy $I(\tilde{\mu})$ than $\mu_0$, in order to get a contradiction. One may compute $I(\mu_0)$ in the following way :
\begin{multline}
I(\mu_0) = \iint_{K \times K} \left[g(x-y) + \f{V}{2}(x) + \f{V}{2}(y)\right] d\mu_0 (x) d\mu_0 (y) \\ +  \iint_{(K \times K)^{c}} \left[g(x-y) + \f{V}{2}(x) + \f{V}{2}(y)\right] d\mu_0(x) d\mu_0(y)  \\ \geq
\mu_0(K)^2 I(\tilde{\mu}) + (1- \mu_0(K)^2) \min_{(K \times K)^{c}}\left[g(x-y) + \f{V}{2}(x) + \f{V}{2}(y)\right].
\end{multline}
By choice of $K$, and since we assumed $\mu_0(K)<1$, this implies  that  
\begin{equation}
 I(\mu_0) \ge  \mu_0(K)^2 I(\tilde{\mu}) + (1- \mu_0(K)^2) (I(\mu_0)  + 1) 
\end{equation} and thus 
$$I(\tilde{\mu}) \le I(\mu_0) + \frac{\mu_0(K)^2 -1}{\mu_0(K)^2} <   I(\mu_0),$$
a contradiction with the minimality of $\mu_0$.
 We thus conclude  that $\mu_0$ has compact support.
 The fact that the support of $\mu_0$ has positive capacity is an immediate consequence of the fact that $I(\mu_0)<\infty$ and the definition of capacity.
\medskip

\noindent
{\bf Step 2.}
 We  turn to the proof of the Euler-Lagrange equations (\ref{EulerLagrange}). For  this, we use the ``method of  variations"  which consists in  continuously deforming $\mu_0$ into  other admissible probability measures. \\
 Let $\nu$ in $\mc{P}(\R^d)$ such that $I(\nu) < + \infty$, and consider the probability measure $(1-t) \mu_0 + t \nu$ for all $t$ in $[0,1]$. Since $\mu_0$ minimizes $I$, we have 
\be
I\((1-t)\mu_0 + t\nu\) \geq I(\mu_0), \mbox{ for all } t \in [0,1].
\ee
By letting $t \rightarrow 0^+$ and keeping only the first order terms in $t$, one obtains the ``functional derivative" of $I$ at $\mu_0$. More precisely, writing  
\be
\iint g(x-y) d((1-t)\mu_0 + t\nu)(x) d((1-t)\mu_0 + t\nu)(y) + \int V(x) d((1-t)\mu_0 + t\nu)(x) \geq I(\mu_0)
\ee
one easily gets that 
\begin{multline}
I(\mu_0) + t \left[ \iint g(x-y) d(\nu- \mu_0)(x) d\mu_0(y) + \iint g(x-y) d(\nu-\mu_0)(y) d\mu_0(x) \right. \\\left. + \int V(x) d(\nu - \mu_0)(x)\right] + O(t^2) \geq I(\mu_0).
\end{multline}
Here, we may cancel the identical order $0$ terms $I(\mu_0)$ on both sides,  and note that 
in view of \eqref{defhmu0} the expression between brackets can be rewritten as $2\int h^{\mu_0}(x)d(\nu - \mu_0)(x) + \int V(x) d(\nu - \mu_0)(x)$. Next, dividing the inequality by $t > 0$, and letting $t \rightarrow 0^+$, it appears that for all $\nu$ in $\mc{P}(\R^d)$ such that $I(\nu) < + \infty$, the following inequality holds :
\be
2 \int h^{\mu_0}(x) d(\nu - \mu_0)(x) + \int V(x) d(\nu- \mu_0)(x) \geq 0
\ee
or equivalently \be \label{EulerLagrange2}
\int \Big(h^{\mu_0}  + \f{V}{2}\Big)(x) d \nu(x) \geq \int \Big(h^{\mu_0} + \f{V}{2}\Big)(x) d\mu_0(x).
\ee
Defining  the constant $c$ by 
\begin{multline} \label{defc2} c = I(\mu_0) - \f{1}{2} \int V d\mu_0= \iint g(x-y) d\mu_0(x)d\mu_0(y) + \f{1}{2} \int V d\mu_0 \\
= \int \Big(h^{\mu_0} + \f{V}{2}\Big) d\mu_0 ,\end{multline}
(\ref{EulerLagrange2}) asserts that
\begin{equation}
\label{EL3}
\int \Big(h^{\mu_0} + \f{V}{2}\Big) d\nu \geq c 
\end{equation}
for all probability measures $\nu$ on $\R^d$ such that $I(\nu) < + \infty$. Note that at this point, if we relax the condition $I(\nu) < + \infty$, then choosing $\nu$ to be a Dirac mass when applying (\ref{EL3}) would yield
\begin{equation} 
\label{EL4}
h^{\mu_0} + \f{V}{2} \geq c
\end{equation} pointwise. However, Dirac masses have infinite energy $I$, and we can only prove that (\ref{EL4}) holds quasi-everywhere, which we do now. 

Assume not, then there exists a set $K$ of positive capacity such that (\ref{EL4}) is false on $K$, and by definition of the capacity of $K$ as as supremum of capacities over compact sets included in $K$, we may in fact suppose that $K$ is compact. By definition, this means that there is a probability measure $\nu$ supported in $K$ such that 
\be \label{ene1}
\iint g(x-y) d\nu(x) d\nu(y) < + \infty.
\ee
Let us observe that $-h^{\mu_0}$ is bounded above on any compact set (this is clear in dimension $d \geq 3$ because the Coulomb kernel is positive and so is $h^{\mu_0}$, and can be easily checked in $d=1,2$ because $\log$ is bounded above on any compact set and $\mu_0$ has compact support). By assumption, equation (\ref{EL4}) is false on $K$, that is $V < 2c - 2 h^{\mu_0}$ on $K$. Integrating this inequality against $\nu$ gives 
\be \label{contrad}
\int V d\nu = \int_K V d\nu < \int_K (2c - 2 h^{\mu_0}) \, d\nu< + \infty
\ee
which, combined with (\ref{ene1}) ensures that $I(\nu)$ is finite. But then  \eqref{contrad} contradicts  (\ref{EL3}). We thus have shown that 
\be \label{EL5}
h^{\mu_0} + \f{V}{2} \geq c\  \mbox{ q.e.}
\ee
which is the first of the relations (\ref{EulerLagrange}). 

For the second one, let us denote by $E$ the set where the previous inequality (\ref{EL5}) fails. We know that $E$ has zero capacity, but since $\mu_0$ satisfies $I(\mu_0) < + \infty$, it does not charge sets of zero capacity (otherwise one could restrict $\mu_0$ on such a set, normalize its mass to $1$ and get a contradiction with the definition of a zero capacity set). Hence we have \be
h^{\mu_0} + \f{V}{2} \geq c \quad \mu_0\mbox{-a.e. } 
\label{EL6}
\ee
Integrating this relation against $\mu_0$  yields 
\be
\int \Big(h^{\mu_0} + \f{V}{2}\Big) d\mu_0 \ge  c,
\ee
but in view of \eqref{defc2} this implies that  
 equality must  hold in (\ref{EL6}) $\mu_0$-almost everywhere. This establishes the second Euler-Lagrange equation.
 \medskip
 
 \noindent
 {\bf Step 3.} We show that the relations \eqref{EulerLagrange} uniquely characterize the minimizer of $I$. 
 Assume that $\mu$ is another probability solving \eqref{EulerLagrange} with some constant $c'$,   and set, for $t \in (0,1)$, $\mu_t := t \mu + (1-t)\mu_0$, hence $ h^{\mu_t} = t h^{\mu} + (1-t) h^{\mu_0}$. We have
 \begin{multline*}
 I(\mu_t) = \int \left( t h^{\mu}(x) + (1-t) h^{\mu_0}(x) + V(x) \right) d\mu_t(x) \\ = \f{t}{2} \int \left(2 h^{\mu}(x) + V(x) \right) d\mu_t(x) \\
+ \f{(1-t)}{2} \int \left(2 h^{\mu_0}(x) + V(x) \right) d\mu_t(x) + \f{1}{2} \int V(x) d\mu_t(x).
\end{multline*}
By assumption, $h^{\mu} + \f{V}{2} \geq c'$ and $h^{\mu_0} + \f{V}{2} \geq c$ almost everywhere. We hence get that 
\begin{multline}
I(\mu_t) \geq t c' + (1-t) c + \f{1}{2} \int V(x) \left( t d\mu(x) + (1-t) d\mu_0(x)\right) \\ = t \left(c' +\hal \int V d\mu\right) + (1-t) \left(c + \hal \int V d\mu_0 \right).
\end{multline}
On the other hand, integrating the second Euler-Lagrange equation in \eqref{EulerLagrange} for both $\mu$ and $\mu_0$,    with respect to $\mu$ and $\mu_0$ respectively, yields, after rearranging terms, 
\begin{displaymath}
I(\mu) = c' +\hal  \int V d\mu  \text{ and } I(\mu_0) = c +\hal \int V d\mu_0.
\end{displaymath}
Hence $I(\mu_t) \geq t I(\mu) + (1-t) I(\mu_0)$, which is impossible by strict convexity of $I$ unless $\mu = \mu_0$. This proves that the  two measures $\mu$ and $\mu_0$ must coincide. 
\end{proof}

\begin{rem} In all this section, we did not use much all the particulars of the Coulomb kernel.  The theorem still holds for a much more general class of $g$'s, say $g$  positive, monotone radial and satisfying $\iint g(x-y)\, dx\, dy<\infty$.  \end{rem}

\begin{defini}\label{def23} From now on, we denote by $\zeta$ the function
\begin{equation} \label{defzeta}
\zeta = h^{\mu_0} + \f{V}{2} - c.
\end{equation}
\end{defini}
\noindent We note that in view of \eqref{EulerLagrange}, $\zeta\ge 0$ a.e. and $\zeta=0$ $\mu_0$-a.e.

\section[The mean field limit]{Proof of the $\Gamma$-convergence and consequences for minimizers of the Hamiltonian : the mean-field limit}\label{seccv}

We now proof Proposition \ref{gammaconvergenceHn}. The  proof uses 
the same ingredients as the proof of the existence of a minimizer of $I$ in the previous section. A statement and a proof with $\Gamma$-convergence in dimension~$2$ for $V$ quadratic  appeared  in \cite[Proposition 11.1]{livre}. It is not difficult to adapt them to higher dimensions and more general potentials. Similar arguments are also found in the large deviations proofs of \cite{bg,bz,cgz}.

In what follows, when considering sequences of configurations $(x_1, \dots, x_n)$  we will make  the slight abuse of notation that consists in neglecting the dependency of the points $(x_1, \dots, x_n)$ on $n$, while one should formally write $(x_{1,n}, \dots, x_{n,n})$.

\begin{proof}[Proof of Proposition \ref{gammaconvergenceHn}]
 In the following, we denote the diagonal of $\R^d \times \R^d$ by $\triangle$ and its complement by $\triangle^{c}$.\\
{\bf Step 1}. We first need to prove that if $\f{1}{n} \sum_{i=1}^n \delta_{x_i} \rightarrow \mu\in \mc{P}(\R^d)$, then $$\liminf_{n\to + \infty} \f{1}{n^2} H_n(x_1, \dots, x_n) \geq I(\mu).$$ Letting $\mu_n$ denote the empirical measure  $\f{1}{n} \sum_{i=1}^n \delta_{x_i}$, we may write 
 \begin{equation} \label{Hnpourempiric}
\f{1}{n^2} H_n(\mu_n) = \iint_{\triangle^{c}} g(x-y) d\mu_n(x) d\mu_n(y) + \int V(x) d\mu_n(x).
\end{equation}
The last term in the right-hand side is harmless~: since $V$ is assumed to be continuous and bounded below, it is lower-semicontinous and bounded below and since the sequence $\{\mu_n\}_n$ converges weakly to $\mu$, we have :
\begin{equation}  \label{convergencesimpleV}
\underset{n\to + \infty}{\liminf} \int V d\mu_n \geq  \int V d\mu.
\end{equation}
In order to treat the first term in the right-hand side of (\ref{Hnpourempiric}), following \cite[Chap. 1]{safftotik}, let us truncate the singularity of $g$ by writing :
\be \label{troncaturediago1}
\iint_{\triangle^{c}} g(x-y) d\mu_n(x) d\mu_n(y) \geq \iint (g(x-y) \wedge M) d\mu_n(x) d\mu_n(y) - \f{M}{n}
\ee
where $M>0$ and  $\wedge$ still denotes the minimum of two numbers. Indeed one has $\mu_n \otimes \mu_n (\triangle) = \f{1}{n}$ as soon as the points of the  configuration $(x_1, \dots, x_n)$ are simple (i.e. $x_i \neq x_j$ for $i\neq j$). The function $(x,y) \mapsto g(x-y) \wedge M$ is continuous, and by taking the limit of (\ref{troncaturediago1}) as $n \rightarrow + \infty$ one gets,  by weak convergence of $\mu_n$ to $\mu$ (hence of $\mu_n \otimes \mu_n$ to $\mu \otimes \mu$)  that for every $M > 0$ :
\begin{equation}
\underset{n \to + \infty}{\liminf} \iint_{\triangle^{c}} g(x-y) d\mu_n(x) d\mu_n(y) \geq \iint (g(x-y) \wedge M) d\mu(x) d\mu(y).
\end{equation}

By the monotone convergence theorem, the (possibly infinite) limit of the right-hand side as $M \rightarrow + \infty$ exists and equals $\iint g(x-y) d\mu(x) d\mu(y)$. Combining with (\ref{convergencesimpleV}) and (\ref{troncaturediago1}), this concludes the proof of the $\Gamma$-$\liminf$ convergence. Let us note that for this part, we really only needed to know that $V$ is lower semi-continuous and bounded below.
\medskip

\noindent
{\bf Step 2.}
We now need to construct a recovery sequence for each measure $\mu$ in $\mc{P}(\R^d)$.  First, we show that we can reduce to measures which are in
$L^{\infty}(\R^d)$, supported in a cube $K$ and such that the density $\mu(x)$ is bounded below by $\alpha > 0$ in $K$.   
Let $\mu $ be an arbitrary measure in $\mc{P}(\R^d)$ such that $I(\mu)<+\infty$. 
Given $\alpha>0$, by tightness of $\mu$,  we may truncate it outside of a compact set which contains all its mass but $\alpha$. Making the compact set larger if necessary, and normalizing the truncated  measure to make it a probability, the argument of Step 1 of the proof of Theorem \ref{theoFrostman} shows that this decreases $I$. In other words we have a family $\mu_\alpha$ with $\mu_\alpha \rightharpoonup \mu$ and $\limsup_{\alpha \to 0 } I(\mu_\alpha) \le I(\mu)$. Thus, by a diagonal argument, it suffices to prove our statement for probability measures $\mu$ which have compact support.

Let us next consider such a probability $\mu$. 
Convoling $\mu$ with smooth mollifiers $\chi_\eta$ (positive of integral $1$), we may approximate $\mu $ by smooth $\mu_\eta$, these converge to $\mu$ in the weak sense of probabilities, as $\eta \to 0$. Let us denote $$\Phi(\mu)= \iint g(x-y)\, d\mu(x)\, d\mu(y).$$ As seen in the proof of 
 Lemma \ref{convexi}, the function $\Phi$ is strictly convex. Writing $\tau_y \mu$ for the translate of $\mu$ by $y$, we deduce, using Jensen's inequality, that
 \begin{equation*}\Phi(\chi_\eta *\mu)= \Phi\( \int \tau_y \mu \chi_\eta (y)  \, dy\) 
 \le \int \chi_\eta(y) \Phi(\tau_y \mu) \, dy.\end{equation*}
  Since $\Phi $ is translation-invariant, we have $\Phi(\tau_y \mu)=\Phi(\mu)$ and thus we have obtained $\Phi(\mu_\eta)\le \Phi(\mu)$.
  On the other hand, $\lim_{\eta\to 0} \int Vd(\mu_\eta) = \int V d\mu$ since $V$ is assumed to be continuous and $\mu_\eta \to \mu$ and they all are supported in the same compact set. We have thus established that $\limsup_{\eta\to 0} I(\mu_\eta) \le I(\mu)$. Thus, by a diagonal argument, it suffices to prove our statement for probability measures $\mu$ which have a smooth density and compact support. 
  
Let us next consider such a probability measure $\mu$. We may find a cube $K$ that contains its support and  consider the probability measure
 $\mu_\alpha=\frac{ \mu +\alpha \indic_K }{1+ \alpha |K|}.$ 
 It is supported in the cube $K$, has an $L^\infty$ density which is bounded below by $\alpha $ in $K$, as desired, and $\mu_\alpha \to \mu$ in the weak sense of probabilities, but also in $L^\infty$. It is easy to deduce from this fact that since $\iint g(x-y)\, dx\, dy<+\infty$ (cf. \eqref{coulombkernel1}) and $V$ is continuous, we have  
$I(\mu_\alpha) \to I(\mu)$.   Again, a diagonal argument allows us to reduce to proving the desired $\Gamma$-limsup statement for such measures.
 \medskip
 
   \noindent
 {\bf Step 3.} Henceforth we assume that $\mu$ is in $L^{\infty}(\R^d)$, supported in a cube $K$ and such that the density $\mu(x)$ is bounded below by $\alpha > 0$ in $K$.  
 Since we are going to construct a configuration of $n$ points in a cube $K$, the typical lengthscale of the distance between two points is $\f{1}{n^{1/d}}$. Let us then choose
a sequence $c_n$ such that\footnote{The notation $a_n \ll b_n$ means that $a_n = o(b_n)$} $\f{1}{n^{1/d}} \ll c_n \ll 1$ as $n \rightarrow + \infty$, and for each $n$, split $K$ into cubes $K_{k}$ (depending on $n$) whose sidelength is in $[c_n, 2c_n]$, cf. Fig. \ref{fig1}.

\begin{figure}[h!]\begin{center}
\includegraphics{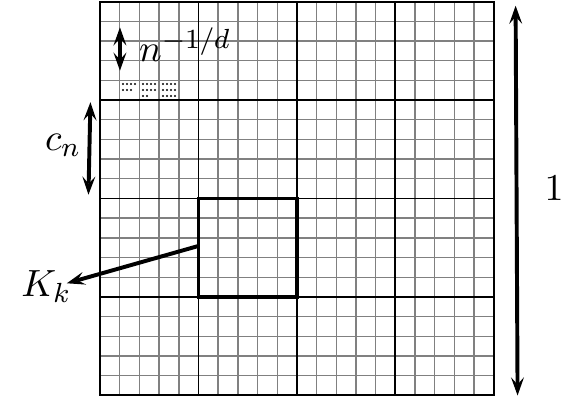}
\caption{Splitting the cube}\label{fig1}
\end{center}
\end{figure}

\begin{claim} We may place $n_k = \lfloor n \mu(K_{k})\rfloor \pm 0, 1$ points in $K_k$ (where $\lfloor \cdot \rfloor$ denotes the integer part),  with $\sum_k n_k = n$, and such that the resulting sets of points $\{x_i\}^n_{i=1}$ satisfy that the balls $B(x_i, \f{4\lambda}{n^{1/d}})$ are disjoint, for some $\lambda > 0$ independent of $n$. 
\end{claim}
It is possible to do so because the density of $\mu$ is bounded above and below on $K$, for a proof see e.g. \cite[Lemma 7.4 and below]{livre}. Then a Riemann sum argument, combined with the facts that $|n_k - n\mu(K_k)| \leq 1$ and $c_n \ll 1$, easily allows to show that the measure $\mu_n := \f{1}{n} \sum_{i=1}^n \delta_{x_i}$ converges weakly to $\mu$. We are then left with estimating $\f{1}{n^2} H_n(x_1, \dots, x_n)$ from above.

For each $0<\eta <1$, let us select $\chi_{\eta}$ a smooth function on $\R^d$, radial and such that $\chi_{\eta} (x) = 0$ if $|x| < \f{1}{2} \eta$ and $\chi_{\eta} = 1$ when $|x| \geq \eta$.
We may write
\begin{multline} \label{resolutionsingularites}
\f{1}{n^2} \iint_{\triangle^{c}} g(x-y) d\mu_n(x) d\mu_n(y) = \f{1}{n^2} \left( \iint_{\triangle^{c}} [(1- \chi_{\eta}) g](x-y) d\mu_n(x) d\mu_n(y) \right. \\ + \left. \iint [\chi_{\eta} g](x-y) d\mu_n(x) d\mu_n(y) \right).
\end{multline}
Since the function $\chi_{\eta} g$ is bounded and  continuous on the cube $K$ where $\mu_n$ and $\mu$ are supported, the last term in the right-hand side converges to $\iint [\chi_{\eta} g](x-y) d\mu(x) d\mu(y)$ by weak convergence of $\mu_n$ to $\mu $. We next turn to the first term in the right-hand side of (\ref{resolutionsingularites}), and show that there is, in fact, no problem near the diagonal because we have sufficient control on the accumulation of points.

Since $1 - \chi_{\eta}$ is bounded by $1$ and vanishes outside $B(0, \eta)$, we may write, by definition of $\mu_n$ and positivity of $g$ in $B(0,1)$ (which is true in all dimensions)~:
\begin{equation}\label{213}
\iint_{\triangle^{c}} [(1- \chi_{\eta}) g] (x-y) d\mu_n(x) d\mu_n(y) \leq \f{1}{n^2} \sum_{i \neq j,\ |x_i - x_j| < \eta} g(x_i - x_j).
\end{equation}

\begin{claim} For all $x, y \in B(x_i, \f{\lambda}{n^{1/d}}) \times B(x_j, \f{\lambda}{n^{1/d}})$, $i\neq j$,  we have $g(x_i - x_j) \leq g\(\hal  (x-y)\)$. 
\end{claim}
This is due to the fact that the balls $B(x_i, \f{4\lambda}{n^{1/d}})$ are disjoint from each other hence for $i\neq j$,  $|x_i-x_j|\ge \f{8\lambda}{n^{1/d}}$, which implies, by the triangle inequality, that if $x \in  B(x_i, \f{\lambda}{n^{1/d}})$ and $y \in B(x_j, \f{\lambda}{n^{1/d}})$, then $|x_i - x_j| \geq \hal  |x-y|$. But in all the cases we consider (cf. \eqref{coulombkernel1}),  $(x,y) \mapsto g(x-y)$ is a decreasing function of the distance between the two points, hence we may write\footnote{We will always use the notation $\fint_U f$ for the average of $f$ on $U$, that is $\f{1}{|U|} \int_U f$.} :
\begin{multline} \label{diracenboules}
g(x_i - x_j) \leq \fint_{B(x_i, \f{\lambda}{n^{1/d}})} \fint_{B(x_j, \f{\lambda}{n^{1/d}})} g\(\hal  (x-y)\) dx dy \\
\leq \frac{{C_d}}{\lambda^{2d}} n^2 \int_{B(x_i, \f{\lambda}{n^{1/d}})} \int_{B(x_j, \f{\lambda}{n^{1/d}})} g\(\hal  (x-y)\) dx dy
\end{multline}
where ${C_d}$ is a  constant depending only on the dimension $d $. Because the balls do not overlap, one may sum the inequalities  (\ref{diracenboules}) for $i \neq j$ to find, with \eqref{213},
\begin{multline} \label{diracenboules2}
\iint_{\triangle^{c}} (1- \chi_{\eta}) g(x-y) d\mu_n(x) d\mu_n(y) \leq \f{1}{n^2} \sum_{i \neq j,\ |x_i - x_j| < \eta} g(x_i - x_j) \\ 
\le\frac{ {C_d}}{\lambda^{2d}}\sum_{i \neq j,\ |x_i - x_j| < \eta}\int_{B(x_i, \f{\lambda}{n^{1/d}})} \int_{B(x_j, \f{\lambda}{n^{1/d}})} g\( \hal  (x-y)\) dx dy\\
\leq \frac{{C_d}}{\lambda^{2d}} \iint_{|x-y| <2 \eta} g\(\hal (x-y)\) dx dy,
\end{multline} for $n$ large enough.  The last term in (\ref{diracenboules2}) is $o(1)$ when $\eta \rightarrow 0$ because  we have $\iint g(x-y) dx dy < + \infty$. Combining with (\ref{resolutionsingularites}) we deduce 
\begin{equation}
\limsup_{n \to + \infty} \f{1}{n^2} \iint_{\triangle^c} g(x-y) d\mu_n(x) d\mu_n(y) \leq \iint_{\R^d\times \R^d} g(x-y) d\mu(x) d\mu(y) + o_{\eta}(1).
\end{equation}
Also $\int V d\mu_n \rightarrow \int V d\mu$, since $V$ is continuous, $\mu_n \rightharpoonup \mu$,  and $\mu_n$ and $\mu$ are supported in the same compact set, so in fact  we have established that $$\limsup_{n\to + \infty} \f{1}{n^2} H_n(\mu_n) \leq I(\mu) + o_{\eta}(1)$$ and letting $\eta \rightarrow 0$ finally gives us the $\Gamma$-$\limsup$ inequality, which concludes the proof.
\end{proof}

\begin{rem}\label{rem11} Again we have not really used the fact that $g$ is  a Coulombic kernel, rather we only used the fact that $g$ is monotone radial, positive in $B(0,1)$ and $\iint g(x-y)\, dx\,dy<\infty$. This shows that the  result  still holds for all such interaction kernels.
\end{rem}

\begin{rem} \label{rem12} To prove the $\Gamma$-liminf relation, we have only used that $V$ is l.s.c. and bounded below. To prove the $\Gamma$-limsup relation, we have assumed that $V$ is continuous for convenience. In fact the construction works for more general $V$'s, for example it suffices to assume that $V$ continuous on the set where it is finite. \end{rem}

We next derive the consequence of the $\Gamma$-convergence 
 Proposition \ref{gammaconvergenceHn}  given by Proposition    \ref{gammaconvmini}. 
 In order to do so, we must  prove the compactness of sequences with suitably bounded energy, as in Remark \ref{gcvcomplete}.   
  \begin{lemme}\label{lemcpthn}
 Assume that $V$ satisfies  \textbf{(A1)}--\textbf{(A2)}.
 Let $\{(x_1, \dots, x_n)\}_n$  be a sequence of configurations in $(\R^d)^n$, and let $\{\mu_n\}_n$ be the associated empirical measures (defined by $\mu_n= \frac{1}{n} \sum_{i=1}^n \delta_{x_i}$).
 Assume $\{\frac{1}{n^2} H_n(\mu_n) \}_n$ is a bounded sequence. Then  the sequence $\{\mu_n\}_n$ is tight, and as $n \to \infty$, it converges weakly in $\mc{P}(\R^d)$ (up to extraction of a subsequence) to some probability measure $\mu$.
 \end{lemme}
 \begin{proof}
 The proof is completely analogous to that of Lemma \ref{coerI}.  First, by assumption, there exists a constant $C_1$ independent of $n$ such that $H_n \le C_1 n^2$, and in view of \eqref{Hnpourempiric}--\eqref{troncaturediago1} we may write, for every $M>0$, 
 \begin{multline}\label{M1}
 C_1\ge  \iint (g(x-y) \wedge M) \, d\mu_n(x)\, d\mu_n(y) - \frac{M}{n} + \int V\, d\mu_n\\
 = \iint \left[g(x-y) \wedge M + \hal V(x) + \hal V(y) \right] \, d\mu_n(x) \, d\mu_n(y)- \frac{M}{n}.\end{multline}
 In view of \eqref{vgrand}, given any constant $C_2>0$, we may find $M$ large enough  and a compact set $K$  such that 
 $$ 
\min_{(K\times K)^c} \left[ g(x-y) \wedge M + \hal V(x) + \hal V(y) \right]>C_2.$$
The rest of the proof is virtually as in Lemma \ref{coerI}.

 \end{proof}

 To conclude, we will make the assumptions on $V$ that ensure both the $\Gamma$-convergence result Proposition \ref{gammaconvmini} and the existence result Theorem \ref{theoFrostman}.  Since we assumed for simplicity that $V$ is continuous and finite, it suffices to assume $\textbf{(A2)}$ to have $\textbf{(A1)} $ and $\textbf{(A3)}$. 
 
 With all the precedes, we may conclude with the following result, which goes back to \cite{choquet}.

 \begin{theo}[Convergence of minimizers and minima of $H_n$] \label{thcvmini} 
  Assume that $V$ is continuous  and satisfies \textbf{(A2)}.   
 Assume that for each $n$,  $\{(x_1, \dots, x_n)\}_n$ is a minimizer of $H_n$. Then, 
\be \label{cvdesmini}\frac{1}{n} \sum_{i=1}^n \delta_{x_i} \to \mu_0 \text{ in the weak sense of probability measures} 
\ee
where $\mu_0$ is the unique minimizer of $I$ as in Theorem \ref{theoFrostman}, and 
\be\label{cvdesmin}
\lim_{n\to + \infty} \frac{H_n(x_1, \dots , x_n) }{n^2} = I(\mu_0).\ee
\end{theo}

\begin{proof}
Applying the  $\Gamma$-limsup part of the definition of $\Gamma$-convergence,  for example to $\mu_0$,  ensures that  $\limsup_{n\to +\infty}\frac{1}{n^2} \min H_n $ is bounded above (by $I(\mu_0)$), hence in particular sequences of minimizers of $H_n$ satisfy  the assumptions of Lemma \ref{lemcpthn}.  It follows that, up to a subsequence, we have 
  $\frac{1}{n} \sum_{i=1}^n \delta_{x_i} \to \mu$ for some $\mu\in \mc{P}(\R^d)$.
  By Propositions \ref{gammaconvmini} and \ref{gammaconvergenceHn}, $\mu$ must minimize $I$, hence, in view of Theorem \ref{theoFrostman}, it must be equal to $\mu_0$. This implies that the convergence must hold along the whole sequence. We also get \eqref{cvdesmin} from Proposition \ref{gammaconvmini}.
\end{proof}
In the language of statistical mechanics or mean field theory, this  result gives the mean-field behavior or average behavior of ground states, and the functional $I$ is called the mean-field energy functional. 
It tells us that  points distribute themselves macroscopically  according to the probability law $\mu_0$ as their number tends to $\infty$, and we have the leading order asymptotic expansion of the ground state energy $$\min H_n \sim n^2 \min I.$$This is not very precise  in the sense that it tells us  nothing 
 about the precise patterns they follow. Understanding this is the object of the following chapters.

\section{Linking the equilibrium measure with the obstacle problem} 
In Section \ref{potentialtheoretic} we described the characterization of the equilibrium measure minimizing $I$ via tools of potential theory. In this section, we return to this question and connect it instead to a well-studied problem in the calculus of variations called the {\it obstacle problem}. This connection is not very much  emphasized in the literature. It is however mentioned in passing in \cite{safftotik} and used intensively in \cite{hedenmakarov2} with the point of view of \cite{sakai}. It allows us to use PDE theory results, such as methods based on the maximum principle methods and regularity theory to obtain
additional information on $\mu_0$. 
\subsection{Short presentation of the obstacle problem}
The obstacle problem is generally formulated over a bounded domain $\Omega \in \R^d$:  given  an $H^1(\Omega)$   function $\psi : \Omega \rightarrow \R$ (called the {\it obstacle}), which is  nonpositive on $\partial \Omega$,  find  the function that achieves 
\begin{equation}
\label{obstaclepb} \min \left\lbrace \int_{\Omega} |\nabla h|^2 , \ h \in H^1_0(\Omega), \ h \geq \psi \right\rbrace.
\end{equation}

For general background and motivation for this problem, see e.g. \cite{ks,friedman,ckinder}. 

Here the space $H^1_0(\Omega)$ is the  Sobolev space of trace-zero functions which is the completion of $C^1_c(\Omega)$ ($C^1$ functions with compact support in $\Omega$) under the $H^1$-Sobolev norm $||h||_{H^1} = ||h||_{L^2} + ||\nabla h||_{L^2}$. The zero trace condition $h \in H^1_0(\Omega)$ may be replaced by different boundary conditions, e.g. a translation $h \in H^1_0(\Omega) + f$, where $f$ is a given function. Note that the minimization problem \eqref{obstaclepb} is a convex minimization problem under a convex constraint, hence it has at most one  minimizer (it is not too hard to show that the minimum is achieved, hence there actually is a unique minimizer).

 An admissible function for \eqref{obstaclepb} has two options at each point : to touch the obstacle or not (and typically uses both possibilities).  If $h$ is the optimizer, the set $\{x\in \Omega| h(x)=\psi(x) \ q.e.\}$ is closed and called the {\it coincidence set} or the {\it contact set}.  It is unknown (part of the problem), and its boundary is called a ``free boundary." The obstacle problem thus belongs to the class of so-called {\it free-boundary problems}, cf. \cite{friedman}.

Trying to compute the Euler-Lagrange associated to this problem by perturbing $h$ by a small function, one is led to two possibilities depending on whether $h = \psi$ or $h > \psi$. 
In a region  where $h > \psi$, one can do infinitesimal variations of $h$ of the form $(1-t)h + tv$ with $v$, say,  smooth  (this still gives an admissible function, i.e. lying above the obstacle, as soon as $t$ is small enough) which shows that $\Delta h =0$  there (since the “functional” derivative of the Dirichlet energy is the Laplacian). 
In the set where  $h=\psi$,  only variations of the same form $(1-t)h + tv$  but with $v\ge \psi$ (equivalent to $v\ge h$ there!) and $t\ge 0$ provide  admissible functions,  and this  only leads to an inequality $-\Delta h\ge 0$ there. These two pieces of information can be grouped in the following more compact form:
\begin{equation}\label{varineq}
\textrm{for all } v \textrm{ in }  H^1_0 \textrm{ such that } v \geq \psi \textrm{ q.e., }  \int_{\Omega} \nabla h \cdot \nabla (v-h) \geq 0.
\end{equation}
This relation is called a {\it variational inequality}, and 
 it uniquely characterizes the solution to \eqref{obstaclepb}, in particular the coincidence set is completely determined as part of the solution.

In Fig. \ref{fig2} below, we describe a few instances of solutions to   one-dimensional obstacle problems, and in Fig. \ref{fig3} to higher dimensional obstacle problems. 

\begin{figure}[h!]
\begin{center}
\includegraphics{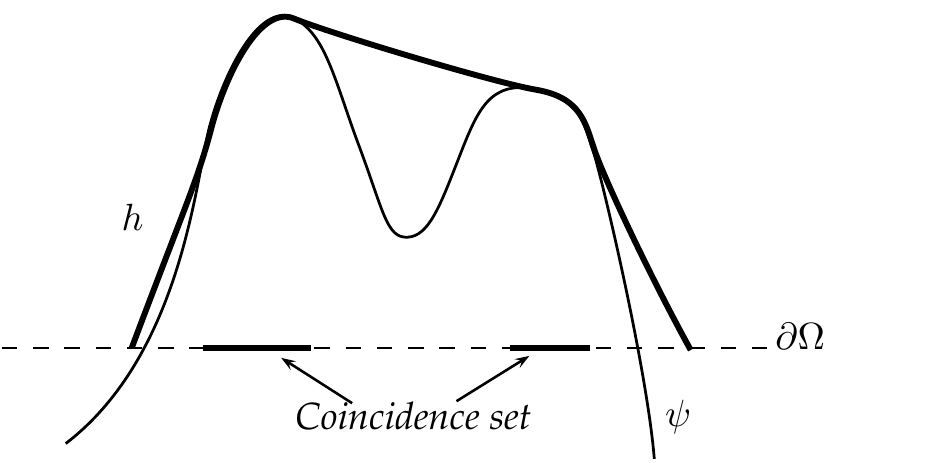}
\includegraphics{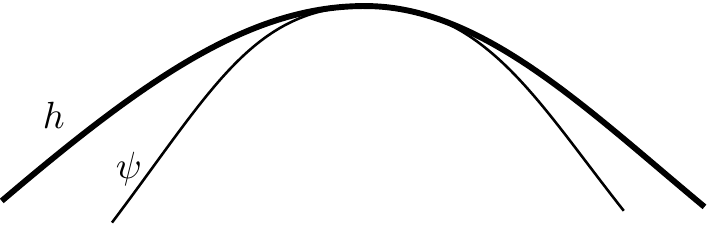}
\caption{The coincidence set for a one-dimensional obstacle problem}
\label{fig2}
\end{center} 
\end{figure}

\begin{figure}[h!]
\begin{center}
\includegraphics[scale=1]{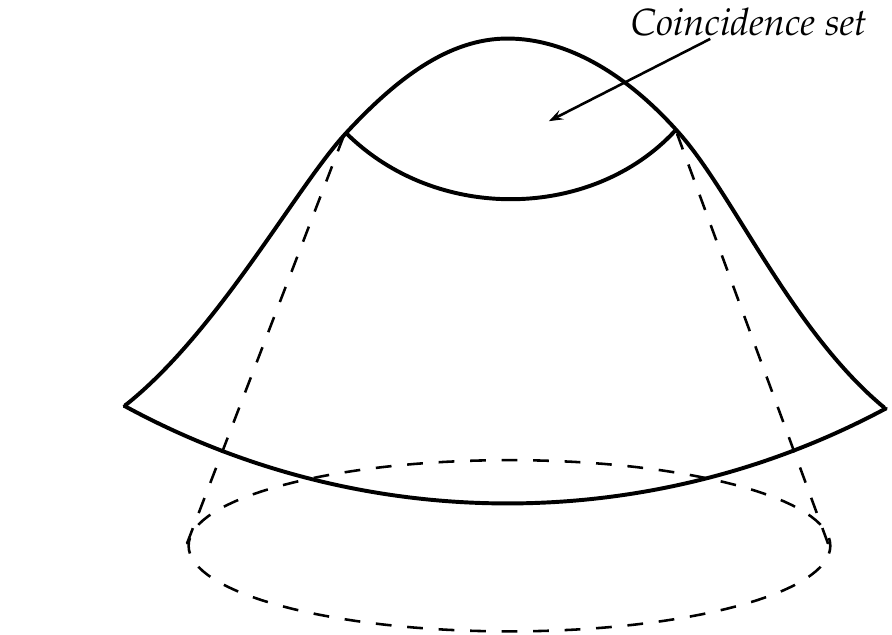}
\end{center}
\caption{A higher-dimensional obstacle problem}\label{fig3}
\end{figure}

The regularity theory of the solutions to obstacle problems and of their coincidence sets has been developed for many years,   culminating with the work of Caffarelli   (for a review see \cite{caff}). This sophisticated PDE theory shows, for example, that the solution $h$ is as regular as $\psi$ up to $C^{1,1}$ \cite{frehse}. The boundary $\partial \Sigma$ of the coincidence set is $C^{1,\alpha}$ except for cusps \cite{caff}. These are points of $\partial \Sigma$ at  which, locally, the coincidence set  can fit  in the region between two parallel planes  separated by an arbitrarily small distance (the smallness of the neighborhood depends of course on this desired distance).
Fig. \ref{fig4} gives examples of coincidence sets, one regular one, and one with cusps.

\begin{figure}[h!]
\begin{center}
\includegraphics[scale=1]{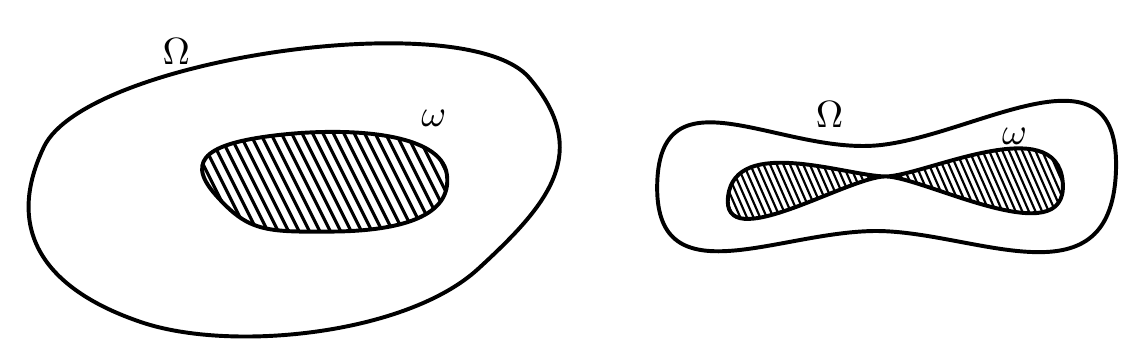}
\caption{Examples of coincidence sets}\label{fig4}
\end{center}
\end{figure}

Moreover, if $\psi$ is $C^{1,1}$, since $\nabla h$ is continuous, the graph of $h$ must leave the coincidence set tangentially. This  formally leads to the following system of equations, where $\omega$ denotes the coincidence set~:
$$\left\{\begin{array}{ll}
-\Delta h =0 & \text{in} \ \Omega\backslash \omega
\\
h=\psi & \text{in} \ \omega\\
\frac{\partial h}{\partial \nu} = \frac{\partial \psi}{\partial \nu} & \text{on} \ \partial \omega\\
h=0 & \text{on} \ \Omega.\end{array}\right.$$
This relation cannot be made rigorous in all cases, because $\omega$ is not an open domain, however it gives the right intuition and is correct when $\omega $ is nice enough. Note that on the boundary of $\Omega \backslash \omega$ we must have a Dirichlet  condition $h=\psi$, together with a Neumann  condition $\frac{\partial h}{\partial \nu} = \frac{\partial \psi}{\partial \nu}$. These two boundary conditions make what is called an {\it overdetermined problem} and this 
overdetermination explains why there is only one possible coincidence set.

\subsection{Connection between the two problems}
The problem we examined, that of the minimization of $I$, is phrased in the whole space, and not in  a bounded domain. While the minimization problem \eqref{obstaclepb} may not have  a meaning over all $\R^d$ (because the integral might not converge), the corresponding variational inequality \eqref{varineq} can still be given a meaning over $\R^d$ as follows~:
given $\psi \in H^1_{\mathrm{loc}}(\R^d)$ solve for $h$ such that 
\begin{equation}\label{obspbrd}
h\ge \psi \ \text{q.e.  and \ } \forall  v\in \mathcal{K}, \quad  \int_{\R^d} \nab h \cdot \nab (v-h) \ge 0\end{equation}where 
 $$\mathcal{K}= 
\left\{ H^1_{\mathrm{loc}}(\R^d) \ \textrm{such that} \   v-h \textrm{ has bounded support and} \  v \ge \psi  \  \text{q.e.}\right\}.$$
Solving this is in fact equivalent to the statement that  for every $R>0$, $h$ is the unique solution to 
$$\min\left\{\int_{B_R} |\nab v|^2, v\in H^1(B_R), v -h\in H^1_0(B_R), v\ge \psi \ \textrm{in} \ B_R\right\},$$
which replaces \eqref{obstaclepb}.
The problem \eqref{obspbrd} is easily seen to have a unique solution~: if there are two solutions $h_1$ and $h_2$  it suffices to apply \eqref{obspbrd} for $h_1$, with $h_2$ as a test-function, and then reverse the roles of the two and add the two relations to obtain $h_1=h_2$.

Let us now  compare  the two problems side by side~:
\begin{description}
\item[Equilibrium measure.]
 $\mu_0$ is characterized by the relations 
 \begin{equation} \left\lbrace \begin{array}{cl} h^{\mu_0} + \f{V}{2} \geq c & \mbox{quasi everywhere} \vspace{2mm} \\ 
h^{\mu_0} + \f{V}{2} = c & \textrm{q.e. in  the support of } \ \mu_0.
\end{array} \right.
\end{equation}

\item[Obstacle problem.]
\begin{equation}
\left\lbrace \begin{array}{cl} 
h \geq \psi & \mbox{q.e.} \\
h = \psi & \mbox{q.e.  in the coincidence set}
\end{array}\right.
\end{equation}
\end{description}

It is then not surprising to expect a correspondence between the two settings, once one chooses the obstacle to be  $\psi = c - \f{V}{2}$. 
\begin{prop}[Equivalence between the minimization of $I$ and the obstacle problem]\label{proequivpb}
\mbox{}
Assume $d\ge 2$, $V$ is continuous and satisfies \textbf{(A2)}. If $\mu_0$ is the equilibrium measure associated to the potential $V$ as in Theorem \ref{theoFrostman}, then its potential $h^{\mu_0}$, as defined in \eqref{defhmu0},  is the unique solution to the obstacle problem with obstacle $\psi = c - \f{V}{2}$ in the sense of \eqref{obspbrd}. 
If in addition  $V\in C^{1,1}$ then $\mu_0=( \hal \Delta V) \indic_{\omega}$.
\end{prop}
Note that the converse might not be true, because a solution of the obstacle problem can fail to provide a {\it probability} measure, however it does in general when shifting $c$ appropriately.

When one works on a bounded domain, this result can be obtained by observing that the problem of minimizing $I$ and that of minimizing \eqref{obstaclepb} are essentially convex duals of each other (see \cite{brezis,breziss}). When working in an infinite domain, the correspondence is probably folklore and could also be worked out by convex duality,  but we were not able to find it completely written in the literature, except for \cite{hedenmakarov2} who follow a slightly different formulation. Here, for the proof, we follow the approach of  \cite{asz} where the result is established in dimension 2 for the particular case of $V$ quadratic (but with more general constraints), the adaptation to any dimension and to general $V$'s is not difficult.

\begin{proof}
{\bf Step 1.} We show that $\nab h^{\mu_0}$ is in $L^2_{\mathrm{loc}} (\R^d, \R^d)$. It is a consequence of the fact that $I(\mu_0)<\infty $ hence, in view of the assumptions on $V$, $\iint g(x-y) \, d\mu_0(x) \, d\mu_0(y)$.
In the case $d\ge 3$, it can be proven that this implies  $\nab h^{\mu_0}\in L^2(\R^d)$  for example by using the Fourier transform, first approximating $\mu_0$ by smooth measures, and combining Corollary 5.10 and Theorem 7.9 in \cite{liebloss}.


In the case $d=2$, we need to consider a reference probability  measure $\bar{\mu}$ for which $h^{\bar\mu}$ is $C^1_{\mathrm{loc}}(\R^2)$. It suffices to consider for example $\bar{\mu}= \frac{1}{\pi } \indic_{B_1}$, the circle law, for which $h^{\bar\mu}$ is radial and can be computed explicitly.
Then, let us  consider $\ro = \mu_0 - \bar{\mu}$. Using the fact that $\int \, d\rho=0$, $\rho$ is compactly supported, and $\iint g(x-y) \, d\mu(x)\, d\mu(y)<\infty$ holds for both $\mu=\mu_0$ and $\mu= \bar\mu$,   we have the following statement
\begin{equation}\label{e.hro}
\iint_{\R^2\times\R^2} - \log |x-y|\, d\rho(x)\, d\rho(y)= \frac{1}{2\pi}\int_{\R^2} |\nabla h^\rho(x) |^2\,dx,
\end{equation} where $h^{\rho}(x)= \int g(x-y) \, d\ro(y)$. 
Indeed, in  the proof of \cite[Lemma 1.8]{safftotik}  it is shown that 
\begin{equation*}
\int_{\R^2} \int_{\R^2} - \log |x-y| \, d\rho(x) \, d\rho(y)= \frac{1}{2\pi} \int_{\R^2} \left(\int_{\R^2} \frac{1}{|x-y|} \, d \rho(y)\right)^2 \, dx,
\end{equation*}
and 
\eqref{e.hro} follows, since $\int_{\R^2} \frac{(x-y) \,d \rho(y)}{|x-y|^2} = - \nab h^{\rho}(x)$ in the distributional sense. 
This shows that $\nab h^{\rho} \in L^2(\R^2)$ and thus, since $ h^{\bar\mu}$ is $C^1_{\mathrm{loc}}(\R^2)$, we deduce that $\nab h^{\mu_0}= \nab h^{\rho}+ \nab h^{\bar \mu}$ is also in $L^2_{\loc} (\R^2, \R^2)$, as desired.
\medskip

\noindent {\bf Step 2.}
Let $v$ be admissible in \eqref{obspbrd}, i.e belong to $\mathcal{K}$, and set $\vp= v- h^{\mu_0}$.  If $\vp$ is smooth and compactly supported, then 
\be \label{vi1}
\int_{\R^d}  \nabla h^{\mu_0} \cdot \nabla (v - h^{\mu_0}) = c_d \int_{\R^d}   \vp \, d \mu_0\ge 0.
\ee 
Indeed, by \eqref{EulerLagrange}, we know that $h^{\mu_0} = \psi$ q.e. in the support of $\mu_0$ and by assumption $v \geq \psi$ q.e. in $\R^d$. Hence $\vp$ is q.e. nonnegative on the support of $\mu_0$  and the inequality \eqref{vi1} follows, since $\mu_0$ does not charge sets of zero  capacity.
To obtain \eqref{vi1} for any   $v\in \mathcal{K}$,  it suffices to show that the subset of $\mathcal{K}$ consisting of $v$'s for which $v-h^{\mu_0}$ is smooth and compactly supported is dense in $\mathcal{K}$  for the topology of $H^1$.   Fix some $v$ in the admissible set and $R>1$ such that $v-h^{\mu_0}$ is supported in $B_{R/2}$. Let $\eta_\eps$ be a standard mollifier and $\chi_R$ a smooth function supported in $B_{2R}$ with $0\le \chi_R \le 1$ and $\chi_R \equiv 1$ in $B_R$. One may check that  
$$v_{\eps,  \delta}= h^{\mu_0}+ (v-h^{\mu_0}) * \eta_\eps  \delta \chi_R$$ satisfies  that $v_{\eps,\delta}- h^{\mu_0}$ is smooth and  approximates  $v$ arbitrarily well in $H^1$ when $\delta $ is small enough, and is $\ge \psi$ when $\eps $ is chosen small enough relative to $\delta$.
This concludes the proof of \eqref{vi1}.
\medskip

\noindent{\bf Step 3.} We prove the statements about $\mu_0$.
First, since the coincidence set $\omega$ is closed, its complement is open, and the function $h^{\mu_0}$ is harmonic on that set. One can note also that in view of   
 \eqref{EulerLagrange} and the definition of the coincidence set $\omega$, the support of $\mu_0$ is included in $\omega$ up to a set of capacity $0$.

If we assume that  $V\in C^{1,1}_{\mathrm{loc}}$, then by Frehse's regularity theorem mentioned above, it follows that $h^{\mu_0}$ is also $C^{1,1}_{\mathrm{loc}}$. In particular $h^{\mu_0}$ is continuous, and so is $V$, so the relations \eqref{EulerLagrange} hold pointwise and not only q.e. This means that we have 
\be \label{eag} h^{\mu_0} + \hal V= c \ \text{ on} \  \omega\ee  and $\supp(\mu_0)\subset \omega$.
Also $C^{1,1}_{\mathrm{loc}}= W^{2,\infty}_{\mathrm{loc}}$ hence $\Delta h^{\mu_0}$ and $\Delta V$  both make sense as  $L^\infty_{\mathrm{loc}}$ functions, and it suffices to determine $\mu_0$ up to sets of measure $0$. We already know that $\mu_0=0$ in the complement of $\omega$ since $h^{\mu_0}$ is harmonic there, and it suffices to determine it in $\interieur{\omega}$.  But taking the Laplacian on both sides  of \eqref{eag},
since $\mu_0= - \frac{1}{c_d} \Delta h^{\mu_0}$,  one finds
$$\mu_0=  \frac{1}{ 2 c_d} \Delta V \ \text{in} \ \interieur{\omega},$$
and the results follows.

\end{proof}

By definition of $\zeta$ \eqref{defzeta}, we have that 
\be \label{zerozeta}
\{x\in \R^d | \zeta(x)=0\}= \omega. \ee
Since $\mu_0$ is a compactly supported probability measure, we have that $h^{\mu_0} = \int g(x-y)\, d\mu_0(y)$ asymptotically behaves like $g(x)$ as $|x|\to \infty$. Since $h^{\mu_0}+ \hal V = c $ q.e. in $\omega$ and since \textbf{(A2)} holds, it follows that $\omega $ must be a bounded, hence compact, set.

We have seen that $\omega$ contains, but is not always equal to,  the support of $\mu_0$.  The latter is called the {\it droplet}   in \cite{hedenmakarov2}, where similar results to this proposition are established. There, it is also discussed how $\supp(\mu_0) $ differs from $\omega$ (they are equal except at ``shallow points", cf. definition there).

\begin{rem}[Note on dimension one] For $d =1$ and  $g = - \log| \cdot |$, as seen before $h^{\mu_0}$ solves 
$$
- \Delta^{1/2} h^{\mu_0} = c_1 \mu_0,$$
and the equivalent of the obstacle problem is instead  a fractional obstacle problem for which a good theory also exists  \cite{caffss}.  One could also write the analogue of Proposition \ref{proequivpb}.
\end{rem}

We have seen how the correspondence between the minimization of $I$ and the obstacle problem thus allows, via the regularity theory of the obstacle problem, to identify the equilibrium measure in terms of $V$ when the former  is regular enough. The known techniques on the obstacle problem \cite{caff} also allow for example  to analyze the rate at which the solution leaves the obstacle (they  say it is subquadratic), which gives us information on the size of the function $\zeta$, defined in \eqref{defzeta}.

\section{Large deviations for the Coulomb gas with temperature} \label{LDPsection}
At this point, we know the $\Gamma$-convergence of $\frac{1}{n^2} H_n$ and its consequence, Theorem~\ref{thcvmini} for ground states  of the Coulomb gas. In this section, we turn for the first time to states with temperature and  derive rather easy consequences of the previous sections on the Gibbs measure, which we recall is defined by 
\be \label{gibbsagain}
d\P_{n, \beta}(x_1, \dots, x_n) := \f{1}{Z_{n, \beta}} e^{-\f{\beta}{2} H_n(x_1, \dots, x_n)} dx_1 \dots dx_n
\ee
with \be \label{defz}
Z_{n, \beta} = \int e^{-\f{\beta}{2} H_n(x_1, \dots, x_n)} dx_1 \dots dx_n.
\ee
Pushing $\mathbb{P}_{n, \beta}$ forward by the map $(x_1, \dots, x_n) \mapsto \f{1}{n} \sum_{i=1}^n \delta_{x_i}$, we may view it as a probability measure on $\mc{P}(\R^d)$, called the Gibbs measure at (inverse) temperature $\beta$. Note that considering $\beta$ that depends on $n$ can correspond to other temperature regimes (very high or very low temperatures) and is also interesting.

Formally, taking $\beta = + \infty$ above reduces to the study to the minimizers of $H_n$, whose behavior when $n \rightarrow + \infty$ we already established  (weak convergence of the empirical measure to the equilibrium measure $\mu_0$). For $\beta < + \infty$ we will see that the behavior of a “typical” configuration" $(x_1, \dots, x_n)$ under the measure $\P_{n, \beta}$ is not very different, and we can even characterize the probability of observing a “non-typical” configuration. 
The sense given to “typical” and “non-typical” will be that  of the theory of large deviations, which we first briefly introduce. For more reference, one can see the textbooks \cite{denholl,deuschel,dz}.

\begin{defini}[Rate function] \label{ratefun} Let $X$ be a metric space (or  a topological space). A rate function is a l.s.c. function $I : X \rightarrow [0, + \infty]$, it is called a “good rate function” if its sub-level sets $\{x, I(x) \leq \alpha\}$ are compact (see Remark \ref{gconvcom}). 
\end{defini}

\begin{defini}[Large deviations]\label{definiLDP} Let $\{P_n\}_n$ be a sequence of Borel probability measures on $X$ and $\{a_n\}_n$ a sequence of positive real numbers diverging to $+ \infty$. Let also $I$ be a (good) rate function on $X$. The sequence $\{P_n\}_n$ is said to satisfy a large deviation principle (LDP) at speed $a_n$ with (good) rate function $I$ if for every Borel set $E \subset X$ the following inequalities hold : 
\be 
- \inf_{\overset{\circ}{E}}I \leq \underset{n \to + \infty}{\liminf} \f{1}{a_n} \log P_n(E) \leq  \underset{n \to + \infty}{\limsup} \f{1}{a_n} \log P_n(E) \leq - \inf_{\bar{E}} I
\ee
where $\overset{\circ}{E}$ (resp. $\bar{E}$) denotes the interior (resp. the closure) of $E$ for the topology of $X$.
\end{defini}
Formally, it means that $P_n(E)$  should behave roughly like $e^{-a_n \inf_{E} I}$. The rate function $I$ is the rate of exponential decay of the probability of rare events, and the events with larger probability are the ones on which $I$ is smaller.

\begin{rem} \label{rem13} At first sight, Definition \ref{definiLDP} looks very close to the $\Gamma$-convergence 
\begin{displaymath}
\frac{\log p_n}{a_n} \overset{\Gamma}{\rightarrow} - I
\end{displaymath} where $p_n$ is the density of the measure $P_n$. However, in general there is no equivalence between the two concepts. For example, in order to estimate the quantity
\be
\log P_n(E) = \log \int_E p_n(x) dx
\ee
it is not sufficient to know the asymptotics of $p_n$, one really also needs to know the size of the volume element $\int_E dx$, which plays a large role in large deviations and usually comes up as an entropy term. There are however some rigorous connections between $\Gamma$-convergence and LDP (see\cite{mar}).
\end{rem}

We will need an additional assumption on $V$~:
\begin{description}
\item[(A4)] There exists $\alpha > 0$ such that
\be \label{conditionintegrabilite}
\int_{\R^d} e^{-\alpha V(x)} dx < + \infty.
\ee
\end{description}

We will also keep the other assumptions that $V$ is continuous and \textbf{(A2)} holds, which ensure the existence of the equilibrium measure  $\mu_0$, and the $\Gamma$-convergence of $\frac{H_n}{n^2}$ to $I$.
 In dimension $d = 2$, the growth assumption \textbf{(A2)}  $\f{V}{2} - \log \rightarrow + \infty$ ensures that the condition \textbf{(A4)} is also satisfied, however in dimension $d \geq 3$ we need to assume (\ref{conditionintegrabilite}), which is a slight strengthening of \textbf{(A2)}, in order to avoid very slow divergence of $V$ such as $V(x) \sim \log \log x$ at infinity.
Note in particular that  \textbf{(A4)} ensures that the integral in \eqref{defz} is convergent, hence
 $Z_{n, \beta}$ well-defined, as soon as $n$ is large enough.

We may now state the LDP for the Gibbs measure associated to the Coulomb gas Hamiltonian. This result is due to \cite{hiaipetz} (in dimension 2),  \cite{bg} (in dimension $1$) and\cite{bz} (in dimension 2) for the particular case of a quadratic potential (and $\beta = 2$), see also \cite{berman} for results in a more general (still determinantal) setting of multidimensional complex manifolds. \cite{cgz} recently treated more general singular $g$'s and $V$'s.    We present here the proof for the Coulomb gas in any dimension and general potential, which is not more difficult. 
\begin{theo}[Large deviations principle for the Coulomb gas at speed $n^2$] \label{LDP}
\mbox{}
Assume $V$ is continuous and satisfies \textbf{(A2)} and \textbf{(A4)}. 
For any $\beta > 0$, the sequence $\{\mathbb{P}_{n, \beta}\}_n$ of probability measures on $\mc{P}(\R^d)$ satisfies a large deviations principle at speed $n^2$ with good rate function $\f{\beta}{2}\hat{I}$ where $\hat{I} = I - \min_{\mc{P}(\R^d)}  I= I - I(\mu_0)$. Moreover 
\begin{equation}
\lim_{n\to + \infty} \f{1}{n^2} \log Z_{n, \beta} = - \f{\beta}{2} I(\mu_0) = - \f{\beta}{2} \min_{\mc{P}(\R^d)} I.
\end{equation}
\end{theo} Here of course, the underlying topology is still that of weak convergence on $\mc{P}(\R^d)$.

The heuristic reading of the LDP is that 
\be  \label{heurisldp}
\P_{n, \beta}(E) \approx e^{-\f{\beta}{2} n^2 (\min_{E} I - \min I)}. \ee 
As a consequence, the only likely configurations of points (under $\P_{n, \beta}$) are those for  which the empirical measures $\mu_n = \frac{1}{n} \sum_{i=1}^n \delta_{x_i}$ converge to $\mu=\mu_0$, for otherwise $I(\mu) > I(\mu_0)$ by uniqueness of the minimizer of $I$, and the probability decreases exponentially fast according to \eqref{heurisldp}. Thus, $\mu_0$ is not only the limiting distribution  of  minimizers of $H_n$, but also the limiting distribution for all ``typical" (or likely) configurations, at finite temperature. 
 Moreover, we can estimate the probability under $\P_{n,\beta}$ of the “non-typical” configurations and see that it has exponential decay at speed $n^2$.
Recall that the cases of the classic random  matrix ensembles GOE, GUE and Ginibre correspond respectively to $d=1$, $\beta =1 ,2$ and $V(x)=x^2$, and $d=2$, $\beta=2$, and $V(x)=|x|^2$. 
The corresponding equilibrium measures were given in Example 2 above. As a consequence of Theorem \ref{LDP}, we have  a proof -- modulo the fact that their laws are given by \eqref{loigue}, \eqref{loiGinibre} -- that the distribution of eigenvalues (more precisely  the spectral or empirical measure) 
 has to follow  Wigner's semi-circle law  $\mu_0 =\frac{1}{2\pi} \sqrt{4-x^2} \indic_{|x|<2}$ for  the GUE and GOE, and the circle law $\mu_0=\frac{1}{\pi}\indic_{B_1} $ for the Ginibre ensemble, in the sense of the LDP (which is in fact stronger than just establishing these laws).  These are the cases originally treated in \cite{hiaipetz,bz,bg}.

In addition, knowing the partition function $Z_{n,\beta}$ is important because it gives access to many physical quantities associated to the system (for e.g. by differentiating $Z_{n,\beta}$ with respect to $\beta$ yields  the average energy, etc), see statistical mechanics textbooks such as \cite{huang}.
In particular $-\frac{2}{\beta}{\log Z_{n,\beta}}$ in our context is physically the {\it free energy} of the system, and 
the existence of a limit for $\frac{1}{n}\log Z_{n,\beta}$ (or for the free energy per particle)  is called the {\it existence of thermodynamic limit}. 
 For the one-dimensional log gas, the value of $Z_{n, \beta}$ is known explicitly for all $\beta>0$ when $V(x)=x^2$ via 
 the exact computation  of  the integral in \eqref{defz}, which uses so-called {\it Selberg integrals} (see e.g. \cite{mehta}). For more general $V$'s an expansion in $n$ to any order is also known \cite{borotguionnet}. In dimension $2$ however, no equivalent  of the 
 Selberg integral  exists and the exact value of $Z_{n,\beta}$ is only known for the Ginibre case  $\beta=2$ and $V(x)=|x|^2$ (there are a few other exceptions). In dimension $d \ge 3$, we know of no such explicit computation or expansion. 

\begin{proof}[Proof of the theorem]
The intuition behind the LDP  might be that since the sequence $\f{1}{n^2}H_n$ $\Gamma$-converges to $I$, we should have 
\begin{displaymath}
\P_{n, \beta} \approx \f{1}{Z_n} e^{-\f{\beta}{2}n^2I}
\end{displaymath} however such an approach is too naive to work directly, for the  reasons  explained in Remark \ref{rem13}. We will use the $\Gamma$-convergence result in a more precise way, also estimating the size of the appropriate sets in configuration space.
\medskip

\noindent
{\bf Step 1.} We first prove the large deviations upper bound, that is : 
\be \label{limsupLDP}
\limsup_{n\to + \infty}\f{\log \P_{n,\beta} (E)}{n^2} \leq -\inf_{\bar{E}} \hat{I}
\ee
up to an estimate on $\log Z_{n,\beta}$, then we will turn to the proof of the lower bound and get in passing the missing estimate on $\log Z_{n,\beta}$.

Let us define $\widehat{H}_n(x_1, \dots, x_n) = H_n(x_1, \dots, x_n) - \sqrt{n} \sum_{i=1}^n V(x_i)$. This amounts to changing $V$ to $(1- \f{1}{\sqrt{n}})V$ in the definition of $H_n$. Of course, $\f{\widehat{H}_n}{n^2}$ still $\Gamma$-converges to $I$, by the same proof as Proposition \ref{gammaconvergenceHn}.  We want to show that 
\be \label{gammaconvsurE}
\liminf_{n \to + \infty} \inf_{E} \f{ \widehat{H}_n}{n^2}  \geq \inf_{\bar{E}} I.
\ee
We may assume that the left-hand side is finite, otherwise there is nothing to prove. Upon passing to a subsequence, suppose that $\mu_n = \frac{1}{n}\sum_{i=1}^n \delta_{x_i}$ is a minimizer (or almost minimizer) of $\widehat{H_n}$ on $E$, more precisely satisfies
\begin{displaymath}
\underset{n \to + \infty}{\lim} \f{ \widehat{H}_n(\mu_n)}{n^2} = \liminf_{n \to + \infty} \inf_{E} \f{ \widehat{H}_n}{n^2}. 
\end{displaymath}
By applying Lemma \ref{lemcpthn} (applied to $\widehat{H_n}$ instead of $H_n$ but this induces no change), 
the sequence $\{\mu_n\}_n$ is tight and  has a subsequence which converges to some $\mu$ in the sense of weak convergence in $\mc{P}(\R^d)$. 
Then by the $\Gamma$-convergence of $\f{\widehat{H}_n}{n^2}$ to $I$, it follows that 
\begin{displaymath}
\liminf_{n \to + \infty} \f{ \widehat{H}_n(\mu_n)}{n^2}  \geq I(\mu) \geq \inf_{\bar{E}} I
\end{displaymath}
and (\ref{gammaconvsurE}) is proven. In particular, (\ref{gammaconvsurE}) implies that $\f{1}{n^2} \widehat{H}_n  \geq \inf_{\bar{E}} I + o(1)$ on $E$, the $o(1)$ being uniform on $E$. Inserting this inequality into the definition of $\P_{n, \beta}$, one gets 
\begin{multline}
\P_{n, \beta}(E) \leq \f{1}{Z_{n, \beta}} \int e^{-\f{\beta}{2} n^2 \inf_{\bar{E}} I + o(n^2)} e^{-\f{\beta}{2} \sqrt{n} \sum_{i=1}^n V(x_i)} dx_1 \dots dx_n
\\ =  \f{1}{Z_{n, \beta}} e^{-\f{\beta}{2} n^2( \inf_{\bar{E}} I + o(1))} \left(\int_{\R^d} e^{-\f{\beta}{2} \sqrt{n} V(x)} dx\right)^n.
\end{multline}
The last integral in the right-hand side is  bounded by a constant for $n$ large enough by the assumption \textbf{(A4)}. Hence, taking the logarithm of both sides, we find 
\be \label{LDPestimee1}
\log \P_{n, \beta}(E) \leq - \log Z_{n, \beta} - \f{\beta}{2} n^2 \inf_{\bar{E}}{I} + o(n^2) + O(n)\quad \text{as} \ n\to +\infty,
\ee
and the $\limsup$ of  (\ref{LDPestimee1}) when $n \rightarrow + \infty$ gives, for each $E \subset \mc{P}(\R^d)$ 
\be \label{limsup1}
\limsup_{n \to + \infty} \f{1}{n^2} \log \P_{n, \beta}(E) \leq \limsup_{n\to +\infty}\Big( - \f{1}{n^2} \log Z_{n, \beta}\Big) - \f{\beta}{2} \inf_{\bar{E}} I.
\ee
This is not exactly the large deviations upper bound relation, because we cannot yet bound the term $- \f{1}{n^2}\log Z_{n, \beta} $. However, by taking $E$ to be the whole space $\mc{P}(\R^d)$ in (\ref{LDPestimee1}), we already get that 
\be \label{estimeeZn1}
\f{1}{n^2} \log Z_{n, \beta} \leq - \f{\beta}{2} I(\mu_0) + o(1)\quad \text{as } \ n\to +\infty.
\ee

\noindent
{\bf Step 2.} We prove the large deviations lower bound.
Let $\mu$ be in the interior $\interieur{E}$. By the $\Gamma$-$\limsup$ result of Proposition 
\ref{gammaconvergenceHn}, there exists a sequence of $n$-tuples $(x_1, \dots, x_n)$ such that the empirical measures $\f{1}{n} \sum_{i=1}^n \delta_{x_i}$ converges weakly to $\mu$ and 
\be
\limsup_{n\to + \infty} \f{H_n}{n^2}(x_1, \dots, x_n) \leq I(\mu).
\ee
Moreover, in the construction for the $\Gamma$-limsup made in Step 2 of the proof of Proposition \ref{gammaconvergenceHn}, the balls  $B(x_i, \f{4\lambda}{n^{1/d}})$ were disjoint. Consequently, if $(y_1, \dots, y_n)$ is such that for each $i = 1 \dots n$, the point $y_i$ is in the (Euclidean) ball $B(x_i, \f{\lambda}{n^{1/d}})$, then the same will hold  with $2\lambda$ instead of $4\lambda$, and one can check that the proof carries through in the same way, yielding that  the empirical measure $\f{1}{n} \sum_{i=1}^n \delta_{y_i}$ converges weakly to $\mu$ and we have 
\be \label{majorationconstruction1}
\limsup_{n + \infty} \f{H_n}{n^2}(y_1, \dots, y_n) \leq I(\mu).
\ee
This is helpful because it shows that there are configurations whose Hamiltonian is not too large (the upper bound (\ref{majorationconstruction1}) holds) and, as shown below, that these configurations occupy enough volume in phase space to contribute significantly to the partition function $Z_{n,\beta}$. 
Denote by $\mc{V}$ the set 
\be \label{defmcV}
\mc{V} := \bigcup_{\sigma \in \mathbb{S}_n} \prod_{i=1}^n B\Big(x_{\sigma(i)}, \f{\lambda}{n^{1/d}}\Big)
\ee
where $\mathbb{S}_n$ is the set of all permutations of $\{1, \cdots, n\}$. Since the Hamiltonian $H_n$ is symmetric, it is invariant under the action of permutation of points in phase space. Thus, in view of (\ref{majorationconstruction1}), we have
\be \label{majorationconstruction2}
\limsup_{n + \infty} \max_{\mc{V}} \f{H_n}{n^2} \leq I(\mu).
\ee
Moreover,  for $n$ large enough the measures $\f{1}{n} \sum_{i=1}^n \delta_{y_i}$ (with $(y_1, \dots, y_n) \in \mc{V}$) are in $E$,  because  $\f{1}{n} \sum_{i=1}^n \delta_{y_i}$ converges weakly to $\mu\in \interieur{E}$. Therefore, we may write
\begin{multline} \label{liminf1}
\P_{n, \beta}(E) \geq \f{1}{Z_{n, \beta}} \int_{\mc{V}} e^{-\f{\beta}{2} H_n(x_1,\dots, x_n)} dx_1 \dots dx_n \\
\ge \frac{|\mc{V}|}{Z_{n, \beta}} e^{-\f{\beta}{2}(  n^2 I(\mu)+ o( n^2))} \quad \text{as } \ n\to +\infty.
\end{multline}
We need to  estimate the volume $|\mc{V}|$ of $\mc{V}$, and  it is easy to see  that (for a certain constant $C$ depending only on the dimension)
\begin{displaymath} 
|\mc{V}| = n! \left(\f{C\lambda^d}{n}\right)^n 
\end{displaymath} hence
\be \label{volumemcV} \log |\mc{V}| = \log n! - n\log n + O(n) = O(n), \ee where we used that, by Stirling's formula, $\log n! = n\log n + O(n)$. 
Taking the logarithm of (\ref{liminf1}) and inserting \eqref{volumemcV}, one gets  
\be \label{LDPestimee2}
\f{1}{n^2} \log \P_{n,\beta}(E) \geq - \f{1}{n^2} \log Z_{n,\beta} - \f{\beta}{2}I(\mu) + o(1).
\ee
Since this is true for any $\mu\in \interieur{E}$, taking the supremum with respect to $\mu \in \interieur{E}$ and the $\liminf_{n\to +\infty}$ of (\ref{LDPestimee2}) gives, for each $E \subset \mc{P}(\R^d)$ :
\be \label{liminf2}
\liminf_{n \to + \infty} \f{1}{n^2} \log \P_{n,\beta}(E) \geq \liminf_{n\to +\infty} \Big(-\f{1}{n^2} \log Z_{n,\beta} \Big)- \f{\beta}{2} \inf_{\interieur{E}} I.
\ee
Moreover, choosing $E = \mc{P}(\R^d)$ and $\mu=\mu_0$  in (\ref{LDPestimee2}), we obtain \be \label{estimeeZn2}
\f{\log Z_{n,\beta}}{n^2} \geq - \f{\beta}{2} I(\mu_0) + o(1).
\ee
Combining the two inequalities (\ref{estimeeZn1}) and (\ref{estimeeZn2}) we get the second part of the theorem, that is the thermodynamic limit at speed $n^2$ :
\be\label{limitethermo}
\lim_{n \to + \infty} \f{\log Z_{n,\beta}}{n^2} = - \f{\beta}{2} I(\mu_0) = -\f{\beta}{2} \min_{\mc{P}(\R^d)} I.
\ee
Finally, inserting (\ref{limitethermo}) into (\ref{limsup1}) and (\ref{liminf2}) completes the proof of Theorem~\ref{LDP}.
\end{proof}


\chapter[Splitting the Hamiltonian]{The next order behavior : splitting the Hamiltonian, first lower bound}\label{chapsplit}
In Chapter \ref{leadingorder}, we examined the leading order behavior of the Hamiltonian $H_n$ of the Coulomb gas, which can be summarized by~: 
\begin{itemize}
\item The minimal energy $\min H_n$ behaves like $n^2 \min I$, where $I$ is the mean-field limit energy, defined on the set of probability measures of $\R^d$.
\item If each $(x_1, \dots, x_n)$ minimizes $H_n$, the empirical measures $\f{1}{n} \sum_{i=1}^n \delta_{x_i}$ converge weakly to the unique minimizer $\mu_0$ of $I$, also known as Frostman's equilibrium measure, which can be characterized via an obstacle problem.
\item This behavior also  holds  when $\beta < + \infty$, except with a very small probability determined by a large deviation principle.
\end{itemize}
The following questions  thus arise naturally~: 
\benu
\item What lies beyond  the term $n^2 I(\mu_0)$ in the asymptotic expansion of $\min H_n$ as $n \to + \infty$ and in the expansion of the partition function $\log Z_{n,\beta} = \f{\beta}{2} n^2 I(\mu_0) + o(n^2)$ ? Is the next term of order $n$ ? 
\item What is the optimal {\it microscopic} distribution of the points ?
\eenu
To study these questions, we wish to  zoom or blow-up the configurations by the factor $n^{1/d}$ (the inverse of the typical distance between two points), so that the points are well-separated (typically with distance $O(1)$), and find a way of expanding the Hamiltonian to next order. 

Henceforth, we will keep the same notation as in Chapter \ref{leadingorder} and in this chapter  we will make the following assumptions:
$V$ is such that the unique Frostman equilibrium  measure $\mu_0$ exists (for example, as we saw it suffices to require that $V$ is continuous and satisfies \textbf{(A3)}), and $\mu_0$ is absolutely continuous with respect to the Lebesgue measure, with a density $m_0(x)$  which is bounded above, or in $L^\infty(\R^d)$. By abuse of notation, we will often write $\mu_0(x)$ instead of  $m_0(x)$.

The results we present originate in \cite{ss1,rs}, but we follow in large part the simplified approach of \cite{ps}, which is also valid for more general interactions than Coulomb.


\section{Expanding the Hamiltonian} \label{sectionaref} The ``splitting" of the Hamiltonian  consists in an exact formula that separates the leading ($n^2$) order term  in $H_n$ from next order terms, by using the quadratic nature of the Hamiltonian. 

 The starting point is the following : given a configuration of points $(x_1, \dots , x_n) \in (\R^d)^n$,  let us set $\nu_n = \sum_{i=1}^n \delta_{x_i}$ (which is not a probability measure anymore, but merely a purely atomic Radon measure). Since we expect $\f{1}{n}\nu_n$ to converge to $\mu_0$, let us expand $\nu_n$  as
\begin{equation} \label{splitting1}
\nu_n = n\mu_0 + (\nu_n - n\mu_0).
\end{equation}
The first term in the right-hand side gives the leading order, and the second one describes the fluctuation of $\nu_n$ around it.  Note that in contrast to the equilibrium measure $\mu_0$ assumed to be a nice measure with a bounded density, the fluctuation $\nu_n - n\mu_0$ is still singular, with an atom at each point of the configuration.

Inserting the splitting \eqref{splitting1} into the definition of $H_n$, one finds  
\begin{eqnarray}
\nonumber H_n(x_1, \dots, x_n)  & = &  \sum_{i \neq j} g(x_i- x_j) + n \sum_{i=1}^n V(x_i)  \\ 
\nonumber
& = & \iint_{\triangle^c} g(x-y) d\nu_n(x) d\nu_n(y) + n \int V d\nu_n \\
\nonumber 
& = &  n^2 \iint_{\triangle^c} g(x-y) d\mu_0(x) d\mu_0(y) + n^2 \int V d\mu_0 
\\ 
\nonumber
 & + &  2n \iint_{\triangle^c} g(x-y) d\mu_0(x) d(\nu_n - n \mu_0)(y)
+ n \int V d(\nu_n - n \mu_0) \\ 
\label{finh}
& + & \iint_{\triangle^c} g(x-y) d(\nu_n - n\mu_0)(x)d(\nu_n - n\mu_0)(y).
\end{eqnarray}
We now recall that $\zeta$ was defined in \eqref{defzeta} by
\be
\zeta = h^{\mu_0} + \f{V}{2} - c = \int g(x-y)\, d\mu_0(y)  + \f{V}{2} - c\ee
and that $\zeta=0$  in $\Sigma$, the support of $\mu_0$  (with the assumptions we made, one can check that  $\zeta$ is continuous, so the q.e. relation can be upgraded to everywhere).

With the help of this we may rewrite the medium line in the right-hand side of \eqref{finh} as  
\begin{multline*}
2n \iint_{\triangle^c} g(x-y) d\mu_0(x) d(\nu_n - n \mu_0)(y) + n \int V d(\nu_n - n \mu_0) \\
 = 2n  \int (h^{\mu_0} + \f{V}{2}) d(\nu_n - n\mu_0) = 2n  \int (\zeta + c) d(\nu_n - n\mu_0)  \\
 = 2n \int \zeta d\nu_n - 2n^2 \int \zeta d\mu_0 + 2 n c \int  d(\nu_n - n\mu_0) = 2n \int \zeta d\nu_n.
\end{multline*}
The last equality is due to the facts that  $\zeta \equiv 0$ on the support of $\mu_0$ and that  $\nu_n$ and $n \mu_0$ have the same mass $n$. We also have to notice that since $\mu_0$ has a $L^{\infty}$ density with respect to the Lebesgue measure, it does not charge the diagonal $\triangle$ (whose Lebesgue measure is zero) and we can include it back in the domain of integration.
By that same argument, one may recognize in the first line of the right-hand side of \eqref{finh}, the quantity $n^2 I(\mu_0)$, cf. \eqref{definitionI}.

We may thus rewrite \eqref{finh} as
\begin{multline} \label{splitting2}
H_n(x_1, \dots, x_n) = n^2 I(\mu_0) + 2n \sum_{i=1}^n \zeta(x_i) \\ + \iint_{\triangle^c} g(x-y) d(\nu_n - n\mu_0)(x)d(\nu_n - n\mu_0)(y).
\end{multline}
Note that this is an exact relation, valid for any configuration of points.
The first term in the right-hand side gives the leading order, i.e.  the energy of the equilibrium measure. In the second term, $\zeta$  plays the role  of an effective confining potential, which is active only outside of $\Sigma $ (recall $\zeta \ge 0$, and $\zeta=0$ in $\Sigma$). 
The last term in the right-hand side  is the most interesting, it 
measures the discrepancy between the diffuse equilibrium measure $\mu_0$ and the discrete empirical measure  $\f{1}{n} \nu_n$. It is an electrostatic (Coulomb) interaction between a ``negatively charged background" $-n\mu_0$ and  the $n$ positive discrete charges at the points $x_1, \dots , x_n$. In the sequel, we will express this energy term in another fashion, and show that it is indeed a lower-order term.

To go further, we  introduce  $h_n$, the potential generated by the distribution of charges $\nu_n - n\mu_0$, defined by
\be \label{defhn}
h_n := g * (\nu_n - n\mu_0) = \int g(x-y) d(\nu_n - n\mu_0)(y).\ee
In dimension  $d \ge 2$ this is equivalent to
\be \label{defhnequiv}
 h_n = - c_d \Delta^{-1}(\nu_n - n\mu_0),\ee and in dimension $1$ to 
 \be\label{defhn1d}
 h_n = - c_1 \Delta^{-1/2} (\nu_n-n\mu_0).\ee
 
Note that $h_n$ decays at infinity, because the charge distribution $\nu_n - n\mu_0$ is compactly supported and has zero total charge, hence, when seen from infinity behaves like a dipole. More precisely, $h_n$ decays like $\nabla g$ at infinity, that is $O(\f{1}{r^{d-1}})$ and its gradient $\nabla h_n$ decays like the second derivative $D^2g$, that is $O(\f{1}{r^{d}})$ (in dimension $1$, like $1/r$ and $1/r^2$).
Formally, using Green's formula (or Stokes' theorem) and the definitions, one would like to say that, at least in dimension $d \ge 2$,
\begin{multline} \label{formalcomputation}
\iint_{\triangle^c} g(x-y) d(\nu_n - n\mu_0)(x)d(\nu_n - n\mu_0)(y) = \int h_n d(\nu_n - n\mu_0)\\ = \int h_n (- \frac{1}{c_d} \Delta h_n) \approx \frac{1}{c_d} \int |\nabla h_n|^2 \end{multline}This is the place where we really use for the first time in a crucial manner the Coulombic  nature of the interaction kernel $g$.
Such a computation allows to replace the sum of pairwise interactions of all the charges and ``background" by an integral (extensive) quantity, which is easier to handle in some sense. 
However, \eqref{formalcomputation} does not make sense because  $\nabla h_n$ fails to be in $L^2$ due to the presence of Dirac masses.  Indeed, near each atom $x_i$ of $\nu_n$, the vector-field $\nabla h_n$ behaves like $\nabla g$ and the integrals $\int_{B(0,\eta)} |\nabla g|^2$  are divergent in all dimensions.   Another way to see this is that the Dirac masses charge the diagonal $\triangle$ and so $\triangle^c$ cannot be reduced to the full space.

\section{The truncation procedure and splitting formula}

In this section we give a rigorous meaning to the formal computation \eqref{formalcomputation} and thus to \eqref{splitting2}.  We restrict for now to the case of $d\ge 2$. The case $d=1$ will be dealt with in the next section.

Given $\eta>0$, we truncate the potential $g$ by setting 
\begin{equation}\label{feta}
f_\eta(x)= \max(g(x)- g(\eta),0).\end{equation}
We note that $\min(g, g(\eta))  = g -f_\eta$, 
and observe that $f_\eta$ solves 
\begin{equation}\label{divf}-\Delta f_\eta = c_d (\delta_0- \delta_0^{(\eta)})\end{equation}
where $\delta_0^{(\eta)}$ denotes the uniform measure of mass $1$ on $\pa B(0, \eta)$. Indeed, it is clear that $-\Delta f_\eta - c_d \delta_0$ can only be supported on $\pa B(0, \eta)$ and has to be uniform since $f_\eta$ is radial. It then suffices to verify that it is a positive measure of mass $1$, which is easily done by integrating against a test-function.

By analogy with \eqref{defhn}--\eqref{defhnequiv}, we may then
define the truncated potential 
\begin{equation}\label{tp}h_{n,\eta}(x) = h_n(x) - \sum_{i=1}^n f_\eta(x-x_i)\end{equation}
and note that it solves 
\begin{equation}\label{hneta}
-\Delta h_{n,\eta}= c_d(\sum_{i=1}^n \delta_{x_i}^{(\eta)} - n\muv).\end{equation}
This way, the truncation of the potential is simply equivalent to ``smearing out" each Dirac charge uniformly onto the sphere of radius $\eta$ centered at the  charge.
 
 We then have the right quantity to make sense of \eqref{formalcomputation}, as shown by the following exact formula~:
\begin{lemme}\label{lem31} Let $d \ge 2$, and $\mu_0$ be a measure with a bounded density. For any configuration of distinct  points $(x_1, \dots, x_n)\in (\R^d)^n$, $h_{n,\eta}$ being defined in \eqref{tp}, the following identity holds~:
\begin{equation} \label{formula}
\iint_{\triangle^c} g(x-y) d(\nu_n - n\muv)(x)d(\nu_n - n\muv)(y)=  \lim_{\eta\to 0}\left(\frac{1}{c_d}\int_{\R^{d}} |\nab h_{n,\eta}|^2 -  n \g(\eta)\right).\end{equation}
\end{lemme}
\begin{proof}
Let us  compute the right-hand side of this relation. Let us choose $R$ so that all the points are in $B(0,R-1)$ in $\mr^{d}$, and $\eta$ small enough that $2\eta<\min_{i\neq j}|x_i-x_j|$.    We note that $f_\eta$ vanishes outside of $B(0,\eta)$ thus  $h_{n,\eta}= h_n$ at distance $\ge \eta$ from the points.  By Green's formula and \eqref{tp}, we thus  have
\begin{multline}\label{greensplit1}
\int_{B_R} |\nabla \hne |^2 = \int_{\pa B_R} h_n \frac{\pa  h_n}{\pa \nu} -
\int_{B_R } \hne    \Delta \hne
  \\= \int_{\pa B_R} h_n \frac{\pa  h_n}{\pa \nu} +
c_d \int_{B_R } \hne    \left(\sum_{i=1}^n \delta_{x_i}^{(\eta)} -n{\muv}\right).
\end{multline}
In view of the decay of $h_n$ and $\nab h_n$ at infinity mentioned above, the boundary integral tends to $0$ as  $R\to \infty$. We thus find
\begin{multline}\label{greensplit2}
\int_{\R^{d} }  |\nabla \hne |^2 =
c_d \int_{\R^{d} }  \hne    \left(\sum_{i=1}^n \delta_{x_i}^{(\eta)} -n\muv\right)\\  =
c_d \int_{\R^{d} }  \left(h_n - \sum_{i=1}^n f_\eta(x-x_i)\right)   \left(\sum_{i=1}^n \delta_{x_i}^{(\eta)} -n\muv\right)  .\end{multline}
 Since $f_\eta(x-x_i) =0 $ on $\pa B(x_i, \eta)=\supp (\delta_{x_i}^{(\eta)})$ and outside of $B(x_i, \eta)$, and since the balls $B(x_i, \eta)$ are disjoint by choice of $\eta$, we may write 
$$
\int_{\R^{d} }   |\nabla \hne |^2 =c_d  \int_{\R^{d}}   h_n\left(  \sum_{i=1}^n \delta_{x_i}^{(\eta)}  - n\muv\right) -
n c_d  \int_{\R^{d}} \sum_{i=1}^n f_\eta(x-x_i) \muv.$$
Let us now use (temporarily) the notation $h^i_n(x) = h_n(x) - \g(x -x_i)$ (for the potential generated by the distribution  bereft of the point $x_i$).
The function $h_n^i $ is regular near $x_i$, hence $\int h_n^i \delta_{x_i}^{(\eta)} \to h_n^i(x_i)$ as $\eta \to 0$.
It follows that 
\begin{multline}\label{estdistinct}
c_d \int_{\R^{d}}  h_n\left(  \sum_{i=1}^n \delta_{x_i}^{(\eta)}  -n \muv\right)-n c_d \int_{\R^{d}} \sum_{i=1}^n f_\eta(x-x_i) \muv
 \\= nc_d  \g(\eta)  + c_d\sum_{i=1}^n h_n^i(x_i)- n c_d \int_{\R^{d}}  h_n\muv  + O(n^2 \|\muv\|_{L^\infty} ) \int_{B(0, \eta)} |f_\eta| +o_\eta(1).
\end{multline}
We can check that    \begin{equation}\label{intfeta}
\int_{B(0, \eta)} |f_\eta|\le C \max(\eta^2, \eta^2 \log \eta) \end{equation}
 according to whether $d\ge 3$ or $d=2$, because $f_\eta$ is bounded by $2|g|$. Thus
\begin{equation}\label{Wsplit11}\lim_{\eta \to 0} \frac{1}{c_d}  \int_{\R^{d} }   |\nabla \hne |^2- n\g(\eta)  = \sum_{i=1}^n h_n^i(x_i)- n \int_{\R^{d}} h_n{\muv}.\end{equation}
 Now, from the definitions it is easily seen that 
\begin{equation}
h^i_n(x_i) = \int_{\R^{d}\backslash \{x_i\}} \g(x_i-y) d(\nu_n - n\muv)(y),
\end{equation}
from which it follows that 
\begin{multline*}
\iint_{\triangle^c} \g(x-y) d(\nu_n - n\muv)(x)d(\nu_n - n\muv)(y)
\\ = \sum_{i=1}^n  \int_{\R^d \backslash \{x_i\}}  \g(x_i-y) d(\nu_n - n\muv)(y) - n \int_{\R^{d}} h_n \, \muv
= \sum_{i=1}^n h_n^i(x_i) - n \int_{\R^{d}} h_n \muv.
\end{multline*} In view of  \eqref{Wsplit11}, we conclude that the formula holds. \end{proof}

The quantity  appearing in the right-hand side of \eqref{formula}  thus provides  a way of computing $\int |\nab h|^2 $ in a ``renormalized" fashion, by truncating the divergent parts of $\nabla h$  and subtracting off from $\int|\nab h|^2$ the expected divergence  $c_d g(\eta)$ corresponding to each point.  
We may write  that near each point $x_i$, we have
\begin{displaymath}
h (x)=   c_d g(x-x_i) + \vp(x) \textrm{ with } \vp \textrm{ of class } C^1.
\end{displaymath}
Thus  $\int_{B(p,\eta)^c}  |\nabla h|^2 \approx c_d^2 \int_{B(p,\eta)^c}  |\nabla g|^2 \approx c_d g(\eta)$ as $\eta \to 0$,  which is the contribution that appears in the second term of the right-hand side in \eqref{formula}. We will sometimes call this quantity a ``precursor to the renormalized energy" which will itself be defined in Chapter \ref{chapdefw}.
The name {\it renormalized energy} originates in \cite{bbh} where a related  way of computing such integrals was first introduced in the context of two-dimensional Ginzburg-Landau vortices, based on cutting out holes instead of truncating. More precisely,  if  one uses the way of ``renormalizing"  of \cite{bbh} as was originally done in \cite{ss1}, then one writes instead of \eqref{formula}
\begin{multline}\label{formulabbh}
\iint_{\triangle^c} g(x-y) d(\nu_n - n\muv)(x)d(\nu_n - n\muv)(y)
\\=
 \lim_{\eta \rightarrow 0} \left( \frac{1}{c_d} \int_{\R^d \backslash \cup_{i=1}^n B(x_i, \eta)}  |\nabla h_n|^2 - n  g(\eta)\right).
\end{multline} The right-hand side is thus equal to that of \eqref{formula}.
 
\begin{rem}\label{mu0faible}
When examining the proof, we see that one does not really need  $\mu_0$ to have bounded density, but only that $\int_{B(x,\eta)}  g(x-y)\, d\mu_0(y)\to 0$ as $\eta \to 0$. A sufficient condition is for example that  for any $x$, $\mu_0(B(x,\eta))$ grows no faster than $\eta^\alpha$ with $\alpha>  d-2$ so that $g(\eta) \eta^\alpha \to 0$ as $\eta\to 0$. This  encompasses a whole class of singular measures, such as measures supported on a rectifiable $(d-1)$-dimensional set. \end{rem}
Before stating the final formula, we want to blow-up  at a scale where the points  are well-separated. The convergence of the empirical measure of the $n$ points to a fixed compactly supported measure suggests that there are typically $n$ points in a bounded domain, so that the distance between two points should be of order $n^{-1/d}$. To get a $O(1)$ distance, one thus  has to change the scale by a factor $n^{1/d}$. We will use a $'$ (prime) symbol to denote the blown-up quantities :
\bm
x' := n^{1/d} x \mbox{ for all } x \in \R^d \mbox{, in particular } x'_i = n^{1/d} x_i.
\em 
Let us also define :
\be\label{hn'}
h_n' := g * \Big(\sum_{i=1}^n \delta_{x'_i} - \mu_0'\Big) \mbox{, so that } h_n'(x') =n^{2/d-1} h_n(x)
\ee
where $\mu_0'$ is the blown-up measure associated to $\mu_0$ :  
\bm
d\mu_0'(x') = m_0(x)\, dx' = m_0(x' n^{-1/d}) \, dx'.
\em
Note that  $\mu_0'$ has total mass $n$.
The function  $h_n'$  is the potential generated by the blown-up distribution $\sum_{i=1}^n \delta_{x'_i} - \mu_0'$   of support $\Sigma' = n^{1/d} \Sigma$. 
As above, we define the truncated version of this potential
\begin{equation}\label{1to1}
h_{n,\eta}'= h_n'-\sum_{i=1}^n f_\eta(x-x_i'),\end{equation}
and of course we have
\begin{equation}\label{eqhnp}
-\Delta h_n' = c_d \( \sum_{i=1}^n \delta_{x_i'} - \muv'\)\qquad - \Delta h_{n,\eta}' = c_d \(\sum_{i=1}^n \delta_{x_i'}^{(\eta)}- \muv'\).\end{equation}

The effect of the rescaling can be seen by 
a change of variables,  indeed
\bm
\int_{\R^d} |\nabla h_{n, n^{-1/d}\eta }|^2 = n^{1-2/d} \int_{\R^d}|\nabla h'_{n, \eta}|^2,
\em so that we have 
\begin{multline}\label{320} \lim_{\eta\to 0}\left(\int_{\R^{d}} |\nab \hne|^2 -  nc_d  \g(\eta)\right)
= \lim_{\eta \to 0} \left(\int_{\R^{d}} |\nab h_{n,  n^{-1/d} \eta}|^2 - nc_d \g(n^{-1/d} \eta) \right)
\\ = \lim_{\eta \to 0} \left( n^{1-2/d} \int_{\R^{d}} |\nab h'_{n,   \eta}|^2 -  nc_d \g(n^{-1/d} \eta)   \right) 
\end{multline}
For $d = 2$, one has $g(\eta n^{-1/d} ) = \f{1}{2} \log n + g(\eta)$ whereas for $d \geq 3$, $g(\eta n^{-1/d} ) = n^{1-2/d} g(\eta)$. Consequently we get  the following formulae for the  change of scale on the expression of  the precursor to the renormalized energy in \eqref{formula}~:  
\begin{multline} \label{identificationenergiechgt}
\lim_{\eta\to 0}\left(\int_{\R^{d}} |\nab h_{n,\eta}|^2 -  n c_d \g(\eta)\right) 
\\=   \lim_{\eta\to 0}\left(\int_{\R^{d}} |\nab h_{n,\eta}'|^2 -  nc_d \g(\eta)\right) 
-\(  \f{c_2}{2}n \log n\)\indic_{d=2} .\end{multline}
Equation \eqref{splitting2} and  Lemma \ref{lem31}, together with (\ref{identificationenergiechgt}) yield  the first important conclusion~:
\begin{prop}[Splitting formula] \label{developpementasympto}
Let $V$ be such that a unique equilibrium measure $\mu_0$ exists, and $\mu_0$ has an $L^\infty$ density (or is as in Remark \ref{mu0faible}).
For any $n\ge 1$, for  any configuration of distinct points  $x_1, \dots, x_n$ in $\R^d$, $d\ge 2$, the following identity holds : 
\bm\boxed{H_n(x_1, \dots, x_n) = n^2I(\mu_0) + 2n \sum_{i=1}^n \zeta(x_i) +
 \(  - \f{n}{2}\log n\) \indic_{d=2}  +  \f{n^{1-2/d}}{c_d} \mathcal H_n(x_1', \dots, x_n') , } 
\em
where  $\zeta$ is as in \eqref{defzeta}, $x_i'=n^{1/d} x_i$, and $\mathcal H_n$ is defined with the help of  \eqref{hn'}--\eqref{1to1} by 
\begin{equation}\label{Wncal}
\mathcal H_n(x_1', \dots, x_n')= \lim_{\eta \to 0} \left(  \int_{\R^{d}} |\nab h'_{n,   \eta}|^2 -  nc_d \g( \eta)\right) .\end{equation}
\end{prop}
We emphasize that here again there is no error term, it is an equality for every $n$ and every configuration.
This formula was first established in \cite{ss1} in dimension $d=2$, and generalized (with the same proof) to higher dimension in \cite{rs}.

Since $\zeta$ plays no other role than confining the points to $\Sigma$ (the support of $\mu_0$), this formula shows that it suffices to analyze the behavior of $\mathcal H_n(x_1', \dots, x_n').$ 
We will show later that for good configurations (those that do not have too much energy)  $ \mathcal H_n(x_1', \dots, x_n')  $ is proportional to $n$, the number of points.  We have thus separated orders in the expansion of $H_n$ : after the leading order term $n^2 I(\mu_0)$ and  an exceptional term $-\hal n \log n$ in dimension $2$, a  next term of order $n^{2-2/d}$ appears.

The analysis of $ \mathcal H_n(x_1', \dots, x_n')$ as $n$ goes to infinity will be the object of Chapter~\ref{derivingw}, but prior to that we will introduce its limit $\mathcal W$ in Chapter~\ref{chapdefw}.

\section{The case $d=1$} 
With our choice of $g(x)=-\log |x|$ in dimension $1$, $g$ is no longer the Coulomb kernel, so the formal computation \eqref{formalcomputation} does not work. 
However $g$ is the kernel of the half-Laplacian, and it is known that the half-Laplacian can be made to correspond to the Laplacian by adding one extra space dimension.
In other words, we should imbed the one-dimensional space $\R$ into the two-dimensional space $\R^2$ and consider the harmonic extension of $h_n$, defined in \eqref{defhn}, to the whole plane. That extension will solve an appropriate Laplace equation, and we will reduce dimension $1$ to a special case of dimension $2$. This is the approach  proposed in \cite{ss2}. 

Let us now get more specific.
 Let us consider $\mu_0$ the one-dimensional equilibrium measure associated to $I$ (with $g=-\log$) as in Theorem \ref{theoFrostman}, and assume it has an $L^\infty$ density $m_0(x)$ with respect to the (one-dimensional) Lebesgue measure (this happens for example when $V(x) = x^2$, the corresponding equilibrium measure being the semi-circle law, cf. Example 2 in Chapter \ref{leadingorder}).  We may now view $\mu_0$ as a singular measure on $\R^2$ by setting 
   \be \label{defembedding}
d\mu_0(x,y) = m_0(x) \delta_{\R}
\ee where $\delta_{\R}$ is the measure of length on the real axis. More precisely, we define $\delta_{\R}$  by its action against smooth test functions $\varphi\in C^0_c(\R^2)$ by  
\be \label{deltar}
\lg \delta_{\R}, \varphi \rg := \int_{R} \varphi(x,0) dx
\ee
which makes  $\delta_{\R}$  a Radon measure on $\R^2$, supported on the real axis. 
Given $x_1, \dots, x_n \in \R$, let us also identify them with  the  points $(x_1, 0), \dots, (x_n,0)$ on the real axis of $\R^2$. 
We may then define the potential $h_n$ and the truncated potential $h_{n,\eta}$  on $\R^2$ by 
\bm
h_n = g * \left(\sum_{i=1}^n \delta_{(x_i,0)} - n m_0 \delta_{\R}\right)\qquad h_{n,\eta} = g * \left(\sum_{i=1}^n \delta_{(x_i,0)}^{(\eta)} - n m_0 \delta_{\R}\right).
\em 
with $g(x)=- \log |x|$ in $\R^2$, which is nothing else than the harmonic extension to $\R^2$, away from the real axis,  of the potential $h_n$ defined in dimension $1$ in \eqref{defhn}. This closely related to  the {\it Stieltjes transform}, a commonly used object in Random Matrix Theory.

Viewed as a function in  $\R^2$,  $h_n$  solves
\bm
-\Delta h_n = c_2 \left(\sum_{i=1}^n \delta_{(x_i,0)} - n m_0 \delta_{\R}\right) \quad \mbox{in}\ \R^2
\em
which is now a local equation, in contrast with \eqref{defhn1d}.
We then observe that, letting $\nu_n= \sum_{i=1}^n \delta_{(x_i,0)}$, we  may also write  
\bm
H_n(x_1,\dots, x_n) = \iint_{\triangle^c \subset \R^2 \times \R^2} g(x-y) d\nu_n(x) d\nu_n(y) + \int_{\R^2} V d\nu_n
\em
where $g$ is always $-\log |x|$, and $V$ is arbitrarily extended to $\R^2$. This is thus formally the same as in dimension $2$, so 
returning to the setting of Section \ref{sectionaref}  an viewing $\nu_n$ and $\mu_0$  as  measures in $\R^2$ thanks to \eqref{defembedding}, we may carry  on with the  proof of the splitting formula as in the case $d = 2$.  We need to use the result of Remark \ref{mu0faible}, which applies because precisely $\mu_0$ is a singular measure but absolutely  continuous and with bounded density with respect to the Hausdorff measure on the real axis. 
The proof  then goes through with no other change and yields the same first splitting formula 
\be \label{splitting1D}
H_n(x_1, \dots, x_n) = n^2 I(\mu_0) + 2n \sum_{i=1}^n \zeta(x_i) + \f{1}{c_2}\( \lim_{\eta \to 0}  \int_{\R^{2}} |\nab h'_{n,   \eta}|^2 -  n c_2 \g( \eta)\) .
\ee
Continuing on with the blow-up procedure, the natural change of scale is then 
 \be\label{dehn'1}
 x'=nx \quad \mu_0'(x')= m_0(x' n^{-1}) \, \delta_{\R} \qquad    h_{n,\eta}'(x')= g* \( \sum_{i=1}^n \delta_{(x_i', 0)}^{(\eta)} - \mu_0'\),
 \ee
   and thus  in the computation analogous to \eqref{320} we obtain the term $g(\eta' n^{-1})=  \log n+g(\eta')  $ which yields the formula for the  change of scales 
 $$ \lim_{\eta\to 0}\left(\int_{\R^{2}} |\nab h_{n,\eta}|^2 -  nc_2 \g(\eta)\right)= 
 \lim_{\eta\to 0}\left(\int_{\R^{2}} |\nab h_{n,\eta}'|^2 -  nc_2 \g(\eta)\right)  -c_2
  n\log n.$$
 We conclude with the following splitting formula for $d=1$ : 
\begin{prop}\label{pro1d} Let $d=1$ and  $V$ be such that a unique equilibrium measure  $\mu_0$ minimizing $I$ exists, and $\mu_0$ has an $L^\infty$ density. Then, for any $n \ge 1$ and any configuration of distinct points $x_1, \dots, x_n\in \R$, the following identity holds :
\bm
\boxed{H_n(x_1, \dots, x_n) = n^2 I(\mu_0) + 2n \sum_{i=1}^n \zeta(x_i) - n \log n + \f{1}{c_2}\mathcal H_n(x_1', \dots, x_n'),}
\em
with $\zeta $  as in \eqref{defzeta}, $x_i'=nx_i$, and 
$$\mathcal H_n(x_1', \dots, x_n')= \lim_{\eta\to 0} \left(\int_{\R^{2}} |\nab h_{n,\eta}'|^2 -  nc_2 \g(\eta)\right)$$ defined via \eqref{dehn'1}.

\end{prop}

There remains again to understand the term $\mathcal H_n$, as before, except for the  particularity that $h'_n$ solves :
\be \label{eq1d1}
-\Delta h'_n = c_2 \left(\sum_{i=1}^n \delta_{(x_i',0)} - \mu_0' \delta_{\R}\right)
\ee
with the extra $ \delta_{\R}$ term. In the rest of these notes, we will not expand much on the one-dimensional case. The main point is that with the above transformations, it can almost be treated  like the two-dimensional case. The interested reader can refer to \cite{ss2}, and \cite{ps} where we showed that  this dimension extension approach can also be used to generalize all  our study to the case of Riesz interaction potentials $g(x)=|x|^{-s}$ with $d-2\le s< d$, via the Caffarelli-Silvestre extension formula for fractional Laplacians.

\section[Almost monotonicity]{Almost monotonicity of the truncation procedure}
In this section, we return to the Coulomb case with $d\ge 2$ and  we  prove that applying the $\eta $ truncation to the energy is essentially  decreasing  in $\eta$. This is natural since one may expect that smearing out the charges further and further should decrease their total interaction energy. 
All the results can easily be adapted to the one-dimensional logarithmic case, cf. \cite{ps}.

\begin{lemme}[\cite{ps}]\label{prodecr} Assume that $\mu_0$ is a measure with a bounded density. 
For any $n$, any  $x_1, \dots, x_n \in \R^d$, and any $1>\eta>\alpha>0$ we have
\begin{multline*}
-C n \|\muv\|_{L^\infty} \eta\le \left(\int_{\R^{d}} |\nab h'_{n,\alpha}|^2 - n c_d\g(\alpha) \right) -
\left(\int_{\R^{d}}|\nab h'_{n,\eta}|^2 - nc_d \g(\eta) \right)
\end{multline*}where  $C$ depends only on $d$.  Moreover, equality holds if $\min_{i\neq j} |x_i'-x_j'| \geq  2\eta$.

\end{lemme}

\begin{proof}
We first let $\fae:= f_\alpha-f_\eta$ and note that  $\fae$  vanishes outside $B(0,\eta)$, $\fae\le 0$ and  
 it  solves (cf. \eqref{divf})
\begin{equation}\label{eqfae}
-\Delta \fae= c_d (\delta_0^{(\eta)}- \delta_0^{(\alpha)}).
\end{equation}
In view of
\eqref{tp}, we have
$\nab \hne'= \nab h_{n, \alpha}' + \sum_{i=1}^n\nab \fae(x-x_i)$
and hence
\begin{multline}
\int_{\R^{d}}  |\nab \hne'|^2 = \int_{\R^{d}}  |\nab h_{n, \alpha}'|^2 + \sum_{i, j} \int_{\R^{d}}  \nab \fae (x-x_i) \cdot \nab \fae (x-x_j) 
\\ +  2 \sum_{i=1}^n \int_{\R^{d}}  \nab \fae(x-x_i)\cdot \nab h_{n, \alpha}'.\end{multline}
We first examine
\begin{align}\nonumber
&  \sum_{i, j} \int_{\R^{d}}  \nab \fae (x-x_i) \cdot \nab \fae (x-x_j) 
\\ \label{faeterm}
& = - \sum_{i,j} \int_{\R^{d}}  \fae (x-x_i) \Delta \fae(x-x_j)= c_d \sum_{i,j} \int_{\R^{d}} \fae (x-x_i) (  \delta_{x_j}^{(\eta)}- \delta_{x_j}^{(\alpha)}) 
.
\end{align}
Next, 
\begin{multline}\label{fae2}
2 \sum_{i=1}^n \int_{\R^{d}}  \nab \fae(x-x_i)\cdot \nab h_{n, \alpha}'
 = -2 \sum_{i=1}^n \int_{\R^{d}} \fae(x-x_i) \Delta  h_{n,\alpha}'
\\=  2
c_d \sum_{i=1}^n\int_{\R^{d}} \fae(x-x_i) \Big( \sum_{j=1}^n \delta_{x_j}^{(\alpha)} - \muv'\Big) . 
\end{multline}
These last two equations add up  to give a right-hand side equal to 
\begin{equation}\label{rhs}
c_d\sum_{i\neq j}  \int_{\R^{d}} \fae (x-x_i) (  \delta_{x_j}^{(\alpha)}+ \delta_{x_j}^{(\eta)}) 
- 2c_d \sum_{i=1}^n \int\fae(x-x_i) \muv'
+ nc_d \int_{\R^{d}} \fae ( \delta_{0}^{(\alpha)} + \delta_0^{(\eta)} )\, . \end{equation}
We then note that $\int\fae( \delta_0^{(\alpha)} + \delta_0^{(\eta)} ) = -\int f_\eta \delta_0^{(\alpha)}=-( \g(\alpha)-\g(\eta))  $ by definition of $f_\eta$ and the fact that $\delta_0^{(\alpha)}$ is a measure supported on $\pa B(0, \alpha)$ and of mass $1$.
Secondly, we bound  $\int_{\R^{d}}\fae(x-x_i) \muv'$ by $$\|\muv\|_{L^\infty} \int_{\R^{d}} |f_\eta| \le  C \|\muv\|_{L^\infty} \max(\eta^{2}, \eta^2 |\log \eta|) $$
according to the cases, as seen in \eqref{intfeta}.
Thirdly, we observe that the first term in \eqref{rhs} is nonpositive and vanishes  if $\min_{i\neq j} |x_i'-x_j'| \geq  2\eta$,  and 
we conclude that 
\begin{multline*}  
- C n \|\muv\|_{L^\infty} \max( \eta^{2}, \eta^2 |\log \eta|)
\\ \le \(\int_{\R^d}  |\nab h_{n, \alpha}'|^2  -nc_d  \g(\alpha)\) -\( \int_{\R^{d}} |\nab \hne'|^2 - nc_d \g(\eta)\)\end{multline*} with equality if $\min_{i\neq j} |x_i'-x_j'| \geq  2\eta$.
Combining all the elements finishes the proof, noting that in all cases we have $\max(\eta^{2}, \eta^2  |\log \eta|) \le \eta$.
\end{proof}
This lemma proves in another way that the limit defining \eqref{Wncal} exists and  provides an immediate lower bound for $\mathcal H_n$: taking the limit $\alpha \to 0$ in the result, we obtain that  for any $0<\eta<1$,  and any $x_1, \dots, x_n$, it holds that 
\begin{equation}\label{lbwnc}
\mathcal H_n(x_1', \dots, x_n') \ge \int_{\R^d} |\nab h_{n,\eta}'|^2 - nc_d g(\eta) - C n \eta \|\muv\|_{L^\infty} \end{equation}
where $C$ depends only on the dimension.

\section{The splitting lower bound}
Combining \eqref{lbwnc} with the splitting formula in Proposition \ref{developpementasympto}, we obtain the following
\begin{prop}[Splitting lower bound \cite{rs,ps}]\label{propRS} 
Assume the equilibrium measure $\mu_0$ exists and has a bounded density.
For every $n\ge 1$,  and for  all configurations of  points $x_1, \dots, x_n$ in $\R^d$, the following inequality holds for all $1>\eta>0$~: 
\begin{multline}
H_n(x_1, \dots, x_n) \geq n^2 I(\mu_0)  + 2n \sum_{i=1}^n \zeta (x_i)-\Big( \f{n}{2} \log n \Big) \indic_{d = 2}  \\
+ \f{n^{2-2/d}}{c_d}\Big(  \frac{1}{n} \int_{\R^d} |\nabla h'_{n, \eta}|^2  -    c_d g(\eta)  - C \eta  \|\mu_0\|_{L^\infty}\Big) ,
\end{multline} where $h_{n,\eta}'$ is as in \eqref{1to1} and $C>0$ depends only on the dimension.
Moreover, there is equality  if $\min_{i\neq j} |x_i'-x_j'| \geq  2\eta$.
\end{prop}
In \cite{rs} we gave a different proof of this result which was based on Newton's theorem and similar to that of    Onsager's lemma \cite{ons}, a tool which has been much used in the proof of stability of matter in quantum mechanics  (see \cite{LO,ls} and references therein).  The proof we have presented here, from \cite{ps}, is somewhat easier and also works in the more general case of Riesz interactions.

\section{Consequences}
The result of   Proposition \ref{propRS}, for fixed $n$ and $\eta$, already yields some easy consequences, such as a lower bound for $H_n$. 
Indeed, taking say $\eta = 1/2$, and using the fact that $\int |\nab h_{n,\eta}'|^2 \ge 0$,  we get as a corollary 
\begin{coroll}[An easy lower bound for $H_n$]
\label{corlb}
Under the same assumptions, we have 
\be \label{lowerboundhamilton}
H_n(x_1, \dots, x_n) \geq n^2 I(\mu_0) + 2n \sum_{i=1}^n \zeta(x_i) -  \Big( \f{n}{2} \log n \Big) \indic_{d = 2}  -  C \|\mu_0\|_{L^\infty}   n^{2-2/d}
\ee
where the constant $C$ depends only on  the dimension. \end{coroll}
Another way of stating this is that $\mathcal H_n(x_1', \dots, x_n') \ge - C n $ and we see that the next order term in the expansion of $H_n$ can  indeed be expected to be of  order $n^{2-2/d}$ --- at least it is bounded below by it.

For illustration, let us show how this lower bound easily  translates  into an upper bound for the partition function $Z_{n, \beta}$ (defined in \eqref{defz})  in the case with temperature.
\begin{coroll}[An easy upper bound for  the partition function]\label{boundlogz}
Assume that $V$ is continuous, such that $\mu_0$ exists, and satisfies \textbf{(A4)}  (see  \eqref{conditionintegrabilite} in Section \ref{potentialtheoretic}).  Assume that $\mu_0$ has an $L^\infty$ density. Then for all $\beta >0$, and for $n$ large enough, we have
$$\log Z_{n,\beta}  \leq - \f{\beta}{2} n^2 I(\mu_0) + \Big(\f{\beta}{4} n \log n\Big) \indic_{d=2} + C\beta n^{2-2/d} + Cn$$
where $C$ depends only on $\mu_0$ and the dimension.
\end{coroll}

To prove this, let us  state a lemma that we will use repeatedly and that exploits assumption \textbf{(A4)}.

\begin{lemme} \label{asymptozeta} Assume that $V$ is continuous, such that $\mu_0$ exists, and satisfies \textbf{(A4)}. For any $\lambda > 0$  we have 
\be
\lim_{n\to  + \infty} \left(\int_{(\R^d)^n} e^{-\lambda \beta n \sum_{i=1}^n \zeta(x_i)} dx_1 \dots dx_n\right)^{1/n} = |\omega|\ee where $\omega=\{\zeta=0\}$, uniformly in $\beta \in [\beta_0, +\infty)$, $\beta_0>0$.
\end{lemme}
\begin{proof} First, by separation of variables, we have
$$ \left(\int_{(\R^d)^n}  e^{-\lambda \beta n \sum_{i=1}^n \zeta(x_i)} dx_1 \dots dx_n\right)^{1/n}= \int_{\R^d}
e^{-\lambda \beta n \zeta(x)} \, dx.$$
Second, we recall that know that since $\mu_0$ is a compactly supported probability measure, $h^{\mu_0}$ must asymptotically behave like $g(x)$ as $|x|\to \infty$, thus $\zeta= h^{\mu_0} + \f{V}{2}-c $  grows like $g(x)+ \hal V -c$.  The assumption \textbf{(A4)} thus ensures that 
there exists some $\alpha>0$ such that $\int e^{-\alpha \zeta(x)}\, dx<+\infty$ and hence
for $\lambda>0$, $\beta>0$,  and $n$ large enough, $\int e^{- \lambda \beta n\zeta(x)}\, dx<+\infty.$ 
Moreover, by definition of $\omega$ (cf. \eqref{zerozeta}),    
$$e^{-\lambda \beta n \zeta}\to \indic_{\omega} \quad \text{as} \ n\to+\infty$$
  pointwise, and $\omega$ has finite measure in view of the growth of $h^{\mu_0}$ and thus of  $\zeta$. Moreover    these functions are dominated by $e^{-\alpha\zeta}$ for $n$ large enough, which is integrable,  so the dominated convergence theorem applies and allows to conclude.
\end{proof}

\begin{proof}[Proof of the corollary]
By definition \eqref{defz} we have 
$$\log Z_{n,\beta}= \log \int e^{ -\frac{\beta}{2}H_n(x_1,\dots, x_n)} \, dx_1 \dots dx_n$$
and inserting \eqref{lowerboundhamilton}, we are led to 
\begin{multline}
\log Z_{n,\beta} \leq - \f{\beta}{2} n^2 I(\mu_0) + \Big(\f{\beta}{4} n \log n\Big)\indic_{d=2} + C\beta n^{2-2/d} \\+ \log \(\int e^{-n\beta \sum_{i=1}^n \zeta(x_i)} dx_1 \dots dx_n\).\end{multline} Using Lemma \ref{asymptozeta} to handle the last term, we deduce that 
$$
\log Z_{n, \beta} \leq - \f{\beta}{2} n^2 I(\mu_0) + 
\Big(\f{\beta}{4} n \log n\Big)\indic_{d=2} + C\beta n^{2-2/d} + n(\log |\omega|+o_n(1))
$$
which gives the conclusion.
\end{proof}
  \section{Control of  the potential and charge fluctuations }
 In this section, we show that $\mathcal H_n$ has good coercivity properties. We note that   $\mathcal H_n (x_1', \dots, x_n') $  can itself be controlled via Proposition \ref{developpementasympto} if a suitable  upper  bound for $H_n$ is known. 
 More precisely, we will see  that the method of truncation, even at a fixed scale $\eta$ (which does not need to go to zero), allows to obtain in a simple way some control on $\nab h_n' $ itself and also on the point discrepancies.
 
 \begin{lemme}[Control of the potential via truncation] \label{lem:fluctu field}
 Assume $\mu_0$ is a measure with an $L^\infty$ density. 
For any $n$, any  $x_1, \dots, x_n\in \R^d$,
 any $1\le q <\frac{d}{d-1} $, any $0<\eta <1$ and any $R>0$ and $K_R=[-R,R]^d$, denoting $\nu_n'= \sum_{i=1}^n \delta_{x_i'}$,
and  letting $h_{n,\eta}'$, $h_n'$ be as in  \eqref{defhn}, \eqref{1to1}, 
we have
\be\label{eq:control Lq field}
 \left\Vert \nabla h'_n \right\Vert_{L ^q (K_R)} \leq |K_R | ^{1/q - 1/2} \left\Vert \nabla h_{n,\eta}'     \right\Vert_{L ^2 (K_R)} + C_{q,\eta} \nu_n ' (K_{R+\eta})
\ee
where $C_{q,\eta}$ depends only on $q$, $\eta$ and $d$ and satisfies $C_{q,\eta}\to 0$ when $\eta \to 0$ at fixed $q$. 
\end{lemme}
\begin{rem}\label{rem41}
The reason for the condition $q<\frac{d}{d-1}$ is that near each $x_i'$, $h_n'$ has a singularity in $g(x-x_i')$, hence $\nab h_n'$ blows up like $|x-x_i'|^{1-d}$ and this is in $L^q$ if and only if $q<\frac{d}{d-1}$. In other words, while $\nab h_{n,\eta}'$ belongs to $L^2$, $\nab h_n'$ belongs at best to such $L^q$ spaces. \end{rem}

\begin{proof}
By \eqref{1to1},  we have 
\[
\nabla h'_n = \nabla h_{n,\eta}' - \sum_{i=1} ^n \nabla f_{\eta} (x-x_i) 
\]
and thus 
\[
\left\Vert \nabla h'_n \right\Vert_{L ^q (K_R)} \leq \left\Vert \nabla h_{n,\eta}' \right\Vert_{L ^q (K_R)} + \nu_n' (K_{R+\eta}) \left\Vert \nabla f_\eta\right\Vert_{L ^q (\R ^d)}
\]
where we used that if $x\in K_R$ and $\eta<1$, then $f_{\eta} (x-x_i) = 0$ if $x_i \in (K_{R+\eta})^c$. A simple application of H\"older's inequality then yields 
\[
\left\Vert \nabla h_{n,\eta}' \right\Vert_{L ^q (K_R)} \leq |K_R| ^{1/q-1/2} \left\Vert \nabla h_{n,\eta}' \right\Vert_{L ^2 (K_R)} 
\]
and concludes the proof of the inequality, with 
$C_{q,\eta} :=  \left\Vert \nabla f_\eta\right\Vert_{L ^q (\R ^d)}.$
%
\end{proof}

 Since  $$-\Delta h_{n,\eta}'= c_d \Big(\sum_{i=1}^n \delta_{x_i'}^{(\eta)} - \mu_0'\Big),$$ 
 controlling $\nab h_{n,\eta}'$ in $L^2$ for some $\eta$ (as we do via  Proposition \ref{propRS}) gives a control on $\sum_i \delta_{x_i'}^{(\eta)}$  (more precisely it gives a control on $\|\sum_i \delta_{x_i'}^{(\eta)}- \mu_0'\|$ in the Sobolev space $H^{-1}_{\mathrm{loc}}$),
   which suffices, say, to control the number of points in a given region, since controlling the Dirac masses or the smeared out Dirac masses is not much different.
 These controls in weak spaces of  the fluctuations $\sum_i \delta_{x_i'}- \mu_0'$, also called discrepancies in numbers of points,  can in fact be improved, again via the smeared out charges at fixed scale $\eta$, and we have the following result~:

 \begin{lemme}[Controlling charge fluctuations or discrepancies,  \cite{rs}]\label{lem:fluctu charge}\mbox{}\\
Assume $\mu_0$ is a measure with an $L^\infty$ density. 
For any $n$, for any $x_1, \dots, x_n\in \R^d$,  let  $\nu_n'=\sum_{i=1}^n \delta_{x_i'}$ and  
$$D(x', R): =  \nu'_n (B(x',R)) - \int_{B(x',R)} d \mu_0'$$ be the point discrepancy.
Then,  $h_{n,\eta}'$ being given by \eqref{1to1}, 
for any $0 < \eta < 1$, $R>2$ and $x' \in \R ^d$, we have
\begin{equation}\label{eq:control fluctu charge}
\int_{B(x',2R)} |\nabla h_{n,\eta}' | ^2 \geq C \frac{D (x',R) ^2 }{R ^{d-2}} \min\left( 1, \frac{D(x',R)}{R ^d} \right),
\end{equation} 
where $C$ is a constant depending only on $d$ and $\|\mu_0\|_{L^\infty}$.
\end{lemme}
\begin{proof}The proof relies on a  Cauchy-Schwarz inequality argument which is the basis of  the Ginzburg-Landau  ball construction that we will see  in Chapter~\ref{toolsgl}. 
 We first consider the case that $D:= D(x', R)>0$. We first  note that if 
\begin{equation}\label{defLe}
R+\eta\le t \le T:= \min \left(  2R, \Big((R+\eta)^d + \frac{D}{2C}\Big)^{\frac1d}\right)\end{equation}
with $C$ well-chosen,  we have 
\begin{eqnarray*}
 -\int_{\pa  B(x',t)  } \frac{\pa  h_{n,\eta}'}{\pa \nu} &=&-\int_{ B(x',t)}\Delta h_{n,\eta}' = c_d \int_{B(x',t)}\Big(\sum_{i=1}^n\delta_{x_i'}^{(\eta)} -\mu_0'(x)\Big)\\
 &\ge &c_{d}\left(D  -   \int_{B(x',t)\backslash B(x',R) } \mu_0'(x)\, dx\right)\ge c_d D- C \left( t^{d}-R^{d}\right)\\
&  \ge & \frac{c_d}{2} D   ,
\end{eqnarray*} with $C$ depends on $d$ and $\|\mu_0\|_{L^\infty}$, if we choose the same $C$ in \eqref{defLe}. 
By the Cauchy-Schwarz inequality, the previous estimate, and explicit integration, there holds 
\begin{eqnarray*}
 \int_{B(x',2R)}|E_\eta|^2&\ge&\int_{R+\eta}^{T}\frac{1}{|\pa B(x',t)|}\left(\int_{\pa  B(x',t)} \frac{\pa h_{n,\eta}'}{\pa \nu}\right)^2dt\\
 &=&C D^2 \int_{R+\eta}^{T}t^{-(d-1)}  \, dt=  CD^2 \left( g( R+\eta)-g(T)\right) .  \end{eqnarray*}
 Inserting the definition of $T$ and rearranging terms, one easily checks that we obtain \eqref{eq:control fluctu charge}.
There remains to treat the case where $D\le 0$.  This time, we let 
$$T \le  t \le R-\eta, \qquad    T:= \Big( (R-\eta)^d - \frac{D}{2C}\Big)^{\frac1d}$$
and if $C$ is  well-chosen  we have \begin{eqnarray*}
 -\int_{\pa  B(x',t)}\frac{\pa h_{n,\eta}' } {\pa \nu} &=&-\int_{ B(x',t)}\Delta h_{n,\eta}' = c_d \int_{B(x',t)}\Big(\sum_{i=1}^n \delta_{x_i'}^{(\eta)} - \mu_0'(x)\Big)\\
 &\le &c_{d}\left(D  -   \int_{B(x',R)\backslash B(x',t) } \mu_0'(x)\, dx\right)\le \frac{c_d}{2} D   ,
\end{eqnarray*}and the rest of the proof is analogous, integrating from $T$ to $R-\eta$. 
\end{proof}

The discrepancy in the number of points measures how regular a point distribution is, and, together with its variance, is a very important quantity from the point of view of the analysis of point processes, see e.g.  \cite{torsti}.
 Also in  approximation theory, the discrepancy is exactly the measure of the accuracy (or error) in the approximation, see  \cite{grabner,brauchartgrabner}.
 
 
\chapter{Definition(s) and properties of the renormalized energy}\label{chapdefw}
In the previous chapter, we have seen   a  splitting of the Hamiltonian  (Proposition    \ref{developpementasympto}) where a lower order term  in the form of a function $\mathcal H_n$, the ``precursor to the renormalized energy", appears.  We have worked so far at fixed $n$.  The ultimate goal is to find the asymptotic limit of this lower order term as $n\to +\infty$: a limiting object will appear which we call the {\it renormalized energy}. This energy is the total Coulomb interaction energy of an infinite configuration  of points in the whole space  in a constant neutralizing background. Such a system in called in physics a {\it jellium}.
This chapter is devoted to the definition(s) of this limiting object  itself and the study of some of its properties, before we proceed in the next chapter with deriving it as the $n\to +\infty$ limit. Thus this chapter can be read independently from the rest.

\section{Motivation and definitions} 
The goal of this chapter is to define a total Coulomb interaction for an infinite system of discrete point ``charges" in a constant neutralizing background of  fixed density $m>0$, related to a potential $h$ that solves (in the sense of distributions)
\begin{equation}\label{eqh}
-\Delta h=c_d \Big (\sum_{p\in \Lambda} N_p \delta_p - m\Big) \quad \text{in}\ \R^d, \ \text{ for} \ d\ge 2
\end{equation}
where $\Lambda $ is a discrete set of points in $\R^d$, and $N_p$ are  positive integers (the multiplicities of the points), respectively
\begin{equation}\label{eqh1d}
-\Delta h= c_d\Big (\sum_{p\in \Lambda} N_p \delta_p - m\delta_\R \Big) \quad \text{in} \ \R^2, \ \text{for} \ d=1\end{equation}
with $\delta_\R$ defined in \eqref{deltar}.

Again, such a system is often called a (classical) {\it jellium} in physics. The jellium model  was first introduced by \cite{wigner1} in the quantum case,  and can be viewed as a toy model for matter~: the points charges are then atoms, which interact (via electrostatic forces) with a cloud of electrons of density $m$.
 
 The reason why we need to consider such systems is that in the  previous chapter we dealt with  functions $h_n'$ that solved  a linear  equation of the  type :
\begin{equation}
\label{defhneta1}
-\Delta h'_n = c_d \Big( \sum_{i=1}^n \delta_{x'_i} - \mu_0'\Big) \quad \text{for } \ d\ge 2,
\end{equation}  in which it is easy, at least formally, to pass to the limit $n\to \infty$. 
Previously, we had chosen to center the blow-up at, say, the origin $0$ (respectively \eqref{eq1d1} for $d=1$).
We note that, in that case, the density $\mu_0'(x')$ equals by definition $\mu_0( x' n^{-1/d})$, so that, at least if $\mu_0$ is sufficiently regular,  $\mu_0'(x') \to \mu_0(0)$ pointwise as $n\to +\infty$, i.e. $\mu_0'$ converges to a constant. It is constant because $\mu_0$ varies much slower than the scale of the configuration of discrete points.
If  we had chosen to blow up around a different point, say $x_0$, then we would obtain instead the constant $\mu_0(x_0)$ as the limit.  This constant is the local density of the neutralizing background charge. As $n\to +\infty$, the number of points  becomes infinite  and they fill up the whole space, so that if we blow-up around an origin which lies in the support of $\mu_0$ (the droplet $\Sigma$), we obtain as a (at least formal) limit  as $n \to + \infty$ of \eqref{defhneta1} an equation of the form \eqref{eqh} (resp. \eqref{eqh1d} for $d=1$). Figure \ref{fig11} illustrates this blow up procedure around a point $x_0$ in $\Sigma$, the support of $\mu_0$. The final goal is to derive $\mathcal{W}$ as the governing interaction  for the limiting infinite point configurations.

\begin{figure}[h!]
\begin{center}
\includegraphics[scale=0.7]{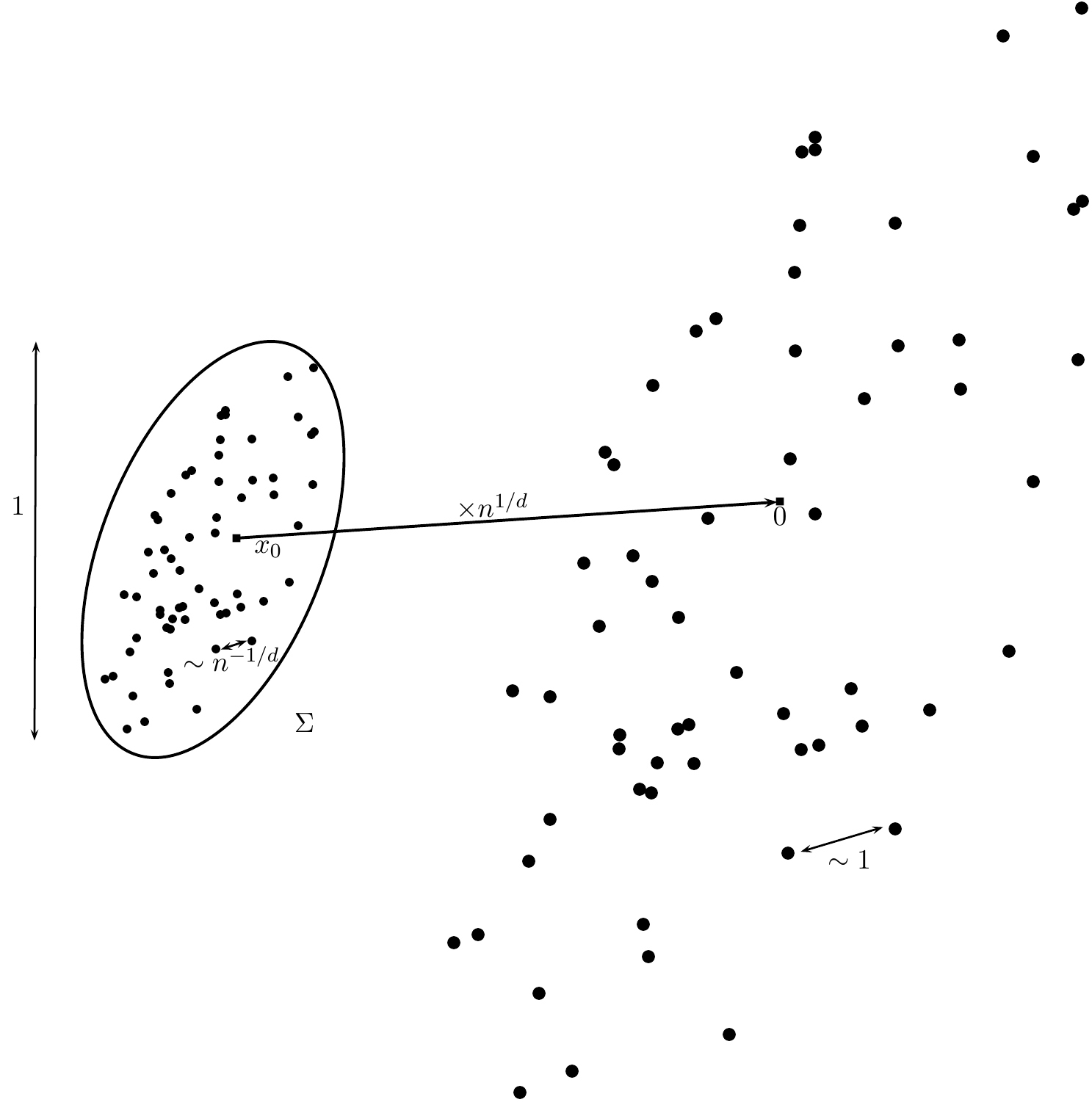}
\caption{An arbitrary blown-up configuration}\label{fig11}
\end{center}\end{figure}

We may  observe that although we do have controls on quantities
\begin{displaymath}
\int |\nabla h'_{n,\eta}|^2
\end{displaymath}
 which in turn give controls on the gradient $\nabla h'_{n,\eta}$ and $\nab h_n'$ (as seen in Lemma \ref{lem:fluctu field}), we do not fully control $h_n'$ itself, and thus we will not know  its limit $h$ itself. Also note that \eqref{eqh} determines $h$ from the data of the points only up to a harmonic function.

 Since we will work a lot with the gradient of $ h$, it is sometimes convenient to denote it by $E$, standing for  ``electric field." Indeed, $\nab h$  physically corresponds to the electric field generated by the  charge distribution  $\sum_p N_p \delta_p- m$. The equation \eqref{eqh} can then be rewritten with left-hand side $-\div E$, since $\div \nab =\Delta$.
 
We will give two definitions of the renormalized energy. They rely on the two different ways of subtracting off the self-interaction energy of each charge, that we have already encountered.  One definition, \`a la Bethuel-Brezis-H\'elein \cite{bbh},  was first introduced in \cite{ssgl} for the study of vortices in the Ginzburg-Landau  model and   was only originally written down  in dimensions 1 and  2. The other one was later introduced in \cite{rs}, it relies on the method of smearing out the charges (or truncating the potential) and works in any dimension $d\ge 2$. It was extended to Riesz interaction kernels $|x|^{-s}$ with $d-2\le s <d$ and to the $d=1$ logarithmic case in \cite{ps}.

 In the sequel, $K_R$ will denote the $d$-dimensional cubes $[-R,R]^d$.
 \begin{defini}[Admissible electric fields] \label{definitionsA} Let $m$ be a positive number. 
If $d \ge 2$,  we let $\bar{\mathcal{A}}_m$ be the set of gradient vector-fields $E = \nabla h$ that belong to $L^{q}_{\mathrm{loc}}(\R^d, \R^d)$ for all $q< \f{d}{d-1}$, \footnote{this is simply the best integrability of the gradient of the Coulomb kernel, as seen in Remark~\ref{rem41}} and such that 
 \begin{equation} \label{definitionabar}
 - \div E = c_d \Big(\sum_{p \in \Lambda} N_p \delta_p - m  \Big) \quad \text{in } \ \R^d
 \end{equation}
 for some discrete set $\Lambda \in \R^d $, and $N_p$ positive integers; resp. if $d=1$, the set of gradient vector-fields $E=\nab h \in L^q_{\mathrm{loc}}(\R^2, \R^2)$ for all $q<2$ such that
 $$ - \div E = c_d \Big(\sum_{p \in \Lambda} N_p \delta_p - m  \delta_\R \Big) \quad \text{in } \ \R^2.$$ 
 We also let $\mathcal{A}_m$ be the set of gradient vector-fields satisfying the same conditions but with all coefficients $N_p \equiv 1$ (i.e. there are no multiple points). 
 \end{defini}

\begin{defini}[Renormalized energy by smearing out the  charges \cite{rs}] \label{def1} Let $d \ge 2$. For any  $E\in \bar{\mathcal{A}}_m$, and any $\eta>0$ 
we define 
\be \label{esmeared}
E_\eta= E- \sum_{p\in \Lambda} N_p \nab f_\eta (\cdot -p)\ee
where $\Lambda $ and $N_p$ are associated to $E$ via \eqref{definitionabar}, and $f_\eta$ is as in \eqref{feta} ; and we let \be \label{Weta}
\mathcal W_\eta (E) = \limsup_{R\to +\infty} \dashint_{K_R} |E_\eta|^2 - m c_d g(\eta).\ee
We then define the renormalized energy $\mathcal{W}$ by 
\begin{equation} \label{definitionW1}
\mathcal{W}(E) := \underset{\eta \rightarrow 0}{\lim}\,  \mathcal W_\eta(E). 
\end{equation}
\end{defini} 
We will see below in Proposition \ref{promono} that this limit exists. 
We note that one may  equivalently define $h_{\eta}$ to be  the truncated  version of the electrostatic potential, that is
\begin{equation}\label{heta}
\nabla h_{\eta} = \nabla h - \sum_{p \in \Lambda} N_p \nabla f_{\eta}(\cdot-p)\end{equation} which solves 
\be \label{divee} -\Delta h_{\eta} = c_d\Big(\sum_{p \in \Lambda} N_p \delta_p^{(\eta)} - m\Big),\ee
 and then $E_\eta= \nab h_\eta$.

\begin{defini}[Renormalized energy by cutting out holes \cite{ssgl,ss2}]\label{def2}
Let $d\ge 2$. For any element $E$ of $\mathcal{A}_m$, we define 
the renormalized energy $W$ by \begin{equation} \label{definitionW2}
W(E) = \underset{R \to + \infty}{\limsup} \f{W(E, \chi_R)}{|K_R|}\end{equation}
where
\be \label{defiW} W(E,\chi_R): =  \underset{\eta \rightarrow 0}{\lim} \int_{\R^d \backslash  \cup_{p \in \Lambda} B(p, \eta)} \chi_R |E|^2 -  c_d  g(\eta)\sum_{p \in \Lambda} \chi_R(p)  
\end{equation} and 
 $\{\chi_R\}_{R > 0}$ is any family of cutoff functions on $\R^2$ such that 
\be\label{condchi}
\chi_R \equiv 1 \mbox{ on } K_{R-1}, \  \chi_R\equiv 0  \mbox{ in  } K^{c}_R \mbox{ and } ||\nabla \chi_R||_{\infty} \mbox{ is uniformly bounded. }
\ee
For  $d=1$, we define $W(E)$ in the same way, except $K_R=[-R,R]$,  $\{\chi_R\}$ is a  family of functions depending only on the first coordinate in $\R^2$, such 
 that $\chi_R \equiv 1 $ in $K_{R-1} \times \R$ and $\chi_R\equiv 0$  on $(K_R \times \R)^c$, and 
\begin{equation} W(E,\chi_R) =  \underset{\eta \rightarrow 0}{\lim} \int_{\R^2 \backslash  \cup_{p \in \Lambda} B(p, \eta)} \chi_R |E|^2 + 2\pi \log   \eta\sum_{p \in \Lambda} \chi_R(p)  
\end{equation} 
\end{defini}


We wrote  down here  an equivalent for   dimension $d \ge 3$ of the $W$ defined in \cite{ssgl} in dimension $2$.  It sufficed to replace $2\pi$ by $c_d$ and $\log \eta$  by $-g(\eta)$.   However, the  good properties we need for this energy, such as  the fact that it is bounded below and its minimum is achieved, have not been written down anywhere. 
It is likely that the methods we present here for showing this for  $\mathcal{W}$ extend to $W$ at least up to dimension $2$. However, it is not completely clear that we would be able
 to derive $W$ from the Coulomb gas Hamiltonian $H_n$ in dimension $d\ge 3$.

\section{First properties}
Let us now  make a few remarks on the definitions and the  comparison between them.
\begin{itemize}
\item Both definitions correspond to computing an average energy per unit volume in the limit of a large box size. It is necessary to do so because the system is infinite and otherwise 
would have an infinite energy.
\item It is not a priori clear how to define a  total Coulomb interaction of such a  ``jellium system", first because of the infinite size of the system as we just saw, second because of the lack of local charge neutrality of the system.
The definitions we presented avoid  having to go through computing  the sum of pairwise interaction between particles (it would not even  be clear how to sum them), but instead replace it with (renormalized variants of) the extensive quantity $\int |\nab h|^2$  (see \eqref{formalcomputation} and the comments following it).
\item In the definition of $W$, the need for the cut-off functions $\chi_R$ is due to the fact that if $\partial K_R$ intersects some ball $B(p, \eta)$, then the value of $W$ oscillates wildly between $+\infty $ and $-\infty$. Cutting off by a $C^1$ function removes the contribution of the points near the boundary, which is generally  negliglible compared to the total volume anyway.  
\item The two definitions correspond to two different ways of ``renormalizing" i.e. subtracting off the infinite contribution of the Dirac masses to the energy, but more importantly they correspond to reversing the order of  the limits $\eta\to 0$ and $R\to +\infty$. As a result the values of $\mathcal{W}$ and $W$ may differ.
This is already seen in the fact that $\mathcal{W}$ accepts multiple points (i.e. $N_p>1$), while $W$ is (formally) infinite for multiple points. Indeed,
let us suppose that $N_p \geq 2$ for some point $p \in \Lambda$ at distance at least $2r_0$ from its neighbors. Then 
\begin{equation*}
\int_{B(p, r_0)} \left|E_{\eta}\right|^2 \approx \int_{B(0, r_0)} |N_p \nabla g_\eta |^2  \approx  N_p^2 c_d g(\eta).
\end{equation*}
A few such multiple points add in $\mathcal{W}$ an extra contribution of $(N_p^2 - N_p)c_d   g(\eta)$, which disappears as $R \rightarrow + \infty$ when dividing by $|K_R|$. 
On the other hand, in the second definition, 
\begin{displaymath}
\int_{B(p, r_0) \backslash B(p, \eta)} |E|^2 \approx N_p^2 c_d g(\eta)
\end{displaymath}
which gives an extra contribution to $\mathcal W$ of order $(N_p^2 - N_p) c_d g(\eta)$ and this term diverges to $+ \infty$ as $\eta$ tends to $0$. 
\item However,  as we will see below, $W$ and $\mathcal{W}$ agree for configurations of points which are ``well-separated" i.e. for which all the points are simple and separated by a fixed minimum distance, because in that case the order of the limits can be reversed.
We will see below  that when dealing with minimizers of these energies, we can reduce to such well-separated configurations.


\item
The functions $W$ and $\mathcal{W}$ are functions of $\nab h$ and not only of the points (recall that $h$ may vary by addition of a harmonic function), however one  can make them functions of the points only by setting, if $\nu = \sum_{p \in \Lambda} N_p \delta_p$,  
\begin{displaymath}
\W(\nu) = \inf \left\lbrace \mathcal{W}(E), - \div E = c_d (\nu - m) \right\rbrace
\end{displaymath}
and the same for $W$, that is taking the infimum of $\mathcal{W}(E)$ on the set of gradient vector-fields $E$ that are compatible with $\nu$. 
Fortunately, these  are still measurable as functions of the points, thanks to a   measurable selection theorem. For more details, we refer to \cite[Sec. 6.6]{ss1}. 

\end{itemize}

We next gather some properties of $W$ and $\mathcal{W}$ whose proofs can be found in \cite{ssgl}, \cite{rs} respectively.
\begin{prop}[Properties of $W$ and $\mc{W}$] \mbox{}
\label{propertiesoftheenergies}
\\
\begin{enumerate}
\item For $d=1,2$, the value of $W$ does not depend on the choice of the family $\{\chi_R\}_{R > 0}$ as long as it satisfies \eqref{condchi}.
\item Both $W$ and $\mathcal{W}$ are Borel-measurable on $\mathcal{A}_m$ (resp. $\bar{\mathcal{A}}_m$) (and over $L^q_{loc}$ when extending the functions by $+ \infty$ outside their domains of definition). 
\item Scaling property : if $E$ belongs to $\mathcal{A}_m$ (resp. $\bar{\mathcal{A}}_m$),  then 
\begin{displaymath}
\widehat{E} := m^{1/d - 1} E\left(\f{\cdot}{m^{1/d}} \right) \in \mathcal{A}_1, \  \text{resp.} \ \bar{\mathcal{A}}_1.
\end{displaymath}
Moreover, we have, if $d\ge 3$,
\begin{equation} \label{scalingW}
\begin{cases}\mathcal{W} (\j) = m ^{2-2/d} \mathcal{W} (\hat{E}) \\
\mathcal{W}_\eta (E) = m ^{2-2/d} \mathcal{W}_{\eta m^{1/d}}  (\hat{E}) \end{cases}
\end{equation} and if $d=1,2$,
\be\label{scalingW2d}
\begin{cases}W(E)= mW(\hat{E}) -\frac{2\pi}{d} m \log m \\
\mathcal{W} (E) = m \mathcal{W} (\hat{E})-\frac{2\pi}{d} m \log m \\
\mathcal{W}_\eta (E) = m \mathcal{W}_{\eta m^{1/d}}  (\hat{E})-\frac{2\pi}{d} m  \log m  \\
\end{cases}\ee
One may  thus  reduce to studying  $W$ and  $\mathcal{W}$ on $\mc{A}_1$, resp. $\bar{\mc{A}_1}$. 
\item 
$\min_{\bar{\mathcal{A}}_1} \mathcal{W}$ is finite and achieved for any $d\ge 2$, and $\min_{\mathcal{A}_1} W$ is finite and achieved for $d=1,2$. Moreover, for $d=2$ the values of these two minima coincide.
\item
The minimum of $\mathcal{W}$ on $\bar{\mathcal{A}}_1$, resp. of $W$ on $\mathcal{A}_1$ for $d=1,2$, coincides with the limit as $N \rightarrow + \infty$ of the minima of $\mathcal{W}$ on vector-fields that are $(N\mathbb{Z})^d$-periodic (i.e. that live on the torus $\T_N = \R^d / (N \mathbb{Z})^d$). 
\end{enumerate}
\end{prop}
It can be expected that  in order to balance charges, the constant $m$, which is the density of the neutralizing background, is also the density of points associated to an $E \in \bar{\mathcal{A}}_m$. This is in fact true on average for configurations with finite energy. Let us show it in the case of $\mathcal{W}$ which is easier.

\begin{lemme} \label{neutralite0} Let $E \in \bar{\mathcal{A}}_m$ be such that $\mathcal{W}(E)<+\infty$. Then, letting  $$\nu = -\div E +m= \sum_{p\in \Lambda } N_p \delta_p,$$ we have
\begin{displaymath}
\lim_{R\to  + \infty} \f{\nu(K_R)}{|K_R|} = m.
\end{displaymath}
\end{lemme}

\begin{proof} First we show that 
\begin{equation}\label{contnu}
\begin{cases}
\nu(K_{R-2})\le  m |K_R|+ C R^{\frac{d-1}{2}} \|E_\eta\|_{L^2(K_R)} \\
\nu(K_{R+1}) \ge m |K_{R-1}| - C R^{\frac{d-1}{2}} \|E_\eta\|_{L^2(K_R)}. \end{cases}\end{equation}
Indeed, by a mean value argument, we may   find 
$t \in [R-1,R]$ such that
\begin{equation}\label{bbord}
\int_{\pa K_t} |E_\eta|^2 \le  \int_{K_R} |E_\eta|^2.\end{equation}
Let us next  integrate \eqref{divee} over $K_t$ and use Green's formula to find
\begin{equation}
\int_{ K_t} \sum_{p \in \Lambda } N_p \delta_p^{(\eta)}  - m |K_t| =  - \int_{\pa K_t} E _\eta\cdot \nu,
\end{equation}where $\nu$ denotes the outer unit normal.
Using the Cauchy-Schwarz inequality and \eqref{bbord}, we deduce that
\begin{equation}\label{feg}
\left|\int_{ K_t} \sum_{p \in \Lambda } N_p \delta_p^{(\eta)}  - m |K_t| \right|\le C R^{\frac{d-1}{2}} \| E_\eta\|_{L^2(K_R)} .\end{equation}
Since $\eta\le 1$, by definition of $\nu$ and since the $\delta_p^{(\eta)}$'s are supported on the $\pa B(p,\eta)$'s, we have $\nu(K_{R-2})\le \int_{ K_t} \sum_{p \in \Lambda } N_p \delta_p^{(\eta)} \le \nu(K_{R+1}) $.  The claim \eqref{contnu} follows. 

Since $\mathcal{W}(E)<+\infty$ then, by definition of $\mathcal{W}$, we have $\mathcal{W}_\eta(E) <+\infty$ for some $\eta<1$ (this is all we  really  use) and it follows
that  $\int_{K_R} |E_\eta|^2 \le C_\eta R^d$ for any $R>1$.  Inserting this into \eqref{contnu}, dividing  by $|K_R|$ and letting $R \to \infty$, we easily get the result.\end{proof}

\section{Almost monotonicity of $\mathcal{W}_\eta$ and lower bound for $\mathcal W$}
In this section, we prove the analogue of Lemma \ref{prodecr} but at the level of the limiting object. By the scaling formula \eqref{scalingW}--\eqref{scalingW2d}, it suffices to prove it over the class $\bar{\mathcal A}_1$.

\begin{prop}\label{promono}
Let  $E \in \bar{ \mathcal{A}}_1$. For any $1>\eta>\alpha>0$ such that $\mathcal W_\alpha(E) <+\infty$  we have 
$$\mathcal W_\eta(E) \le \mathcal W_\alpha (E )+ C\eta$$
where $C$ depends only on  $d$; and thus   $\lim_{\eta\to 0}\mathcal W_\eta(E)= \mathcal W(E)$ always exists. Moreover,
$\mathcal W_\eta$ is bounded below on $\bar{\mathcal A}_1$ by  a constant depending only on $d$.
\end{prop}

\begin{proof} 
Let   $E\in \bar{\mathcal{A}}_1$ and let $E_\eta$ be associated via \eqref{esmeared}. Assume $1>\eta>\alpha>0$. 
Let $K_R= [-R/2,R/2]^d $.
Let $\chi_R$ denote a smooth cutoff function equal to $1$ in $ K_{R-3}$  and vanishing outside $K_{R-2}$, with $|\nab \chi_R|\le 1$. As in the  proof of Lemma \ref{prodecr}  we note that 
$$E_\eta= E_\alpha + \sum_{p\in \Lambda} N_p \nab \fae(x-p)$$
and insert to expand
\begin{multline}
\int_{\R^{d}} \chi_R   |E_\eta|^2- \int \chi_R  |E_\alpha|^2 
\\=\sum_{p, q\in \Lambda} N_p N_q \int_{\R^{d}}  \chi_R  \nab \fae(x-p)\cdot \nab \fae(x-q) + 2\sum_{p\in \Lambda} N_p \int_{\R^{d}} \ \chi_R \nab \fae(x-p)\cdot E_\alpha.\end{multline} Using an integration by parts, we may write 
\begin{equation}\label{errormain}
\int_{\R^{d}}\chi_R   |E_\eta|^2- \int_{\R^{d}} \chi_R  |E_\alpha|^2 \le Error+ Main\, , 
\end{equation} where \begin{multline*}
Main:=
- \sum_{p, q \in\Lambda}N_pN_q  \int_{\R^{d}}  \chi_R \fae(x-p)\Delta \fae(x-q) \\- 2\sum_{p\in\Lambda}N_p \int_{\R^{d}}  \chi_R   \fae(x-p)\div E_\alpha.\end{multline*}
and 
\begin{multline*}
Error:= C\sum_{p, q\in K_{R-1} \backslash  K_{R-4}} N_p N_q \int_{\R^{d}} |\fae(x-p)||\nab \fae(x-q)|
\\+ \sum_{p \in K_{R-1} \backslash  K_{R-4}}N_p  \int_{\R^{d}}    | \fae(x-p)| |E_\alpha|.\end{multline*}
We will work on controlling $Error$ just below, and for now, using the fact that $\fae $ is supported in $B(0,\eta)$,  we may write similarly  as in the proof of Lemma \ref{prodecr}:
\begin{multline*}
Main
= c_d\sum_{p,q \Lambda} N_pN_q\int \chi_R  \fae (x-p) ( \delta_q^{(\eta)}- \delta_q^{(\alpha)} ) 
\\+ 2 c_d \sum_{p,q\in  \Lambda}N_p N_q  \int \chi_R \fae(x-p) \delta_q^{(\alpha)}    -2 c_d \sum_{   p \in  \Lambda} \int\chi_R  \fae(x-p) \\
\le c_d \sum_{p,q\in \Lambda} N_pN_q\int  \chi_R \fae (x-p) ( \delta_q^{(\eta)}+ \delta_q^{(\alpha)} )  +C \eta \sum_{p\in K_{R-2} \cap \Lambda} N_p
\end{multline*} where we have used that $\int |\fae|\le C \eta$ in view of  \eqref{intfeta}.
Since $\fae\le 0$ it follows that 
\begin{equation}\label{m01}
Main \le C  \eta \sum_{p\in K_{R-2} \cap \Lambda} N_p  .\end{equation}

Next, to control $Error$, we partition $K_{R-1}\backslash K_{R-4}$ into disjoint cubes $\mathcal C_j$  of sidelength centered at points $y_j$ and we denote by $\mathcal N_j =\sum_{p  \in \Lambda \cap \mathcal C_j} N_p$.
By Lemma \ref{lem:fluctu charge},     we have that 
$\mathcal N_j^2  \le C + C e_j$ where $$e_j:= \int_{B_2(y_j)}  |E_\alpha|^2 .$$
Using that the overlap of the $B_2(y_j)$ is bounded, we may write 
$$ \sum_j \mathcal N_j^2  \le C R^{d-1} + \sum_j e_j \le CR^{d-1} + \int_{K_{R}\backslash K_{R-5}} |E_\alpha|^2 .$$
We then may deduce, by separating the contributions in each $\mathcal C_j$  and using the Cauchy-Schwarz inequality and 
$\int |f_{\alpha, \eta} |^2  \le C \eta^d \g^2(\alpha)$, 
that
\begin{multline}
Error\le  C \eta^d \g^2(\alpha) \sum_j \mathcal N_j^2 + C \sum_j \eta^{d/2}\g( \alpha)\mathcal  N_j e_j^{1/2}
 \le C \eta^d \g^2( \alpha) \left( \sum_j \mathcal  N_j^2 + \sum_j e_j\right) \\
\le C \eta^d \g^2( \alpha)      \left( R^{d-1} + \int_{ K_{R}\backslash  K_{R-5}}  |E_\alpha|^2 \right)  .\end{multline}
Returning to \eqref{errormain} and \eqref{m01}  we have found that 
\begin{multline}\label{estm} 
\(\int   \chi_R  |E_\alpha|^2 - c_d \sum_{p \in\Lambda\cap K_{R-4}} N_p  \g(\alpha)\)-\( \int \chi_R  |E_\eta|^2 - c_d \sum_{p \in\Lambda\cap K_{R-4}} N_p  \g(\eta)\)\ge
\\
- C  \eta  \sum_{p \in\Lambda\cap K_{R-4}} N_p - C\eta^d  \g^2( \alpha)     \left( R^{d-1} + \int_{ K_{R}\backslash  K_{R-5}} |E_\alpha|^2 \right),
\end{multline} where $C$ depends only on $d$. 
In view of Lemma \ref{neutralite0} we have that $$ \lim_{R\to \infty} \frac{1}{R^d}\sum_{p\in \Lambda \cap K_R} N_p =m \qquad \lim_{R \to \infty}\frac{1}{R^d} \sum_{p\in  \Lambda \cap (K_R\backslash K_{R-5})}N_p =0.$$
In addition,  since $\mathcal W_\alpha(E)  <\infty$ and by definition of $\mathcal W_\alpha$,  we must have $$\lim_{R\to \infty} \frac{1}{R^d} \int_{ K_{R}\backslash  K_{R-5}}  |E_\alpha|^2=0.$$ Dividing \eqref{estm} by $R^d$ and letting $R \to +\infty$, we deduce the desired result $$\mathcal W_\alpha(E)- \mathcal W_\eta(E)  \ge -    C \eta .$$
It then immediately follows that $\mathcal W_\eta$ has a limit (finite or infinite) as $\eta \to 0$, and that $\mathcal W_\eta(E)$ is bounded below by, say, $\mathcal W_{1/2}(E) - Cm$, which in view of its definition is obviously bounded below by $-c_d - C$. 
\end{proof}
We deduce the following~:
\begin{coroll}\label{minoW}
$\mathcal W$ is bounded below on $\bar{\mathcal A}_m$ by a constant depending only on $m$ and $d$.\end{coroll}This property, in addition to its intrinsic interest, will turn out  to be crucial for us  in the next chapter when deriving rigorously $\mathcal{W}$ from the Coulomb gas  Hamiltonian, in the limit $n \to \infty$.

The almost monotonicity property of $\mathcal W_\eta$ has allowed for a rather simple proof of the boundedness  from below of $\mathcal W$. This is a place where the analysis differs a lot from that developed for $W$ in \cite{ssgl}: there, it is also proved that $W$ is bounded below, in dimension $d=2$ (it also works for $d=1$) but by a different method relying on a ``ball construction", \`a la Jerrard \cite{jerrard} and Sandier \cite{sandier}, which only worked in dimensions $1$ and $2$. The reason why the methods cannot be interchanged is that the order of the limits in the definitions of $W$ and $\mathcal W $ are reversed.

\section{Well separated and periodic configurations}
In this section, we are going to see how to compute explicitly the renormalized energies for periodic configurations of points. We  first remark the equivalence of the two ways of computing the renormalized energy for configurations of points which are well-separated.

\begin{lemme}[The energy of well-separated configurations]\label{lemequiv}\mbox{}\\
Let $d\ge 2$. Assume that $h$ solves  a relation of the form
\be \label{eqconsti2}
- \Delta h = c_d \Big(\sum_{p\in \Lambda} \delta_{p} - \mu(x)\Big) \mbox{ in an open set } \Omega \subset \R^d
\ee in the sense of distributions, 
for some discrete set $\Lambda$, and $\mu\in L^\infty(\Omega)$. 
Assume that the points are well-separated in the sense that for some $r_0>0$, 
\begin{equation}\label{wellsep}\min\left(\min_{p\neq p'\in \Lambda\cap \Omega} |p-p'|, \min_{p\in \Lambda\cap \Omega} \dist (p, \partial \Omega)\right)\ge 2 r_0>0.\end{equation}
Then, letting $h_\eta= h+ \sum_{p\in \Lambda} f_\eta(\cdot -p)$, we have
\begin{multline}\label{lequi}
 \int_{\Omega} |\nab h_\eta|^2 -  \# ( \Lambda\cap \Omega)c_d g(\eta)\\
 = \int_{\Omega \backslash \cup_{p\in \Lambda} B(p,\eta)} |\nab h|^2-  \# ( \Lambda\cap \Omega)c_d g(\eta)  + \# (\Lambda \cap \Omega) o_\eta(1)\|\mu\|_{L^\infty(\Omega)}, \end{multline}
where $o_\eta(1) \to 0$ as $\eta \to 0$ is a function that depends only on the dimension. \end{lemme}
\begin{rem} The result is in fact true with appropriate modification for $\mu$ as in Remark \ref{mu0faible}.
\end{rem}
\begin{proof}This is quite similar to the proof of Lemma \ref{lem31} or Lemma \ref{prodecr}.

Since $h_\eta$ is defined by \eqref{heta},  the $B(p, r_0) $ are disjoint and included in $\Omega$, and $f_{\eta}$ is identically $0$  outside of $B(0, \eta)$ we may write for any $\eta<r_0$, and any $0<\alpha<\eta$,
\begin{multline}\label{pu12}
\int_{\Omega\backslash \cup_{p\in \Lambda} B(p, \alpha)}  |\nab h_\eta|^2 = \int_{\Omega
 \backslash \cup_{p\in \Lambda} B(p, \alpha)}
 |\nab h|^2 + \# (\Lambda \cap \Omega) \int_{B(0,\eta)\backslash B(0,\alpha)}  |\nab f_{\eta} |^2 \\ - 2 \sum_{p \in \Lambda} \int_{B(p, \eta)\backslash B(p,\alpha)}  \nab f_\eta(x-p) \cdot \nab h.
 \end{multline}
First we note that since $f_\eta = g(\alpha)- g(\eta)$ on $\pa B(0, \alpha)$ and $f_\eta=0$ on $\pa B(0, \eta)$, using Green's formula and \eqref{divf},   ${\nu}$ denoting the outwards pointing unit normal to $\pa B(0, \alpha)$, we have
\begin{equation}
\label{nabfe}
\int_{B(0,\eta)\backslash B(0,\alpha)} |\nab f_{\eta}|^2 = -\int_{\pa B(0, \alpha)} (g(\alpha)- g(\eta))\f{\pa  f_\eta }{\pa \nu}= c_d(g(\alpha)- g(\eta) ).\end{equation} Next,
  using Green's formula and \eqref{eqconsti2} we have
\begin{multline*}\int_{B(p, \eta)\backslash B(p,\alpha)}  \nab  f_\eta(x-p) \cdot \nab h\\= - c_d \int_{B(p,\eta)\backslash B(p,\alpha) }  f_\eta(x-p) \mu(x)\, dx
- (g(\alpha) -g(\eta))\int_{\pa B(p, \alpha)}  \frac{\pa h}{\pa \nu} .\end{multline*}
For the first term in the right-hand side we write  as in \eqref{intfeta}\begin{equation*}
\left|\int_{B(p,\eta)\backslash B(p,\alpha) }  f_\eta(x-p) \mu(x)\, dx
\right|
 \le
 \|\mu\|_{L^\infty}o_\eta(1),\end{equation*}
where $o_\eta(1)$ depends only on  $d$. For the second term, by \eqref{eqconsti2} and Green's theorem again,  we have
$$- \int_{\pa B(p, \alpha)}\f{\pa h }{\pa \nu}=  c_d + O(\|\mu\|_{L^\infty} \alpha^d)  .$$
Inserting these two facts we deduce
\begin{equation*}\int_{B(p, \eta)\backslash B(p,\alpha)}  \nab f_\eta(x-p) \cdot \nab h= c_d    (g(\alpha)-g(\eta)) +  O(\|\mu\|_{L^\infty} \alpha^d g(\alpha) ) +  \|\mu\|_{L^\infty}o_\eta(1).\end{equation*}
Combining this with  \eqref{pu12} and  \eqref{nabfe}, we find
\begin{multline*}
\int_{\Omega\backslash \cup_{p\in \Lambda} B(p, \alpha)}  |\nab h_\eta|^2 = \int_{\Omega
 \backslash \cup_{p\in \Lambda} B(p, \alpha)}
 |\nab h|^2\\ + \# ( \Lambda\cap \Omega) \( c_d g (\eta)  -c_d g (\alpha) +o_\alpha(1) +\|\mu\|_{L^\infty} o_\eta(1) +  \|\mu\|_{L^\infty}O( \alpha^d g (\alpha) ) \).
\end{multline*}Letting $\alpha \to 0$, in view of the definition of $W$ in \eqref{defiW}, we obtain the result.

\end{proof}

\begin{coroll}[$W$ and $\mathcal{W}$ coincide in 2D for well-separated points]\label{propequiv}\mbox{}\\
Assume $d=2$, and let $E \in \mathcal{A}_1$ be such that $\mathcal{W}(E) <+\infty$ and the associated set of points satisfies $\min_{p\neq p' \in \Lambda } |p-p'| \ge 2r_0>0$ for some $r_0>0$. Then 
$\mathcal{W}(E)= W(E)$.
\end{coroll}
For the proof, see \cite[Prop. 3.3]{rs}. It is very likely that this can be extended at least to dimension $d=3$.

We now turn to periodic configurations, and show that, for them, $W $ or $\mathcal{W}$ can be computed and expressed  as a sum of pairwise Coulomb-like interactions between the points. By periodic configuration, we mean a configuration on the fundamental cell of a torus, repeated periodically, which can be viewed as a configuration of $N$ points on a torus (cf. Fig. \ref{fig6}).

\begin{figure}[h!]
\begin{center}
\includegraphics[scale=0.9]{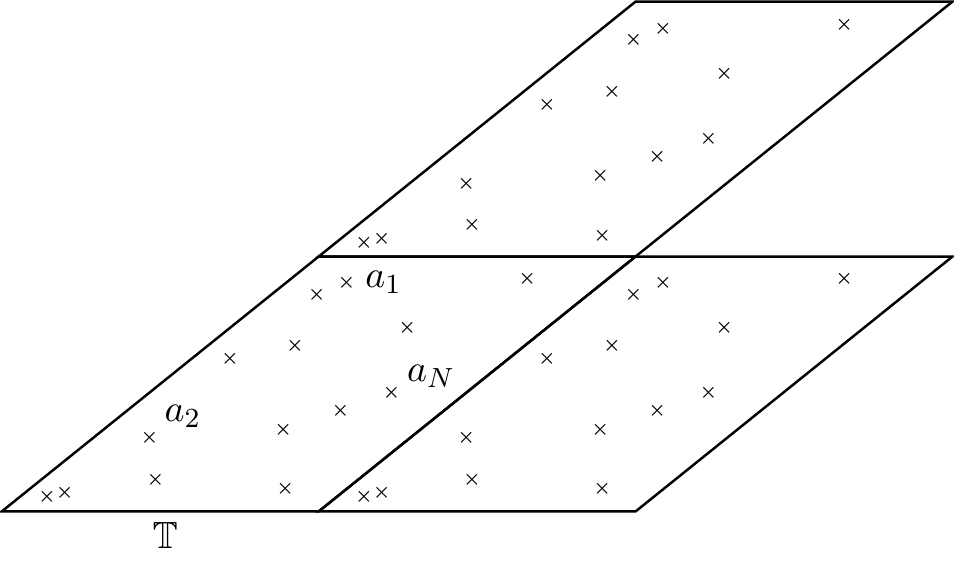}
\caption{Periodic configurations}\label{fig6}
\end{center}
\end{figure}
 
\begin{prop} \label{propperiodic} Let $a_1, \dots, a_N$ be $N$ points in a torus $\T$ of volume $N$ in $\R^d$, $d\ge 2$. 
\begin{enumerate}
\item If there is a multiple point, then for any $E$ compatible with the points (i.e. such that $-\div E = c_d \left(\sum_{i=1} \delta_{a_i} - 1 \right)$), we have  $ \mathcal{W}(E) = + \infty$.
\item If all points are distinct, letting $H$ be the periodic solution to 
\be\label{perH}
- \Delta H = c_d \Big(\sum_{i=1}^N \delta_{a_i} -1 \Big), \quad \int_{\T} H = 0;
\ee
then any other periodic $E$ compatible with the points  satisfies 
\begin{equation}\label{compweh}
\mathcal{W}(E)\geq \mathcal{W}(\nabla H).
\end{equation}
 Moreover, 
\begin{equation} \label{energieperiodique}
\boxed{
 \mathcal{W}(\nabla H) = \f{c_d^2}{N} \sum_{i \neq j} G(a_i - a_j) + c_d^2 \lim_{x \rightarrow 0} \Big(G - \f{g}{c_d}\Big)}
\end{equation}
where $G$, the Green function of the torus, solves 
\begin{equation}\label{527}
-\Delta G = \delta_0  - \f{1}{|\T|} \mbox{ over } \T,\quad  \int_{\T} G  = 0.
\end{equation}
\item If $d=2$, the same results hold true with $W$ instead of $\mathcal{W}$.
\item If $d=1$ and $g(x)=- \log |x|$, the same results hold true with $W$,  $c_d=2\pi$ and $H$ the solution on $\R/(N\mathbb{Z})\times \R$ with  zero-mean on the real axis  of
\begin{displaymath}
- \Delta H =2\pi  \Big(\sum_{i=1}^N \delta_{a_i} -\delta_{\R} \Big),
\end{displaymath}
\be \label{formulG} 
G(x)= -\frac{1}{2\pi} \sum_{i\neq j} \log \left|2 \sin \frac{\pi x}{N}\right| , \ee i.e. we have 
\begin{equation} \label{energieperiodique1d}
\boxed{ W(\nabla H) = - \f{2\pi }{N} \sum_{i \neq j}   \log \left|2 \sin \frac{\pi (a_i-a_j) }{N}\right| -  2\pi  \log \frac{2\pi}{N}.}
\end{equation}

\end{enumerate}
\end{prop}

\begin{proof} We start with the first assertion : let $N_i$ be the multiplicity of $a_i$. We have $\sum_i N_i=N$.  On the torus, for any $E$ compatible with the points, we have the basic lower bound
\begin{equation}\label{lbbasic}
\int_{\T}  |E_{\eta}|^2 \geq c_d\sum_{i} N_i^2 g(\eta) - CN.\end{equation} This can be proven with calculations similar to the proof of Lemma \ref{lemequiv}, but we give instead one relying on a similar argument to 
the proof of Lemma \ref{lem:fluctu charge}.
Letting $r>0$ be the minimal distance from $a_i$ to the other points, we may write, using the Cauchy-Schwarz inequality and Stokes's theorem
\begin{eqnarray*}
\int_{B(a_i, r)} |E_\eta|^2 & \ge &  \int_{\eta}^r \( \frac{1}{|\pa B(0, t)|}  \( \int_{\pa B(a_i, t)} E_\eta \cdot \nu \)^2 \)\, dt\\
 & = & \frac{1}{|\pa B(0,1)|} \int_\eta^r \frac{1}{t^{d-1}} (c_d( N_i- |B(0,1)| t^d) )^2 \, dt,
\end{eqnarray*}
where we used the relation  $-\div E_\eta=c_d \(\sum_{i=1}^n \delta_{a_i}^{(\eta)}-1\)$.
Using then \eqref{defcd} and integrating explicitly, we  obtain
\begin{equation}\label{minocs}
\int_{B(a_i, r)} |E_\eta|^2 \ge\frac{d-2}{c_d}  c_d^2 \frac{g(\eta) - g(r) }{d-2}(N_i^2 -CN_i)\end{equation}
where $C$ depends on $d$ and $r$, and this implies \eqref{lbbasic}.

Meanwhile,  by periodicity, we  have that 
\begin{displaymath}
\mathcal{W}(E) = \underset{\eta \rightarrow 0} \liminf \left( \f{1}{|\T|} \int_{\T} |E_{\eta}|^2 - c_d g(\eta)    \right).
\end{displaymath}
Since the volume $|\T| = N$, it follows that \begin{displaymath}
\mathcal{W}(E) \geq \f{1}{N} \Big( \sum_i N_i^2 - N\Big) c_d g(\eta) - C
\end{displaymath}
and the limit as $\eta\to 0$  of this quantity is $+ \infty$ unless $\sum_i N_i^2 = N$, which imposes that   all the multiplicities $N_i$ be  equal to $1$. The same goes for $W$ (for which we already know that multiple points give infinite value).

Let us now prove the second assertion.

Assume now that $E_1$ and  $E_2$ are two admissible periodic gradient vector-fields, with $E_1 = \nabla h_1$ and $E_2 = \nabla h_2$. Since $\Delta (h_1 - h_2) = 0$,  $E_1$ and $E_2$ differ by the gradient of a harmonic function, but they are also periodic so this difference must  in fact be a constant vector $\vec{c}$, and the same for $(E_1)_\eta$ and $(E_2)_\eta$. We can then compute 
\begin{equation*}
 \int_{\T} |(E_1)_\eta|^2 - \int_{\T}|(E_2)_\eta |^2 
=   \int_{\T } |\vec{c}|^2 + 2 \vec{c} \cdot \int_{\T} (E_2)_\eta  .
\end{equation*}
If $E_2 = \nabla H$ for some $H$ periodic, then $\int_{\T} (E_2)_\eta = \int_{\T} \nabla H_\eta = 0$, hence we 
deduce \eqref{compweh}.

Let us now turn to the proof of \eqref{energieperiodique}. Let $H$ be the periodic solution with mean zero. It is easy to see that $H(x) = c_d \sum_{i=1}^N G(x-a_i)$ with $G$ the Green function defined in the proposition, and thus $ H_\eta(x)= c_d \sum_{i=1}^N G(x-a_i)- \sum_{i=1}^n f_\eta(x-a_i)$. Also $G=\frac{1}{c_d}g + \phi$ with $\phi $ a continuous function.  Inserting  all this and using Green's formula, we find 
\begin{multline}
\int_{\T} |\nab H_\eta|^2 = - \int_{\T} H_\eta \Delta H_\eta\\= c_d \int_\T\( c_d \sum_{i=1}^N G(x-a_i)  - \sum_{i=1}^N f_\eta(x-a_i)\) \Big(\sum_{j=1}^N \delta_{a_j}^{(\eta)} - \frac{1}{N}\Big)\\
= N c_d (g(\eta)+c_d \phi(0))  +c_d^2  \sum_{i\neq j} G(a_i-a_j) +c_d  \int_{B(0,\eta)} f_\eta+o(1)
\end{multline}
as $\eta \to 0$, where we have used that $f_\eta $ vanishes on $\pa B(0, \eta)$ where $\delta_0^{(\eta)}$ is supported, and that $\int_{\T} G=0$. Using \eqref{intfeta}, letting  $\eta \rightarrow 0$ and dividing by $N$ gives (\ref{energieperiodique}).

For the third  assertion, if all the points are simple and the configuration is periodic, it follows that the points are well-separated, i.e. satisfy the assumptions of Lemma \ref{lemequiv}. Thus we know that 
\begin{equation}
\int_{\T} |E_{\eta}|^2 - N c_d g(\eta) =
\lim_{\eta \to 0} \int_{\T\backslash \cup_{i=1}^N B(a_i, \eta)}  |E|^2  - N c_d g(\eta)+ o_{\eta}(1)
\end{equation} 
and   this  immediately proves the identity 
\begin{equation}\label{identw}
\mathcal{W}(E) = \f{1}{|\T|} \(\lim_{\eta \to 0} \int_{\T\backslash \cup_{i=1}^N B(a_i, \eta)}  |E|^2   - N c_d g(\eta)\)= W(E, \indic_{\T}),
\end{equation}with the notation of \eqref{defiW}.
At this point, one can  also check that in dimension $d=2$, the right hand side is also equal to $W(E)$ in this setting  (this requires a little more care to show that the effect of the cut-off function is negligible, see\cite{ssgl} for details).   This implies the results of the third item.

 The proof of item 4 is very similar, we refer to \cite{ssgl,ss2} for details.
In the case of dimension $1$, the suitable Green function $G$ can be explicitly solved by Fourier series, and one finds the  formula \eqref{formulG}.
\end{proof}

In dimension $d \ge 2$, equation \eqref{527} can also be solved somewhat explicitly. For   a torus $\T = \R^d / (\mathbb{Z} \vec{u}_1 + \dots + \mathbb{Z} \vec{u}_d)$ of volume $N$, corresponding to the lattice $\Lambda = \mathbb{Z} \vec{u}_1 + \dots \mathbb{Z} \vec{u}_d$ in $\R^d$, one may first  express $G$ solving \begin{equation} \label{fourierreseau}
- \Delta G = \delta_0 - \f{1}{|\T|} = \delta_0 - \frac{1}{N}, \quad \int_{\T} G = 0
\end{equation}
as a Fourier series : 
\begin{equation}
G = \sum_{\vec{k} \in \Lambda^*} c_{\vec{k}} e^{2i\pi \vec{k} \cdot \vec{x}},
\end{equation}
where $\Lambda^*$ is the dual lattice of $\Lambda$, that is 
\begin{displaymath}
\Lambda^* = \{\vec{q} \in \R^d, \vec{q} \cdot \vec{p} \in \mathbb{Z} \textrm{ for all } \vec{p} \in \Lambda \}.
\end{displaymath}
Plugging this  into \eqref{fourierreseau}, one sees that the coefficients $c_{\vec{k}}$ must satisfy the relations 
\begin{displaymath}
-(2i \pi)^2 |\vec{k}|^2 c_{\vec{k}} = 1 - \delta_{\vec{k},0}
\end{displaymath} where $\delta_{\vec{k},0}$ is $0$ unless $\vec{k}=0$,
and $c_0  = \int_{\T} G = 0$ by assumption. This is easily solved by $c_{\vec{k}} = \f{1}{4\pi^2 |\vec{k}|^2}$ for $\vec{k} \neq 0$, hence  the formula
\begin{equation}\label{Geisenstein}
G(\vec{x}) = \sum_{\vec{k} \in \Lambda^* \backslash \{0\}} \f{e^{2i \pi \vec{k} \cdot \vec{x}}}{4\pi^2 |\vec{k}|^2}.
\end{equation}
Such a series is called an Eisenstein series, cf. \cite{lang} for reference and formulas on Eisenstein series.


\section{Partial results on the minimization of $W$ and $\mathcal{W}$, and the crystallization conjecture} \label{sec54}
We have seen in item 5 of Proposition \ref{propertiesoftheenergies} that the minima of $\mathcal{W}$ and $W$ can be achieved as limits of the minima over periodic configurations (with respect to larger and larger tori). On the other hand, Proposition \ref{propperiodic} provides  a more explicit expression for periodic configurations. 
In dimension $d=1$ (and in that case only) we know how to use this expression \eqref{energieperiodique1d}  to identify 
 the minimum over periodic configurations~:  a convexity argument (for which we refer to \cite[Prop. 2.3]{ss2}) shows that the minimum is achieved when the points are equally spaced, in other words for the lattice or crystalline distribution $\mathbb{Z}$ (called ``clock distribution" in the context of orthogonal polynomials, cf. \cite{simon}).  Combining with the result of item~5  of Proposition \ref{propertiesoftheenergies} allows to identify $\min_{\mathcal{A}_1} W$~:
\begin{theo}[The regular lattice is the minimizer in 1D \cite{ss2}]\label{thm1d}
If $d=1$, we have $$\min_{\mathcal{A}_1} W= -2\pi \log (2\pi)$$ and this minimum is achieved by gradients of periodic potentials $h$  associated to the lattice (or clock) distribution $\Lambda = \mathbb{Z}$.  
\end{theo}
Of course, the minimum over any $\mathcal{A}_m$ is deduced from this by scaling (cf. \eqref{scalingW2d}).
There is no uniqueness of minimizers, however a uniqueness result can be proven when viewing $W$ as a function of stationary point processes, cf. \cite{leble}.

In higher dimension, determining the value of $\min W$ or $\min \mathcal{W}$  is an open question, even though it would suffice to be able to minimize in the class of periodic configurations with larger and larger period, using the formula \eqref{energieperiodique}. The only question that we can answer so far is that of the minimization over the restricted class of pure lattice configurations, in dimension $d=2$ only,
i.e.   vector fields which are gradient of functions that are periodic with respect to a lattice $\mathbb{Z} \vec{u} + \mathbb{Z} \vec{v}$
 with $det(\vec{u}, \vec{v}) = 1$, corresponding to configurations of points that can be identified with $\mathbb{Z} \vec{u} + \mathbb{Z} \vec{v}$. In this case, we have :
\begin{theo}[The triangular lattice is the minimizer over lattices in 2D] \label{minimisationreseau} \mbox{}
The minimum of $W$, or equivalently $\mathcal{W}$, over this class of vector fields is achieved uniquely by the one corresponding to the  triangular ``Abrikosov"  lattice.
\end{theo}
Here the triangular lattice means $ \mathbb{Z} + \mathbb{Z} e^{i \pi/3}$, properly scaled, i.e. what is called the Abrikosov lattice in the context of superconductivity, cf. Chap. \ref{chap-intro}.

When restricted to lattices, 
$W$ corresponds to a ``height" of the associated flat torus (in Arakelov geometry). With that point of view, 
the  result was already known since  \cite{osp}, a fact we had not been aware of. The same  result was also obtained  in \cite{chen-oshita} for a similar energy.
We next give a sketch of the proof from \cite{ssgl}, which is not very difficult thanks to the fact that  it reduces (as \cite{osp} does) to 
the same question for a certain modular function, which was solved by number theorists in the 50's and 60's.

\begin{proof}[Proof of Theorem \ref{minimisationreseau}] 
Proposition \ref{propperiodic}, more specifically \eqref{energieperiodique},    provides an explicit formula for the renormalized energy of such  periodic configurations. Using  \eqref{Geisenstein} to express $G$, and denoting by $H_\Lambda$ the periodic  solution associated with \eqref{perH}, we find that 
\be\label{viaeis}
\mathcal{W}(\nab H_\Lambda) = \lim_{x \rightarrow 0} \left( \sum_{\vec{k} \in \Lambda^* \backslash \{0\}  } \f{e^{2i \pi \vec{k} \cdot \vec{x}}}{4\pi^2 |\vec{k}|^2}  + 2\pi \log x\right).
\ee
By using either the “first Kronecker limit formula” (cf. \cite{lang}) or a direct computation, one shows that in fact
\begin{equation}\label{viaepsteinzeta}
\mathcal{W}(\nab H_\Lambda) =C_1 + C_2 \lim_{x \rightarrow 0, x > 0} \left( \sum_{\vec{k} \in \Lambda^* \backslash \{ 0\} } \f{1}{|\vec{k}|^{2+x}} - \int_{\R^2}\f{dy}{1 + |y|^{2+x} }\right),
\end{equation}
where $C_1$ and $C_2>0$ are constants.
The series $\sum_{\vec{k} \in \Lambda^* \backslash \{ 0\} } \f{1}{|\vec{k}|^{2+x}} $  that appears is now the ``Epstein Zeta function" of the dual lattice $\Lambda^*$. The first Kronecker  limit formula allows to pass from one modular function, the Eisenstein series, to another, the Epstein Zeta function. Note that both formulas \eqref{viaeis} and \eqref{viaepsteinzeta}, when $x \to 0$,  correspond to two different ways of regularizing the divergent series $\sum_{p \in \Lambda^*\backslash \{0\}} \frac{1}{|p|^2}$, and they are in fact explicitly related.

The question of minimizing $\mathcal{W}$ among lattices is then reduced to minimizing the Epstein Zeta function 
\begin{displaymath}
\Lambda \mapsto \zeta_{\Lambda} (x) : =  \sum_{\vec{k} \in \Lambda \backslash \{0\} 0} \f{1}{|k|^{2+x}}
\end{displaymath}
as $x \rightarrow 0$. But results from \cite{cassels,rankin,ennola,ennola2,diananda,montgomery} assert that 
\begin{equation}
\zeta_{\Lambda} (x) \geq \zeta_{\Lambda_{triang}} (x), \ \forall x > 0
\end{equation}
 and the equality holds if and only if $\Lambda = \Lambda_{triang}$ (the triangular lattice).  Because that lattice is self-dual, it follows that it is the  unique minimizer.
 \end{proof}

One may ask  whether this triangular lattice does achieve the global minimum of $W$ and $\mathcal{W}$.  The fact that the Abrikosov lattice is observed in superconductors, combined with the fact -- which we will see later -- that $W$ can be derived  as 
 the limiting minimization problem of Ginzburg-Landau, justify to conjecture this~:
\begin{conjecture}
In dimension $d=2$, the value of $\min_{\mathcal{A}_1} W= \min_{\bar{\mathcal{A}}_1}\mathcal{W} $ is equal to the value at the vector field associated to  the triangular lattice of volume 1. 
\end{conjecture}

It was recently proven in \cite{betermin} that this conjecture is equivalent to a conjecture of Brauchart-Hardin-Saff \cite{bhs} on the next order term in the asymptotic expansion of the minimal logarithmic energy on the sphere (an important problem in approximation theory, also related to Smale's ``7th problem for the 21st century"), which is obtained by  formal analytic continuation, hence by very different arguments. This thus reinforces the plausibility of this conjecture.

 In dimension $d\ge 3$ the  computation of the renormalized energy restricted to the class of lattices holds but the meaning  of \eqref{viaepsteinzeta} is not clear. The minimization of the Epstein Zeta function over lattices is then  an open question (except in dimensions 8 and 24). In dimension $3$, both the FCC (face centered cubic) and BCC (boundary centered cubic) lattices  (cf. Fig. \ref{fig8}) could play the role of the triangular lattice, but it is only conjectured that FCC is a local minimizer (cf. \cite{sarns}), and so by duality BCC can be expected to  minimize $\mathcal W$.
 
 \begin{figure}[h!]
 \begin{center}
 \includegraphics[scale=0.5]{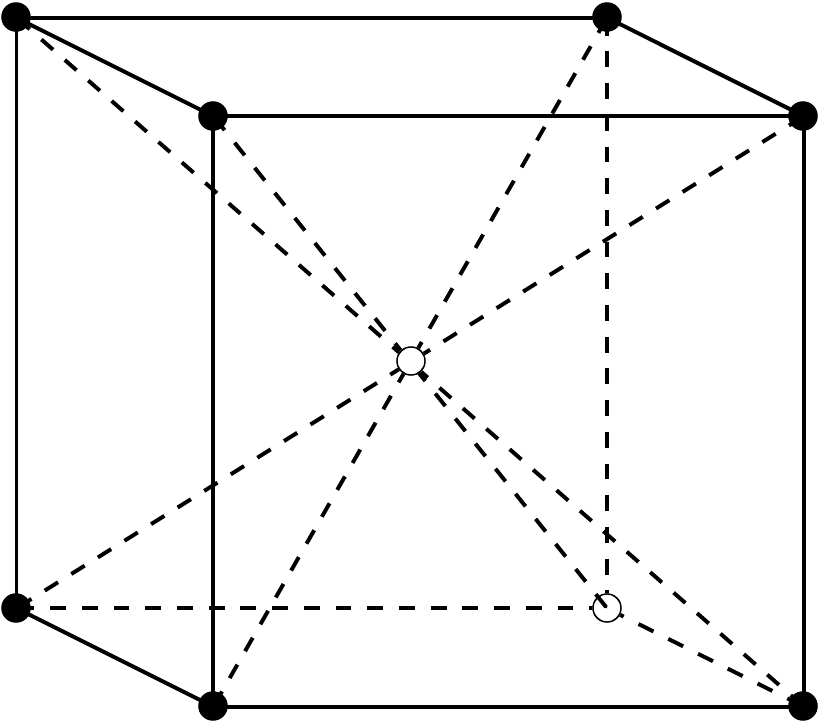}
 \includegraphics[scale=0.5]{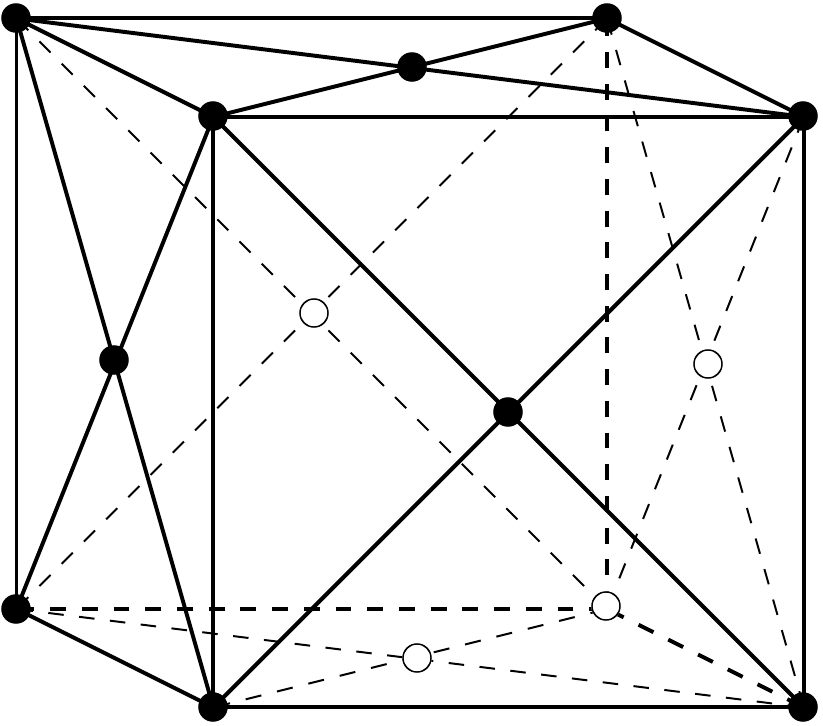}
\caption{BCC and FCC lattices}\label{fig8}
\end{center}
\end{figure}

Similarly, one may conjecture that, at least (and probably only) in low dimensions, the minimum of $\mathcal W$ is achieved by some particular lattice.

Proving the conjecture belongs to the wider class of crystallization problems. A typical question of this sort is,  given a potential $V$ in any dimension, to determine the point positions that minimize
$$\sum_{i\neq j} V(x_i-x_j)$$ (with some kind of boundary condition), or rather
$$\lim_{R\to \infty}\frac{1}{|B_R|}\sum_{i\neq j, x_i, x_j \in B_R} V(x_i-x_j),$$
and to determine whether the minimizing configurations are perfect lattices.  Such questions are fundamental in order to understand
the  crystalline structure of matter.
 They also arise in  the arrangement of   Fekete points \cite{sk} and the  ``Cohn-Kumar conjecture" \cite{ck}.
One should immediately stress that there are very few positive results in that direction in the literature  (in fact it is very rare to have a proof  that the solution to any minimization problem is periodic). Some exceptions include the two-dimensional sphere packing  problem, for which Radin \cite{radin} showed that the minimizer is the triangular lattice, and an extension of this by Theil \cite{theil} for a class of very  short range Lennard-Jones potentials.
The techniques used there do not apply to Coulomb  interactions, which are much longer range. Let us mention  another recent positive result. 
 The question of minimization of $\mathcal{W}$ can also  be very informally rephrased as that  of finding
$$\min`` \|\sum_p \delta_{p}-1\|_{(H^{1})^*}"  $$
where the quantity is put between brackets to recall that $\delta_{p}$ does not really belong to the dual of the Sobolev space  $H^{1}$ but rather has to be computed in the renormalized way that defines $\mathcal{W}$. A closely related problem is to find
 $$ \min \|\sum_p \delta_{p} -1\|_{\mathrm{Lip}^*},$$ and it  turns out to be  much easier. It is  shown by Bourne-Peletier-Theil in \cite{bpt} with a relatively short proof that again the triangular lattice achieves the minimum.

 We finish by referring to  some extra results. 
  
With Rota Nodari, in \cite{rns}, we  showed the equivalence between several ways of phrasing the  minimization of $W$ in dimension $2$ over a finite size box~: minimization with prescribed boundary trace and minimization among periodic configurations. In all cases, we were  able to prove, in the spirit of \cite{aco},  that the energy density and the points were uniformly distributed at any scale $\gg 1$, in good agreement with (but of course much weaker than!) the conjecture of periodicity of the minimizers.

Even though the minimization of $\mathcal W$ is only conjectural, it is natural to  view it as (or expect it to be)  a quantitative ``measure of disorder" of a configuration of points in the plane.
    In this spirit,  with Borodin   \cite{bs},
    we used $W$ (or rather a variant of it) in dimensions $1$ and $2$ to quantify and compute explicitly  the disorder of some classic random point processes in the plane and on the real line.

\chapter{Deriving $\mathcal{W}$ as the large $n$ limit~: lower bound via a general abstract method}
\label{derivingw}

Our goal in this chapter is to pass to the limit $n \to \infty$ in the results obtained in Chapter \ref{chapsplit}, starting from Proposition \ref{propRS}, in order to extract $\mathcal{W}$ as a limiting energy. The main task is to obtain a lower bound in the limit $n\to \infty$, which is expressed in terms of an average of  $\mathcal{W}$ with respect to a suitable measure that encodes all the possible blow-up profiles. This is accomplished via a general method which can be formulated abstractly, and which we start by presenting.

\section{Lower bound for $2$-scales energies} \label{frameworkergodique}
In this section we present the abstract framework  which serves to  prove lower bounds on energies  containing two scales (one much smaller than the other). The question is to deduce from a $\Gamma$-convergence (as defined in section \ref{sectiongammaconvergence}) result at a certain scale a statement at a larger scale. The framework can thus be seen as a type of $\Gamma$-convergence result for 2-scale energies. The lower bound is expressed in terms of a probability measure, which can be seen as a Young measure on profiles (i.e. limits of the configuration functions viewed in the small scale).  The method is similar in spirit to that of Alberti-M\"uller \cite{am}, where they introduce what they call  ``Young measures on micropatterns," but  differs a bit, in particular in the fact that it is based on the use of Wiener's multiparameter ergodic theorem, following a suggestion of S. R. S. Varadhan.

 Let us first give a rough idea of the type of situation we wish to consider.
Let us assume we want to bound from below an energy which is the average over large (as $\varepsilon \to 0$) domains $\Omega_{\varepsilon}$ of some nonnegative energy density $e_{\varepsilon}(u)$, defined on a space of functions $X$ (functions over $\R^n$), $\fint_{\Omega_{\varepsilon}} e_{\varepsilon}(u(x))\, dx$, and we know the $\Gamma$-liminf behavior of $e_{\varepsilon}(u)$ on small (i.e. here, bounded) scales --- here the two scales are the finite scale $1$ and the large scale corresponding to the diameter of the large domain $\Omega_\eps$.
  By this we mean that  we know how to obtain bounds from below independent of $\eps$, say for example  we can prove that \be\label{lbb}\liminf_{\varepsilon \to 0}\int_{B_R}e_{\varepsilon}(u) \, dx \geq \int_{B_R} e(u)\, dx.\ee 

However we cannot always directly apply such a knowledge to obtain a lower bound on the average over large domains :  A natural idea is to cut the domain $\Omega_{\varepsilon}$ into boxes of fixed size $R$, to obtain lower bounds on each box (say of the type \eqref{lbb}) and add them together. By doing so, we may lose some  information on the behavior of the function on  the boundary of the boxes, which would be necessary to obtain a nontrivial lower bound. Moreover, we only get a lower bound by  a number (related to the minimal value that the lower bound can take, say e.g.  $\min_u \int_{B_R} e(u) \, dx$), while we would prefer instead a lower bound which is still a function depending on the $u$'s, i.e. on limits of the configuration $u_{\varepsilon}$. This is achieved by using the multiparameter ergodic theorem, as we shall now describe.

Let us  turn to more precise statements. 
Let $\Omega$ be a compact set of positive measure in $\R^d$, satisfying
\be\label{assomt}
 \lim_{\eps \to 0} \frac{|(\Omega + \eps x) \triangle \Omega|}{|\Omega|}=0\ee
(where $\triangle$ denotes the symmetric difference between sets).  For each $\varepsilon$, let $f_{\varepsilon}(x,u)$ be a functional depending on $x$, defined on a space of functions on $\R^d$, assumed to be a Polish space, and  denoted $X$. We require  $f_\eps$ to be  measurable functions on $\Omega \times X$.  Let us emphasize that $u$ lives on the blown-up sets $\frac{1}{\eps} \Omega$, i.e. on  the large scale, whereas $x\in \Omega$ lives  on the small scale.
\begin{example}\label{ex3}
The function given by 
\begin{displaymath}
f_{\varepsilon}(x,u) = \int p(x) e_{\varepsilon}(u) \chi(y) dy
\end{displaymath}
where $\chi$ is a cut-off function supported in $B(0,1)$, $e_{\varepsilon}$ is the energy density, and $p$ is a function on $\Omega$. The function $p$ can be interpreted as a weight depending on $x$, if $p$ is constant then the functionals $f_{\varepsilon}(x,\cdot)$ do not depend on $x$.
\end{example}

We denote by $\theta_{\lambda}$ the action of $\R^d$ on the space $X$ by translations, i.e.  $\theta_{\lambda}u = u(\lambda + \cdot)$ (it could be a more general action, but for the applications we have in mind, the action of translations is really what we need), and we require that $(\lambda, u) \mapsto \theta_{\lambda} u$ is continuous with respect to each variable. We also define the following groups of  transformations on $ \R^d \times X$~:
\begin{displaymath}
T_{\lambda}^{\varepsilon} (x,u) = (x + \varepsilon \lambda, \theta_{\lambda}u), \qquad T_{\lambda}(x,u) = (x, \theta_{\lambda}u).
\end{displaymath}
We assume we are looking at a global energy of the form
\begin{equation}
F_{\varepsilon} (u) = \fint_{\Omega} f_{\varepsilon}(x, \theta_{\f{x}{\varepsilon}} u) dx.
\end{equation}
\begin{example} If the local functional $f_{\varepsilon}(x, u)$ is given by
\begin{displaymath}
f_{\varepsilon}(x,u) = \int_{y \in \R^d} e_{\eps}(u)\chi(y)  dy
\end{displaymath}
where $\chi(y)$ is a cut-off function of integral $1$ supported in say $B(0,1)$, and $e_\eps$ is the local energy density (this is the simpler case of Example \ref{ex3} where $f_{\eps}(x,u)$ does not depend on $x$), then, with the previous definition, $F_{\eps}$ is equal to 
\begin{eqnarray*}
F_{\eps}(u) & = &  \fint_{x \in \Omega} \left[ \int_{y \in \R^d} \chi(y) e_{\eps}\left( u\left(\f{x}{\eps} + y\right)\right) dy \right] dx \\& 
= & \f{\eps^d}{|\Omega|} \int_{y \in \R^d} \int_{\frac{1}{\eps}\Omega + y} \chi(y) e_{\eps}\left(u(z)\right) dy dz \\
& \approx & \f{\eps^d}{|\Omega|} \int_{\frac{1}{\eps}\Omega} \left[ \int_{y} \chi(y) dy \right] e_{\eps}(u(z)) dz=  \fint_{\frac{1}{\eps}\Omega} e_{\eps}(u(z)) dz.
\end{eqnarray*}
The first equality is simply a change of variables $z = \f{x}{\eps} + y$. Between the second and the third line, we note that  the sets $\frac{1}{\eps}\Omega + y$ over which we integrate are almost constant : $\frac{1}{\eps}\Omega$ is of size $\f{1}{\eps} \gg 1$ and we translate it by a small $y \in B(0,1)$. Therefore, an application of Fubini's theorem and the use of \eqref{assomt} allow us to exchange the integration over $y \in \R^d$ and the one over $z \in \frac{1}{\eps}\Omega + y \approx \frac{1}{\eps} \Omega$. 
Writing $\fint_{\frac{1}{\eps}\Omega} e_{\eps}(u(z)) dz$ in this fashion  can be seen as a way to use a smooth partition of unity.
\end{example}

We will make the following assumptions : 
\begin{description}
\item[(i)] (bound from below) The functionals $f_{\eps}$ are bounded below by a constant independent of $\eps$ (for convenience we  suppose, up to adding a constant,  that $f_{\eps} \geq 0$). 
\item[(ii)] (coercivity and $\Gamma$-liminf)  There exists a nonnegative measurable function  $f$ on $\Omega \times X$, such that the following holds~: 
if the quantities
\begin{displaymath}
\int_{K_R} f_{\eps}\left(T^{\eps}_{\lambda}(x_\eps, u_\eps)  \right)d\lambda 
\end{displaymath}
are bounded (when $\eps \rightarrow 0$) for any $R$, then $(x_\eps, u_\eps)$ has a convergent subsequence, converging to some $(x,u)$ and 
$$ \liminf_{\eps \rightarrow 0} f_{\eps}(x_{\eps}, u_{\eps}) \geq f(x,u).$$
\end{description}

The next step is to define what we announced as ``Young measures on profiles."
For $u$ in $X$, we let $P_{\eps}$ be the probability measure on $\Omega \times X$ obtained by pushing forward the normalized Lebesgue measure on $\Omega$ by the map 
\begin{equation*}\left\{\begin{array}{l}
\Omega \rightarrow \Omega \times X\\
\ x \mapsto (x, \theta_{\f{x}{\eps}}u).\end{array}\right.
\end{equation*}
It is equivalent to define $P_{\eps}$ as the probability measure such that for any $\Phi \in C^0(\Omega \times X)$ : 
\begin{equation}
\int \Phi(x,v) dP_{\eps}(x,v) = \fint_{\Omega} \Phi(x, \theta_{\f{x}{\eps}} u) dx.
\end{equation}
We are thus  considering the probability measures on the  translates of the blow-ups of a given  function $u$ with  the average obtained by centering the blow-up uniformly over the points of $\Omega$. Formally one can write : 
\begin{displaymath}
P_{\eps} = \fint_{\Omega} \delta_{\left( x, \theta_{ \f{x}{\eps}} u\right) }dx.
\end{displaymath}
One can also be more precise by viewing $u\mapsto P_{\eps}$ as an embedding
\begin{equation}\label{defii}
\left\{\begin{array}{rl}
 \phi_\eps: & X\to \mathcal{P}(\Omega\times X)\\
  & u\mapsto \displaystyle  \fint_{\Omega} \delta_{\left( x, \theta_{ \f{x}{\eps}} u\right) }dx,
\end{array}\right.
\end{equation}
where $\mathcal{P} (S)$ denotes the space of  Borel probability measures on    $S$.
Note that the first  variable  $x$ is just there to keep the memory of the blow-up center. The first marginal of $P_\eps$ is always equal to the normalized Lebesgue measure on $\Omega$, regardless of the function $u$. Also, the probability here is that of an analyst~: the embedding $\phi_\eps$ is completely deterministic.

If $P_\eps =\phi_\eps(u_\eps) $ for a sequence of  functions $u_\eps$  has a limit  as $\eps \to 0$, that limit can be seen as a Young measure, but  encoding the whole blown-up profiles $u= \lim_{\eps \to 0}  \theta_{ \f{x}{\eps}} u_\eps$ rather than  only the limiting values of $u_\eps$ at $x$, as is the case with the usual definition of Young measures (for which we refer to 
\cite{evanswc}).  For example, if the functions $u$ represent distributions of  points, and if these  form a lattice packed at scale $\varepsilon$, the result is the average over a fundamental domain of the lattice “seen” from every possible origin. In a more general situation, $P$ encodes the respective weights of the possible point patterns that emerge locally. One could imagine for example  in dimension 2 a probability with weigth $p$ on triangular lattice configurations and weight $1-p$ on square lattice configurations. 

By definition of $P_\eps$,  we can rewrite the global energy $F_{\eps}$ as the integral of the local energy $f_{\eps}$ with respect to $P_{\eps}$ :
\begin{equation}
F_{\eps}(u_\eps) = \fint_{\Omega} f_{\eps}(x, \theta_{\f{x}{\eps}}  u_\eps) dx = \int_{\Omega \times X} f_{\eps}(x,v) dP_{\eps}(x,v).
\end{equation}
Now, if we are able to find a limit $P$ to the probability measures $P_{\eps}$ as $\eps \rightarrow 0$, we may hope to write 
\begin{displaymath}
\liminf_{\eps \rightarrow 0} \int f_{\eps} dP_{\eps} \geq \int f dP
\end{displaymath}
where $f$ is given by  assumption \textbf{(ii)}. This will indeed hold and is  reminiscent of Fatou's lemma (indeed the sequence $\{f_{\eps}\}$ is, by assumption, bounded below).  
The last step is to combine this with the   multiparameter ergodic theorem of Wiener (see\cite{becker}), whose statement we recall~:
\begin{theo}[Multiparameter ergodic theorem]\label{th6.1} Let $X$ be a Polish (complete separable metric)  space with a continuous $d$-parameter group $\Theta_{\lambda}$ acting on it. Assume $P$ is a $\Theta$-invariant probability measure on $X$. Then for all $f \in L^1(P)$,  we have 
\begin{displaymath}
\int f(u) dP (u)= \int f^*(u) dP(u)
\end{displaymath}
where 
\begin{displaymath}
f^*(u): = \lim_{R \to + \infty} \fint_{K_R} f(\theta_{\lambda} u) d\lambda \quad \text{$P$-a.e.}
\end{displaymath}
 We may replace the cubic domains $K_R$ by any  family of reasonable shapes, such as balls, etc (more precisely a Vitali family, see\cite{becker} for the conditions). 
\end{theo}

Let us now  give the statement of the abstract result. It originally appeared in \cite{ssgl} in the case where the energy density does not depend on the blow-up center $x$, and was then generalized in \cite{ss1}.

\begin{theo}[Lower bound for two-scale energies \cite{ss1}] \label{thabstrait} Assume $\Omega, X,\{\theta_\lambda\}, \{f_\eps\}_\eps$,  $f$, $\{F_\eps\}_{\eps}$  are as above and    satisfy assumptions \textbf{(i)}--\textbf{(ii)}.
Assume $\{u_\eps\}_{\eps} $, a family of elements of $X$, is such that $\{F_\eps (u_\eps)\}_{\eps}$ is bounded, and  let $P_\eps= \phi_\eps(u_\eps)$.
Then one may extract a subsequence $\{P_\eps\}_\eps $ such that 
\begin{enumerate}
\item  $ \{P_\eps\}_{\eps}$ converges weakly in the sense of  probabilities, to some probability measure $P\in \mathcal{P}(\Omega \times X)$, whose first marginal is the normalized Lebesgue measure on $\Omega$. 
\item The limit $P$ is $T_{\lambda}$-invariant.
\item For $P$-almost every point $(x,u)$, there is some  $x_\eps$ such that  $(x_{\eps}, \theta_{\f{x_\eps}{\eps}} u_{\eps})\to (x,u)$. (Thus $P$ is indeed an average over possible local limits.) 
\item The following $\liminf$ holds : 
\begin{equation}
\liminf_{\eps \rightarrow 0} F_{\eps}(u_{\eps}) \geq \int f(x,u) dP(x,u) = \int f^*(x,u) dP(x,u)
\end{equation}
with 
\begin{displaymath}
f^*(x,u) := \lim_{R \to + \infty} \fint_{K_R} f(T_{\lambda}(x,u))d\lambda =  \lim_{R \to + \infty}\fint_{K_R} f(x, \theta_{\lambda}u) d\lambda.
\end{displaymath}
\end{enumerate}
\end{theo}
We now indicate the ingredients of the proof (details can be found in \cite{ssgl,ss1}).

\begin{proof}  1. The main point is to show that $\{P_{\eps}\}_{\eps}$ is tight, i.e. for any $\eta > 0$ there exists a compact  set $K_\eta$ such that $P_{\eps}(K_\eta) \geq 1 - \eta$ for small $\eps$. This comes as a consequence of  the assumption that $\{F_{\eps}(u_\eps)\}_\eps$ is bounded and the coercivity assumption \textbf{(ii)} on the functionals.
The fact that the first marginal is the normalized Lebesgue measure is obvious since it is true for each $P_\eps$ and thus  remains true in the limit.\\
2. The invariance by $T_{\lambda}$ is a straightforward consequence of the definition of $P_\eps$. Consider a test-function $\Phi\in C^0(\Omega \times X)$ and $\lambda \in \R^d$.   On the one hand~:
\begin{displaymath}
\lim_{\eps \rightarrow 0} \fint_{\Omega}\Phi(x, \theta_{\f{x}{\eps} + \lambda} u_{\eps}) dx = \lim_{\eps \rightarrow 0} \fint_{\Omega} \Phi(x, \theta_{\f{x}{\eps}} u_{\eps} )dx
\end{displaymath}
because $\f{x}{\eps} + \lambda \approx \f{x}{\eps}$ for any $\lambda$ fixed when $\eps$ goes to zero (this uses the assumption \eqref{assomt}).
 But on the other hand, for any $\lambda$, we have by definition of $P_\eps$,
\begin{displaymath}
\lim_{\eps \rightarrow 0} \fint_{\Omega}\Phi(x, \theta_{\f{x}{\eps} + \lambda} u_{\eps}) dx = \lim_{\eps \rightarrow 0} \int \Phi( x,\theta_\lambda u) \, dP_\eps (x,u)= 
\int \Phi(x, \theta_{\lambda} u) dP(x,u).
\end{displaymath}
We deduce that  we must  have, for all continuous $\Phi$, 
\begin{equation}
\int \Phi(x, \theta_{\lambda} u) dP  =  \int \Phi(x, u) dP,
\end{equation}
which exactly means that $P$ is $T_\lambda$-invariant.
\\
3. This is a rather direct consequence of the definition of $P_\eps$.
\\
4. This is a result that uses the fine topological information provided by assumptions 
 \textbf{(ii)},  and combines it with the weak convergence of $P_\eps $ to $P$, assumption \textbf{(i)} and Fatou's lemma, cf. \cite[Lemma 2.2]{ssgl}.
 \end{proof}

As desired, this result provides a lower bound on functionals of the type $F_\eps (u_\eps)$, which is expressed in terms of the probability $P$, i.e. in terms of the limits of $u_\eps$. 
As a corollary, it implies the weaker result of lower bound of $F_\eps$ by a number~: 
\be\label{lbab}
\liminf_{\eps \rightarrow 0} F_\eps(u_\eps) \geq \int \inf_{u} f^*(x,u) dP(x,u) = \fint_{\Omega} \inf_{u} f^*(x,u) dx.
\ee
The minimization of the function $f^*$ is similar to a ``cell problem" in homogenization (cf. e.g. \cite{braides4}).

Once this result is proved, it remains to show, if possible,  that such a lower bound is sharp, which requires  
 constructing a family  $\{u_\eps\}$ such that 
\be\label{ub}
\limsup_{\eps \rightarrow 0}F_\eps (u_\eps)  \leq \fint_{\Omega} \inf f^*(x, \cdot) ) dx.
\ee
This certainly requires at least that the $\Gamma$-liminf relation in assumption \textbf{(i)} be also a $\Gamma$-limsup, i.e. that there exist recovery sequences.  This really depends on the specifics of the local  functionals.
If \eqref{ub} can be shown,  then, just as in Proposition \ref{gammaconvmini}, comparing \eqref{lbab} and \eqref{ub} implies that if $u_\eps $ minimize $F_\eps$ for every $\eps$ and $\min F_\eps $ is bounded, then letting  $P$ be as in Theorem \ref{thabstrait}, we must have
$$P-a.e. (x,u), \  u \text{ minimizes the local functional }\   f^*(x,\cdot).$$

We will next see  how to apply this abstract result in the context of the Coulomb gas Hamiltonian. It has also been used for vortices in Ginzburg-Landau in \cite{ssgl}, as we will see in Chapter \ref{glnext},  and droplets in the Ohta-Kawasaki model \cite{gms2}. In all these cases, we were able to conclude because the corresponding upper bound \eqref{ub} turned out to be provable.

\section[Next order for the Coulomb gas]{Next order lower bound for the Coulomb gas Hamiltonian}

\subsection{Assumptions}\label{secassump}

To conclude with  our final results in Chapter \ref{chapub} 
we will make additional assumptions on $V$ which we already state~:
\benu
\item the strongest ones  made in Chapter \ref{leadingorder}, i.e. that $V$ is continuous, finite-valued and satisfies \textbf{(A3})--\textbf{(A4)}.
This  in particular guarantees that the equilibrium measure $\mu_0$ exists and has compact support.
\item \textbf{(A5)} The support  $\Sigma$ of the equilibrium measure has a $C^1$ boundary.
\item \textbf{(A6)} The equilibrium measure $\mu_0$ has an $L^{\infty}$ density which is bounded below of class $C^1$ on its support~: 
\be
\mu_0(x) = m_0(x) \mathbf{1}_{\Sigma}(x) dx
\ee
where $\Sigma$ is the support of $\mu_0$ and $m_0\in C^1(\Sigma)\cap L^\infty(\R^d)$ is its density, which satisfies 
\begin{equation}\label{bornesm0}
0 < \underline{m} \leq m_0 \leq \overline{m}.
\end{equation} 
\eenu
Again, by abuse of notation, we will confuse $m_0(x)$ and $\mu_0(x)$.

If $V$ is smooth enough, these assumptions are sufficiently generic. They are nonempty: an easy example is the case when $V$ is a multiple of $|x|^2$ and $\mu_0$ is a multiple of the  characteristic function of a ball (see Example 2 in Chapter \ref{leadingorder}) -- in fact any $V$ positive quadratic works as well.   
Recall  also that from Proposition \ref{proequivpb},  when $V$ is $C^2$ $\mu_0$ is a measure with  density $m_0= (\hal \Delta V\indic_{\omega}) \in L^\infty$, thus if $\Delta V $ is bounded below by a positive constant, \eqref{bornesm0} is satisfied. 
If in addition $V$ is $C^3$ on $\Sigma$,  then $\mu_0 $ is $C^1 $ in $\Sigma$ and \textbf{(A6)} is fully   satisfied. This strong assumption is assumed mostly for convenience, to simplify our upper bound construction. For the lower bound, the assumption that $\mu_0\in C^0(\Sigma)$ (and probably even less) suffices.
Note that when \textbf{(A6)} holds, by continuity of $\Delta V$, $\Sigma$ and the coincidence set $\omega$ must coincide.

The assumption \textbf{(A5)} can be investigated in light of the regularity theory for the obstacle problem, for which $C^1$ regularity of the boundary of the coincidence set  is generic in some sense in dimension 2  \cite{schaeffer,monneau}, or is true if the coincidence set if convex.
Note also that a result \cite{kindernirenberg,isakov} shows that if the boundary of the coincidence 
set is $C^1$, it is in fact analytic.
Again, weaker conditions should suffice. 


\subsection{Lower bound}
As already mentioned, we now return to Proposition \ref{propRS}, in order to extract $\mathcal{W}$ as a limiting lower bound, using the abstract framework of Section \ref{frameworkergodique}.

In view of the results of Proposition \ref{developpementasympto}, in order to bound from below $H_n$ at the next order, it suffices to bound from below $\mathcal H_n$ and by monotonicity in $\eta$ to bound from below
 $\frac{1}{n} \int_{\R^d} |\nab h_{n,\eta}'|^2 $,  where $h_{n,\eta}'$ is given by \eqref{1to1}. This will be done
according to the scheme of Section \ref{frameworkergodique}. We first consider $\eta$ as fixed and let $n\to \infty$, and later let $\eta \to 0$.
The setup to use the abstract framework is to take $\Omega= \Sigma$, $X= 
L^q_{\mathrm{loc}} (\R^d, \R^d) $ for some $q<\frac{d}{d-1}$, and $\eps= n^{-1/d}$.
Assumption \textbf{(A5)} ensures in particular that the condition \eqref{assomt} is satisfied.

We wish to obtain lower bounds for sequences of configurations $(x_1, \dots, x_n)$. In all that follows, the configuration depends implicitly on $n$, i.e. we mean $(x_{1,n}, \dots, x_{n,n})$ but drop the second index in the notation.
The main lower bound result that we obtain with the method outlined above is~: 

\begin{theo}[Lower bound at next order for the Coulomb gas Hamiltonian]  \label{thlbnext} \mbox{}Assume that $V$ is continuous and such that the equilibrium measure $\mu_0$ exists and satisfies \textbf{(A5)}--\textbf{(A6)}.   For any $x_1, \dots, x_n\in \R^d$, let $h_n' = g * (\sum_{i=1}^n \delta_{x_i'} -\mu_0') $ as in \eqref{dehn'1}.
Let $P_n\in \mathcal{P}(\Sigma \times X)$ be the push-forward of
the normalized Lebesgue measure on $\Sigma $ by 
\begin{displaymath}
x \mapsto (x, \nabla h'_n (n^{1/d} x + \cdot)).
\end{displaymath}
Assume $\mathcal H_n(x_1', \dots, x_n') \le C n$ for some constant $C$ independent of $n$, where $\mathcal H_n$ is as in \eqref{Wncal}. 
Then, up to extraction of a subsequence, $P_n$ converges weakly in the sense of probabilities  to a probability measure $P\in \mathcal{P}(\Sigma \times X)$ such that 
\begin{itemize}
\item[(i)] $P$ is translation-invariant, and its first marginal is  the normalized Lebesgue measure on $\Sigma$.
\item[(ii)]  For $P$-almost every  $(x,E)$, $E$ belongs to the class $\bar{\mc{A}}_{\mu_0(x)}$.
\item[(iii)] We have the following $\Gamma-\liminf$ inequality~: 
\begin{equation} \label{liminfconclu1}
\liminf_{n \to+ \infty} \frac{1}{n} \mathcal H_n(x_1', \dots, x_n')  \geq    \widetilde{\mc{W}}(P),
\end{equation}where $\widetilde{\mc{W}}$ is defined over the set of probability measures   $P\in \mathcal{P}(\Sigma \times X)$
satisfying $(i)$ and $(ii)$ 
 by 
\be\label{tildew}
\widetilde{\mc{W}}(P):= \frac{|\Sigma|}{c_d}\int \mc{W}(E)\, dP(x,E).\ee
\end{itemize}
\end{theo}
This result was proven in this form in \cite{rs}. The same also holds with $\mathcal{W}$ replaced by $W$ in dimension $d=2$, as was previously proven in \cite{ss1}, and also in dimension $d=1$ in \cite{ss2}.

One may guess the value of the minimum of $\widetilde{W}$ on its domain of definition: by property $(i)$ on $P$, we have 
\begin{equation}\label{mintw}
\min \widetilde{\mc{W}}\ge \f{|\Sigma|}{c_d} \int\Big( \min_{\bar{\mc{A}}_{\mu_0(x)}} \mc{W}\Big) \ dP(x,E) = \f{1}{c_d} \int_{\Sigma} \Big(\min_{\bar{\mc{A}}_{\mu_0(x)}   } \mathcal{W}\Big)dx,
\end{equation} 
and by the scaling relations 
 \eqref{scalingW} and \eqref{scalingW2d}, we thus get 
\begin{multline}\label{eq:gamma d}
 \min \widetilde{\mc{W}}\ge \xi_d := \frac{1}{c_d} \int_{\Sigma} \min_{\bar{\mc{A}}_{\mu_0(x)}} \mc{W}\, dx
 \\
 =
\begin{cases}
  \displaystyle  \f{1}{c_d} \left(\displaystyle \int_{\Sigma} \mu^{2 - 2/d}_0(x) dx \right) \min_{\bar{\mc{A}}_1} \mc{W}& \text{for } d\ge 3\\
\displaystyle  \f{1}{2\pi } \min_{\bar{\mc{A}}_1} \mc{W}  - \displaystyle \f12 \int_{\Sigma} \mu_0(x)  \log \mu_0(x) dx & \text{for $d=2$}.\end{cases}\end{multline}
It turns out, as we will see below, that these inequalities are equalities. 
In view of the splitting formula in Proposition \ref{developpementasympto}, and dropping the term $\sum_i \zeta(x_i) $ which is always nonnegative, 
Theorem \ref{thlbnext} has the following 
\begin{coroll}\label{cor61}
 We have
$$
\liminf_{n\to  + \infty} n^{2/d - 2}\left(\min H_n - n^2I(\mu_0)  +\Big(\f{n}{2}\log n \Big) \indic_{d=2}  \right) \geq \xi_d ,$$where $\xi_d $ is as in \eqref{eq:gamma d}.
\end{coroll}

\begin{proof}[Proof of the theorem]
 As announced, we apply the abstract framework of Section~ \ref{frameworkergodique} for fixed $\eta$.
 We will need the following notation~: given a gradient  vector field $E=\nab h$ satisfying a relation of the form 
 $$-\div E = c_d \Big(\sum_{p\in \Lambda} N_p \delta_p - \mu(x)\Big)\ \text{in } \R^d,$$
 whether an element of $\bar{\mathcal{A}}_m$, or the gradient of a potential defined by \eqref{defhn}, 
 we define $E_\eta$ to be  as in \eqref{esmeared}, and we denote by
 $\Phi_\eta$ the map $E \mapsto E_\eta$, which to a vector field corresponding to singular charges assigns the vector field corresponding to smeared out charges.

  Let us  define $P_{n, \eta}$ as the push-forward of the normalized Lebesgue measure on $\Sigma$ by the map
\begin{displaymath}
x \mapsto (x, \nabla h'_{n, \eta} (n^{1/d} x + \cdot)).
\end{displaymath}
In other terms, $P_{n, \eta}$ is the push-forward of $P_n$ by $\Phi_\eta$. Then, we take $\chi$ to be a  nonnegative cut-off function supported in $B(0,1)$ and of integral $1$, and  set 
\begin{displaymath}
f_n(x, E) = \left\lbrace \begin{array}{cl}\displaystyle \int \chi(y)|E|^2(y)dy & \text{ if } E = \nabla h'_{n, \eta}(n^{1/d}x + \cdot) \\ + \infty & \text{otherwise.}
\end{array} \right.
\end{displaymath}
This gives the ``local" energy at the small scale. The definition ensures that we only consider a class of vector fields that are of the interesting form.

We then let  $F_n(E)$ be given, as in Theorem \ref{thabstrait},  by  
\begin{equation} \label{defFnE}
F_n(E) = \fint_{\Sigma} f_n(x, \theta_{n^{1/d} x} E) dx
\end{equation}
and we may observe that by Fubini's theorem and a change of variables, 
\begin{multline*}
F_n(E) =  \f{1}{|\Sigma|} \int_{\R^d} \int_\Sigma  \chi(y) |\nabla h'_{n, \eta}(n^{1/d}x +y)|^2 dx dy\\ \le  \f{1}{n|\Sigma|} \int_{\R^d}    \int_{\R^d} |\nabla h'_{n, \eta}(z)|^2 \chi(y) \, dy  dz.
\end{multline*}
Since $\int \chi=1$, it follows that 
\begin{equation}\label{615}
  \f{1}{n} \int_{\R^d} |\nabla h'_{n, \eta}|^2\geq |\Sigma|F_n(E)
\end{equation}
so in order  to bound from below $\frac{1}{n} \int_{\R^d} |\nab h_{n,\eta}'|^2 $ as desired, it does suffice to bound from below $F_n$, which will be done via Theorem \ref{thabstrait}.

To do so, we have to check the three assumptions  \textbf{(i)}--\textbf{(ii)} stated in  Section~\ref{frameworkergodique}. First, it is true that the energies $f_n$ are nonnegative,  this gives the condition \textbf{(i)}.
 To check condition \textbf{(ii)}, we use the following lemma

\begin{lemme}[Weak compactness of local electric fields]\label{lemprel}\mbox{}\\
Let $h_{n,\eta}'$ be as above, and let $\nu_n'=\sum_{i=1}^n \delta_{x_i'}$ . 
Assume that  for every $R>1$ and for some  $\eta\in(0,1)$, we have
\begin{equation}\label{eq:asum h}
\sup_n  \int_{K_R} |\nab h_{n,\eta}'(n^{1/d}\underline{x_n}+\cdot) |^2 \le C_{\eta, R},
\end{equation}
and $\underline{x_n}\to x\in\mr^d$ as $n\to \infty$ (a sequence of blow-up centers).
Then $\{\nu'_n(n^{1/d}\underline{x_n}+\cdot)\}_{n}$ is locally bounded and up to extraction converges weakly as $n\to \infty$ in the sense of measures to
$$\nu = \sum_{p\in\Lambda}N_p \delta_p$$
where $\Lambda $ is a discrete set of $\R^d$ and $N_p\in \mathbb{N}^*$.
In addition, there exists $E\in L^q_{\mathrm{loc}}(\R^d, \R^d) $, $q<\frac{d}{d-1}$, $E_\eta\in L^2_{\mathrm{loc}} (\R^d, \R^d)$, with $E_\eta=\Phi_\eta(E)$,  such that, up to extraction, as $n \to \infty$, 
\begin{eqnarray}\label{cvh2}
& \nab h_n'(n^{1/d} \underline{x_n}+\cdot)  \rightharpoonup   E \ \text{weakly in } \ L^q_{\mathrm{loc}}(\R^d, \R^d)   \ \text{for}\ q<\frac{d}{d-1}
\\
\label{cvh}
& \nab h_{n,\eta}' (  n^{1/d}\underline{x_n}+\cdot)   \rightharpoonup  E_\eta \ \text{ weakly in } \ L^2_{\mathrm{loc}} (\R^d, \R^d). \end{eqnarray}
Moreover $E $ is the gradient of a function $ h$, and if $x\notin \pa\Sigma$, we have 
\begin{equation}\label{eqhl}
- \Delta h = c_d(\nu -\mu_0(x))\quad \text{in} \ \R^d
 \end{equation} hence $E\in \bar{\mathcal{A}}_{\mu_0 (x)}$.
\end{lemme}

\begin{proof}
First, from \eqref{eq:asum h} and \eqref{dehn'1}, exactly as in the proof of \eqref{feg} we have that for some $t \in (R-1,R)$, for every $n$, 
$$\left| \int_{ K_t}    \sum_i  \delta_{x_i'}^{(\eta)} (n^{1/d} \underline{x_n} + x')    - \int_{K_t} \mu_0(\underline{ x_n}  +  n^{-1/d}  x')\, dx' \right| \le          C_{\eta, R},$$
for some constant depending only on $\eta$ and $R$. 
It follows that, letting $\underline{\nu_n'}:=\nu_n'(n^{1/d} \underline{x_n} + \cdot) $, we have
$$\underline{\nu_n'}(K_{R-1}) \le C_d \|\mu_0\|_{L^\infty}  R^d  +C_{\eta,R}.$$
This establishes that $\{\underline{\nu_n'}\}$ is locally bounded independently of $n$. In view of the form of $\underline{\nu_n'}$, its limit can only be of the form $\nu= \sum_{p\in \Lambda} N_p \delta_p$, where $N_p$ are positive  integers and $\Lambda $ is a discrete set in $\R^d$.

Up to a further extraction we also have \eqref{cvh} by \eqref{eq:asum h} and weak compactness  of $\nab h_{n,\eta}'$ in $L ^2_{\mathrm{loc}}$. The compactness and  convergence \eqref{cvh2}  follow from Lemma~\ref{lem:fluctu field}. The weak local  convergences of both $\underline{\nu_n'}$ and $\nabla h_n ' (n^{1/d}\underline{ x_n} + \cdot)  $ together with the continuity of $\mu_0$ away from $\pa \Sigma$ (cf. \textbf{(A6)}),   imply after passing to the limit  in
$$-\Delta h_{n}'(n^{1/d} \underline{x_n} + \cdot) = c_d \(\underline{\nu_n'} - \mu_0 (\underline{x_n} + n^{-1/d} \cdot)\)
$$ (which is obtained by centering \eqref{hn'} around $\underline{x_n}$) 
that $E$ must be a gradient and that \eqref{eqhl} holds. Finally $E_\eta=\Phi_\eta(E)$ because one may check that $\Phi_\eta$ commutes with the weak convergence in $L^q_{\mathrm{loc}}(\R^d, \R^d) $ for the $\nab h_n'  (n^{1/d} \underline{x_n} + \cdot)$ described above. \end{proof}

 To check condition \textbf{(ii)}, let us thus 
 assume that 
 $$\forall R>0, \quad \limsup_{n\to  + \infty} \int_{K_R} f_n (T_\lambda^n (\underline{x_n}, Y_n))\, d\lambda<\infty, \quad \underline{x_n} \in \Sigma.$$
 By definition of $f_n$, this condition is equivalent to 
 $$\forall R>0, \forall n \ge n_0, \  Y_n= \nab h_{n,\eta}'(n^{1/d} \underline{x_n} + \cdot)  \text{ and } \limsup_{n\to +\infty} \int \chi * \indic_{K_R} |Y_n|^2 <+\infty, \ \underline{x_n} \in \Sigma.$$
This implies the assumption of Lemma  \ref{lemprel}. We may also assume, up to extraction of a subsequence that $\underline{x_n} \to x \in \Sigma$ (since $\Sigma $ is compact). Applying Lemma \ref{lemprel}, we have $Y_n \rightharpoonup Y$ weakly in $L^2_{\mathrm{loc}}(\R^d, \R^d)$, and all the other results of the lemma. 
This weak convergence implies in particular that \begin{multline*}
\liminf_{n \to + \infty} f_n(\underline{x_n}, Y_n)  \ge f(x,Y)\\:=\begin{cases}
 \int \chi(y) |Y|^2 (y) \, dy & \text{if} \ x\in \Sigma \backslash \pa \Sigma \ \text{ and } Y=\Phi_\eta( E) \ \text{for some }  E \in \bar{\mathcal{A}}_{\mu_0(x)}\\
 0 & \text{if} \ x\in \pa \Sigma\\
 + \infty & \text{otherwise}.\end{cases} 
\end{multline*}
This completes the proof that condition \textbf{(ii)} holds.

Theorem \ref{thabstrait} then yields the convergence (up to extraction) of $P_{n,\eta}$ to some $P_\eta$, and, in view of \eqref{615},  
\be\label{510}\liminf_{n\to + \infty} \frac{1}{n} \int_{\R^d} |\nab h_{n,\eta}'|^2\ge \liminf_{n\to + \infty}
|\Sigma| F_n(E)\ge |\Sigma| \int f^*(x,Y) dP_\eta(x,Y)\ee
where 
$$f^*(x,Y)= \lim_{R\to + \infty} \fint_{K_R} f(x, \theta_\lambda Y) \, d\lambda.$$
By definition of $f$ and since $\pa \Sigma$ is of Lebesgue measure $0$ by \textbf{(A5)}, for $P_\eta$-a.e. $(x,Y)$, we must have $Y = \Phi_\eta(E)$ with some $E\in\bar{\mc{A}}_{\mu_0(x)}$. Pushing forward by $\Phi_\eta^{-1}$, we get the convergence of $P_n$ to $P$ satisfying the first two stated properties.
Moreover, applying Fubini's theorem,  we may write, 
$$f^*(x,Y)= \lim_{R\to +\infty} \dashint_{K_R} \chi * \indic_{K_R} |Y|^2  \ge 
\lim_{R \to + \infty} \dashint_{K_R} |Y|^2$$
where we used that  $\chi* \indic_{K_R} \ge \indic_{K_{R-1}}$.
By definition of the push-forward, it follows from \eqref{510} that 
$$ \liminf_{n\to + \infty} \frac{1}{n} \int_{\R^d} |\nab h_{n,\eta}'|^2\ge |\Sigma| \int \Big(\lim_{R \to + \infty} \dashint_{K_R} |\Phi_\eta(E)|^2\Big) \, dP(x,E).$$
Using that $\int \mu_0=1$, the fact that the first marginal of $P$ is the normalized Lebesgue measure, that $P$-a.e., $E\in \bar{\mc{A}}_{\mu_0(x)}$, 
 and the definition of $\mathcal{W}_\eta$ (Definition \ref{def1}),  we deduce that 
\begin{multline*}
\liminf_{n\to + \infty} \frac{1}{n} \int_{\R^d} |\nab h_{n,\eta}'|^2-c_d g(\eta) \\
\ge
|\Sigma|  \int \Big(\lim_{R \to + \infty}\dashint_{K_R} |\Phi_\eta(E)|^2 - c_d g(\eta) \mu_0(x)\Big) \, dP(x,E)\\=
|\Sigma| \int \mathcal{W}_\eta(E)\, dP(x,E).\end{multline*}
Inserting into \eqref{lbwnc}, we find that 
$$\liminf_{n\to \infty}\frac{1}{n}  \mathcal H_n(x_1', \dots, x_n')  \ge |\Sigma| \int \mathcal{W}_\eta(E)\, dP(x,E)- C\eta\|\mu_0\|_{L^\infty}.
$$
Since $\mathcal{W}_\eta$ is bounded below as seen in Corollary \ref{minoW} (and the constant remains bounded when $\mu_0$ does in view of  \eqref{scalingW}--\eqref{scalingW2d}), we may apply Fatou's lemma and take the $\liminf_{\eta\to 0}$ on both sides, and  we obtain, by definition of $\mathcal{W}$, that
$$\liminf_{n\to \infty}\frac{1}{n}  \mathcal H_n(x_1', \dots, x_n') \ge |\Sigma| \int \mathcal{W} (E)\, dP(x,E).$$

\end{proof}

\begin{rem}
Note that for this lower bound, we do not really need the full strength of \textbf{(A5)}--\textbf{(A6)}: it suffices that $\mu_0$ be continuous on its support,  that $\pa \Sigma $ has zero measure, and that $\Sigma$ satisfies \eqref{assomt}. 
\end{rem}


\chapter{Deriving $\mathcal W$ as the large $n$ limit: screening, upper bound, and consequences} 
  \label{chapub}
  
  In this chapter, we obtain our final results on the Coulomb gas. First,  we describe how to obtain the upper bound that optimally  matches  the lower bound obtained in Chapter \ref{derivingw}. This upper bound relies on an important construction, which we call the ``screening" of a point configuration. 
 Once the upper and lower bounds match, it follows that  the prefactor governing the $n^{2-2/d}$ order term in $H_n$ is indeed $\widetilde {\mathcal W}$, defined from $\mathcal W$ as in \eqref{tildew}. As a consequence we obtain  
    an asymptotic expansion of $\min H_n$ up to order $n^{2-2/d}$, with  prefactor  $\min \widetilde {\mathcal W}$,  and the fact that minimizers of $H_n$ have to 
   converge to minimizers of $\widetilde{\mathcal W}$. 
  We also  derive  consequences on the statistical mechanics, with an expansion up to order $n^{2-2/d}$ of $\log \Z$, which becomes sharp as the inverse temperature $\beta \to \infty$, and some large deviations type results, which show that the Gibbs measure concentrates on minimizers of $\widetilde {\mathcal W}$ as $ \beta \to \infty$.
  
\section{Separation of points and screening}

In order to construct test configurations which will almost achieve equality in the lower bound of Theorem \ref{thlbnext}, we need to start from a minimizer of $\mathcal W$, and be able to truncate it in a finite box, so as to then copy and paste such finite configurations. The tool to be able to do this is the screening result. 
To screen a (possibly infinite) configuration  means to  modify it near  the boundary
of a cube $K_R$  to make the normal component of the electric  field vanish on $\pa K_R$, still keeping the points well-separated all the way to the boundary of the cube. This modification needs to add only a negligible energy cost.  The vanishing normal component will in particular impose
the total number of points in the cube, but it also makes the configurations ``boundary compatible" with each other, which  will allow to copy and paste them together, e.g. to periodize them.  Physically,  ``screening" roughly means here that   a particle sitting outside of $K_R$ does not ``feel" any electric field coming from $K_R$.

The method we originally used consists in first  reducing to configurations with points that are simple and ``well-separated"  in the sense seen previously, which simplifies the screening construction. This uses 
an unpublished result of E. Lieb \cite{lieb2} which states roughly that 
\begin{theo}[Lieb] Points minimizing the Coulomb interaction energy must be well-separated. \end{theo}
A more precise statement and a proof in the setting with confining potential  can be found in dimension $d = 2$ in  \cite[Theorem 4]{rns}. This was readapted to the  setting of smeared out charges in \cite{rs} and gives  the following
\begin{prop}[Reducing to well-separated points] \label{avantscreening}
Let $\Lambda$ be a discrete subset of $K_R$ and let $h$ satisfy
\begin{equation}\label{eq:Laplace h}
-\Delta h =  c_d \Big( \sum_{p\in \Lambda} N_p \delta_p ^{(\eta)} -1\Big) \quad \text{in} \ K_R.
\end{equation}
Denote $\Lambda_R = \Lambda \cap K_{R-1}$. There exist three positive constants $\eta_0,r_0,C$ such that if $\eta < \eta_0$, $R$ is large enough and one of the following conditions does not hold~:
\begin{eqnarray}\label{cond1}& \forall p  \in \Lambda_R, \quad  N_{p} = 1\\
\label{cond2} & \forall p \in \Lambda_R, \ \dist(p,\Lambda_R \setminus \{ p \})\ge 2 r_0,
\end{eqnarray}
then there exists $\tilde{\Lambda}$, a discrete subset of $K_R$ and an associated potential $\tilde{h}$ satisfying
\begin{equation}\label{dhtilde}
-\Delta \tilde{h} =  c_d \Big( \sum_{p\in \tilde{\Lambda}} N_p \delta_p ^{(\eta)} -1\Big)\quad \text{in}\  K_R
\end{equation}
such that
\[
\int_{K_R} |\nabla \tilde{h}| ^2 \leq  \int_{K_R} |\nabla h| ^2 - C.
\]


\end{prop}
Another way of phrasing the proposition  is that  if a configuration has points that are too close to one another, we can always replace it by one that has a smaller energy, thus we can always take a   minimizing sequence for
\begin{equation}\label{fetar}
F_{\eta, R} := \inf \left\lbrace\int_{K_R} |\nabla h|^2, -\Delta h = c_d \left(\sum N_p \delta^{(\eta)}_p - 1\right) \textrm { in } K_R\right\rbrace.
\end{equation}
with  points that are simple and well-separated (at least in $K_{R-1}$).
Note that in the proposition, the configuration of points $\Lambda$ may depend on $\eta$.

We may now state the screening result.
\begin{prop}[Screening] \label{propscreening} 
There exists $\eta_0>0$ such that the following holds for all $\eta<\eta_0$. Let $h_\eta$ satisfy  \eqref{eq:Laplace h}--\eqref{cond1}--\eqref{cond2} and
\begin{equation}\label{eq:bound ener cube}
\int_{K_R} |\nab h_\eta| ^2 \leq C_\eta R ^d.
\end{equation}
Then there exists $\hat{\Lambda}$ a configuration of points and $\nab\hat{h}$ an associated gradient vector  field (both possibly also depending on $\eta$)  defined in $K_R $ and satisfying
\begin{equation}\label{eq:proscreen field}
\left\{\begin{array} {ll}
-\Delta \hat {h} = c_d\Big( \displaystyle \sum_{p\in \hat{\Lambda}} \delta_p - 1\Big) &  \mbox{ in } K_R\\ 
\displaystyle\frac{\pa \hat{h}}{\pa  \nu}  =0 \quad & \text{on} \ \pa K_R\end{array}\right. \end{equation}
such that for any $p\in \hat{\Lambda}$
\begin{equation}\label{eq:mindistr plus}
\min \left( \dist(p,\hat{\Lambda}\setminus \{p\}), \dist(p,\partial K_R) \right) \geq \frac{ r_0}{10}
\end{equation} with $r_0$ as in \eqref{cond2}, 
and
\begin{equation}\label{majow}
\int_{K_R} |\nab \hat{h}_\eta  |^2 \le \int_{K_R} |\nab h_\eta| ^2 + o( R^d)
\end{equation}
as $R\to \infty$, where the $o$ depends only on $\eta$.

\end{prop}
\begin{rem} \label{remneutraliteglobale} The vanishing of the normal derivative $\frac{\partial{\hat{h}}}{\partial \nu} $ on $\partial K_R$ implies, by Green's theorem, that $\# (\hat{\Lambda} \cap K_R )= |K_R|$ exactly.
\end{rem}
\begin{rem}In \cite{ps} we prove that arbitrary configurations with bounded energy can be screened, not only those whose points are well-separated. The separation makes the construction easier however, and is also needed to finish the proof of the upper bound energy.\end{rem}

Let us start by giving the idea of the proof of Proposition \ref{propscreening}. For details, cf. \cite{rs}.
\begin{proof}

\noindent {\bf Step 1.} It consists in selecting, by a mean value argument, a ``good boundary"  $\pa K_t$ (cf. Fig \ref{fig9})  at distance $L$ with $1\ll L\ll R$ (as $R \to \infty$) from $\pa K_R$, and such that 
\be\label{estbor1}\left\lbrace
\begin{array}{cc}
\int_{\partial K_t} |\nabla h_{\eta}|^2 & \leq\frac{ C_{\eta} R^d }{L}=o(R^d) \\
\int_{K_{t+1} \backslash  K_{t-1}}  |\nabla h_{\eta}|^2 & \leq \frac{ C_{\eta} R^d}{L} = o(R^d).
\end{array} \right.
\ee
\noindent
{\bf Step 2.} Taking $\eta < r_0 / 4$ so that the balls $B(p, \eta)$ are disjoint, we can modify the boundary $\partial K_t$ into $\partial \Gamma$ (cf. Fig \ref{fig9})  so that $\partial \Gamma$ intersects no ball $B(p, \eta)$ and still satisfies 
\be\label{estbor}
\int_{\partial \Gamma} |\nabla h_{\eta}|^2 \leq o( R^d).\ee
We do not move the points whose associated smeared charges intersect $\pa K_t$. Instead,
we isolate them in small cubes and leave unchanged all the points lying in $\Gamma$, defined as  the
union of $K_t$ with these small cubes.
\medskip 

\noindent
{\bf Step 3.} 
We build a new configuration of points and potential in $K_R \backslash \Gamma$, to replace the previous one.
 To do so, we partition this region  into hyperrectangles  $\mathcal{K}_i$, each centered at some point $x_i$, on each of which we  solve 
 \be\label{equatui}
\left\lbrace \begin{array}{c}
- \Delta u_i =  c_d (\delta_{x_i} - 1) \textrm { in } \mathcal{K}_i \\
\textrm{ the normal derivatives } \frac{\partial u_i}{\partial \nu} \textrm{ are compatible }
\end{array}
\right.
\ee
so that the normal derivatives “connect” nicely, meaning that they agree (with suitable  orientation) on any two adjacent hyperrectangles, as well as on the boundary of $\Gamma$. 
The new set $\hat{\Lambda}$ is defined as $\cup_i \{x_i\} \cup( \Lambda\cap K_{R-1})$. 
One checks that the hyperrectangles have sidelengths which are bounded below in such a way that the new set $\hat{\Lambda}$ satisfies \eqref{eq:mindistr plus}
\medskip

\begin{figure}[h!]
\begin{center}
\includegraphics[scale=0.6]{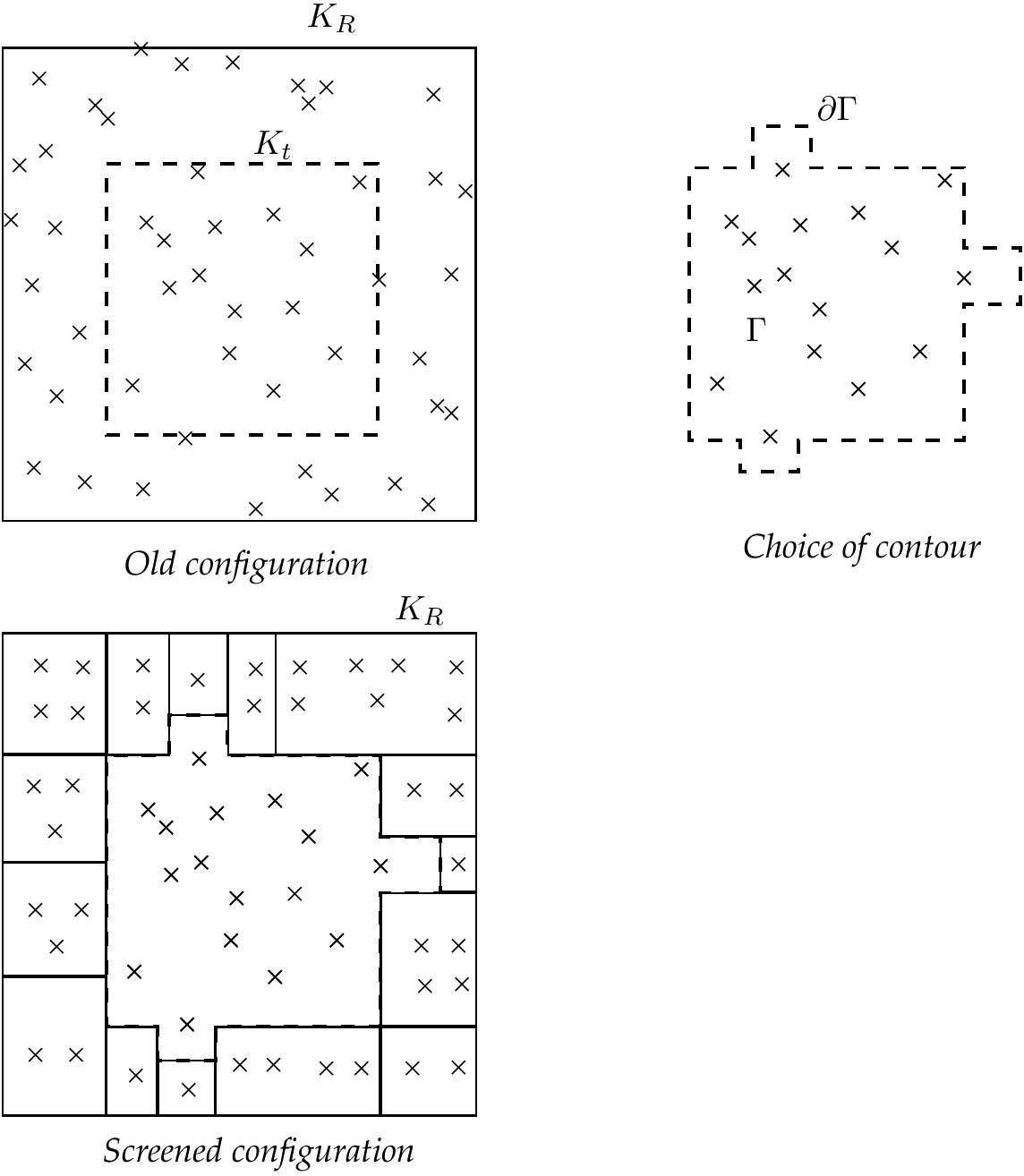}
\end{center}
\caption{The screening construction}\label{fig9}
\end{figure}

\noindent
{\bf Step 4.} We define a global vector field and estimate its energy. First of all, to $\nab h_\eta$ given in the statement of the proposition  corresponds a  $\nab h$ via \eqref{heta}. Since $\Lambda$ may possibly depend on $\eta$, so does $\nab h$, but this is not important.
Defining then $E$ to be $\nab u_i$ on each $\mathcal{K}_i$,  and $E= \nab h$   in $\Gamma$, thanks to the compatibility condition we may check that $E$ satisfies 
\be\label{diveee}
-\div E = c_d \Big(\sum_{p\in\hat{ \Lambda}}  \delta_{p} -1\Big) \ \textrm{ in } \ K_R.\ee
Indeed, one may check that the divergence of $E$ (or any vector field) in the sense of distributions on each interface is given by the jump in normal derivative (here constructed to be zero) while the divergence on each cell $\mathcal{K}_i$ is given by \eqref{equatui}.

By elliptic estimates, and using \eqref{estbor}, one can  evaluate $\int |\nab u_i|^2$ and we claim that such a construction can be achieved with
\begin{equation*}
\int_{K_R\backslash \Gamma} |E_\eta|^2 \le   \sum_i \int_{\mathcal{K}_i} |\nabla u_{i,\eta}|^2 \leq o(R^d),
\end{equation*}
with $E_\eta= E+ \sum_{p \in \hat{\Lambda}} \nab f_\eta(\cdot -p)$, 
i.e. the modification in the boundary layer $K_R \backslash \Gamma$ can be made with a negligible energy, so that 
\be\label{estene}
\int_{K_R} |E_\eta|^2\le \int_{K_R} |\nab h_\eta|^2 + o(R^d).\ee
We would  like to define $\nab \hat{h}$ as $E$ in $K_R$,  but the problem is that $E$ is not a gradient. To remedy this, we use a kind of Hodge (or Helmoltz if $d \le 3$) decomposition, which consists in adding a vector field to $E$ to make it a gradient, without changing its divergence,  while not deteriorating the energy estimate \eqref{estene}.
Let us now show this more precisely.

The Hodge decomposition tells us that 
we may find a vector field $X$ defined over $K_R $, such that 
$\div X = 0$ in $K_R$, the normal component $X \cdot {\nu} = 0$ on $\partial K_R$ and $E + X$ is the  gradient of a function, which we call  $\hat{h}$,  hence so is $E_{\eta} + X= \nab \hat{h}_\eta$. An easy computation then yields
\begin{displaymath}
\int_{K_R} |E_{\eta}|^2 = \int_{K_R} |E_{\eta} + X|^2 + \int_{K_R} |X|^2 - 2 \int_{K_R} (E_{\eta} + X) \cdot X 
\end{displaymath}
But we can apply Green's theorem on the right-most term and find \begin{displaymath}
\int_{K_R} (E_{\eta} + X) \cdot X  = 0 
\end{displaymath}
since we saw that   $E_\eta +X $ is a gradient and by assumption $\div X = 0$ and $X \cdot {\nu}=0$. We conclude that 
\begin{equation}
\int_{K_R} |E_{\eta}|^2 \geq \int_{K_R} |E_{\eta} + X|^2 = \int_{K_R} |\nabla \hat h_\eta|^2
\end{equation} 
which combined with \eqref{estene}, proves that $\hat{h}$ satisfies all the desired properties (cf. \eqref{diveee}).
\end{proof}

\begin{coroll}
\label{corollscreen}
Given $\eta<1$, for any $R$ large enough  and such that $|K_R|\in \mathbb{N}$, there exists  an $\bar{h}$ satisfying \eqref{eq:proscreen field}--\eqref{eq:mindistr plus}, 
such that 
\be\label{645} \limsup_{R\to +\infty} \dashint_{K_R} |\nab\bar{h}_\eta|^2- c_d g(\eta) \le \inf_{\bar{\mathcal{A}}_1} \mathcal{W}_\eta.\ee 
\end{coroll}
\begin{proof} In view of  Propositions \ref{avantscreening} and \ref{propscreening}, for each given $R$, we may choose a $\nab \bar{h}$ approximating $F_{\eta, R}$ (cf. \eqref{fetar}) and associated to simple well-separated points (separated by $r_0$ which may depend only on $d$)  and with $\frac{\pa \bar{h}}{\pa \nu}=0$ on $\pa K_R$, which implies by Remark~\ref{remneutraliteglobale} that $\# (\Lambda \cap K_R)= |K_R|$. 
Then, letting $R \to \infty$, we have
$$\limsup_{R\to +\infty} \dashint_{K_R} |\nab\bar{h}_\eta|^2- c_d g(\eta) \le \limsup_{R \to + \infty} \frac{F_{\eta, R}}{|K_R|} -c_dg(\eta)  . $$
Then \eqref{645} follows as  a consequence of the definitions of $\mathcal{W}_\eta$ and $F_{\eta, R}$. 
\end{proof}
Note that taking the $\bar{h}$ given by this corollary and periodizing it after reflection allows to show that $\inf_{\bar{\mathcal{A}}_1} \mathcal{W}_\eta$ has a minimizing sequence made of periodic vector fields. 
\\

Let us  now give a sketch of the proof of Proposition \ref{avantscreening}, whose argument is based on \cite{lieb2}. 
\begin{proof}
We treat the case of simple points, the case of multiple points can be ruled out in the same way as a limiting case of simple points. 

Let $\Lambda$ be the set  of points, and  let us look a particular point of $\Lambda$, which we can assume, up to translation, to be the origin. For simplicity of the presentation,  we will neglect the boundary effects and do as if the configuration lived in the whole space $\R^d$. We wish to show that if  $\Lambda\backslash \{0\}$ contains a point $x$ very close to $0$, then the point $x$ can be moved away from $0$ to a point $y$ while decreasing the energy. 

The first step is to show that for a minimizer,  each point is at the minimum of the potential generated by the rest.
Let $U$ be the potential generated by all the points of $\Lambda$ except $x$ (in particular $U$ is regular in a neighborhood of $x$) and by the  background distribution (which here is constant, but weaker assumptions would suffice). Suppose that we modify the configuration by moving  $x$ to $y$,  and let $\bar{h}$ be the perturbation induced on the electrostatic potential, i.e.
\begin{displaymath}
\left\{\begin{array}{ll} -\Delta \bar{h} = c_d \left( \delta^{(\eta)}_y - \delta^{(\eta)}_x \right) & \text{in} \ K_R\\ \bar{h}=0 & \text{on} \ \partial K_R.  \end{array}\right.
\end{displaymath}
We want to estimate the energy of the perturbed configuration. Writing $g_\eta=\min(g,g(\eta))$, we have $\bar h= g_\eta(\cdot -y)-g_\eta (\cdot- x)$, and $h= U+g_\eta(\cdot -x)$.  Several integrations by parts allow to see that 
\begin{multline}\label{minolocU}
\int_{K_R} |\nabla (h+ \bar{h})|^2 - \int_{K_R} |\nabla h|^2 =
\int_{K_R} |\nab \bar{h}|^2 + 2 c_d \int_{K_R} ( \delta^{(\eta)}_y - \delta^{(\eta)}_x) h\\
= 
c_d \int (\delta_y^{(\eta)}- \delta_x^{(\eta)} ) (g_\eta (\cdot-y)-g_\eta(\cdot -x)) + 2 c_d \int ( \delta^{(\eta)}_y - \delta^{(\eta)}_x)( U+ g_\eta (\cdot -x))
\\
\simeq  2c_d \int U (\delta_{y}^{(\eta)} - \delta_{x}^{(\eta)}) \quad \text{as} \ \eta \to 0
\end{multline} where we have used that $g_\eta=g(\eta)$ on $\pa B(0, \eta)$.
Thus, if the configuration minimizes the energy, then both sides must always be nonnegative hence  $x$ must be at a  local (even global) minimum of $U$. This is true, at least formally for $\eta = 0$, but can be adapted to smeared out charges with errors that become negligible as $\eta$ gets small. 
On the other hand, the potential $U$ may be decomposed as 
\begin{displaymath}
U = U^{ext} + U^{int}
\end{displaymath} 
where 
\begin{displaymath}
U^{int} = c_d \Delta^{-1}\left(\delta_{0}^{(\eta)} - \textbf{1}_{B(0,2 r_0)}\right)
\end{displaymath}
is the potential created by the singular charge at $0$ and by the background distribution in the ball, and 
\begin{displaymath}
U^{ext} =  c_d \Delta^{-1}\left(\sum_{p\in \Lambda\backslash \{x,0\} } \delta_p^{(\eta)}  - \textbf{1}_{B(0, 2r_0)^c}\right)
\end{displaymath}
is the potential generated by all the other charges, without $x$ (and the background distribution).  We may then observe  that $U^{ext}$ is super-harmonic in the ball $B(0,2r_0)$ and thus achieves its minimum at some point $\bar{x}$ that belongs to the boundary $\partial B(0, 2r_0)$. Moreover, $U^{int}$ is radial and explicitly computable, and if $r_0$ is small enough one can check that $U^{int} (r)$ is decreasing and thus   $U^{int}$ achieves its minimum over $B(0,2r_0)$ on the boundary $\partial B(0,2 r_0)$. But then  if $x\in B(0,2r_0)$ and $x\neq 0$, $x$ can be moved to $y=\bar{x}$, this decreases $U$, hence  in view of \eqref{minolocU}, this decreases the energy (how much it can be decreased can  be better  estimated, and this quantitative version of the argument 
allows to adapt the proof to the case of a bounded set and with $\eta$ nonzero).
This shows the desired result: if $0$ and $x$ (hence two arbitrary points in the configuration) are not separated by a distance $2r_0$ depending only on $d$, then the energy can be decreased. 
\end{proof}


\section{Upper bound and consequences for ground states}\label{sec63}
As already mentioned, the  lower bound of Theorem \ref{thlbnext} needs to be complemented by a corresponding upper bound, proving that the lower bound was   sharp. As in $\Gamma$-convergence, this is accomplished by an explicit construction.

The following proposition states the result we can obtain. It shows that we can find some configuration of points for which the lower bound of Corollary \ref{cor61} is sharp. Because we have in view the statistical mechanics problem as well, it will be useful to show that this is true not only for that configuration, but for a ``thick enough" neighborhood of it.

\begin{prop}[Upper bound at next order] 
\label{proub}
For any $\eps>0$ there exists $r_1>0$ and for any $n$ a set $A_n \subset (\mr^d)^n$ such that 
\begin{equation}\label{eq:volume min}
|A_n|\ge n! \left(\pi (r_1)^d/n\right)^n  
\end{equation}
and for any $(y_1, \dots, y_n) \in A_n$ we have
\begin{equation}\label{bsup 3d}
\limsup_{n\to \infty}  n^{2/d-2} \left(H_n(y_1, \dots, y_n) - n^2 I(\mu_0) + \Big( \frac{n}{2}\log n\Big) \indic_{d=2}      \right) \le \xi_d+ \eps,
\end{equation}where $\xi_d$ is defined by \eqref{eq:gamma d}.
\end{prop}

\begin{proof} 
We sketch the main steps of the construction, which relies  on pasting together vector fields obtained via  the screening construction of  Proposition  \ref{propscreening}, more precisely those given by Corollary \ref{corollscreen}. That corollary was stated for hypercubes but it applies to hyperrectangles as well.
\medskip

\noindent
{\bf Step 1.}
We fix some large $R$ and, thanks to assumption  \textbf{(A5)} (cf. Section \ref{secassump}), partition $\Sigma'$ (the blown-up of the set $\Sigma$) into hyperrectangles $K_i$ of sidelengths in  $[R,2R]$ and such that $\int_{K_i} \mu_0' \in \mn$.   This is not very difficult to do, cf. \cite[Lemma 7.5]{ss1}, it leaves however a (layer) region $\Sigma_{\mathrm{bound}}'$ near the boundary of $\Sigma'$ which cannot exactly be partitioned into hyperrectangles.
We let $m_i = \fint_{K_i} \mu_0'$.
\medskip

\noindent
{\bf Step 2.}
We  paste  in each  $K_i$, copy of the $\bar{h}$ given by Corollary \ref{corollscreen}, translated and rescaled  by a factor $m_i^{1/d}$, so that we have a  solution to 
$$\left\{\begin{array}{ll}
-\Delta \bar{h}_i= c_d(\sum_p \delta_p - m_i) & \text{in} \ K_i\\
\frac{\pa \bar{h}_i}{\pa \nu} =0 & \text{on} \ \pa K_i,\end{array}\right.$$ and 
\be\label{estbhi}
\fint_{K_i} |\nab (\bar{h}_i)_\eta|^2- m_i c_d g(\eta)  \le \min_{\bar{\mathcal{A}}_{m_i} }\mathcal{W}+o_R(1).\ee
Note that the rescaling factor does not degenerate by assumption \textbf{(A6)} (cf. Section \ref{secassump}),  that the total number of points in $K_i$ is $m_i |K_i|= \int_{K_i} \mu_0'$, and that the points are separated by a distance depending only on $\overline{m}$ in assumption \textbf{(A6)}, as provided by Corollary \ref{corollscreen}.
\smallskip 

\noindent
{\bf Step 3.} Since $m_i$ is not the desired weight $\mu_0'$,  we correct $\bar{h}_i$ by adding a solution to 
$$
\left\{\begin{array}{ll}
-\Delta u_i= c_d(m_i- \mu_0') & \text{in} \ K_i\\
\frac{\pa u_i}{\pa \nu} =0 & \text{on} \ \pa K_i.\end{array}\right.$$
Thanks to assumption \textbf{(A6)}, we know that  $\mu_0'$ varies slowly (more precisely $|\nab \mu_0'|\le C n^{-1/d}$ at the scale considered), so $\|m_i - \mu_0'\|_{L^\infty(K_i)}$ is small, which allows to prove by elliptic regularity estimates that $u_i$ is small in a strong sense. 
We note that this is the point where we use the $C^1$ regularity of $\mu_0$ on its support, but that it could easily be replaced by a weaker statement showing slow variation (such as a H\"older continuity assumption).
 We then let $E_i= \nab \bar{h}_i+ \nab u_i$ in each $K_i$.
 \medskip
 
 \noindent
 {\bf Step 4.} We complete the construction by defining a vector field of the right form in the region $\Sigma_{\mathrm{bound}}'$ near the boundary of $\Sigma'$. Because that region has a negligible volume, it is not important to use an approximate minimizer of $\mathcal{W}$, it suffices to construct some vector field $E_{\mathrm{bound}}$ associated to well-separated points, and  satisfying 
 $$\left\{\begin{array}{ll}
-\div E_{\mathrm{bound}}= c_d (\sum_{p\in \Lambda_{\mathrm{bound}}}\delta_p - \mu_0') &\text{in} \ \Sigma_{\mathrm{bound}}'\\
E_{\mathrm{bound}}\cdot \nu= 0 & \text{on} \ \pa \Sigma_{\mathrm{bound}}'.\end{array}\right.
 $$
 and \be\label{ebound}
 \int_{\Sigma_{\mathrm{bound}}'}  |E_{\mathrm{bound}, \eta}|^2 - \# \Lambda_{\mathrm{bound}}c_d g(\eta)\le o(n).
 \ee

 \noindent
 {\bf Step 5.} We paste together the vector fields $E_i$ and $E_{\mathrm{bound}}$ defined in all the regions that make up $\Sigma'$, extend them by $0$ outside $\Sigma'$, and call the result $E$. Because the normal components of these vector fields are continuous across the interfaces between the regions, $E$ globally satisfies  a relation of the form
 $$-\div E = c_d (\sum_{p\in \Lambda_n} \delta_p - \mu_0')\quad \text{in} \ \R^d,$$
 for a collection of points $\Lambda_n$ in $\Sigma'$ which are simple and well separated, and for which we can check that $\# \Lambda_n=n$ (since $\int_{\Sigma'} \mu_0'= n\int\mu_0= n$).
 The vector field $E$ is no longer a gradient, however we can keep the points of $\Lambda_n$ and define $h_{n,\eta}'$ associated to them via \eqref{dehn'1}.
 Computing exactly as in Step 4 of the proof of Proposition \ref{propscreening}.
 shows that $$
 \int_{\R^d} |\nab h_{n,\eta}'|^2\le \int_{\R^d} |E_\eta|^2 $$
 where $E_\eta$ is as in \eqref{esmeared}  and 
 in view of the  above ($\nab u_i$ is  negligible), we may write 
 \be\label{1est}
 \int_{\R^d} |\nab h_{n,\eta}'|^2\le \int_{\R^d} |E_\eta|^2\le \sum_i \int_{K_i}|\nab(\bar{ h}_i)_{\eta}|^2 +\int_{\Sigma'_{\mathrm{bound}}} |E_{\mathrm{bound}, \eta}|^2 +   \text{ negligible terms}.\ee 
 \noindent
 {\bf Step 6.} We estimate the energy of the constructed configuration. 
Adding  the contributions given by \eqref{estbhi}, \eqref{ebound},  and inserting into \eqref{1est},  we obtain that for $R$ and $\eta$ given, we have  a configuration of points $\{x_1',\dots, x_n'\}=\Lambda_n$ (and a corresponding blown-down configuration ($x_1, \dots, x_n$))  for which 
 $$\int_{\R^d} |\nab h_{n,\eta}'|^2- nc_d g(\eta)
 \le \sum_i |K_i| \min_{\bar{\mathcal{A}}_{m_i} }\mathcal{W}+no_R(1).$$
 Using  a Riemann sum argument and  the continuity of $m \mapsto \min_{\bar{\mathcal{A}}_m} \mathcal{W}$ which follows from \eqref{scalingW}--\eqref{scalingW2d}, we conclude that 
 $$\int_{\R^d} |\nab h_{n,\eta}'|^2- n c_d g(\eta) \le \int_{\Sigma'} \Big(\min_{\bar{\mathcal{A}}_{\mu_0'} (x)} \mathcal{W}_\eta\Big) \, dx + o_R (n).$$
 Since the points are well-separated, as soon as $\eta $ is small enough, we are in the case of equality of Proposition \ref{propRS}, which means that we may write (using also that $\zeta(x_i)=0$ by Definition \ref{def23} since all the $x_i$'s are in $\Sigma$ )
 \begin{multline*}
 H_n (x_1,\dots, x_n) \le n^2 I(\mu_0)- \Big(\frac{n}{2}\log n\Big) \indic_{d=2} 
\\ + \frac{n^{2-2/d}   }{c_d}        \( \frac{1}{n}     \int_{\Sigma'} \min_{\bar{\mathcal{A} }_{\mu_0' (x)} }\mathcal{W}_\eta \, dx +o_R(1)\) .\end{multline*}
Taking $R$ large enough and $\eta$ small enough, and using the fact that $\min_{\bar{\mathcal{A}}_m} \mathcal{W}_\eta \to \min_{\bar{\mathcal{A}}_m}\mathcal{W}$ as $\eta \to 0$,  and by definition of $\xi_d$, we can find $(x_1, \dots, x_n)$ so that \eqref{bsup 3d} holds, i.e. so that we have the desired right-hand side up to an error $\eps$.
\smallskip 
 
 \noindent 
 {\bf Step 7.}
 The statement about the volume of the set $A_n$ follows by noting that if $y_i\in B(x_i,r_1 n ^{-1/d})$ for $r_1$ small enough (depending on $r_0$) then $y_1,\dots, y_n$ are also well separated and we may perform the same analysis for $H_n(y_1,\dots, y_n)$, except with an additional error  depending only on $r_1$ and going to $0$ when $r_1 \to 0$, so which can be made $<\eps$ by taking $r_1$ small enough. It is  in addition clear that the set of such $y_i$'s has volume $n! (r_1^d/n)^n$ in configuration space: the $r_1 ^d /n$ term is the volume of the ball $B(x_i,r_1 n ^{-1/d})$, it is raised to the power $n$ because there are $n$ points in the configuration, and multiplied by $n !$ because permuting $y_1,\dots,y_n$ does not change the energy.

\end{proof}

\begin{rem}
If we view these results as a $\Gamma$-convergence at next order of $H_n$, then Theorem \ref{thlbnext} provides  the $\Gamma$-liminf relation, but Proposition \ref{proub} only provides the $\Gamma$-limsup relation at the level of minimizers, which is enough to conclude about these, but does not provide a full $\Gamma$-convergence result.  To get one, we would need to construct recovery sequences for all probabilities $P$ satisfying the first two properties of Theorem \ref{thlbnext}. This is technically more complicated, because such $P$'s need to be approximated by a single sequence, and it was accomplished in \cite{ss1} for $W$ in dimension $2$. 
In the setting of $\mathcal{W}$, we are limited by the fact that our screening procedure is written only for configurations with well-separated points. \end{rem}

Comparing with the lower bound result of Theorem  \ref{thlbnext}, we immediately obtain the following result on the minimum and  minimizers of the Coulomb gas Hamiltonian.

\begin{theo}[Ground state energy expansion and microscopic behavior of minimizers]\label{th1}\mbox{}\\
Assume $V$ is continuous and such that the equilibrium measure exists  and satisfies \textbf{(A5)}--\textbf{(A6)} (cf. Section \ref{secassump}).
As $n \to \infty$ we have 
\be\label{expansionmin}
\boxed{
\min H_n =   \displaystyle n^2 I(\mu_0) - \Big(\frac{n}{2}\log n \Big) \indic_{d=2} 
+ n^{2-2/d}\xi_d + o(n^{2-2/d}),}\ee
where $\xi_d$ is as in \eqref{eq:gamma d}, and it holds that 
$\xi_d = \min \widetilde{\mc{W}}.$
In addition, if $(x_1, \dots, x_n)\in (\mr^d)^n$ minimize $H_n$, \footnote{again, the configuration depends implicitly on $n$}letting $h_n' $ be associated via \eqref{hn'} and  $P_n$ and $P$  be as in Theorem \ref{thlbnext},  then  $P$ minimizes $\widetilde{\mathcal{W}}$ and for  $P$-a.e. $(x, \j)$, $\j$ minimizes $\mathcal{W}$ over $\bar{\mathcal{A}}_{\mu_0(x)}$.
 \end{theo}
This result was obtained in \cite{rs} for $d\ge 2$ and previously for $d=2$ in \cite{ss1} with $W $ instead of $\mathcal{W}$, which gives another proof  that in that case $\min_{\mathcal{A}_1} W= \min_{\bar{\mathcal{A}}_1} \mathcal{W}$ (we have seen other ways of justifying this in Corollary  \ref{propequiv}). In \cite{ss2} the corresponding result for $d=1$ is obtained:
\be
 \boxed{
\min H_n =   \displaystyle n^2 I(\mu_0) - n\log n   + n\Big(  \frac{1}{2\pi} \min_{\mathcal{A}_1} W    - \int \mu_0(x) \log \mu_0(x) \,dx\Big) +o(n)  } 
  \ee
and we recall that in that case the value of the minimum  of $W$ is known, cf. Theorem \ref{thm1d}.

These expansions of the minimal energy are to be compared to Theorem \ref{thcvmini}~: we have 
obtained as announced a next order expansion of the minimal energy, in terms of the unknown (except in 1D) constants $\min W$, $\min \mathcal{W}$, and we have seen that -- modulo the logarithmic terms in dimensions $1$ and $2$ -- this next order term lives at the order $n^{2-2/d}$, something which was not immediate.  Moreover, with the statement ``for $P$-a.e. $(x,E)$,  $E$ minimizes $\mathcal{W}$ over $\bar{\mathcal{A}}_{\mu_0(x)}$," we have obtained a characterization of the minimizers at the microscopic level~: after blow-up around a point $x\in \Sigma$ chosen uniformly, one sees a jellium of points with  density (or  ``background")  $\mu_0(x)$ (note that this is the only way -- through the equilibrium measure -- that the result depends on the potential $V$), which has to almost everywhere be a minimizer of $\mathcal{W}$. 
This reduces the study of the minimizers of $H_n$ to the minimization of $\mathcal{W}$. If one believes that minimizers of $\mathcal{W}$ are lattices, as one is led to believe in dimension $2$ (recall in dimension $2$ the Abrikosov triangular lattice is the unique best one), then one expects that minimizers of $H_n$ microscopically look ``almost-everywhere" like such  lattices.  

In \cite{rns} we showed that in dimension $2$, for true minimizers (and not just configurations whose energy is asymptotically minimal) the ``almost-everywhere" part of the statement can be replaced by ``everywhere".  We also showed as in the proof of Proposition \ref{avantscreening} that for minimizers, points are separated by a distance $c n^{-1/d}$ with $c>0$ depending only on the dimension and $\|\mu_0\|_{L^\infty}$, and that  the discrepancies of the numbers of points (cf. Lemma \ref{lem:fluctu charge}) are $O(R^{d-1})$ (this can be compared to \cite{ao}). Those results  should be adaptable to general dimensions.

\section{Consequences on the statistical mechanics}
 Now that we have completed the expansion at next order of the Hamiltonian $H_n$, we may, just as in Section \ref{LDPsection}, apply it to obtain information on the thermal states and the partition function, by simply inserting this expansion into 
 \eqref{gibbsagain}.
 In this context, the result of Theorem \ref{th1} on ground states can be viewed as a zero temperature  (or $\beta=\infty$) result.

\begin{theo}[Next order asymptotic expansion of the partition function]
\label{thz}
Assume $V$ is continuous and such that $\mu_0$ exists and satisfies \textbf{(A4)}--\textbf{(A6)} (cf. \eqref{conditionintegrabilite} and Section \ref{secassump}). Assume  $\bar{\beta}:= \limsup_{n\to +\infty}\beta n^{1-2/d}>0$.
Then there exists $C_{\bar{\beta}}$ (depending on $V$ and $d$) such that $\lim_{\bar{\beta} \to +\infty} C_{\bar{\beta}}=0$  
and 
\be  \label{631}
 \boxed{ \left|\log Z_{n,\beta} +    \frac{\beta}{2} \(  n^2  I(\mu_0) - \Big(\f{ n}{2}\log n\Big) \indic_{d=2}   + n^{2-2/d}\xi_d\)  \right| \le C_{\bar{\beta}}  \bar{\beta} n,}\ee
 where $\xi_d $ is as in \eqref{eq:gamma d}.
\end{theo}
In \cite{ss2} one finds the analogous result in dimension $1$: there exists $C_\beta$ with $\lim_{\beta\to  +\infty} C_\beta=0$ such that 
$$
 \boxed{ \left|\log Z_{n,\beta} +    \frac{\beta}{2} \(  n^2  I(\mu_0) -  n\log n   + n \xi_1\)  \right| \le C_\beta \beta n  ,}$$ with $\xi_1= \frac{1}{2\pi} \min_{\mathcal{A}_1}  W - \int_{\R } \mu_0 \log \mu_0 . $

These results should be compared to the expansion stated in Theorem \ref{LDP}~: again, we obtain here an expansion to next order. The expansion gets more precise as $\bar{\beta} \to + \infty$ i.e. $\beta\gg n^{2/d-1}$ (and then one essentially recovers the minimal energy as in Theorem \ref{th1}). However it already identifies the conjectured crystallization regime as $\beta \gg n^{2/d-1}$, and complements the result of Corollary \ref{boundlogz} by a lower bound of similar form.
Except in dimension  $1$, where as already mentioned, $Z_{n,\beta}$ is exactly known via Selberg integrals (at least  for $V$ quadratic), these results improve on  the known results.

We conclude with a result of large deviations type, which improves on Theorem~\ref{LDP} by bounding the probability of rare microscopic events. 

\begin{theo}[Rare events at the microscopic scale]\label{th4}
Assume $V$ is continuous and satisfies \textbf{(A3)}--\textbf{(A6)}.
Let $i_n$ be the map which to  any $x_1, \dots, x_n\in \R^d$ associates    $P_n\in \mathcal{P}(\Sigma\times X)$ as in Theorem \ref{thlbnext}. 
For any $n>0$ let $A_n\subset (\mr^d)^n$ and $$A_\infty= \{ P \in \mathcal{P}(\Sigma \times X)|
\exists (x_1, \dots , x_n) \in A_n, i_n(x_1, \dots, x_n) \rightharpoonup P \  \text{up to a subsequence} \}.$$ Let $\bar{\beta}>0$ be as in Theorem \ref{thz},   $\xi_d$ be as in \eqref{eq:gamma d}, and $\widetilde{\mc{W}}$ as in \eqref{tildew}. There exists $C_{\bar{\beta}}$ such that $\lim_{\bar{\beta} \to + \infty} C_{\bar{\beta}}=0$, and 
\begin{equation}
\label{ldr}
\limsup_{n\to \infty} \frac{ \log \Q(A_n)}{n ^{2-2/d}} 
\le -\frac{\beta}{2} \left(\inf_{P\in A_\infty}\widetilde{\mc{W}}  -  \xi_d - C_{\bar{\beta}}  \right).
\end{equation} 
\end{theo}
Here $A_\infty$ is the set of accumulation points of the $P_n$'s associated to configurations in $A_n$.
Another way of phrasing this result is that in the limit $n\to \infty$
$$\widetilde{\mc{W}}(P) \le \xi_d+C_{\bar{\beta}} = \min\widetilde{\mc{W}} + C_{\bar{\beta}}$$
except with exponentially decaying probability. This means that we have a threshhold phenomenon: the Gibbs measure concentrates on configurations whose $\widetilde{\mathcal{\mc{W}}}$ is below the minimum plus $C_{\bar{\beta}}$. The threshhold tends to $0$ as $\beta \gg n^{2/d-1}$ and, in that regime,  the Gibbs measure concentrates on minimizers of $\widetilde{\mc{W}}$, a weak crystallization statement.  
This indicates that even at nonzero temperature, configurations have some order (if one believes of seeing $\mathcal{W}$ as a measure of disorder, cf. the end of Section \ref{sec54}), in the sense that their $\widetilde{\mc{W}}$ cannot be too large. 
  \begin{proof}[Proof of Theorems \ref{th1} and \ref{th4}]
 The proof is a direct consequence of the definition of $\Q$ and Theorem \ref{thlbnext}, following the same method as in Theorem \ref{LDP}.
 We start by obtaining  a lower bound for $Z_{n,\beta}$ from Proposition \ref{proub}~: let $\eps>0$ be given. By   definition of $Z_{n,\beta}$, we may write 
$$
 Z_{n,\beta} = \int_{(\R^d)^n} e^{-\frac{\beta}{2}H_n(x_1,\dots , x_n)} \, dx_1 \dots dx_n
 \ge \int_{A_n} e^{-\frac{\beta}{2}H_n(x_1,\dots , x_n) }\, dx_1 \dots dx_n
$$
 where $A_n$ is the set given by Proposition \ref{proub}.
Taking the logarithm and  inserting \eqref{bsup 3d}, we obtain
 $$
 \log Z_{n,\beta} \ge \log |A_n| - \frac{\beta}{2}\( n^2 I(\mu_0) +\Big ( \frac{n}{2}\log n\Big) \indic_{d=2} +n^{2-2/d}  (\xi_d +\eps) \) .$$
 But $\log |A_n| \ge \log  (n! \pi (r_1 )^d/n)^n \ge -C_\eps n  -C $ using Stirling's formula, where $C_\eps$ depends on $r_1$, itself depending on $\eps$.
 Inserting, we obtain the lower bound corresponding to \eqref{631}.
 
 Conversely, let $A_n$ be an arbitrary set in $(\R^d)^n$. 
 Assume  $(x_1, \dots, x_n)$ minimizes (or almost minimizes) $H_n$ over $A_n$, then Theorem \ref{thlbnext} gives us that 
 \begin{multline*}
  H_n(x_1, \dots, x_n) - 2n \sum_{i=1}^n \zeta(x_i)  + \Big(\frac{n}{2} \log n\Big) \indic_{d=2}\\ \ge n^{2-2/d}\( \widetilde{\mathcal{W}}(P) +o_n(1)\)\ \ 
 \ge n^{2-2/d}\(  \inf_{P \in A_{\infty} }\widetilde{\mc{W}}(P)+o_n(1)\).\end{multline*}
 It follows that, inserting this estimate into \eqref{gibbsagain},  we obtain an upper bound on the probability of $A_n$ by writing
\begin{multline*}
 \log \Q(A_n) \le - \log Z_{n,\beta} \\- \frac{\beta}{2} \( -  \Big(\frac{n}{2} \log n\Big) \indic_{d=2}+ n^{2-2/d} \( 
  \inf_{P \in A_{\infty} }\widetilde{\mathcal{W}}(P)+o_n(1)\) \) \\
 + \log 
  \int_{(\R^d)^n}e^{-\beta n \sum_{i=1}^n \zeta(x_i)}  \, dx_1 \dots dx_n.\end{multline*}
  Inserting the lower bound on $\log Z_{n,\beta}$ obtained above, and the result of 
   Lemma~\ref{asymptozeta} (since our assumptions ensure that $\beta n \to \infty$ as $n \to \infty$) 
 we find 
 \begin{multline}\label{qan}
 \log \Q(A_n) \le - \log Z_{n,\beta} -  \frac{\beta}{2}  n^{2-2/d} \( -\xi_d +
  \inf_{P \in A_{\infty} }  \widetilde{ \mathcal{W}}(P) +o_n(1)\)  \\
 + n (|\omega|+o_n(1)).\end{multline}
Taking in particular $A_n = (\R^d)^n$, we have $\Q(A_n)=1$,  and we can check that \begin{multline*}
A_\infty\subset  \Big\{ P| E\in \bar{\mathcal{A}}_{\mu_0(x)}  P\text{-.a.e, }  \text{and the first marginal of $P$ is the normalized } \\
\text{Lebesgue measure on } \ \Sigma\Big\}\end{multline*} by Theorem \ref{thlbnext}. It follows as previously  that 
$$\inf_{P \in A_{\infty} }\widetilde{ \mathcal{W}}(P)  = \frac{1}{c_d} \int\min_{\bar{\mathcal{A}}_{\mu_0(x)}}  \mathcal{W}  \, dx= \xi_d$$ by the change of scales formula.
Inserting into \eqref{qan}, 
we deduce the upper bound for $\log Z_{n,\beta}$ stated in \eqref{631}.  This completes the proof of \eqref{631} and of Theorem \ref{th1}. 
Using then again \eqref{qan} for a general $A_n$, and plugging in \eqref{631}, we obtain the result of Theorem \ref{th4}.
 \end{proof}
 The analogues of Theorems \ref{th1}, \ref{thz} and \ref{th4} are obtained in \cite{ps} for the more general case of Riesz interaction kernels $g(x)=|x|^{-s}$ with $d-2\le s <d$.
 
 In \cite{lebles} we are able to go further and obtain, for the Riesz as well as the Coulomb interaction kernel,  a complete next order Large Deviations Principle  on the limiting objects $P$, with rate function 
 $$\widetilde{ \mathcal W }(P)+ \frac{1}{\bar\beta} Ent (P)$$ where $Ent (P)$ is a specific relative entropy with respect to the Poisson process. This implies  
 the existence of  an exact asymptotic expansion of $\log Z_{n,\beta}$ up to order $n$, hence of a thermodynamic limit, as well as several other results.
 
\chapter{The Ginzburg-Landau functional: presentation and heuristics}
\label{glheuris}

In this chapter, we present some nonrigorous heuristics on the Ginzburg-Landau model that allow to see how and when  vortices are expected to form in minimizers.
A detailed presentation of the functional and of the related physics was already given in \cite[Chap. 2]{livre}, so we will try here to focus more on the new additions compared to that reference. We also refer to the classic book of  \cite{dg}.

\section{The functional}
Let us start by  recalling the expression of the  Ginzburg-Landau functional that was introduced in \eqref{GL} in Chapter \ref{chap-intro}:
\be \label{GL2e}
G_{\eps}(u, A)  = \f{1}{2} \int_{\Omega} |(\nabla - iA)u|^2 + |\curl  A - \he|^2 + \f{(1 - |u|^2)^2}{2\eps^2},
\ee
and the Ginzburg-Landau equations 
\begin{equation*}
(\text{GL})\left\lbrace 
\begin{array}{cc}
- \nabla^2_Au = \f{1}{\eps^2} u(1- |u|^2) & \text{in} \ \om
 \\ [2mm]
- \nabla^{\perp}h = \langle i u, \nabla_A u \rangle & \ \text{in} \ \om
\end{array} \right.
\end{equation*} with boundary conditions
\begin{equation*}
\left\lbrace 
\begin{array}{cc}
\nab_A u \cdot \nu=0 & \text{on} \ \bo\\
h=\he& \text{on} \ \bo.\end{array}\right.
\end{equation*}
The precise meaning of the quantities appearing here was given in Chapter \ref{chap-intro}. It is an important fact that this is a $\mathbb{U}(1)$ gauge theory, i.e. all the physically meaningful quantities, such as  the  energy  and the equations (GL) are invariant under the gauge-transformations
$$
\left\lbrace \begin{array}{l}
 u \mapsto u  e^{i\Phi}\\
 A\mapsto A + \nab \Phi \end{array}\right.$$
where $\Phi$ is any smooth real-valued function. One may easily check that, in addition to the energy,  gauge-invariant quantities include $|u|^2$ which represents the density of superconducting electrons, the magnetic field $h= \nab \times A$, and the superconducting current $j:= \langle iu , \nab_A u \rangle$.
 For more general gauge theories in theoretical physics, in particular non-Abelian ones, we refer to \cite{jt,manton}.
\section{Types of states, critical fields}

\subsection{Types of solutions and phase transitions}
Three types of solutions (or states) to (GL) can be found:
\begin{enumerate}
\item  the normal solution : $(u\equiv 0, \curl A \equiv \he)$.  This is a true solution to $(GL)$ and its energy is  very easily computed: it is  $ \frac{|\om|}{4\eps^2}$.
\item
the Meissner solution (or superconducting solution) :  $(u \equiv 1, A\equiv 0)$, and all its gauge-transforms. This is a true solution if $\he=0$, and a solution close to this one (i.e. with $|u|\simeq 1$ everywhere) persists if $\he $ is not too large.
Its energy is approximately
$G_\eps(1,0)=\frac{\he^2}{2}|\om|$.
By comparing these energies, we see that the Meissner solution is more favorable when $\he$ is small, while the normal solution is more favorable when $\he$ is large enough, more precisely when $\he > \frac{1}{\eps \sqrt{2}}$.
\item the vortex solutions: there is another state, with vortices, called the {\it mixed state} where normal and superconducting phases co-exist, and which is more favorable for  intermediate values of $\he$.
\end{enumerate}

The physics gives us the following more precise results. 
There are three main critical values of $\he$ or {\it critical
fields} $\hci$, $\hcii$, and $\hciii$, for which phase-transitions
occur. 
\begin{itemize}
\item For $\he<\hci$ there are no vortices and  the energy minimizer is  the superconducting state $(u \equiv 1, A\equiv 0)$. (This is a true solution if $\he=0$, and a solution close to this one (i.e. with $|u|\simeq 1$ everywhere) persists if $\he $ is not too large.)
It is said that the superconductor ``expels" the applied magnetic field, this is the ``Meissner effect", and the corresponding solution is called the Meissner solution. 
\item  For $\he= \hci$, which is of the order of $\lep$ as $\eps \to 0$,   the first vortice(s) appear.
\item For $\hci<\he<\hcii$ the superconductor is in the ``mixed phase" i.e. 
there are vortices, surrounded by superconducting phase where $|u|\simeq 1$.
 The higher $\he>\hci$, the more vortices
there are. The vortices  repel each other so they tend to arrange in
 triangular Abrikosov lattices in order to minimize their
repulsion.
\item  For $\he= \hcii\sim \frac{1}{\eps^2} $, the vortices are so densely packed that they
 overlap each other, and  a second phase transition
 occurs, after which $|u|\sim 0$ inside the sample, i.e. all
 superconductivity in the bulk of the sample is lost.
\item For $\hcii<\he<\hciii$
 superconductivity
persists only  near the boundary, this is called {\it surface
superconductivity}.  More details and the mathematical study of this transition are found in \cite{fh} and references therein.
\item 
 For $\he> \hciii=O(\frac{1}{\eps^2}) $ (defined
in decreasing fields), the sample is completely in the normal
phase, corresponding to the ``normal" solution $u\equiv 0, h\equiv \he$ of (GL). See \cite{gp} for a proof. 
\end{itemize}
The picture below represents the phase diagram for two-dimensional superconductors which is found in physics textbooks.

\begin{figure}

\begin{center}
\includegraphics[scale=0.6]{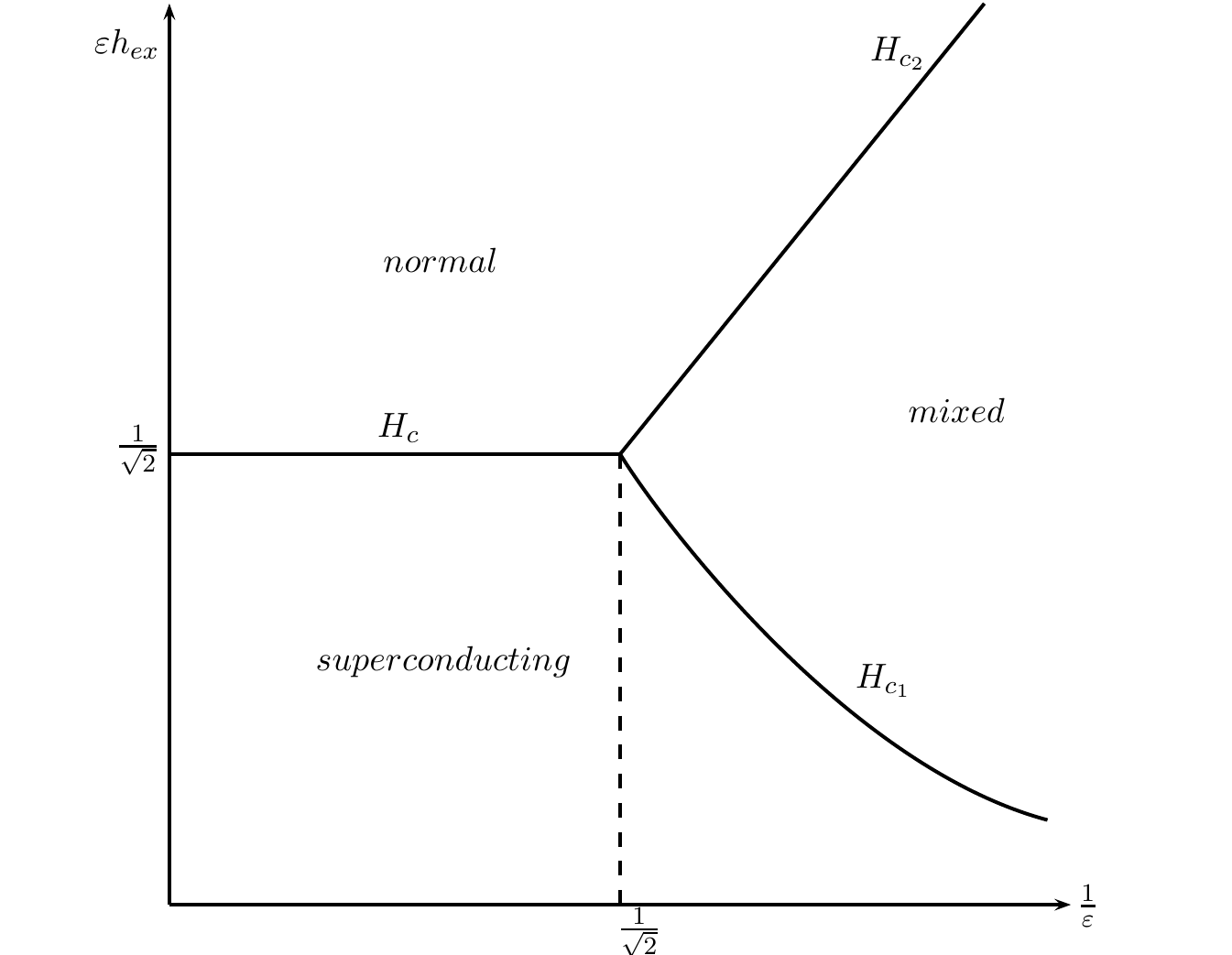}
\end{center}
 \caption{Phase diagram in $\mr^2$}
 \end{figure}

\subsection{Vortex solutions}\label{722}

We have already mentioned in Chapter \ref{chap-intro} that vortices are isolated zeroes of the order parameter $u$, and that they come with an integer topological {\it degree}. 

 When $\eps$ is small, the potential term in \eqref{GL2e} implies that  any discrepancy between $|u|$ and $1$ is strongly penalized,   and a scaling
argument hints that $|u|$ is different from $1$ only  in regions of
characteristic size $\eps$.
 A typical vortex centered at a point $x_0$
``looks like" $ u= \rho e^{i\varphi} $ with $\ro (x_0)=0$ and
$\ro(x)= f(\frac{|x-x_0|}{\eps}) $ where $f(0)=0$ and $f$ tends to $1$
as $r \to + \infty$, i.e. its characteristic core size is $\eps$,
and
\be\label{degvor}\frac{1}{2\pi} \int_{\pa B(x_0,R\eps)} \frac{\pa \varphi}{\pa \tau} = d \in \mz\ee is
its degree (note that the phase $\vp$ can only be understood as a ``multi-valued function"). For example
$\vp= d \theta$ where $\theta $ is the polar angle centered at
$x_0$ yields a vortex of degree $d$.

True radial solutions in $\mr^2$ of the Ginzburg-Landau equations
 of degree $d$, of
 the form
$$u_d(r, \theta)=f_d(r)e^{id \theta}, \quad  A_d(r, \theta)=
g_d(r) (-\sin \theta, \cos\theta) $$ have been shown to exist  \cite{plohr,plohr2,bergerchen}. In \cite{gustsigal}, it was shown that they are stable for $\eps<1/\sqrt{2}$ and $|d|\le 1$ and for $\eps \ge 1/\sqrt{2}$, unstable otherwise.

Let us compute the approximate energy of the rescaled version of such a solution, in a bounded size domain (say $B_R$)~: letting $\widetilde{u_d}(r,\theta)=u_d(\frac{r}{\eps}, \theta)$ and $\widetilde{A_d}(r,\theta)= \frac{1}{\eps} A_d(\frac{r}{\eps}, \theta)$, 
\begin{eqnarray}\nonumber
\hal \int_{B_R} |\nab_{\widetilde{A_d}}\widetilde{u_d}|^2 & = & 
\hal \int_{B_{R/\eps}   }|\nab_{A_d}u_d|^2\\ \nonumber
& = & 
 \hal \int_{B_{R/\eps}   } |\nab f_d|^2 + f_d^2 |\nab (d\theta)- A_d|^2 \\
\nonumber
& \approx & \hal \int_{B_{R/\eps}} |\nab f_d|^2 +\hal \int_0^{R/\eps}   f_d^2 \left(\frac{d}{r}\right)^2 2\pi r \, dr\, d\theta\\ \label{costvor}   & = & \pi d^2 \int_1^{R/\eps} \frac{dr}{r}= \pi d^2 \log \frac{R}{\eps} ,\end{eqnarray}
as $\eps \to 0$. Here we have used the fact that $f_d$ is expected to have a cut-off effect in balls of lensgthscale $\eps$ near the vortex center (here the origin). The error in the above estimate is in fact $O(1)$ as $\eps \to 0$.
We thus see that with such an ansatz, in a solution with vortices, each of them ``costs"  at leading order an energy $\pi d^2 \lep$, with $d$ its degree.

In a bounded domain, there are indeed solutions with several such vortices glued together, for example arranged along a triangular lattice (their existence is proved at least as a bifurcation from the normal solution, see \cite{c,almog4}).

\subsection[Related models]{Related models: superfluids and rotating Bose-Einstein condensates}

For comparison, it is interesting to mention the energy functional corresponding to the Gross-Pitaevskii model (in the so-called Thomas-Fermi regime) of superfluids such as Bose-Einstein condensates, in rotation with velocity vector $\Omega$, and in a confining potential $V$ (cf. \cite{fetter} for general reference and \cite{lseir} for the derivation from quantum mechanics)
$$\text{GP}(u)=\int_{\mr^2}|\nab u|^2 - (\Omega \times x)\cdot \langle iu,\nab u \rangle +V(x)|u|^2 +G|u|^4$$
which, after completing the squares can equivalently be written as
\be 
\text{GP} (u)=  \int_{\mr^2} |\nab u- i \Omega \times x u |^2 + (V(x)- \Omega^2) |u|^2  + G |u|^4
\ee
or
\be
\text{GP}_\eps(u)= \int_{\mr^2} |\nab u- i \Omega \times x u |^2 + \frac{(V_{\mathrm{eff}}(x)  - |u|^2)^2}{2\eps^2}  \ee
for some effective potential $V_{\mathrm{eff}}$.
The well-known  analogy with the Ginzburg-Landau model is readily visible: the role of $A $ is replaced by that of the angular momentum vector-field $\Omega \times x$, and that of $\he$ by $\Omega$. The  effective potential $V_{\mathrm{eff}}$ (which depends on $V, \Omega, $ and $G$) does not create significant differences from the constant $1$ when $\eps$ is small compared to other characteristic constants. The limit $\eps \to 0$ is called in this context the Thomas-Fermi regime. 
As a result  of this strong analogy, the techniques developed for Ginzburg-Landau adapt well to treat such functionals, cf. the review article \cite{nicolas}    and references therein. The analogy functions  well for rotation angles which are not too large, i.e. for the equivalent of the regime of $\he$ (much) smaller than the  second critical field $\hcii$, but significantly breaks down after that, i.e. the physics is very different for very high rotation (but the mathematical tools are still useful).

\section{Heuristics}
\subsection{Rough heuristics}
Let us examine the competition between all the (nonnegative) terms appearing in \eqref{GL2e}. We will write $u$ in trigonometric form as $u=\ro e^{i\vp},$ with again $\vp$ a ``multi-valued" function.
For a configuration with vortices, we have the formal relation 
\be\label{nabphi}
\nab \times \nab \vp = 2\pi \sum_i d_i \delta_{a_i}\ee 
where the $a_i$'s are the vortex centers, and the $d_i\in \mz$ their degrees. This relation is  true in the sense of distributions.  To check it, it suffices to test against a smooth function and use \eqref{degvor} (note also that $\nab \times \nab =0$ for true functions).

As we have seen, the term $\int \f{ (1-|u|^2)^2}{2\eps^2}$ prefers Meissner states $\ro \simeq 1$, or states with vortices of lengthscale scale $\eps$, while it disfavors the normal state $\ro=0$.
By explicit computation, the  quantity $ |\nab_A u|^2 $ is   in trigonometric form 
\be \label{nabautrig}
  |\nab_A u|^2=  |\nab \ro|^2 +\ro^2 |\nab \vp - A|^2.\ee
The term $\io |\nab_A u|^2$ thus favors $\ro$ to be constant and 
\be\label{nabp}
\nab \vp \approx  A.\ee
The term $\int |\nab \times A - \he|^2$ ``prefers" the induced field $h = \nab \times A$ to ``align" with the applied field $\he$:
\be\label{ah}
\nab \times A\approx \he
\ee
Taking the curl of \eqref{nabp} and combining with \eqref{nabphi} and \eqref{ah} leads to the formal relation
\be\label{close} 2\pi \sum_i d_i \delta_{a_i} \approx  \he,\ee
which indicates at least heuristically that when $\he$ is not small, there should be vortices (otherwise the left-hand side would vanish).
The question of understanding what 
\eqref{close} exactly means and  in what sense, and what configurations of points $a_i$ satisfy this assertion, is the core of the matter of our study. Based on what we have seen in previous chapters, we can expect that    the configuration of vortex points $a_i$ (with degrees $d_i=1$)  which best approximates the uniform distribution of density   $\he$ in \eqref{close}, is the triangular Abrikosov lattice of density $\he$. 
 We will see that this becomes true only when $\he $ is large enough, because as seen above, it  costs a fixed amount  $\pi d^2 \lep$ in the term  $\int |\nab \vp|^2$ to create one vortex. This way, the size of $A$, which is of order $\he$ by \eqref{ah}, has to become significantly larger than $\lep$ for this heuristics to be completely correct. Below this threshold, boundary effects are also important, as we shall see, and the true optimal distribution of the vortices is a constant distribution in a subdomain of $\Omega$, analogous to the equilibrium measure for Coulomb gases. We will next give more precise heuristics related to these facts, and give a complete proof in Chapter \ref{glleading}.

\subsection{The vorticity measure and the London equation}
The precise  meaning to give to relations of the form \eqref{nabphi} and \eqref{close} is given via the {\it vorticity measure} (or vorticity) of a configuration, defined by 
\be\label{vorticity}
\mu(u, A) = \curl \langle iu, \nab_A u\rangle + h\ee
which is a gauge-invariant quantity. It was first introduced in this form in \cite{ss3}, and is the gauge-invariant analogue of the {\it Jacobian determinant} of  $u$ seen as a map from $\mr^2 $ to $\mr^2$ in the outlook popularized by Jerrard and Soner \cite{js}, itself previously viewed  in \cite{bbh} as the Hopf differential of the map $u$. It is also the analogue of the vorticity of a fluid. 
Note that in trigonometric representation  we have 
$$\mu(u, A) = \curl (j+ A) = \curl (\ro^2  (\nab \vp-A)+A) \approx \curl \nab \vp$$
at least when $\ro $ is close to $1$, as is expected in the limit $\eps \to 0$.
This is why, in view of \eqref{nabphi}, we may write the heuristic relation 
\be \label{jacheur}
\mu(u, A)\approx 2\pi \sum_i d_i \delta_{a_i}.\ee
This relation holds in the asymptotics of $\eps \to 0$, and its proper meaning will be given in the next chapter. Suffice it for now to say that it is more correct to replace the sum of Dirac masses in the right-hand side of \eqref{jacheur} by Diracs smeared out at the scale $\eps$ -- characteristic lengthscale of the vortices -- that we denoted $\delta_{a_i}^{(\eps)}$, as alluded to in \eqref{london}, and as done for the Coulomb gas. 
\medskip

Turning again to the functional \eqref{GL2e}, one may observe that for a fixed $u$, the energy $G_\eps$ is a  positive quadratic function of $A$, thus always has a unique critical point in terms of $A$,  and that critical point is  a minimum. We may thus always consider that without loss of generality, $G_\eps$ has been minimized with respect to $A$, this decreases the energy and does not affect the zeroes of $u$, i.e. the vortices and their degrees, which are the objects we wish to understand.  This way, we may assume that the Euler-Lagrange equation associated with this minimization, which is the second relation in (GL), is satisfied, together with its boundary condition: 
\be\label{gl2}\left\lbrace \begin{array}{ll}
-\nab^\perp h = j =\langle iu, \nab_A u \rangle & \text{in} \ \Omega\\
h=\he & \text{on} \ \pa \Omega.\end{array}\right.
\ee
Taking the curl of this equation, one obtains  $-\Delta h = \curl j$, which we can rewrite, by definition of $\mu$,  as 
\be\label{eqlondon}
\left\lbrace \begin{array}{ll}
-\Delta h + h =  \mu(u, A) & \text{in} \ \Omega\\
h=\he & \text{on} \ \pa \Omega.\end{array}\right.\ee
This is exactly the rigorous version of the  London equation \eqref{london}, which directly relates the vorticity $\mu$ and the induced magnetic field $h$. Another way of writing it is that 
\be
h(x)=\he+ \io G_\Omega(x,y) \, (\mu(u,A)-\he)(y) \, dy
\ee
where $G_\Omega $ is a Green-type function (or more correctly a Yukawa potential) of the domain with Dirichlet boundary condition, solution to 
\be\label{eqgreen}
\left\lbrace \begin{array}{ll}
-\Delta G_\Omega (\cdot, y) + G_\Omega(\cdot, y) =  \delta_y & \text{in} \ \Omega\\
G_\Omega(\cdot, y) =0 & \text{on} \ \pa \Omega.\end{array}\right.\ee

 This  shows that $h$ can be mathematically seen as the potential generated by the vorticity distribution $\mu(u, A)$ via $G_\Omega$. This kernel 
 depends on the domain, but its leading contribution is the Coulomb kernel in dimension $2$, $-\frac{1}{2\pi}\log |x|$, hence the origin of the analogy with the 2D Coulomb gas, as explained in Chapter \ref{chap-intro}.

Note that when the  vorticity vanishes, the London equation  reduces to 
\be\label{eqlondon0}
\left\lbrace \begin{array}{ll}
-\Delta h + h =  0 & \text{in} \ \Omega\\
h=\he & \text{on} \ \pa \Omega,\end{array}\right.\ee
hence (up to dividing by $\he$)  we can expect a particular role  to be played by the function $h_0$, solution of 
\be\label{h0}
\left\lbrace \begin{array}{ll}
-\Delta h_0 + h_0 =  0 & \text{in} \ \Omega\\
h_0=1 & \text{on} \ \pa \Omega,\end{array}\right.\ee
which depends only on the domain $\Omega$ and exhibits exponential decay away from $\pa \Omega$.
The situation expected when there are no vortices is to have $h\approx \he h_0$, which physically corresponds to the Meissner effect, for which it is said that the applied magnetic field is expelled by the superconducting sample and only penetrates in it in a layer localized near the boundary (in our normalization, this layer has characteristic lengthscale $1$, but physically, it is the so-called {\it penetration depth}).

\subsection[Formal derivation of the first critical field]{Approximation to the energy and formal  derivation of the first critical field}
We may now justify \eqref{apprx}, which we recall here~:
\be\label{appr2}
G_\eps(u, A)\approx \hal \io |\nab h|^2 + |h-\he|^2
\ee
with $h=\nab \times A$ solving \eqref{eqlondon}.
 Taking the norm of \eqref{gl2} we may compute that  in trigonometric form
$|\nab h|^2 = |j|^2= \ro^4 |\nab \vp - A|^2 $. Comparing with
   \eqref{nabautrig} we thus have 
$$\io |\nab_A u|^2 =\io |\nab \ro|^2 + \frac{|\nab h|^2}{\ro^2} . $$
But for any solution of $\text{(GL)}$, it holds that $|u|\le 1$ (this can be checked using the maximum principle on the equation satisfied by $|u|$). We may thus bound from below 
$$\io |\nab_A u|^2 \ge \io |\nab h|^2 $$ and we expect almost equality in view of the heuristic relation $\rho \approx 1$. The difference will turn out to be indeed negligible  as $\eps \to 0$ as a by-product of our analysis, based on comparing ansatz-free lower bounds and upper bounds obtained by explicit constructions.

Once \eqref{appr2} is established, deriving the first critical field can be done formally: at the point where the first vortices appear, we can expect that the induced magnetic field is well approximated to leading order  by the magnetic field generated in the situation with no vortices, i.e.  $\he h_0$. One may then expand around that function by setting $h=\he h_0 +h_{1,\eps}$ where $h_{1,\eps}$ is seen as a correction term, insert this into \eqref{appr2} and expand in terms of this splitting.  
This yields
\begin{eqnarray*}
G_\eps (u_\eps,A_\eps) &  \approx \D\f{\he^2}{2} \D\io |\nab h_0|^2 + |h_0-1|^2 + \hal\D \io |\nab h_{1,\eps}|^2 +|h_{1,\eps}|^2\\ &
 + \he\D \io  (-\Delta h_{1,\eps}+h_{1,\eps}) (h_0-1)
\\ & =  \D\f{\he^2}{2}\D \io |\nab h_0|^2 + |h_0-1|^2 + \hal\D \io |\nab h_{1,\eps}|^2 +|h_{1,\eps}|^2 \\
& +\he \D \io (h_0-1) \mu(u_\eps, A_\eps),
\end{eqnarray*}
where for the cross-term we have used an integration by parts, and \eqref{eqlondon} with \eqref{h0}.
With the approximate relation \eqref{jacheur} and estimating $\hal \io|\nab h_{1,\eps}|^2 +h_{1,\eps}^2 $ as the cost to create a vortex, i.e. $\pi \sum_i d_i^2 \lep$ by the heuristic of Section \ref{722}, we are led to 
$$G_\eps (u_\eps,A_\eps)  \approx \f{\he^2}{2} \io |\nab h_0|^2 + |h_0-1|^2+\pi \sum_i d_i^2 \lep + 2\pi \he\sum_i d_i (h_0-1)(a_i).$$
The energy of a configuration with vortices thus becomes smaller than that of the vortex-free Meissner solution if we can have 
$$\pi \sum_i d_i^2 \lep +2\pi \he\sum_i d_i (h_0-1)(a_i)\le 0.$$
Noting that $h_0-1\le 0 $ in $\Omega$ by the maximum principle applied to the equation \eqref{h0}, a quick examination shows that this can be first achieved when 
$$\he \ge \frac{\lep}{2\max_\Omega |h_0-1|}$$
and with vortices that have degrees $d_i=+1$, located at the point(s) where $h_0-1$ achieves its minimum (or equivalently $|h_0-1|$ achieves its maximum).
This gives the leading order value of the first critical field 
\be\label{vhci} \boxed{\hci\sim \frac{1}{2\max_\Omega |h_0-1|}\lep \quad \text{as}\  \eps \to 0.}\ee
This heuristic is in fact correct, it first appeared (in a slightly different but equivalent form) in \cite{br2}, it was then justified rigorously in \cite{s1}. A proof of the most precise result can be found in \cite[Chap. 12]{livre}.

This expansion of $\hci$ confirmed and made more precise the expansion known to physicists, e.g.  in \cite{dg}, by giving the exact prefactor in terms of the domain $\Omega$, and locating the points of nucleation of the first vortices. 
As soon as there is more than one vortex accumulating near one of the optimal point(s), the Coulomb repulsion between vortices starts to kick in, and slightly delays the onset of more vortices. Again for details we refer to \cite{livre}. 

 We will see in Chapter \ref{glleading} how to derive more information  about the number and optimal distribution of vortices above $\hci$.

\chapter{Main mathematical tools for Ginzburg-Landau}\label{toolsgl}
Mathematicians started to get interested in the  Ginzburg-Landau  model mostly in the 90's, with Berger-Chen, Chapman, Rubinstein, Schatzman,  Du, Gunzburger, Baumann, Phillips, cf. e.g. to the review papers  \cite{ch,dugu}. Then, Bethuel-Brezis-H\'elein \cite{bbh} were the first to introduce  tools to  systematically study vortices, their exact profile and their asymptotic energy (with important input from Herv\'e-Herv\'e \cite{hh} and Mironescu \cite{miro}). They did it in the simplified context of the  two-dimensional  Ginzburg-Landau equation not containing the magnetic gauge (set $A\equiv 0$, $\he \equiv 0$ in (GL)), and  under an a priori bound $C\lep$ on the energy, which allows only for a fixed number of vortices as $\eps \to 0$. The analysis of that model was completed by many works, including the precise study of solutions by Comte-Mironescu \cite{cm1,cm2}, and the monograph of Pacard-Rivi\`ere \cite{pr}.  
 
 The analysis of the simplified model was   adapted to the model with gauge but with a different boundary condition  and still the same a priori bound, by Bethuel and Rivi\`ere \cite{br,br2}, still in dimension 2. The three-dimensional (more physical) versions were first studied in Rivi\`ere \cite{ri1}, and later in \cite{lr,js,bbo,bjos}, among others. For a (slightly outdated) review, we refer to  \cite[Chap. 14]{livre}.

An important  challenge after these works was to be able to treat the case where the a priori bound  is released and the number of vortices blows up as $\eps \to 0$, as really happens in  the full model with magnetic field.  There are two main related technical tools that have been widely used and applied in such a situation. The first is the ``vortex balls construction" method introduced independently by Jerrard \cite{jerrard} and Sandier  \cite{sandier}, which allows to get completely general lower bounds for the energy of a configuration in terms of its vortices (regardless of their number and degrees).
 The second is the so-called ``Jacobian estimate" which gives a quantitative estimate  and meaning for \eqref{jacheur}, i.e. relates 
 the vorticity of an arbitrary configuration, as defined in \eqref{vorticity} (or the Jacobian in the gauge-independent version) to its underlying vortices.

\section{The ball construction method}
As mentioned, the ball construction method was first introduced in two slightly different variants in \cite{sandier} and \cite{jerrard}, and it was reworked and  improved over the years e.g. \cite{ss2,livre,compagnon,jspirn}, and extended to higher dimensions \cite{js,sandier3d}. It would be too long here to prove the best-to-date results, but we will give an idea of the method and a statement of  results. The main question in the end, for what we need here,  was to obtain estimates on the energy  
that allow for only an error of a constant per vortex.

\subsection{A sketch of the method}

 The method consists in starting by understanding  lower bounds for unit-valued complex functions in the plane. 
 \subsubsection{A lower bound on an annulus}
 Assume that $u$ is a (complex-valued) function mapping an annulus (say centered at the origin) to the unit circle, in other words $|u|=1$ in $B_R\backslash B_r$.
 If $u$ is sufficiently regular (say, continuous, for more refined assumptions, see \cite{bn1,bn2} and references therein), we can define its degree as the integer
 $$d=deg(u, \pa B_t):= \frac{1}{2\pi} \int_{\pa B_t} \langle iu, \frac{\pa u}{\pa \tau}\rangle = \frac{1}{2\pi} \int_{\pa B_t} \frac{\pa \vp}{\pa \tau},$$
 which is constant over $t\in [r, R]$ and where we have written $u=e^{i\vp}$ for $\vp$ some real-valued lifting of $u$ (for questions of  existence and regularity of a lifting see \cite{brezismiro} and references therein).
We may then write 
\begin{eqnarray}
\int_{B_R \backslash B_r} |\nab u|^2 & = & \int_{B_R\backslash B_r} |\nab \vp|^2 
\ge \int_r^R \int_{\pa B_t} \left|\frac{\pa \vp}{\pa\tau} \right|^2 \, dt\\
& \ge & \int_r^R \left( \int_{\pa B_t}  \frac{\pa \vp}{\pa \tau}\right)^2 \frac{dt}{2\pi t}
\end{eqnarray}
where the second relation follows by an application of Cauchy-Schwarz's inequality (note the similarity with the proof of Lemma \ref{lem:fluctu charge}).
We then recognize the degree $d$ and may write 
\be\label{bb}
\int_{B_R \backslash B_r} |\nab u|^2\ge \int_r^R \frac{4\pi^2 d^2}{2\pi t}\, dt= 2\pi d^2  \log \frac{R}{r},\ee
and there is equality if and only if $\frac{\pa \vp}{\pa \tau}$ is constant on each circle $\pa B_t$, which amounts  in the end to $u$ being of the form $e^{i(\theta+\theta_0)}$ in polar coordinates centered at the center of the annulus.
The lower bound \eqref{bb} is general and is the building block for the theory.
It does show how a vortex of degree $d$ induces a logarithmic cost, as in \eqref{costvor}.

In a Ginzburg-Landau configuration  with vortices, we will not have $|u|=1$ everywhere, but we can expect that $|u|\approx 1$ except in small regions of scale $\eps$ around the vortex cores. We can expect to be able to localize the ``bad regions" where $|u|$ is far from $1$, which must contain all the vortices,  in balls of size  $C\eps$. We may then center around each such ball an annulus of inner radius $C\eps$ and outer radius $R$ (the largest possible so that it does not intersect any other vortex), and then the estimate \eqref{bb} yields on such an annulus a lower bound
by $\pi d^2 \log \frac{R}{C\eps}\sim \pi d^2 \log \frac{1}{\eps}$  at leading order as $\eps \to 0$.
 If we can build such annuli that are disjoint, then we may add the lower bounds obtained this way  and obtain a global lower bound of the form 
 $$\int |\nab u|^2 \ge \pi \sum_i d_i^2 ( \lep  +O(1))\quad \text{as } \eps \to 0$$
 where $d_i$ are the degrees of the vortices.
 Two questions remain: first  to find an algorithm to build such {\it disjoint} annuli in some optimal way, and second to handle the fact that we do not really have $|u|=1$ outside of small balls, but rather $|u|\approx 1$, with a control via the energy term $\io\frac{(1-|u|^2)^2}{2\eps^2}.$
 These questions are answered in a slightly different way 
  by both the methods of \cite{jerrard} and \cite{sandier}; we now give the main elements.
  
 \subsubsection{Construction of initial balls}
 To initiate the ball construction, one does need some weak upper bound on the energy, of the form 
 $$G_\eps(u,A) \le C\eps^{\alpha -1 }, \quad \alpha \in (0,1)$$
 which implies 
 \be \label{contrro}
 \io |\nab \ro|^2 + \frac{(1-\ro^2)^2}{2\eps^2} \le C \eps^{\alpha -1}\ee
 with $\ro = |u|$.
 This control implies, for $0<\delta <1$,  a control of the total perimeter of the bad set $\{\ro \le 1- \delta\}$ via the co-area formula (cf. \cite{eg}), this is the argument used by Sandier: by Cauchy-Schwarz, we have 
 $$\io |\nab \ro | (1-\ro^2) \le C \eps^{\alpha} $$
 and the left-hand side is equal to 
 $$ \int_0^{+\infty} (1-t^2) \mathcal{H}^1(\{ |\ro(x)| =t\}) \, dt,$$ where $\mathcal{H}^1$ denotes the one-dimensional Hausdorff measure. 
 The upper bound thus allows to find many level sets $\{\ro \le 1-\delta\}$, with $\delta$ as  small as  a power of $\eps$, whose perimeter is small. A compact set of small perimeter can then be covered by disjoint closed  balls $B_i$ of radii $r_i$, with $\sum_i r_i $ controlled by that perimeter.
 
 In Jerrard's construction, the initial balls are obtained differently. The use of the co-area formula is replaced by the following lemma, based on elementary arguments:
 \begin{lemme}[Lower bound on circles \cite{jerrard}, Lemma 2.3]\label{lemmjcercle}
 Letting $\ro $ be a real-valued function defined over $B(x,r)\subset \mr^2 $ with $ r\ge \hal \eps$, if $m= \min_{\pa B(x,r)} \ro(x)$, we have 
 $$\int_{\pa B(x, r)}|\nab \ro|^2+ \frac{(1-\ro^2)}{2\eps^2} \ge c_0 \frac{(1-m)^2 }{\eps}$$
 for some universal constant $c_0>0$. \end{lemme}
 This is another way of quantifying the cost of  $|u|$ being away from $1$.
 
Then Jerrard only  covers, again by disjoint closed balls of radii $r_i$,  the connected components of the set $\{\ro \le \hal\}$ on the boundary of which $u$ has nonzero degree, and is able to do it in such a way that the radius $r_i$ of each ball $B_i$  is bounded above by  $\eps$ times the energy  that the ball contains.

\subsubsection{Ball construction method}
Consider a collection of (disjoint closed) initial balls $\mathcal{B}_0=\{B_i\}$ of radii $r_i$, and let us assume to fix ideas that $|u|=\ro =1$ outside of these balls. 
If we have disjoint annuli centered around these same balls, of inner radii $r_i$ and outer radii $R_i$, then we may add the lower bounds given by \eqref{bb} to obtain 
$$\hal \io |\nab u|^2 \ge \pi \sum_i d_i^2 \log \frac{R_i}{r_i},$$
where $d_i$ is the degree of $u$ on each annulus.
We then see that these lower bounds combine nicely {\it if} all the ratios  $\frac{R_i}{r_i}$ are equal, because then 
$$ \log \frac{R_i}{r_i}= \log \frac{\sum_i R_i}{\sum_i r_i}=\log s$$
where $s$ is the common ratio $R_i/r_i$, in other words the common conformal factor of the annuli. Let us underline that this is the point where the construction is purely {\it two-dimensional}: in higher dimensions the energy $\int |\nab u|^2$ is not conformally invariant and the estimates on annuli would not involve logarithms and not combine well.

In order to apply this reasoning, the annuli need to all be disjoint. The idea of the ball construction method is to build such disjoint annuli by continuously growing jointly all the initial balls, keeping their centers fixed, and multiplying their radii by the same factor $s\ge 1$, until $s$ is large enough (typically of order $1/\eps$). This way the previous reasoning applies, and at least for $s$ close enough to $1$, the balls (hence the annuli) remain disjoint. 

At some point during the growth process, two  (or more) balls  can become tangent to each other. The method then is to {\it merge} them into a ball that contains them both and has a radius equal to the sum of the radii of the merged ball.  In other words, if $B_1= B(a_1, r_1)$ and $B_2= B(a_2, r_2)$ are tangent, we merge them into $B'=B(\frac{a_1 r_1+a_2 r_2}{r_1+r_2}, r_1+ r_2)$.
 (The resulting ball can then intersect other balls, in which case one proceeds to another merging, etc, until all the balls are disjoint).
The merging process preserves the total sum of the radii, and as for the degrees we have 
\be\label{addingdeg} deg(u, \pa B')= deg (u, \pa B_1)+ deg(u, \pa B_2).\ee
 Thus the only problem is that the  $ d_i^2$ do not add up nicely during merging. 
The price to pay is to give up on obtaining a lower bound with a  $\sum_i d_i^2$ factor, but rather to keep a lower bound by the smaller factor $\sum_i |d_i|$ (here we use that the $d_i$'s are all integers). Such factors do add up nicely through merging since we have \eqref{addingdeg} during a merging, thus $|d|=|d_1+d_2|\le |d_1|+|d_2|$. 

The fact that we need to abandon the hope of lower bounds by $\sum_i d_i^2$ is completely natural, due to possible cancellations of 
singularities (or vortices) of $u$. If $u$ has a vortex of degree $+1$, and a nearby vortex of degree $-1$ at distance $r$, once the associated balls have been merged, the total degree is $0$, and one does not expect any substantial energy to lie in the annuli surrounding the merged balls. 

After mergings, the old collection is replaced by the new collection (with merged balls) which is still made of disjoint balls,  and the growth process is resumed, until some next intersection and merging happens, etc. The construction can then be stopped at any value of the parameter $s$, depending on the desired final total radius of the balls, and the desired final lower bound. 

Using this method and combining it with \eqref{bb},  one arrives at the following result, where for any ball $B$, $r(B)$ will denote its radius. Also if $\mathcal{B}$ is a collection of balls, $\lambda \mathcal{B}$ is the collection of balls with same centers, and radii multiplied by $\lambda$.
\begin{prop}[Ball construction]\label{propunit}
Let $\mathcal{B}_0$ be a collection of disjoint closed balls in the plane. Assume $u:\Omega \backslash \cup_{B\in \mathcal{B}_0} \to \mathbb{S}^1$. 
For any $s\ge 1$ there exists a family of balls $\mathcal B(s)$ such that the following holds.
\begin{itemize}
\item 
$\mathcal{B}(1)= \mathcal B_0.$
\item
For any $s_1\le s_2$ we have 
$$\cup_{B\in \mathcal B(s_1)} B \subset \cup_{B\in \mathcal B(s_2)} B.$$
\item
There exists a finite set $T$ (the set of merging ``times") such that 
if $[s_1, s_2]\subset [1, +\infty) \backslash T$, we have 
$\mathcal{B}(s_2) = \frac{s_2}{s_1} \mathcal{B}(s_1)$.
\item 
$$\sum_{B\in \mathcal{B}(s)}  r(B)= s \sum_{B \in \mathcal{B}_0} r(B).$$
\item
For any $B \in \mathcal{B}(s)$ such that $B \subset \Omega$, denoting $d_B= deg(u, \pa B)$ we have 
$$\hal \int_B |\nab u|^2 \ge \pi |d_B| \log s= \pi |d_B| \log \frac{\sum_{B\in \mathcal{B}(s)}  r(B)}{  \sum_{B \in \mathcal{B}_0} r(B)}.$$\end{itemize}
\end{prop}
This is the building block estimate. As mentioned, one needs to control the  initial total radius by some small factor; typically, one can expect it to be  $n\eps$ where $n$ is  the number of initial balls. Then one may choose the parameter $s$  according to the needs,  so that the final sum of the radii  be not too large, but still large enough for  the factor in the right-hand side to be    at  leading order $\pi |d|\log \frac{1}{n\eps}$. For example, a good choice may be $\sum_{B\in \mathcal{B}(s)} r(B)= 1/\lep$, which is  a $o(1)$ quantity (guaranteeing small balls) but such that $\log (\sum_B  r(B))^{-1}= \log \lep $ is negligible compared to $\lep$.

\subsubsection{Dealing with non unit-valued functions}
The  main technical difficulty that remains is to handle the fact that $|u|$ is not really equal to $1$ outside of small ``initial balls." One then needs to use the fact that outside of the initial balls $u$ does not vanish and one may write 
\be \label{ecrit}
\int|\nab u|^2 = \int \ro^2\left|\nab \left(\frac{u}{\ro}\right)\right|^2 + |\nab \ro|^2 + \frac{(1-\ro^2)^2}{2\eps^2},\ee
use a bound from below for $\ro $, and then  bounds from below for unit vector fields to bound from below $\int |\nab ( \frac{u}{\ro} )|^2$.

In Sandier's construction, this is handled by combining the result of Proposition \ref{propunit} with  a co-area argument as outlined above, but in a rather sophisticated manner since the argument has to be applied to all sub-level sets at once. For details, we refer to \cite[Chap. 4]{livre}.

In Jerrard's construction, this is handled by  combining \eqref{ecrit} with the result of Lemma \ref{lemmjcercle} and \eqref{bb}
to obtain 
$$\int_{\pa B_r}  |\nab u|^2 + \frac{(1-|u|^2)^2}{2\eps^2}\ge m^2 \frac{2\pi d^2}{r} + c_0 \frac{(1-m)^2}{\eps},$$
with $m= \min_{\pa B_r} |u|$. 
 Optimizing over $m$ yields
$$\hal \int_{\pa B_r}  |\nab u|^2 + \frac{(1-|u|^2)^2}{2\eps^2}\ \ge \lambda_\eps (\frac{r}{|d|}) $$
with $\lambda_\eps  (s)  $ that behaves like $\min (\frac{c}{\eps}, \frac{\pi}{s})$, and whose antiderivative $\Lambda_\eps$ satisfies $\Lambda_\eps(s) \ge \pi \log \frac{s}{\eps} -C.$
 The balls are grown and merged (in the same way as explained before) in such a way that it's not the factor of sum of radii which is constant, but rather the parameter $s=r(B)/|d_B|$, to preserve $r\ge s|d|$.
One checks  that the estimate $$\hal \int_B  |\nab u|^2 + \frac{(1-|u|^2)^2}{2\eps^2}
\ge r(B) \frac{\Lambda_\eps(s)}{s}     \ge\pi |d| \( \log \frac{s}{\eps} -C\) $$
is true initially and is preserved through the growth and merging process, yielding the desired estimate at the end of the growth process. 
 For more details, we refer to \cite{jerrard}.

In all cases, the presence of the gauge $A$ does not change substantially the situation, the method consists in controlling the error that it creates via the term $\int |\curl A-\he|^2$.

\subsection{A final statement}\label{sec812}
To give a more precise idea, let us now finish with  the statements of a result on the complete Ginzburg-Landau functional.  If one is interested in the Ginzburg-Landau functional without magnetic gauge, then it suffices to set $A \equiv 0$ in the following result. 

It is borrowed from \cite[Theorem 4.1]{livre}.
 A similar result (slightly stronger in some sense, slightly weaker in some other) and following Jerrard's construction \cite{jerrard}, can be found in \cite[Proposition 5.2]{compagnon}.   We let $F_\eps $ denote the Ginzburg-Landau energy with $\he$ set to $0$.
\begin{theo}[Ball construction lower bound]\label{thboules}
For any $\alpha\in (0,1)$, there exists $\eps_0(\alpha)$ such that for any $\eps<\eps_0$, if $(u,A)$ is such that $\io |\nab |u||^2 + \frac{(1-|u|^2)^2}{2\eps^2}\le \eps^{\alpha-1}$, the following holds.

For any $1>r>\eps^{\alpha/2}$, there exists a finite collection of disjoint closed balls $\mathcal{B}=\{B_i\}$ such that 
\begin{itemize}
\item 
$\sum_{B\in \mathcal{B}} r(B)=r$
\item
$$\{x\in \Omega| \dist (x, \bo) >\eps, ||u(x)|-1|\ge \eps^{\alpha/4}\} \subset \cup_i B_i.$$
\item Writing $d_i= \deg (u, \pa B_i)$ if $B_i\subset \{x\in \Omega| \dist(x, \bo ) >\eps\}$,  $d_i=0$ otherwise, and $D=\sum_i |d_i|$, we have 
\be\label{mainboules}\hal \int_{\cup_i B_i} |\nab_A u|^2 + |\curl A|^2 + \frac{(1-|u|^2)^2}{2\eps^2}\ge \pi D \(\log \frac{r}{D\eps}-C\)\ee where $C$ is a universal constant.
\item If in addition $F_\eps(u,A) \le \eps^{\alpha-1}$ then 
$$D\le C \frac{F_\eps(u,A)}{\alpha\lep}.$$
\end{itemize}
\end{theo}
In practice the last item already gives a rough lower bound on the energy $F_\eps$ (without optimal constants) which can serve to provide a first control on $D$, which  can then be inserted into the main result \eqref{mainboules}. 
Compared to the heuristic lower bound of $\pi \sum_i d_i^2 \lep$, this lower bound 
\begin{itemize}
\item loses $d_i^2$ and replaces it by $|d_i|$: as explained this is normal due to possible cancellations between vortices happening at small scales. 
\item introduces an error $-D\log D$: this is also normal due to the possibility of many vortices accumulating near a point, or near the boundary (think of the case of $n$ vortices of degree $1$ regularly placed at small distance from the the boundary of the domain).
\item introduces an error $-CD$ where $C$ is an unknown constant. When one knows that the number of vortices is bounded, the analysis derived from \cite{bbh} allows to identify the constant order term  in the energy of a vortex. It is (at least in the case of degree $\pm 1$ vortices), a constant that they denote $\gamma$, and which depends on the explicit optimal profile of the modulus of $u$ for a radial vortex (identified in \cite{hh,miro}).
One thus usually proceeds in two steps: first control the number of vortices via ball construction lower bounds which give the correct leading order energy, then if one can show that the number of vortices is locally bounded,  recover this constant  order term $\gamma$.
\item In order to accomplish this program, one may need (we needed) to localize the above lower bound over finite size balls in a (possible large) domain, and eliminate the $-D \log D$ error.  We have seen that the energy carried by vortices is not only located at the vortex centers, it is spread over relatively large annuli surrounding them. 
In case of vortices accumulating  near a point, the ball construction (because it stops at finite total radius) is missing some energy $\pi D^2 \log \frac{R}{r}$ (as given by \eqref{bb}) which is lying in even larger annular regions. Although $\log \frac{R}{r}$ is then of order $1$, such an energy suffices to compensate $-D \log D$ thanks to the power $2$ in $D^2$ which beats $-D \log D$ when $D$ gets large.
The method to do this and combine it with the ball construction, itself properly localized, is quite technical   in its details, and is the object of \cite{compagnon}, to which we refer the interested reader.
\end{itemize}

 \section{The ``Jacobian estimate"}
 Let us  now turn to the ``Jacobian estimate," which allows to give a rigorous meaning to 
  \eqref{jacheur}, in terms of the result of a ball construction. Estimates of the same nature already  appeared in \cite{br,s1,ss3}, the estimate was optimized and its name popularized 
 by the work of Jerrard and Soner \cite{js}.  
 Let us state it in the version presented in \cite[Chap. 6]{livre}

 The case without gauge $A$ is again contained in what follows by taking $A\equiv 0$.
 \begin{theo}[Jacobian estimate]\label{thjac}
 Let $u : \Omega \to \mathbb{C}$ and $A : \Omega \to \mr^2$ be $C^1$. Let $\mathcal{B}$ be a collection of disjoint closed  balls with centers $a_i$ and radii $r_i$ such that 
 $$\left\{x\in \Omega, \dist(x, \bo) >\eps,  ||u|-1|\ge \hal \right\} \subset \cup_{B\in \mathcal{B}} B
 .$$ Then, letting $d_i=\deg(u, \pa B(a_i,r_i))$ if $B(a_i,r_i) \subset\{x\in \Omega, \dist(x, \bo) >\eps\}$ and $0$ otherwise, defining $\mu(u,A)$ by \eqref{vorticity}, if $\eps $ and $r$ are less than $1$,   we have for $C>0$  some universal constant
 \begin{equation}\label{jacest}
 \left\| \mu(u,A)  - 2\pi \sum_i d_i \delta_{a_i} \right\|_{(C^{0,1}_0(\Omega))^*} \le C \max(\eps, \sum_i r_i) (1+ F_\eps (u,A)).\ee
 Moreover $\|\mu(u,A)\|_{(C^0(\Omega))^*} \le C F_\eps(u,A).$
 \end{theo}
 The spaces $(C^0)^*$ and $(C^{0,1})^*$ here are the space of bounded Radon measures and 
the dual of Lipschitz functions, respectively.
For estimates in the dual of H\"older spaces, see the statement in \cite{livre}. 
Note that this result is naturally meant to work with a collection of disjoint balls obtained by a ball-construction. The total radius chosen to end the construction has to be taken small enough if one wants the estimate to be precise  ---  this is of course in competition with the lower bound estimate which improves as the total radius gets larger. So we see why  $\sum_i r_i$ has to be optimized according to the needs. 
Note that it's the  centers of the final balls in the construction (that may depend on the final total radius chosen $r$)  which play the role of approximate vortex centers.
More precise estimates can be obtained, but with points $a_i$ that do not correspond to  centers of balls obtained in a  ball construction, this was done in \cite{jspirn}.

The proof is easy enough that we can give its main argument.
   \begin{proof} 
   We set $\mu= \mu(u,A)$.
   First, let us consider the function $\chi$ on $[0,+\infty]$ defined by $\chi(t)= 2t$ if $t \le 1/2$, $\chi(t)= 1$ if $t\in [\hal, \frac{3}{2}]$, $\chi(t)=t$ if $t\ge 2$, and $\chi $ is continuous and piecewise affine. It satisfies 
   \be \label{chmt}
   |\chi(t)^2 -t^2|\le 3 t|1-t|.\ee

   We may then set 
   $\tilde{u}= \chi(|u|)\frac{u}{|u|} $. 
   By assumption on the balls, we have $|\tilde{u}|=1$  outside of $\cup_{B\in\mathcal{B}} B$.
   We then define $\tilde{\mu}= \curl \langle i\tilde{u}, \nab_A \tilde{u}\rangle+\curl A$ and check two facts~:
   \be\label{peerja}
   \|\mu- \tilde{\mu}\|_{C^{0,1}(\Omega)^*}\le C \eps F_\eps(u,A),\ee
    \be\label{muvan}
    \tilde{\mu}=0 \quad \text{outside} \ \cup_{B\in\mathcal{B}} B.\ee
    For the first fact, it suffices to use an integration by parts: let $\zeta$ be a smooth test-function vanishing on $\bo$. By definition  of $\tilde{u}$ we have 
    \begin{multline}\label{rhhs}
    \left|\io \zeta(\mu-\tilde{\mu})\right| = \left|\io \nab^\perp \zeta \cdot (\langle iu, \nab_{A} u\rangle -\langle i\tilde{u}, \nab_A \tilde{u}\rangle)\right|\\ \le \|\nab \zeta\|_{L^\infty(\Omega)} \io \frac{||u|^2 - |\tilde{u}|^2 |}{|u|} |\nab_A u|\le  3\|\nab \zeta\|_{L^\infty(\Omega)} \io |1-|u|| |\nab_A u|
    \end{multline}
    where we used the formal relation $\langle iu, \nab_A u\rangle = \ro^2 (\nab \vp - A)$ and $|\nab_A u|\ge \ro |\nab \vp - A|$ together with \eqref{chmt}.
    It then suffices to apply Cauchy-Schwarz to control the right-hand side of \eqref{rhhs}
    by $\eps F_\eps (u,A)$ and conclude \eqref{peerja}.
    \eqref{muvan} is a consequence of the simple observation that wherever $|u|=1$, we have $\curl \langle iu, \nab_A u \rangle +\curl A= \curl \nab \vp=0$. We thus know  that $\tilde{\mu}$ is supported in the (disjoint) balls only, and thus we may write, for any smooth test-function $\zeta$ vanishing on $\bo$, 
  \begin{multline*}
   \io \zeta \tilde{\mu}= \sum_i \int_{B(a_i,r_i)}
\zeta \tilde{\mu}\\
= \sum_i \zeta(a_i) \int_{B(a_i, r_i) } \curl \( \langle i\tilde{u},\nab_A  \tilde{u}\rangle + A \)  + \sum_i \int_{B(a_i,r_i)}  (\zeta-\zeta(a_i) ) \tilde{\mu}.\end{multline*}   
   The first term of the right-hand side of this relation can be handled by Stokes'  theorem, and recalling that $|\tilde{u}|=1$ on the boundary of each ball and the definition of the degree, we find 
   \be\label{premt}
   \sum_i \zeta(a_i) \int_{B(a_i, r_i) } \curl \( \langle i\tilde{u},\nab_A  \tilde{u}\rangle + A \)  =   
   2\pi \sum_i d_i \zeta(a_i)\ee (for the balls that are $\eps$-close to $\bo$ we need to replace $a_i$ by the nearest point on the boundary).
   The second term can be bounded above thanks to the Lipschitz continuity of $\zeta$ by 
   $$\|\zeta\|_{C^{0,1}(\Omega)} \sum_i r_i \int_{B(a_i,r_i)} |\tilde{\mu}|.$$
   Noting that 
   $\tilde{\mu}= 2(\pa_x \tilde{u} - i A_x \tilde{u} ) \times  (\pa_x \tilde{u} - iA_y \tilde{u}) + \curl A$ (this is the same as using  the formal relation $\curl \langle i\tilde{u}, \nab \tilde{u}\rangle= \curl (\tilde{\ro}^2 (\nab\vp -A) ) = \nab^\perp \tilde{\ro}^2 \cdot \nab \vp$), we can bound $|\tilde{\mu}| $ by 
   $|\nab_A \tilde{u}|^2 + |h|$, and in view of the definition of $\chi$ we are led to the control of the second term by $\|\zeta\|_{C^{0,1}(\Omega)} (\sum_i r_i +\eps) (1+F_\eps(u,A))$ (a little discussion is again needed for the balls that are very close to the boundary). Combining this with 
   \eqref{peerja} and \eqref{premt}, we obtain the result.
   
   \end{proof}


   \chapter{The leading order behavior for Ginzburg-Landau} \label{glleading}
   In this chapter, thanks to the tools presented in the previous chapter, we carry out the same program as in Chapter \ref{leadingorder} i.e.  the program of obtaining a mean-field limit or leading order behavior of minimizers (or ground states) of the Ginzburg-Landau functional (without temperature).
   The content of this chapter is essentially that of \cite{ss3} or \cite[Chap. 7]{livre}, but we will try here to highlight the analogy with the Coulomb gas.
   
\section{The $\Gamma$-convergence result}

   In what follows, the space $H^{-1}(\Omega)$ denotes the dual of the Sobolev space $H^1_0(\om)$, and  $\mathcal{M}(\om)$  denotes the space of bounded Radon measures over $\om$, i.e. $C^0(\om)^*$. For a measure $\mu$ in $\mathcal{M}(\om)$, $|\mu|(\om)$ denotes its total variation. 
   
   We admit the fact (see e.g. \cite[Sec 7.3.1]{livre} or \cite[Lemma 3.2]{gms}) that if $\mu\in H^{-1}(\om)$ then  
    $U^\mu(x)= \int G_\Omega(x,y)\, d\mu(y)$, with $G_\Omega$ given by \eqref{eqgreen}  makes sense and we have
\be\label{nabug}
\io |\nab U_\mu|^2 + |U_\mu|^2 = \iint_{\om \times \om } G_\om (x,y) \, d\mu(x)\, d\mu(y).\ee

   \begin{theo}[$\Gamma$-convergence of the Ginzburg-Landau functional \cite{ss3}, \cite{livre}, Chap. 7]\label{th9.1}
   Assume \be\label{asshe}
   \lim_{\eps \to 0 } \frac{\he}{\lep}= \lambda>0.\ee  Then, as $\eps \to 0$, the functional $\frac{G_\eps}{\he^2}$ $\Gamma$-converges as $\eps \to 0$, for  the sense of the convergence  of $\frac{\mu(u_\eps, A_\eps)}{\he} $ to $\mu$ in $\mathcal{M}(\Omega)$, to the functional 
  \be\label{elmu}
  E_\lambda(\mu)= \frac{1}{2\lambda} |\mu|(\om)+ \hal \io |\nab h_{\mu}|^2 + |h_\mu-1|^2\ee
  defined over $\mathcal{M}(\om) \cap H^{-1}(\om)$,
  where $h_\mu$ is the potential generated by $\mu $ as follows~:
  \be\label{hmu}
  \left\{ \begin{array}{ll}
  -\Delta h_\mu+ h_\mu= \mu & \text{in} \ \Omega\\
  h_\mu=1& \text{on} \ \bo.\end{array}\right.
  \ee
     \end{theo}
\begin{rem}  
\begin{enumerate}
\item   Note that here we use a sense of convergence of $(u,A)$ that is the convergence of a nonlinear function of $(u,A)$, as  in Remark \ref{rem4} in Chapter~\ref{leadingorder}. This is otherwise the counterpart of Proposition \ref{gammaconvergenceHn} for Ginzburg-Landau.
\item We could obtain convergence in a stronger sense for $\mu(u_\eps, A_\eps)/\he$, we refer to \cite[Chap. 7]{livre}.
\item We can in fact obtain the same result when $\lambda=\infty$, provided $\he \ll \frac{1}{\eps^2}$ as $\eps \to 0$. This is done in \cite[Chap. 8]{livre}.\end{enumerate}
\end{rem}
   \section{The proof of $\Gamma$-convergence}
   As in every proof of $\Gamma$-convergence, we need to prove a lower bound and an upper bound through the construction of a recovery sequence.
  \subsection{Lower bound} We in fact prove the stronger  result of  $\Gamma$-liminf   + compactness (cf. Remark \ref{gconvcom} in Chapter \ref{leadingorder})  i.e. that if $\frac{1}{\he^2} G_\eps(u_\eps, A_\eps)$ is bounded, then $\mu(u_\eps, A_\eps)/\he$ has a convergent subsequence and the $\Gamma$-liminf relation holds. 
   The lower bound relies on the estimates given by the ball construction method, and some lower semi-continuity arguments. Compared to the situation of the Coulomb gas, we do not have to worry about removing the diagonal terms, these are naturally smoothed out (at the scale $\eps$) in the Ginzburg-Landau  functional, but in turn we have  to estimate these terms, corresponding to the self-interaction -- or cost -- of each vortex, and we do so via the ball construction method.  Also we do not have an energy in the form of a sum of pairwise interactions but rather in integral form, as an integral of the potential generated by the charges, equivalent to \eqref{formalcomputation}. Note also that the fact that the vortex degrees do not have fixed sign would create difficulties in using the same method as in the proof of Proposition~\ref{gammaconvergenceHn}.
   
Let us start from an arbitrary family of configurations $(u_\eps, A_\eps)$, assuming $\frac{G_\eps(u_\eps, A_\eps)}{\he^2} \le C $ for some $C$ independent of $\eps$. Since we assume $\he \sim \lambda \lep$ with $\lambda>0$, this also implies that $\he \le C \lep$ and thus $G_\eps (u_\eps, A_\eps) \le C \lep^2$.  We may then apply Theorem \ref{thboules} with $\alpha=\hal $ and  final radius $r= \lep^{-10}.$
 It yields a collection of balls $\{B_i\}$ covering $\Omega_\eps:=\{x\in \Omega, \dist(x, \pa \om) >\eps\}$, outside of which we have $||u_\eps|-1|\le \eps^{1/4}$, with\footnote{recall that $F_\eps$ is the Ginzburg-Landau functional $G_\eps$ with $\he$ set to $0$}  $C F_\eps (u,A)\ge  D\lep$, and 
 \be \label{resb}
 \int_{\cup_i B_i} |\nab_{A_\eps} u_\eps|^2 + |h_\eps|^2 + \frac{(1-|u_\eps|^2)^2}{2\eps^2} 
 \ge \pi \sum_i |d_i|\(\log \frac{1}{\sum_i |d_i| \eps }- C \log \lep  \).
 \ee
  One immediately checks that the bound $D \lep \le C F_\eps (u,A)$, the bounds on $G_\eps $ and on $\he$ yield $D\le C \lep$ for some constant $C$ (depending only on $\lambda$). Plugging this into \eqref{resb}, we get
  \be \label{resb2}
 \int_{\cup_i B_i} |\nab_{A_\eps} u_\eps|^2 + |h_\eps-\he|^2 + \frac{(1-|u_\eps|^2)^2}{2\eps^2} 
 \ge \pi \sum_i |d_i|(\lep - C \log \lep  ).
 \ee
 It also implies that, defining $\nu_\eps= 2\pi \sum_i d_i \delta_{a_i}$ (the discrete approximate Jacobian), we have  that $\frac{\nu_\eps}{\he}$ is bounded in the sense of measures (since  \eqref{asshe} holds).   Thus, up to extraction, we can assume that $\frac{\nu_\eps}{\he} \to \mu$ for some bounded Radon measure $\mu\in \mathcal{M}(\om)$, and we have 
 \be\label{mucov}
 \liminf_{\eps \to 0} \frac{2\pi \sum_i  |d_i|}{\he} \ge |\mu|(\om).\ee
 In addition, with the Jacobian estimate Theorem \ref{thjac}, by choice of the  final radius $r= \lep^{-10}$, we find that $\frac{1}{\he}  (\mu(u_\eps, A_\eps)- \nu_\eps) \to 0$ in $(C^{0,1}(\om))^*$ and thus we also have \be\label{cvjaco}
 \frac{  \mu(u_\eps, A_\eps)}{\he}\to \mu\quad \text{in} \  (C^{0,1}(\om))^*.\ee
 
Next, since we are looking for a lower bound for $G_\eps(u,A)$ we can assume without loss of generality  that $G_\eps(u, \cdot)$ has been minimized with respect to $A$, which ensures, as explained in Chapter \ref{glheuris}, that the second Ginzburg-Landau equation \eqref{gl2} is satisfied, hence also the  London equation \eqref{eqlondon}. 
Dividing \eqref{eqlondon} by $\he$ and using \eqref{cvjaco}, we find that $\frac{h_\eps}{\he}$ (where $h_\eps= \curl A_\eps$) converges  (say in the sense of distributions) to some $h_{\mu}$ which is related to $\mu$ via \eqref{hmu}. 

We recall that \eqref{gl2} implies that $|u_\eps|^2 |\nab_{A_\eps} u_\eps|^2\ge |\nab h_\eps|^2 $. Since $|u_\eps|=1$ outside of the balls modulo an error $\eps^{1/4}$, we may bound from below  $\int_{\Omega \backslash \cup_i B_i}  |\nab_{A_\eps} u_\eps|^2 $  by $\int_{\Omega \backslash \cup_i B_i} |\nab h_\eps|^2 + o(1)$, and thus with \eqref{resb2}  we are led to 
$$G_\eps(u_\eps,A_\eps) \ge  \pi \sum_i |d_i|\(\lep - C \log \lep  \) + \int_{\Omega \backslash \cup_i B_i} |\nab h_\eps|^2 +|h_\eps-\he|^2 .$$
The last step is to divide by $\he^2$ and pass to the liminf. For the second term, we observe that since $\sum_i r_i\to 0$ as $\eps$, we may extract a sequence $\{\eps_n\}_n$ such that 
  $\mathcal{A}_N :=\cup_{n\ge N} (\cup_i B_i) $ satisfies $|\mathcal{A}_N|\to 0$ as $N\to \infty$. In other words, there exists an arbitrarily  small set which contains all the balls for all  $\eps$'s  along the subsequence.  We may then write 
  \begin{multline*}C \ge  \liminf_{n\to \infty}\frac{
  G_{\eps_n}(u_{\eps_n},A_{\eps_n})}{\he^2} \ge \liminf_{\eps \to 0} \frac{\lep}{\he} \liminf_{\eps \to 0}\frac{\pi \sum_i |d_i|}{\he}\\ +\hal  \liminf_{n\to \infty} \int_{\Omega \backslash \mathcal{A}_N}
   \left|\nab\frac{ h_\eps}{\he}\right|^2 +\left|\frac{h_\eps}{\he}-1\right|^2 .\end{multline*}
   Using \eqref{asshe}, \eqref{mucov}, the weak convergence of $h_\eps/\he$ to $h_\mu$,  and weak lower semi-continuity (up to a further extraction), we deduce that for every $N$,
   $$  \liminf_{n\to \infty}\frac{
  G_{\eps_n}(u_{\eps_n},A_{\eps_n})}{\he^2} \ge \frac{1}{2\lambda}|\mu|(\Omega)+ \hal 
   \int_{\om \backslash \mathcal{A}_N} |\nab h_{\mu}|^2 + |h_\mu-1|^2.
   $$
   
   Letting then $N\to \infty$, since $|\mathcal{A}_N|\to 0$  we deduce that $h_\mu\in L^2(\om)$ and the lower bound 
   $  \liminf_{\eps\to 0}\frac{
  G_{\eps}(u_{\eps},A_{\eps})}{\he^2} \ge E_\lambda(\mu)$ holds along that subsequence.
 This proves the $\Gamma$-liminf relation (together with the compactness). 
\subsection{Upper bound}
 We prove the $\Gamma$-limsup inequality via the construction of a recovery sequence when $\mu$ is a nonnegative measure (the general case is not much different, see \cite[Chap. 7]{livre} for details).
 We first split $G_\om(x,y)$ (defined by \eqref{eqgreen}) as $G_\om(x,y)=  \frac{1}{2\pi} ( -\log |x-y| +S_\om(x,y))$ with $S_\om\in C^1(\om\times \om)$. 
 \medskip
 
 \noindent
{\bf Step 1.} Determining the vortex locations\\
 Since $\mu$ is a positive measure of finite mass in $\om$, we may apply the $\Gamma$-limsup part of Proposition \ref{gammaconvergenceHn} to the probability measure $\frac{\mu}{|\mu|(\om)}$, with potential $V=0$ and with 
 \be\label{defnnnn}
 n = \left[\frac{1}{2\pi} \he 
|\mu|(\om)\right] ,\ee with $[\cdot]$ the integer part.  This yields the existence of points $a_i$ (depending on $n$ hence on $\eps$), such that 
\be\label{prem1}
\nu_\eps:=\frac{\sum_i \delta_{a_i} }{[\frac{1}{2\pi} \he 
|\mu|(\om)] } \to  \frac{\mu}{|\mu|(\om)}\ee
and 
\be\label{prem2}\limsup_{n\to \infty} \iint_{\mr^2 \backslash \triangle} - \log |x-y|\, d\nu_\eps(x)\, d\nu_\eps(y)\le \frac{1}{|\mu|(\om)^2}  \iint - \log |x-y|\, d\mu(x)\, d\mu(y).\ee
Moreover, examining the proof in Proposition \ref{gammaconvergenceHn}, we see that the points $a_i$ are separated by a distance $Cn^{-1/2}  \ge c\he^{-1/2} \gg \eps$, we   may also check that the points can be assumed to all lie in $\om$, and that the same results hold when replacing $\nu_\eps$ by $\frac{\frac{1}{2\pi} \sum_i \mu_i  }{[\frac{1}{2\pi} \he 
|\mu|(\om)] }$ where $\mu_i$  is the uniform measure of mass $2\pi$ supported in $\pa B(a_i, \eps)$ (note the $\mu_i$'s have disjoint support by the previous observation).
In other words, we have 
\be\label{prem3}
\frac{\mu_\eps}{\he} :=\frac{\sum_i \mu_i }{\he} \to \mu\ee
in the weak sense of measures, and 
\be\label{prem4}\limsup_{\eps \to 0 } \sum_{i\neq j} \frac{1}{\he^2} \iint - \log |x-y|\, d\mu_i(x)\, d\mu_j(y)\le  \iint - \log |x-y|\, d\mu(x)\, d\mu(y).\ee
Since \eqref{prem3} holds, by weak convergence  and regularity of $S_\om$ we also have 
\be \label{prem5}
\limsup_{\eps \to 0 }  \frac{1}{\he^2} \iint S_\om(x,y)\, d\mu_\eps(x)\, d\mu_\eps(y)\le  \iint S_\om(x,y)\, d\mu(x)\, d\mu(y).\ee

We can also easily estimate the contributions of diagonal terms, by definition of $\mu_i$~:
\begin{multline}\label{prem6}
 \sum_{i=1}^n  \iint - \log |x-y|\, d\mu_i(x)\, d\mu_i(y)= - n \int_{[0,2\pi]^2} \log |\eps e^{i\theta}-\eps e^{i\phi}| \, d\theta\, d\phi\\
 = 4\pi^2 n \lep +C.\end{multline}
Combining \eqref{asshe}, \eqref{defnnnn}, \eqref{prem4}---\eqref{prem6} and the splitting of $G_\om$ we obtain 
\begin{multline}\label{prem7}
 \frac{1}{\he^2} \iint_{\om\times \om} G_\om(x,y)\, d\mu_\eps(x)\, d\mu_\eps(y)\\
 \le  \iint_{\om \times \om} G_\om(x,y)\, d\mu(x)\, d\mu(y) + \frac{ |\mu|(\om)}{\lambda} +o(1).\end{multline}
{\bf Step 2.}  Constructing the configuration.\\
In this step we sort of reverse-engineer the configuration $(u_\eps,A_\eps)$ from the vortices we have  constructed.
First we let $h_\eps$ be the solution of 
\be \label{eqforh}
\left\{\begin{array}{ll}
 - \Delta h_\eps + h_\eps= \mu_\eps & \text{in} \ \om\\
h_\eps=\he & \text{on} \ \bo.\end{array}\right.\ee
Then, we let $A_\eps$ be any vector field such that $\curl A_\eps=h_\eps$ in $\om$ and define $u_\eps=\ro_\eps e^{i\vp_\eps}$ as follows. 
We let 
\begin{equation}\label{917}\ro_\eps(x)= \begin{cases} 
0 & \text{in} \cup_{i=1}^n B(a_i, \eps)\\
\frac{|x-a_i|}{\eps} -1 & \text{in } B(a_i, 2\eps) \backslash B(a_i, \eps)\\
1& \text{otherwise.}
\end{cases}\end{equation}
For any $x \in \om \backslash \cup_{i=1}^n B(a_i, \eps)$, we let 
$$\vp_\eps(x)= \oint_{(x_0,x)} (A_\eps- \nab^{\perp} h_\eps) \cdot \tau \, d\ell,$$
where $x_0$ is any reference point in $\om \backslash \cup_i B(a_i, \eps)$, and $(x_0,x)$ is any curve joining $x_0$ to $x$ in $\om  \backslash \cup_i B(a_i, \eps).$
From \eqref{eqforh} we see that this definition does not depend modulo $2\pi$ on the curve chosen to join $x_0$ to $x$, thus $e^{i\vp}$ is well-defined in $(\cup_i B(a_i, \eps))^c$.
Indeed, if $\gamma= \pa U$ is a closed curve in $(\cup_i B(a_i, \eps))^c$, using Stokes' theorem and $\curl A_\eps=h_\eps$, we find 
$$\oint_\gamma (A_\eps-\nab^\perp h_\eps) \cdot \tau \,d\ell= \int_U (-\Delta h_\eps +h_\eps) = \int_U \sum_i \mu_i\in 2\pi \mn.$$ 
The (multi-valued) function $\vp_\eps$ satisfies
\be\label{pre7}
-\nab^\perp h_\eps = \nab \vp_\eps - A_\eps \quad \text{in} \  \om  \backslash \cup_i B(a_i, \eps).\ee
Finally, we may define $u_\eps=\ro_\eps e^{i\vp_\eps}$ and we notice that the fact that $\vp_\eps$ is not defined on $\cup_i B(a_i, \eps)$ is not important since $\ro_\eps $ is zero there.  
\medskip

\noindent
{\bf Step 3.}  Computing the energy of the test-configuration.\\
To finish it suffices to evaluate $G_\eps(u_\eps,A_\eps)$. 
First we notice that by construction $$\io |\nab |u_\eps||^2 + \frac{(1-|u_\eps|^2)^2}{2\eps^2} \le Cn \le o(\he^2)$$ by \eqref{defnnnn}. Using that $|\nab_{A_\eps} u_\eps|^2 = |\nab |u_\eps||^2 + |u_\eps|^2 |\nab \vp_\eps -A_\eps|^2 $ and that $|u_\eps|=0$ in $\cup_i B(a_i, \eps)$ and $|u_\eps|\le 1$ everywhere and \eqref{pre7}, 
we deduce that 
$$\io |\nab_{A_\eps} u_\eps|^2 +\frac{(1-|u_\eps|^2)^2}{2\eps^2} \le \int_{( \cup_i B(a_i, \eps))^c} |\nab \vp_\eps - A_\eps|^2  +O(n) \le \io |\nab h_\eps|^2+o(\he^2) .$$
It follows that for this configuration we have the inequality corresponding to \eqref{apprx}
i.e. 
$$G_\eps(u_\eps,A_\eps) \le \hal \io |\nab h_\eps|^2 + |h_\eps-\he|^2 +o(\he^2).
$$
In view of \eqref{eqforh}, we have $h_\eps(x)= \he +\io G_\om(x,y) \, (\mu_\eps(y)-\he)\, dy$ and using \eqref{nabug}, we have that 
$$ \io |\nab h_\eps|^2 + |h_\eps-\he|^2=  \iint_{\om\times \om} G_\om(x,y) \, d(\mu_\eps(x)-\he) (x)\, d(\mu_\eps(y)-\he)(y),
$$as in \eqref{apprx2}.
Evaluating this integral is now a direct consequence of \eqref{prem7} and \eqref{prem3}
and leads us to 
\begin{multline*}\limsup_{\eps \to 0}\frac{G_\eps (u_\eps,A_\eps) }{\he^2} \le\hal \iint_{\om\times \om} G_\om(x,y) \, d(\mu-1) (x)\, d(\mu(y)-1)(y)+ \frac{1}{ 2\lambda} |\mu|(\om)\\
= \hal \io |\nab h_\mu|^2 + |h_\mu-1|^2 +  \frac{1}{ 2\lambda} |\mu|(\om)=E_\lambda(\mu).\end{multline*}
Indeed, there is no problem in passing to the limit in terms of the form 
$$\iint_{\om\times \om} G_\om(x,y) \, dx \ d\frac{\mu_\eps}{\he} (y) $$ since one may check that the function $\io G_\om (x,y) \, dx $ is a continuous function of $y$.
This concludes the proof of the $\Gamma$-limsup, provided we check that we do have $\frac{\mu(u_\eps,A_\eps)}{\he} \rightharpoonup \mu$. But this can be checked  from 
$\mu(u_\eps,A_\eps) = \curl (|u_\eps|^2 (\nab \vp_\eps-A_\eps)) + h_\eps = \curl (|u_\eps|^2 \nab^\perp h_\eps) +h_\eps$, \eqref{eqforh}, \eqref{917} and \eqref{prem3}.

 \section[Mean-field limit and obstacle problem]{Minimization of the mean-field limit and connection to the obstacle problem}
 Once the $\Gamma$-convergence result is obtained, it immediately implies the leading-order behavior from Proposition \ref{gammaconvmini} and Remark \ref{gcvcomplete}: we have
 \begin{coroll}[Limit of Ginzburg-Landau minimizers]
 Assume \eqref{asshe}. Let $(u_\eps, A_\eps)$ minimize $G_\eps$, then as $\eps \to 0$ we have 
$$\frac{\mu(u_\eps, A_\eps)}{\he}\rightharpoonup \mu_\lambda$$
where $\mu_\lambda $ is the unique minimizer of $E_\lambda$.
\end{coroll}
The fact that $E_\lambda$ has a unique minimizer is a consequence of its obvious convexity in $\mu$. 
This result is indeed a mean-field limit, since it describes the limit of the (suitably normalized) vorticity for which  \eqref{jacheur} holds. Since $\he \to +\infty $ as $\eps \to 0$, the number of vortices  is expected to blow-up like $\he$ too, and they arrange themselves according to the distribution $\mu_\lambda$, which plays the role of the equilibrium measure $\mu_0$ for the Coulomb gas in the Ginzburg-Landau context, cf. Figure \ref{fig12}.

The limiting energy $E_\lambda(\mu)$ is of similar nature as the mean-field limit Hamiltonian $I$ in Chapter \ref{leadingorder}. This is more readily visible if one rewrites $E_\lambda$ as 
\be\label{rewre} E_\lambda (\mu)= \frac{1}{2\lambda}|\mu|(\om)+\iint_{\om \times \om} G_\om(x,y) d(\mu-1)(x)\,d(\mu-1)(y) .\ee
Compared to $I$ in \eqref{definitionI}, the confining potential is replaced by the fact of working in a bounded domain with a Green-Dirichlet function, and the constraint that $\mu$ be a probability is replaced by the penalization term in $|\mu|(\om)$, which behaves like a Lagrange multiplier term.

We may now identify the minimizer of $E_\lambda$ with the solution of an obstacle problem, just like for the minimization of $I$. The correspondence here is in some sense even easier due to the fact that we are in a bounded domain. 

\begin{prop}[Identification of the optimal density]\label{pro91} The minimizer $\mu_\lambda$ of $E_\lambda$  is uniquely characterized by the fact that the associated potential $h_{\mu_\lambda}$ given by \eqref{hmu} is the solution of the following obstacle problem :
\be\label{obsgl}\min_{\substack{ h\ge 1-\frac{1}{2\lambda}\\ h-1\in H^1_0(\om)}}\io |\nab h|^2 + h^2.\ee
The function $h_{\mu_\lambda}$ is in turn characterized by the variational inequality 
$$\forall v \ge 1-\frac{1}{2\lambda}, v-1\in H^1_0(\om), 
\text{we have}  \quad \io \nab h_{\mu_\lambda} \cdot \nab (v-h_{\mu_\lambda}) + h_{\mu_\lambda} (v-h_{\mu_\lambda})  \ge 0 $$
or by the relations 
\be\label{elobs}
\begin{cases}
h_{\mu_\lambda}  \ge 1-\frac{1}{2\lambda} & \text{in} \ \om\\
h_{\mu_\lambda} = 1- \frac{1}{2\lambda} & \text{q.e. in the support of } \mu_\lambda\\
h_{\mu_\lambda}= 1 & \text{on} \ \bo.\end{cases}
\ee
\end{prop}The constant function $1-\frac{1}{2\lambda}$ thus plays the role of the obstacle. 
We will write 
\be\label{coinset}
\omega_\lambda:=
\left\{h_{\mu_\lambda}=1-\frac{1}{2\lambda}\right\}\ee for  the coincidence set. 
From the above characterizations we deduce that $\mu_\lambda $ is a nonnegative measure, and that  \be\label{mul}
\mu_\lambda = -\Delta h_{\mu_\lambda} + h_{\mu_\lambda}= (1-\frac{1}{2\lambda} ) \indic_{\omega_\lambda} .\ee
Thus,  in this Ginzburg-Landau context, the optimal measure always has a density, and that density is always constant on its support, cf. Fig. \ref{fig12}.
 
\begin{figure}[h!] 
\begin{center}
\includegraphics[scale=0.6]{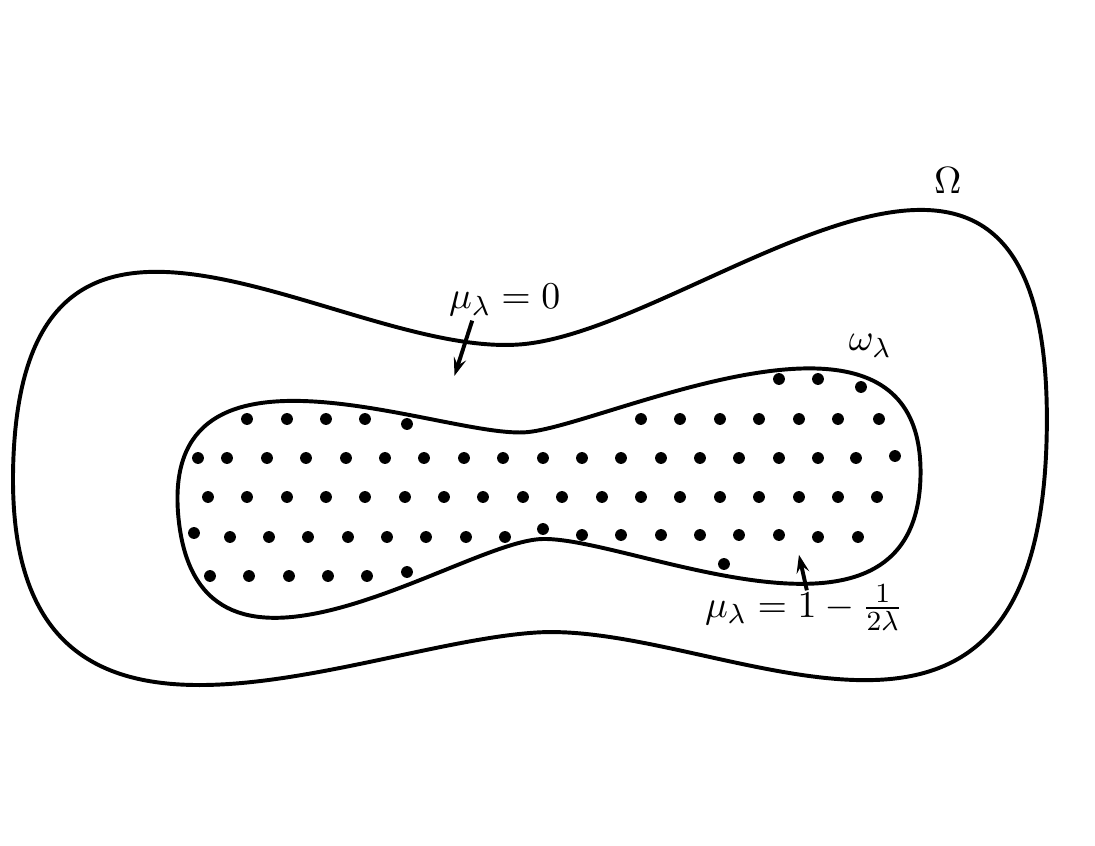}
\caption{The optimal measure for Ginzburg-Landau}\label{fig12}
\end{center}\end{figure}

 \begin{proof}[Proof of the proposition] To give an alternate proof to that of  Chapter \ref{leadingorder}, we may obtain this by convex duality. It suffices to observe that 
 the minimization of $E_\lambda$ viewed as the following function of $h_\mu-1$:
 $$\frac{1}{2\lambda} \io |-\Delta h+h +1| + \hal \io |\nab h|^2+ h^2
 $$ 
 is dual in the sense of convex duality to the minimization problem
 $$\min_{\substack{ |h|\le \frac{1}{2\lambda}\\ h\in H^1_0(\om)}} \hal \io |\nab h|^2 + h^2 + 2h
$$ in the sense that they have the same minimizer (and minimum). For a proof of this fact, cf. \cite[Chap. 7]{livre}.
 One  can then check that by the maximum principle the solution $h$ satisfies $h\le 0$, so the upper constraint $h\le \frac{1}{2\lambda}$ is not active, and thus $h+1$ solves 
  \eqref{obsgl}. 
 The alternate way is to start from \eqref{rewre} and  make variations on $\mu$ as in the proof of Theorem \ref{theoFrostman}. This leads to the equations \eqref{elobs}, and one can then check that they uniquely characterize the solution to \eqref{obsgl}.
 
 The regularity theory is exactly as in  Proposition \ref{proequivpb}: since the obstacle is constant hence smooth, $h_{\mu_\lambda}$ is $C^{1,1}$ by Frehse's regularity theorem, and we can deduce that \eqref{mul} holds.
 \end{proof}
 
 The solution of \eqref{obsgl} when the constraint of being above the obstacle is omitted, is obviously the function $h_0$ solution to \eqref{h0}. It then follows that $h_0$ is also the solution of the problem \eqref{obsgl}, if and only if $h_0$ lies above the obstacle i.e. $h_0 \ge 1-\frac{1}{2\lambda}$, equivalent to $\lambda \ge 1/(2\min (h_0-1))$. Whether this condition is satisfied depends on the value of $\lambda$, which we recall is $\lim_{\eps \to 0} \frac{\he}{\lep}$, i.e. encodes the intensity of the applied magnetic field.
 
 We deduce the following result on the description of $\mu_\lambda$
   as the external magnetic field is increased.
 
 \begin{prop}\label{pro92}
 \begin{itemize}
 \item $\omega_\lambda$ is increasing with respect to $\lambda$ and $\cup_{\lambda >0} \omega_\lambda= \om$. \item For $\lambda \le \lambda_\om:= \frac{1}{2\max |h_0-1|} $ we have $\omega_\lambda =\varnothing$, $\mu_\lambda=0$ and $h_{\mu_\lambda}= h_0$.
\item For $\lambda >\lambda_\om:=\frac{1}{2\max |h_0-1|} $ we have $\mu_\lambda \neq 0$,  and \eqref{mul} holds. \end{itemize}
\end{prop}

This way we recover in a weaker sense the value of the first critical field for which vorticity first appears:
\be
\label{hc1approx}
\hci \sim \lambda_\om\lep\ee
i.e. we give a first rigorous justification of \eqref{vhci}.

We can also deduce from the proposition that the onset of $\omega_\lambda$ is located near the point(s) of minimum of $h_0$, i.e.  that the vorticity first appears there, as formally derived in Chapter \ref{glheuris}.

\section{The intermediate regime near $\hci$}
Understanding what happens more precisely near $\hci$ (exact number and locations of the vortices as $\he \sim \hci$) requires a finer analysis than this leading order one: one needs to make more precise expansions around $\he h_0$ as in Section 7.4.2. This is done in \cite[Chap. 9, Chap. 12]{livre}. It is found that the vortices appear one by one near the point(s) of minimum of $h_0$ in $\om$, with new vortices appearing each time $\he$ is incremented by an order $\log \lep$, as long as $\he \sim \frac{\lep}{2\max|h_0-1|}$ i.e. $\lambda= 1/(2\max|h_0-1|)$.
The locations of the vortices tend to minimize exactly a Coulomb gas type of interaction:
if their number  $n$ remains bounded as $\eps \to 0$, their location (suitably blown-up near the point(s) of minimum of $h_0$) minimize
\be\label{wQ}-\sum_{i\neq j\in [1,n]} \log |x_i-x_j|+ n\sum_{i=1}^n Q(x_i)\ee
while if their number becomes  unbounded, their density (again after suitable scaling) tends to minimize among probability measures
\be\label{IQ}
I(\mu)= \iint  - \log |x-y|\, d\mu(x)\, d\mu(y)+ \int_{\mr^2} Q(x)\, d\mu(x)\ee
with $Q$ a nonnegative quadratic function  equal to the Hessian of $h_0$ at the point of minimum of $h_0$.
Once $\lambda>1/(2\max|h_0-1|)$ the optimal description is the one just given in Propositions \ref{pro91}, \ref{pro92}.

In other words, in the critical regime near $\hci$,  the interaction of vortices is precisely that of a 2D Coulomb gas with quadratic confining potential $Q$. One remembers from Chapter \ref{leadingorder} that the equilibrium measures associated to quadratic potentials $V$ are always of constant density, just like   what happens for $\mu_\lambda$.

 
\chapter{The splitting and the next order behavior for Ginzburg-Landau}
\label{glnext}
In this chapter we sketch the method that allows us to derive the renormalized energy $W$ from the minimization of Ginzburg-Landau, at the next order, beyond the mean-field limit that we just saw in Chapter \ref{glleading}.
This is, in a simplified form, the content of the paper \cite{ssgl}.

\section{Splitting}
In order to extract the next order energy, it is very important to have an exact splitting of the Ginzburg-Landau functional which ``algebraically" decouples the orders. 
We saw in the previous chapter that the magnetic field $h_\eps$ satisfies $\frac{h_\eps}{\he} \to h_{\mu_\lambda}$ in some weak sense, with $\mu_\lambda$ the ``mean-field" limit, i.e. the minimizer of $E_\lambda$. 
Since $h_\eps$ plays the role of the potential $h_n$ for the Coulomb gas, the splitting for the Coulomb gas Hamiltonian in Chapter \ref{chapsplit} can then give us a hint~: to split $h_\eps$ as $\he h_{\mu_\lambda}+h_1$ where $h_1$ is a remainder.
This rough idea is correct, however there are two difficulties: first \eqref{apprx} is only an approximate relation, and not an identity, so we really need to work starting from the configurations $(u,A)$ themselves. Secondly the approximation $h\sim \he h_{\mu_\lambda}$ is correct at leading order, but a correction needs to be introduced in order to extract the right next order.

The fact is that we need to know the number of vortices (or their total degree) more precisely than through $\he \mu_\lambda$.
$\mu_\lambda$ was found by minimizing the limiting energy $E_\lambda$, which was itself
 derived by bounding below the cost of the self-interaction of each vortex via the ball construction lower bound, which gives a cost $\sim\pi |d|\lep$ per vortex, resulting  after taking the limit $\eps \to 0$ in the term 
$\lambda |\mu|(\om)$ in $E_\lambda$.
But one needs to be more precise: we will restrict ourselves to the situation where $\lambda>\lambda_\om$ (cf. \eqref{hc1approx}), i.e. the limiting measure $\mu_\lambda$ and the coincidence set $\omega_\lambda$ are  nontrivial.  The number of expected vortices in that regime is thus proportional to $\he$. Since vortices are uniformly distributed in $\omega_\lambda$, their mutual distances are of order $1/\sqrt{\he}$. 
We can thus think of each vortex as being alone in a box of size $C/\sqrt{\he}$, and remembering that its core size is always $\eps$, a lower bound of the type \eqref{bb} leads us to expecting a cost
$\pi |d|\log \frac{C}{\eps\sqrt{\he}}$ per vortex. In other words, there should be a correction of order $\log \he\sim \log \lep$ per vortex, which has not been accounted for in $E_\lambda$ (and did not matter at the leading order level).
This heuristically justifies introducing this correction in the self-interaction cost,  by minimizing instead of $E_\lambda$ the following problem:
\be\label{ecorr}
\hal \log \frac{1}{\eps\sqrt{\he} } \io |\mu|+ \frac{1}{2} \io |\nab h|^2 + |h-\he|^2\ee
where $\mu$ and $h$ are related via 
 \be\label{hmu2}
  \left\{ \begin{array}{ll}
  -\Delta h+ h= \mu & \text{in} \ \Omega\\
  h=\he& \text{on} \ \bo.\end{array}\right.
  \ee
This is a problem of the same form as the minimization of $E_\lambda$ (to see it, just divide everything by $\he^2$), except that the parameter $\lambda$ is replaced by the correction $\frac{\he} {\log \frac{1}{\eps\sqrt{\he}}}$, which is equivalent to $\lambda$ as $\eps \to 0$ (but here we define this energy for each fixed $\eps$).
As a result, the same proof as Proposition \ref{pro91} applies, and asserts that the minimization problem \eqref{ecorr} is equivalent to the obstacle problem 
\be \label{obscorr}
\min_{\substack{h\ge \he - \frac{1}{2}\log \frac{1}{\eps \sqrt{\he} } \\
 h- \he \in H^1_0(\om)}}    \io |\nab h|^2 +h^2.\ee
We will denote the solution by $h_{0,\eps}$ and the associated measure
\be\label{mu0e}
\mu_{0,\eps}=-\Delta h_{0,\eps} + h_{0,\eps}.\ee
It is clear that  $\frac{\mu_{0,\eps}}   {\he} \to \mu_\lambda$ and $\frac{h_{0,\eps}}{\he} \to h_{\mu_\lambda}$ as $\eps \to 0$, however these objects are a little more precise, and as explained above,  they are the ones with respect to which we should do the splitting. 

We will also denote by  $\omega_{0,\eps}= \{x, h_{0,\eps}(x)= \he - \hal \log \frac{1}{\eps \sqrt{\he} }\}$ the corresponding coincidence set, and  note that 
$$\mu_{0,\eps}= \(\he - \hal \log \frac{1}{\eps \sqrt{\he} }\) \indic_{\omega_{0,\eps}},$$
and recall that $\mu_{0,\eps}\ge 0$.

\begin{prop}[Splitting formula for Ginzburg-Landau \cite{ssgl}]
\label{splittinggl}
Let $(u,A)$ be an arbitrary configuration and  for any $\eps$ and   $\he$, set
$$A_{1,\eps}= A- \nab^\perp h_{0,\eps},$$ where $h_{0,\eps}$ is the solution of \eqref{obscorr}.
Then we have 
\be\label{splitgl}
\boxed{G_\eps(u,A)= G_\eps^0 + G_\eps^1(u,A_{1,\eps}) -\hal \io (1-|u|^2) |\nab h_{0,\eps}|^2 }\ee
where 
\be\label{ge0}
G_\eps^0= \hal \log \frac{1}   {\eps\sqrt{\he} }\io \mu_{0,\eps}+ \hal \io |\nab h_{0,\eps}|^2 + |h_{0,\eps}-\he|^2 \ee
and 
$$G_\eps^1 (u,A) = \hal \io |\nab_A u|^2 + |\curl A- \mu_{0,\eps}|^2 + \frac{(1-|u|^2)^2}{2\eps^2} + \io  (h_{0,\eps} - \he) \mu(u,A).$$
\end{prop}

\begin{proof}[Proof of the splitting formula]
First, in view of the  definition of $A_{1,\eps}$ and \eqref{mu0e}, we may write
$$|\nab_A u|^2= |\nab_{A_{1,\eps}} u|^2 + |u|^2 |\nab h_{0,\eps}|^2- 2\nab^\perp h_{0,\eps}\cdot \langle iu, \nab_{A_{1,\eps}} u\rangle,$$
and
$$\curl A= \curl A_{1,\eps}+\Delta h_{0,\eps}= \curl A_{1,\eps} +h_{0,\eps}- \mu_{0, \eps}.
$$
Inserting into the expression of $G_\eps$ and expanding the squares, we find
\begin{multline*}
G_\eps(u,A)= \hal \io |u|^2 |\nab h_{0,\eps}|^2 +|h_{0,\eps}-\he|^2
\\+ \io - \nab^\perp h_{0,\eps} \cdot \langle iu, \nab_{A_{1,\eps}} u\rangle +(\curl A_{1,\eps}- \mu_{0,\eps}) (h_{0,\eps}  -\he)\\+ \hal \io |\nab_{A_{1,\eps}} u|^2 + |\curl A_{1,\eps} - \mu_{0,\eps}|^2 + \frac{(1-|u|^2)^2}{2\eps^2}\\
=\hal \io |\nab h_{0,\eps}|^2 +|h_{0,\eps}-\he|^2 + \io (h_{0,\eps}-\he) (\curl \langle iu, \nab_{A_{1,\eps}} u\rangle   +\curl A_{1,\eps}-\mu_{0,\eps})
\\+\hal \io (|u|^2 -1) |\nab h_{0,\eps}|^2
\end{multline*}
where we have used an integration by parts and the fact that $h_{0,\eps}=\he$ on $\bo$.
We next observe that 
$$ \io (h_{0,\eps}-\he) \mu_{0,\eps}= - \hal \log \frac{1}{\eps\sqrt{\he}} \io \mu_{0,\eps} $$
since that is the value of $h_{0,\eps} - \he$ on the support of $\mu_{0,\eps}$, and 
$$\curl \langle iu, \nab_{A_{1,\eps}} u\rangle   +\curl A_{1,\eps}= \mu(u,A_{1,\eps}) .$$
Inserting into the above,  we obtain the result.
\end{proof}

As desired, we have obtained an exact splitting formula for the Ginzburg-Landau energy. The first term $G_\eps^0$ is a constant independent of $(u,A)$ and  easily seen to be asymptotically  equivalent to $\he^2 E_\lambda(\mu_\lambda)$, i.e. to the leading order of the energy. The last term is generally $o(1)$ because thanks to the potential term in the energy we may control $\io (1-|u|^2) $  by $\eps \sqrt{G_\eps(u,A)}$ via Cauchy-Schwarz.
The middle term  $G_\eps^1$ is the interesting one : it is the difference between an energy functional which is very similar to Ginzburg-Landau, except with external field replaced by the non constant function $\mu_{0,\eps})$, and a term which, thanks to the Jacobian estimate  (Theorem  \ref{thjac}) can be evaluated by 
$$ \io  (h_{0,\eps} - \he) \mu(u,A)\simeq 2\pi \sum_i d_i (h_{0,\eps} - \he) (a_i)$$
(indeed, one may easily check that $\mu(u,A_{1,\eps}) \simeq \mu(u,A)$).
In this term, all vortices with positive  degree bring a negative contribution, since $h_{0,\eps} \le \he$ in $\om$ by the maximum principle. In other words they will allow to gain energy, while negative degree vortices will not and hence will not be favorable.  Moreover, $h_{0,\eps}- \he$ is minimal and equal to $-\hal \log \frac{1}{\eps\sqrt{\he}}$ in the coincidence set $\omega_{0,\eps}$, hence it will be most favorable to have vortices there. Setting 
\be\label{defzgl}
\zeta_\eps(x)= h_{0,\eps}- \he+ \hal \log \frac{1}{\eps\sqrt{\he}}\ee 
we have $\zeta_\eps \ge 0$ in $\om$, $\{\zeta_\eps= 0\}=\omega_{0,\eps}$ and we may thus rewrite the splitting formula formally as 
\begin{multline}
\label{splitformel}
G_\eps(u,A)\simeq G_\eps^0 + 2\pi \sum_i   d_i\zeta_\eps(a_i)\\+  \hal \io |\nab_{A_{1,\eps}} u|^2 + |\curl A_{1,\eps}- \mu_{0,\eps}|^2 + \frac{(1-|u|^2)^2}{2\eps^2}  - \pi n  \log \frac{1}{\eps\sqrt{\he}}+o(1),
\end{multline}
with $n=\sum_i d_i$ is  the number of vortices (assumed positive), and rigorously as 
\be 
G_\eps(u,A) = G_\eps^0 +  \io \zeta_{\eps} \mu(u,A_{1,\eps})+ \mathcal{F}_\eps(u,A_{1,\eps}) +o(1)\ee 
where 
 \begin{multline}\label{Feps}
 \mathcal{F}_\eps(u,A_{1,\eps}):=
  \hal \io |\nab_{A_{1,\eps}} u|^2 + |\curl A_{1,\eps}- \mu_{0,\eps}|^2 + \frac{(1-|u|^2)^2}{2\eps^2} \\ - \hal \log \frac{1}{\eps\sqrt{\he}}    \io \mu(u,A_{1,\eps}) .\end{multline}
This is of the same form as the splitting of the Coulomb gas Hamiltonian $H_n$ in Proposition \ref{developpementasympto}. The role of $n^2 I(\mu_0)$ is played by $G_\eps^0 \sim \he^2 E_\lambda(\mu_\lambda)$, $\zeta_\eps $ plays the same role of a confining potential as $\zeta$, confining  the points to the support of the optimal measure ($\mu_{0,\eps}$ for Ginzburg-Landau, the equilibrium measure for the Coulomb gas), and plays no role otherwise. The remaining term  $\mathcal{F}_\eps$ behaves as the precursor to the renormalized energy   $W(\nab h_n, \indic_{\mr^2})$, although this is more delicate to see : the term $ \hal \io |\nab_{A_{1,\eps}} u|^2 + |\curl A_{1,\eps}- \mu_{0,\eps}|^2 + \frac{(1-|u|^2)^2}{2\eps^2}   $ behaves like a Ginzburg-Landau energy, hence it will include the interaction between vortices, plus the cost of each vortex, which can be estimated via a ball-construction method, while the term 
 $- \pi n  \log \frac{1}{\eps\sqrt{\he}}$ in effect subtracts off the cost of each vortex, i.e.  ``renormalizes" the Ginzburg-Landau energy.
 
 All of the analysis from that point consists in showing rigorously that this is true, and that the term 
 $\mathcal{F}_\eps$
  will effectively converge to the (average of the)   renormalized energy $W$.  This will be more technical than for the Coulomb gas, because for instance we have to truly get rid of the possibility of (too many) negative vortices. Also we will  need to have very precise ball construction lower bounds for the cost of each vortex, to show that the compensation 
$-  \pi n  \log \frac{1}{\eps\sqrt{\he}}$, which includes the correction in $\log \lep$ that we inserted, is the right one.

\begin{rem}
In \cite{ssgl} the computations are made more complicated by the fact that we also treat the case of $\he$ possibly very close to $\hci$, which requires more precise estimates, themselves requiring  the mass of the measure which respect to which one splits to be quantized.
\end{rem}
\section{Deriving $W$ from Ginzburg-Landau}
\subsection{Rescaling and notation}
In view of  the splitting formula \eqref{splitgl}, and the fact that the last term is very small, in order to study $G_\eps$,   it suffices to study $G_\eps^1(u,A_{1,\eps})$.
We may introduce $h_{1,\eps}=\curl A-h_{0,\eps} $. In view of the fact (which we may assume) that for $(u,A)$ the second Ginzburg-Landau equation \eqref{gl2}, hence the London equation \eqref{eqlondon} is satisfied, we check that (by definition of $A_{1,\eps}$ and $h_{0,\eps}$), the function $h_{1,\eps}$ satisfies 
 \be\label{london1}
  \left\{ \begin{array}{ll}
  -\Delta h_{1,\eps}+ h_{1,\eps}= \mu(u,A)-\mu_{0,\eps}
   & \text{in} \ \Omega\\
  h_{1,\eps}=0 & \text{on} \ \bo.\end{array}\right.
  \ee
Thus $h_{1,\eps}$ is the analogue in the Ginzburg-Landau context of the potential $h_n$ defined in \eqref{defhn}. 

As in the case of the Coulomb gas, the next step is to blow up at the scale of the inter-vortex distance, here of order $1/\sqrt{\he}$.
Figure \ref{fig10} illustrates how we blow up around a center point belonging to the coincidence set of the obstacle problem, i.e. the support of $\mu_\lambda$ (or $\mu_{0,\eps}$), and wish to to find a triangular lattice distribution of vortices after blow-up in the limit $\eps \to 0$.

\begin{figure}[h!]

\begin{center}
\includegraphics[scale=.7]{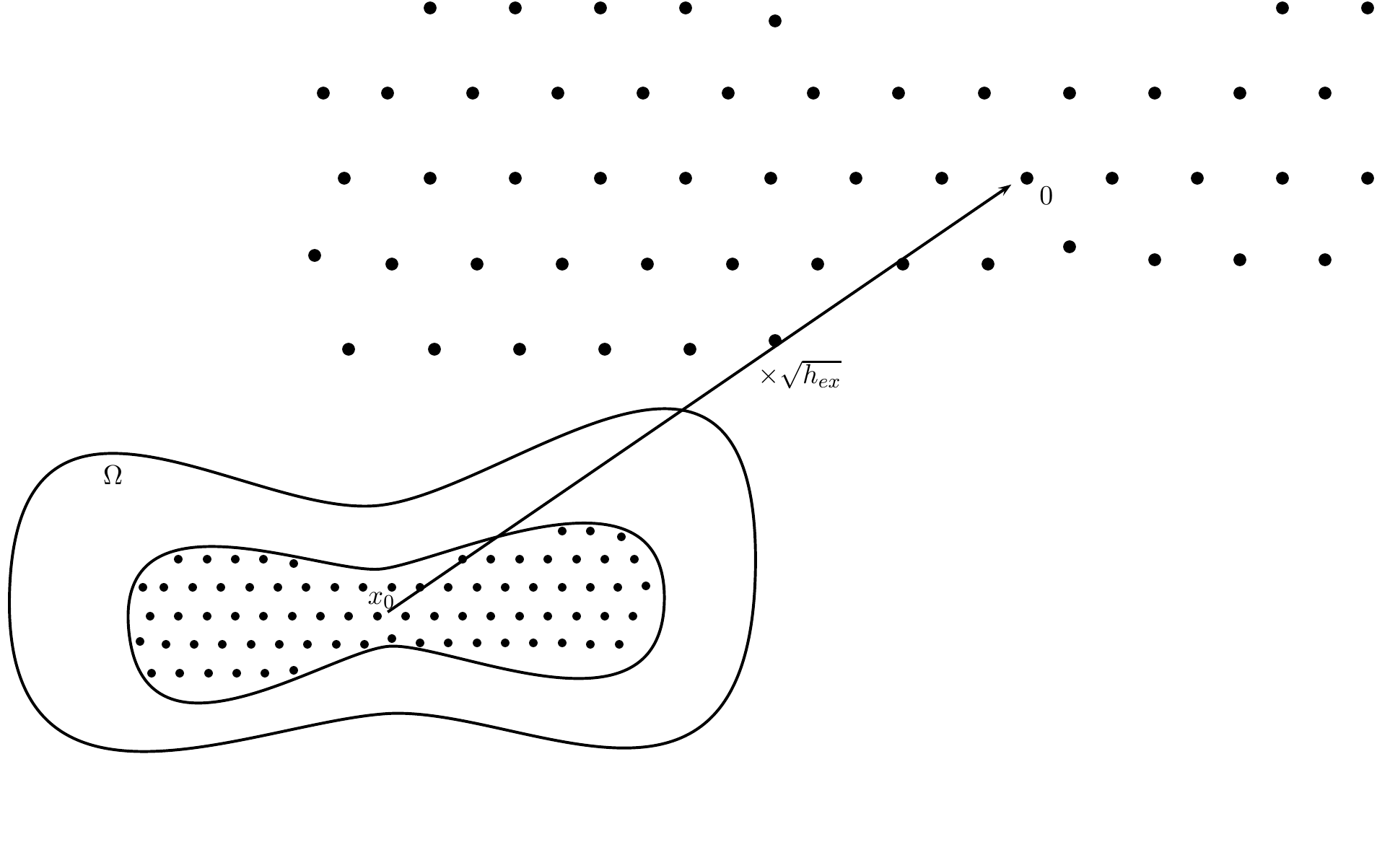}
\caption{Blow up to the Abrikosov lattice}\label{fig10}
\end{center}
\end{figure}

We thus define  $$x'= \sqrt{\he} x,  \ \eps'= \sqrt{\he} \eps, \ u'(x')= u(x),$$
$$ \ A'(x')=\sqrt{\he} A_{1,\eps}(\frac{x'}{\sqrt{\he}}),
  \ h'(x')= h_{1,\eps}(\frac{x'}{\sqrt{\he}}),\  \mu'(x')= \mu(u,A) (\frac{x'}{\sqrt{\he}} ) $$ $$\om'= \sqrt{\he} \om, \ \omega_{0,\eps}' = \sqrt{\he} \omega_{0,\eps},
\  \mu_{0,\eps}'= (1- \frac{1}{2\he} \log \frac{1}{\eps\sqrt{\he}} ) \indic_{\omega_{0,\eps}'} .$$
 We note that the density of $\mu_{0,\eps}' $ tends to $m_\lambda:=1-\frac{1}{2\lambda}$ as $\eps \to 0$.
 Rescaling the equation \eqref{london1} yields 
 $$- \Delta h'+ \frac{1}{\he} h'= \mu' - \mu_{0,\eps}'$$
 and thus if one centers the blow-up in a point of $\omega_{0,\eps}$ we will have in the limit $\eps \to 0$
 $$-\Delta h= 2\pi \sum_p \delta_p- m_\lambda \quad \text{in } \mr^2$$ for some points $p\in \mr^2$ (which we will prove appear with single-multiplicity), 
 i.e, $\nab h$ belongs to the admissible class $\mathcal{A}_{m_\lambda/2\pi}$ defined in Chapter \ref{chapdefw}, for which we can define $W(\nab h)$.
 
 Inserting this change of scales into the energy $\mathcal{F}_\eps$, we find
 $$\mathcal{F}_\eps(u,A_{1,\eps})= \hal \int_{\om_\eps'}|\nab_{A'} u'|^2+ \he|\curl A'- \mu_{0,\eps}'|^2 + \frac{(1-|u'|^2)^2}{2(\eps')^2}  - \hal \log \frac{1}{\eps'} \int_{\om_\eps'} \mu(u',A') .$$

As always, the main result will be obtained by proving first a lower bound, and second a matching upper bound via an explicit construction. 
\subsection{Lower bound}
It is to be obtained by the abstract method presented in Section \ref{frameworkergodique}, applied on the large sets $\omega_{0,\eps}'$, with the ``local" energy being naturally
\begin{multline}\label{fep}
f_\eps(u,A)=  \int_{\R^2} \chi \left[\hal |\nab_A u|^2 + \frac{\he}{2} |\curl A- \mu_{0,\eps}'|^2 + \frac{(1-|u|^2)^2}{4(\eps')^2} - \hal | \log {\eps'}| \mu(u,A)\right],
\end{multline}
with $\chi$ some smooth nonnegative cutoff function supported in $B(0,1)$ and of integral $1$.
Note that this local energy does not depend on a centering point, it is translation-invariant, so we may apply the method  of Section \ref{frameworkergodique} dropping the dependence in  the centering point $x$.
The energy outside of the coincidence set $\omega_{0,\eps}$ will simply be discarded,  it is indeed negligible for minimizers.

In order to apply the abstract framework of Section \ref{frameworkergodique}, a first  step is to show  that if 
\be\label{bornaprio}
\forall R>0, \quad \int_{K_R} f_\eps(\theta_\lambda (u_\eps,A_\eps))\, d\lambda  \le C_R\end{equation}
then $(u_\eps, A_\eps)$ has a  subsequence converging to some $(u,A)$ with  
$$\liminf_{\eps \to 0} f_\eps(u_\eps, A_\eps)\ge f(u,A).$$
 We also recall that \eqref{bornaprio} is equivalent to 
$$\int_{\R^2} \chi* \indic_{K_R}  \left[\hal |\nab_A u|^2 + \frac{\he}{2} |\curl A- \mu_{0,\eps}'|^2 + \frac{(1-|u|^2)^2}{4(\eps')^2} - \hal |\log \eps'| \mu(u,A)\right]\le C_R.$$

One of the important steps of the proof is to show that  the energy density $f_\eps$ controls the number of vortices, so that an upper bound of the form  \eqref{bornaprio} implies that the total degree of vortices in $K_{R-1}$ is  bounded by a constant (depending on $C_R$). Such a bound then easily implies that 
$$\int_{K_{R-1}} \hal |\nab_A u|^2 + \frac{\he}{2} |\curl A- \mu_{0,\eps}'|^2 + \frac{(1-|u|^2)^2}{4(\eps')^2} \le C_R| \log {\eps'}|.$$
As explained in Chapter \ref{toolsgl}, such an upper bound, which controls the number of vortices independently of $\eps'$ makes our life much easier, since it puts us in the framework of \cite{bbh,br2}, for which we can compute sharp and precise lower bounds (up to $o(1)$) for the Ginzburg-Landau energy. It also  completely rules out the possibility of vortices with degrees different from $+1$.
This type of analysis leads to the following lower bound: if \eqref{bornaprio} holds with $C_R$ replaced by $C R^2$  and the second Ginzburg-Landau equation  \eqref{gl2} holds, then 
up to extraction of a subsequence, we have that $h' $ converges to some $h \in \mathcal{A}_{m_\lambda/2\pi}$ and 
$$\liminf_{\eps \to 0} f_\eps(u_\eps, A_\eps) \ge  W(\nab h,\chi)   +\frac{\gamma}{2\pi} m_\lambda,
$$ where $\gamma$ is the constant from \cite{bbh,miro}, and $W$ is the precursor to the renormalized energy as in \eqref{defiW}.
The heuristic for this is quite natural: once the vortices, at points $p$,  have been shown to be of degree $1$ and in bounded number, we split the positive part of the energy as 
\begin{multline*}
\int_{\R^2}  \chi      \left[\hal |\nab_A u|^2 + \frac{\he}{2} |\curl A- \mu_{0,\eps}'|^2 + \frac{(1-|u|^2)^2}{4(\eps')^2}
\right] \\    \ge   \int_{\mr^2 \backslash \cup_p B(p, \eta)} \hal \chi   |\nab_A u|^2    + \int_{\cup_p B(p,\eta)\backslash B(p, M\eps')} \hal \chi |\nab_A u|^2  \\ +\int_{\cup_p B(p,M\eps')}   \chi   
   \left[\hal |\nab_A u|^2 + \frac{\he}{2} |\curl A- \mu_{0,\eps}'|^2 + \frac{(1-|u|^2)^2}{4(\eps')^2}    \right].   \end{multline*}
Outside the $B(p, M \eps)$ with $M$ large, we expect that $|u|\simeq 1$ (which allowed us to  discard some positive terms expected to be negligible). Then from the second Ginzburg-Landau equation, as seen in Chapter \ref{glheuris}, we have $|\nab_A u|\approx |\nab h|$. In  the annuli $B(p, \eta)\backslash B(p, M\eps)$, the energy is then expected to be bounded from below as in \eqref{bb}, which yields 
$\pi \log \frac{\eta}{M\eps'}$ per vortex. In the ``vortex cores" $B(p, M\eps)$ all the energy terms will matter, and the energy depends on the optimal radial profile for the Ginzburg-Landau energy, as given in \cite{hh,miro}, which gives  a contribution  $\gamma$ per vortex, in addition to the cost $\pi \log M$.
Combining all these terms, and multiplying by the number of vortices, which is expected to be $m_\lambda/2\pi$ per unit volume, we formally get 
\begin{multline*}f_\eps(u,A) \ge \int_{\mr^2 \backslash \cup_p B(p, \eta)}  \chi |\nab h|^2 +\frac{\gamma}{2\pi} m_\lambda+\sum_p \chi(p) \( \pi \log \frac{\eta}{\eps'} - \pi \log \frac{1}{\eps'}\)   \\
=  \int_{\mr^2 \backslash \cup_p B(p, \eta)}  \chi |\nab h|^2 +\sum_p \chi(p) \pi \log  \eta+\frac{\gamma}{2\pi} m_\lambda\end{multline*} which, modulo taking the limit $\eta \to 0$,  is exactly the stated result.

We may then define $f(h)= W(\nab h,\chi)+ \frac{\gamma}{2\pi} m_\lambda$, condition $\textbf{(ii)}$ of Section \ref{frameworkergodique} is then satisfied,  and it is also straightforward that $f^*$  defined as in Theorem \ref{th6.1}
satisfies  \begin{multline*}
f^*(h)= \lim_{R\to \infty} \dashint_{K_R} f(h(\lambda + \cdot ) ) d\lambda \\ = \lim_{R\to \infty} \dashint_{K_R}  W(\nab h, \chi *\indic_{K_R})+  \frac{\gamma}{2\pi} m_\lambda= W(\nab h)+  \frac{\gamma}{2\pi} m_\lambda\end{multline*}
where $W$ is now the full renormalized energy as in Definition \ref{def2}.
This allows us to use the abstract framework of Section \ref{frameworkergodique} except there is one major assumption which is not satisfied: namely the assumption $\textbf{(i)}$ that $f_\eps $ must be bounded below by a constant independent of $\eps$.  However this assumption can easily seen not to be true!
What is true is only  that $f_\eps$ is bounded below on average, but not pointwise. 
This causes one of the most serious and technical difficulties in the proof.

We point out that this problem  
 does not occur in the truncation approach to Coulomb gases where one takes the $R\to \infty$ limit {\it before} the $\eta \to 0$ limit, which is one of the main advantages of that approach. By contrast, it  occurs in the approach of  renormalizing by  first ``cutting out holes" and letting $\eta \to 0$.  In the Ginzburg-Landau setting, we see no analogue of the truncation method that could remedy this, in particular due to the fact that the signs of vortices can a priori be arbitrary, which makes the monotonocity of the truncation break down. Instead, we rely on the ball construction lower bounds to remedy this (and this is a completely two-dimensional remedy, as is the ball construction). 
The lower bound stated in  Theorem \ref{thboules} do not suffice, rather we need the ``improved lower bounds" that we introduced in \cite{compagnon} and which are roughly explained at the end of Section \ref{sec812}. Thanks to these lower bounds, which are sharp up to a constant error per vortex, we are able to show that even though the energy density associated to $f_\eps$ is not bounded below pointwise, the negative part of $f_\eps$ (corresponding to the subtracted  vortex costs) 
can be ``displaced" into the positive part of $f_\eps$ in order to replace $f_\eps$ by an equivalent energy density $g_\eps$ which is pointwise bounded below, without making too much error, in the sense that  
$\|f_\eps - g_\eps\|_{\mathrm{Lip}^*} $ is well controlled. We call this ``mass displacement." For that we have to look for energy that compensates the negative $- \pi| \log {\eps'}|$, and this energy is found in balls surrounded the vortices, as well as in annuli  (as described at the end of Section 8.1.2) that can be up to distance $O(1)$ away.
This is done in \cite{compagnon}, and at the same time it is shown that $g_\eps$ (hence $f_\eps$)  controls the number of vortices, a crucial fact whose need we mentioned above.

\subsection{Upper bound}
The upper bound, at least the one needed to construct a recovery sequence for minimizers, can be obtained with the same ideas as  for the Coulomb gas in the proof of Proposition \ref{proub}.
There, we considered the support of the equilibrium measure and partitioned it into rectangles on which we pasted the ``screened" minimizers obtained in Proposition~\ref{propscreening}, rescaled to have the proper density.  In the Ginzburg-Landau context, this is easier  since $ \mu_{0,\eps}'$ the analogue of the blown-up equilibrium measure, has a uniform density on its support. What we thus need to do is assume that this support $\omega_{0,\eps}'$ is nice enough (has $C^1$ boundary), which we can show is ensured for example by the strong requirement that $\om$ be convex (we could certainly remove that condition and replace it by the assumption that $\omega_{0,\eps}$ has no cusps), and then partition it (up to a small boundary layer) into squares of size $R\times R$. In each square we need to paste a solution of 
$$\left\{ \begin{array}{ll} -\Delta h= 2\pi \sum_{p\in \Lambda} \delta_p - m_\lambda & \text{in} \ K_R\\
\frac{\pa h}{\pa \nu} = 0 & \text{on} \ \pa K_R\end{array}\right.$$
with $$\limsup_{R\to \infty} \frac{W(\nab h, \indic_{K_R})}{|K_R|} \le \min_{\mathcal{A}_{m_\lambda/2\pi}} W.$$
This is a screened minimizer of $W$ over $\mathcal{A}_{m_\lambda/2\pi}$. The fact that such an $h$ can be found is proven in \cite{ssgl} and follows the same outline as  the proof of  Proposition \ref{propscreening} in the case of $\mathcal{W}$, except that we do not need to first reduce to configurations with well-separated points (the screening can be accomplished for essentially generic configurations). It is however  complicated by the lack of lower bound on the energy density associated to $W( \cdot, \chi)$, which requires  to go through another ``mass displacement" to transform the energy density into one that is bounded below, just as we did for $f_\eps$ above.
Once these screened minimizers are pasted to almost cover $\omega_{0,\eps}'$, we obtain  a resulting set of points (the vortices)  and an associated vector field $\nab h$, which we extend by $0$ outside $\omega_{0,\eps}$, and to which we add $\nab h_{0,\eps}$.  After projecting this vector field onto gradients (which can only decrease its energy, as seen in Step 5 of the proof of Proposition \ref{proub}) and  blowing down, this defines the test induced magnetic field   $h_\eps$. 
Then there remains to build a corresponding $(u,A)$, which can be done as in the proof of the upper bound for Theorem \ref{th9.1}. Because we now need the energy to be optimal at next order, we need to be more precise near the vortex cores, and plug in exactly the optimal radial profiles for vortices which give the energy $\gamma$ per point.

\subsection{A statement of main result}

Modulo the technical difficulties mentioned above, for which we refer to \cite{compagnon,ssgl}, we can state the main  result which holds in the setting we chose to describe. The result can be written in complete $\Gamma$-convergence form. However for simplicity, we only state it as the analogue of Theorem \ref{thlbnext}, together with the consequences for minimizers as in Theorem \ref{th1}.

\begin{theo}[Next order behavior of the Ginzburg-Landau functional \cite{ssgl}]
Assume  $\om$ is convex and  \eqref{asshe} holds.  Let $(u_\eps, A_\eps)$ be such that $G_\eps(u_\eps, A_\eps)\le G_\eps^0+ C\he$, and let  $P_\eps$ be the push-forward of the normalized Lebesgue measure on $\omega_{0,\eps}$ by 
$$x\mapsto \frac{1}{\sqrt{\he}} \nab h_\eps (x+ \frac{\cdot}{\sqrt{\he}} ) $$
where $h_\eps $ is implicitly extended by  $0$ outside the domain $\om$. Then, up to extraction of a subsequence, we have  $P_\eps \to P$ in the weak sense of probabilities,  where $P$ is some probability measure concentrated on $\mathcal{A}_{m_\lambda/2\pi}$ and 
\be\label{mingno}
G_\eps(u_\eps, A_\eps)\ge G_\eps^0 + \he |\omega_{0,\eps}| \left(\int W(\nab h) \,d P(\nab h) + \frac{m_\lambda \gamma}{2\pi}\right)+o(\he).\ee
If in addition $(u_\eps, A_\eps)$ minimizes $G_\eps$ then $P$-a.e $\nab h$ minimizes $W$ over $\mathcal{A}_{m_\lambda/2\pi}$ and
$$G_\eps(u_\eps, A_\eps)= G_\eps^0 + \he |\omega_{0,\eps}|  \left(\min_{\mathcal{A}_{m_\lambda/2\pi} } W  + \frac{m_\lambda \gamma}{2\pi}\right)+o(\he).$$
\end{theo}

As announced, this provides a next order expansion of the Ginzburg-Landau energy, in a similar fashion as what we have seen for Coulomb gases with an error  $o(\he)$ which is also $o(1)$ per vortex. Moreover, it connects the question of minimizing Ginzburg-Landau to that of minimizing $W$.  If  Conjecture 1 (the conjecture that the triangular lattice minimizes $W$) was proven, then it would rigorously justify why the Abrikosov lattice appears in experiments on superconductors.  The result here says that almost all the blown up configurations ressemble minimizers of $W$ as $\eps \to 0$. Using the method of \cite{rns} (cf. end of Section \ref{sec63}) it should be possible to obtain a stronger result for all blow ups, i.e. a result of equidistribution of energy.
 
The assumption \eqref{asshe} has been made here for simplicity of presentation, the result in \cite{ssgl} is more general and works as long as the number of vortices diverges to infinity which happens as soon as $\he-\hci \gg \log \lep$, and as long as  $\he \ll \frac{1}{\eps^2}$.

\end{document}